\documentclass{article}

\PassOptionsToPackage{final}{showkeys} 

\usepackage[utf8]{inputenc}
\usepackage{anyfontsize}
\usepackage[english]{babel}
\usepackage{csquotes} 
\usepackage[T1]{fontenc}
\usepackage[theoremfont,largesc,p,osf]{newpxtext}

\usepackage{amsmath}
\usepackage{amsfonts}
\usepackage{amssymb}
\usepackage{amsthm}

\usepackage{dsfont} 
\usepackage{newpxmath}
\usepackage[cal=cm]{mathalfa}

\usepackage{bm} 

\usepackage{mathtools}
\mathtoolsset{centercolon}

\usepackage{xcolor}
\usepackage{centernot}
\usepackage{chngcntr}

\definecolor{darkblue} {rgb}{0,0,0.5}
\definecolor{darkgreen}{rgb}{0,0.5,0}
\usepackage{tikz}
\usetikzlibrary{cd}
\usetikzlibrary{arrows}
\usepackage{pst-node}
\usepackage{ytableau}
\usepackage{tikz-network}

\SetVertexStyle[Shape=circle,InnerSep=0pt,MinSize=1.5pt,FillColor=black]
\SetEdgeStyle[LineWidth=0.75pt]

\usepackage[useregional]{datetime2}
\usepackage[shortlabels]{enumitem}
\usepackage{titlesec}

\usepackage{fullpage}
\usepackage[font=small,labelfont=bf]{caption}
\usepackage{float}

\let\q\quad
\let\qq\quad

\usepackage[notref,notcite]{showkeys}

\usepackage[pdfusetitle]{hyperref}
\hypersetup{
  colorlinks = true,
  urlcolor  = blue,         
  linkcolor = darkblue,     
  citecolor = darkgreen,    
  filecolor = darkred,       
  anchorcolor = green
}
\usepackage{cleveref}
\usepackage{nameref}
\usepackage[
  backend=biber,
  style=alphabetic,
  doi=false,
  isbn=false,
  url=false,
  maxbibnames=6,
  maxcitenames=2,
  giveninits=true
]{biblatex}
\renewbibmacro{in:}{}

\allowdisplaybreaks

\newtheorem{theorem}{Theorem}\crefname{theorem}{Theorem}{Theorems}
\newtheorem{proposition}[theorem]{Proposition}\crefname{proposition}{Proposition}{Propositions}
\newtheorem{lemma}[theorem]{Lemma}\crefname{lemma}{Lemma}{Lemmas}
\newtheorem{corollary}[theorem]{Corollary}\crefname{corollary}{Corollary}{Corollaries}

\crefname{fact}{Fact}{Facts}

\newtheorem*{theorem*}{Theorem}
\newtheorem*{proposition*}{Proposition}
\newtheorem*{lemma*}{Lemma}
\newtheorem*{corollary*}{Corollary}
\newtheorem*{conjecture*}{Conjecture}

\theoremstyle{definition}
\newtheorem{definition}{Definition}\crefname{definition}{Definition}{Definitions}
\newtheorem*{definition*}{Definition}
\newtheorem{remark}{Remark}\crefname{remark}{Remark}{Remarks}
\newtheorem*{remark*}{Remark}
\newtheorem{example}{Example}\crefname{example}{Example}{Examples}
\newtheorem*{example*}{Example}
\newtheorem*{fact*}{Fact}
\newtheorem*{note*}{Note}
\newtheorem*{notation*}{Notation}
\newtheorem*{problem*}{Problem}

\renewcommand{\mid}{\;\middle\vert\;}

\newcommand{\mc}{\mathcal}

\newcommand{\id}{\textrm{id}}

\newcommand{\iso}{\cong}
\usepackage{xifthen}
\newcommand{\1}[1][]{
  \ifthenelse{\isempty{#1}}
  {\mathds{1}}
  {\mathds{1}_{\{#1\}}}
}

\makeatletter
\AtBeginDocument {%
          \def\resetMathstrut@{%
           \setbox\z@\hbox{\the\textfont\symoperators\char40}%
           \ht\Mathstrutbox@\ht\z@ \dp\Mathstrutbox@\dp\z@}%
}%
\makeatother
\newcommand*\autoop{\left(}
\newcommand*\autocp{\right)}
\newcommand*\autoob{\left[}
\newcommand*\autocb{\right]}
\DeclareRobustCommand*\{{\ifmmode \left\lbrace \else \textbraceleft \fi }
\DeclareRobustCommand*\}{\ifmmode \right\rbrace \else \textbraceright \fi }

\newif\ifinsidebracedgroup
\AtBeginDocument {%
   \let\originalbardelimiter\|
   \def\myleftbar{\ifinsidebracedgroup
                        \right\originalbardelimiter
                        \global\insidebracedgroupfalse
                  \else
                        \left\originalbardelimiter
                  \fi
                  \global\let\|\myrightbar}%
   \def\myrightbar{\ifnum\currentgrouptype=9
                        \left\originalbardelimiter
                        \global\insidebracedgrouptrue
                  \else
                        \right\originalbardelimiter
                  \fi
                 \global\let\|\myleftbar}%
   \let\|\myleftbar
   \mathcode`( 32768
   \mathcode`) 32768
   \mathcode`[ 32768
   \mathcode`] 32768
   \begingroup
       \lccode`\~`(
       \lowercase{%
   \endgroup
       \let~\autoop
   }\begingroup
       \lccode`\~`)
       \lowercase{%
   \endgroup
       \let~\autocp
   }\begingroup
       \lccode`\~`[
       \lowercase{%
   \endgroup
       \let~\autoob
   }\begingroup
       \lccode`\~`]
       \lowercase{%
   \endgroup
       \let~\autocb
   }}

\DeclarePairedDelimiter{\ceil}{\lceil}{\rceil}
\DeclarePairedDelimiter{\floor}{\lfloor}{\rfloor}
\DeclarePairedDelimiter{\abs}{\lvert}{\rvert}%
\DeclarePairedDelimiter{\norm}{\lVert}{\rVert}%
\newcommand{\cbnorm}[1]{{\left\vert\kern-0.25ex\left\vert\kern-0.25ex\left\vert #1 
    \right\vert\kern-0.25ex\right\vert\kern-0.25ex\right\vert}}

\makeatletter
\let\oldabs\abs
\def\abs{\@ifstar{\oldabs}{\oldabs*}}
\let\oldnorm\norm
\def\norm{\@ifstar{\oldnorm}{\oldnorm*}}
\let\oldceil\ceil
\def\ceil{\@ifstar{\oldceil}{\oldceil*}}
\let\oldfloor\floor
\def\floor{\@ifstar{\oldfloor}{\oldfloor*}}
\makeatother

\let\bar\overline

\newcommand{\bC}{\mathbb{C}}

\newcommand{\bM}{\mathbb{M}}
\newcommand{\bN}{\mathbb{N}}
\newcommand{\bR}{\mathbb{R}}
\newcommand{\bQ}{\mathbb{Q}}
\newcommand{\bZ}{\mathbb{Z}}

\newcommand{\cA}{\mathcal{A}}
\newcommand{\cB}{\mathcal{B}}
\newcommand{\cC}{\mathcal{C}}
\newcommand{\cD}{\mathcal{D}}
\newcommand{\cE}{\mathcal{E}}
\newcommand{\cF}{\mathcal{F}}

\newcommand{\cM}{\mathcal{M}}
\newcommand{\cN}{\mathcal{N}}

\newcommand{\cP}{\mathcal{P}}

\newcommand{\cR}{\mathcal{R}}
\newcommand{\cS}{\mathcal{S}}
\newcommand{\cT}{\mathcal{T}}
\newcommand{\cU}{\mathcal{U}}

\let\oldexists\exists
\renewcommand{\exists}{\oldexists\,}

\DeclareMathOperator{\conv}{conv}

\DeclareMathOperator{\rank}{rank}

\DeclareMathOperator{\supp}{supp}

\DeclareMathOperator{\Tr}{Tr}
\newcommand{\Span}{\mathrm{span}}

\newcommand{\embeds}{\hookrightarrow}

\usepackage{braket} 
\newcommand{\ketbra}[2]{|#1\rangle\langle #2|}
\renewcommand{\braket}[2]{\langle #1 | #2 \rangle}
\newcommand{\ip}[2]{\left\langle #1,\;#2\right\rangle} 
\newcommand{\pure}[1]{|{#1}\rangle\hspace{-1pt}\langle{#1}|} 

\numberwithin{equation}{section}
\numberwithin{theorem}{section}   
\makeatletter
\let\c@definition\c@theorem  
\let\c@remark\c@theorem          
\let\c@fact\c@theorem              
\makeatother
\numberwithin{example}{section}

\usepackage{fullpage}
\usepackage{parskip}
\usepackage{titlesec}
\usepackage{centernot}

\usepackage[normalem]{ulem}
\usepackage{cancel}
\usepackage{bm}

\newcommand{\cric}[1]{\stackrel{\circ}{#1}}

\newcommand{\SN}{\mathrm{SN}}
\newcommand{\SR}{\mathrm{SR}}

\newcommand{\cb}{\mathrm{cb}}
\newcommand{\cp}{\mathrm{cp}}
\newcommand{\op}{\mathrm{op}}
\newcommand{\sa}{\mathrm{sa}}
\newcommand{\HS}{\mathrm{HS}}

\renewcommand{\Re}{\mathrm{Re}}

\usepackage{subfiles} 

\definecolor{cool_green}{rgb}{0.0, 0.5, 0.0}

\definecolor{fuchsia}{HTML}{FF00FF}

\newcommand{\struck}{\bgroup\markoverwith{\textcolor{blue}{\rule[0.5ex]{2pt}{0.7pt}}}\ULon}  
\titleformat*{\paragraph}{\normalfont\normalsize\bfseries\boldmath}  
\definecolor{deepgreen}{rgb}{0.0, 0.35, 0.12}

\newcommand{\uprod}[1]{{(#1)}_{\cU}}

\newcommand{\kket}[1]{\ket{#1}\hspace{-0.1em}\rangle}
\newcommand{\bbra}[1]{\langle\hspace{-0.1em}\bra{#1}}

\newcommand{\POVM}{\mathrm{POVM}}
\newcommand{\stoc}{\mathrm{stoc}}
\newcommand{\minten}{\otimes_{\min}}

\newcommand{\CB}{\mathrm{CB}}

\DeclareMathOperator*{\esssup}{ess\,sup}

\usepackage{longtable,booktabs,array}

\newcommand{\defn}[1]{\textbf{#1}}        
\newcolumntype{S}{>{$}l<{$}}      

\usepackage{authblk}

\begin{document}

\title{Compatibility of quantum instruments}
\author[1]{Chloe Kim\thanks{\href{mailto:kim705@illinois.edu}{kim705@illinois.edu}}}
\author[4,5]{Yujie Zhang}
\author[1,3]{Eric Chitambar}
\author[2,3]{Marius Junge}
\author[2,3]{Felix Leditzky}
\affil[1]{Department of Electrical and Computer Engineering, University of Illinois Urbana-Champaign}
\affil[2]{Department of Mathematics, University of Illinois Urbana-Champaign}
\affil[3]{Illinois Quantum Information Science and Technology Center (IQUIST), University of Illinois Urbana-Champaign}
\affil[4]{Institute for Quantum Computing and Department of Physics \& Astronomy, University of Waterloo}
\affil[5]{Perimeter Institute for Theoreticla Physics}

\maketitle

\begin{abstract}
    Measurement compatibility asks whether a family of quantum measurements can be simulated by one parent measurement followed by classical post-processing. 
    We extend this measurement compatibility to the notion of \emph{$d$-compatibility} for measurements, channels, and instruments whose parent passes a $d$-dimensional quantum system and unlimited classical information to the post-processing. 
    We formulate $d$-compatibility as a common linear factorization of operator spaces and operator systems and show that $d$-compatible families form a closed set, even with infinitely many settings and a separable infinite-dimensional input.  
    To certify $d$-incompatibility, we define the compatibility functional as the smallest product of completely bounded norms over common factorizations through a $d$-dimensional register.  
    Using factorization theory of operator spaces, we show that measure-and-prepare channels defined in terms of mutually unbiased bases and pinching channels defined in terms of Heisenberg-Weyl operators are $d$-incompatible in certain regimes of $d$.
    We also link $d$-compatibility of multi-instruments to the preparability of Choi assemblages using channel-state duality, and use this connection to show that for bipartite states the (in)compatibility of instruments can be used to certify Schmidt numbers and the implementability of the instrument using one-way LOCC protocols.
    Finally, our compatibility functional also bounds the generalized robustness of compatibility, which we interpret through a memory-bounded nontransient preparation game. 
\end{abstract}

\tableofcontents

\section{Introduction}\label{sec:intro}

The question of which quantum measurements can be implemented simultaneously lies at the heart of quantum theory.
Already in the early works of Heisenberg and Bohr \cite{Heisenberg_1927,Bohr_1928}, it was recognized that certain pairs of observables cannot be jointly measured without disturbance.
In classical physics, by contrast, any two observables can in principle be measured simultaneously.
This feature is now understood through the notion of measurement compatibility, or joint measurability, which asks whether several positive-operator valued measures (POVMs) can be simulated by a single measurement followed by classical post-processing~\cite{Busch1986, Heinosaari_2016}.
Compatibility and simulability also have operational roles in quantum information.
An incompatible family can demonstrate Einstein--Podolsky--Rosen steering, and the correspondence between steering problems and joint-measurability problems is one-to-one \cite{Quintino2014, Uola2014, Uola_2018}.
Incompatibility likewise provides advantages in state discrimination \cite{Carmeli2019,SSC19,UKSYG19, Buscemi2020}, higher-order quantum processes \cite{sudarsanan2024higher}, is necessary for Bell nonlocality \cite{WolfPerezGarciaFernandez2009, Quintino2015, Hirsch2018, Plavala2025} and is closely related to generalized contextuality~\cite{TavakoliUola2020,Selby2023,ZhangSchmidYingSpekkens2026,Zhang2026quantifiers}.
More generally, measurement simulability characterizes the resources needed to reproduce a family of measurements, from classical processing of prescribed measurement devices \cite{Guerini2017,BLN25} to the use of quantum systems of bounded dimension \cite{Ioannou2022,Jones2023}.

The starting point of this paper is a reformulation of the notion of incompatible measurements into operator algebraic language.
Measurement compatibility is often formalized as the following.
Let
\begin{equation}
\cM = \{\{M_{a|x}\}_{a\in\mathsf{A}}\}_{x\in\mathsf{X}}
\end{equation}
be a family of POVMs acting on a Hilbert space $H$, where $x \in \mathsf{X}$ labels the measurement setting and $a \in \mathsf{A}$ labels the outcome.
Following the terminology introduced in~\cite{GourHeinosaariSpekkens2018} and used in~\cite{BLN25}, we call \defn{multi-meter} for a family a collection of measurements on a common input and ouput systems.
The multi-meter $\cM$ is said to be \defn{compatible} or \defn{jointly measurable} if all its POVMs can be reproduced by a single ``parent'' POVM through classical post-processing.
More precisely, $\cM$ is compatible if there exists a POVM $\{\Pi_\omega\}_{\omega\in\Omega}$ and conditional probability distributions $p(a|x,\omega)$ such that
\begin{equation}
    M_{a|x} = \sum_{\omega\in\Omega} p(a|x,\omega)\, \Pi_\omega. 
    \label{eq:intro-parent-sum}
\end{equation}
We can view each measurement as a completely positive trace-preserving linear map $M_x: S_1(H) \to \ell_1^{|\mathsf{A}|}$, given by $M_x(\rho)=(\Tr[M_{a|x}\rho])_{a\in\mathsf{A}}$, that maps density operators to probability distributions~\cite[Thm.~2.37]{Watrous2018}.
$S_1(H)$ is the space of trace class operators whose norm is given by $\Tr|X|$ and $\ell_1^{|\mathsf{A}|}$ is a vector space $\bC^{|\mathsf{A}|}$ with the absolute sum as its norm.
We identify each $p\in\ell_1^{|\mathsf{A}|}$ with the diagonal matrix $\sum_a p_a\ketbra{a}{a}$ in $S_1^{|\mathsf{A}|}$, so the trace of $p$ is $\sum_a p_a$.
Trace preservation of $M_x$ is then the normalization $\sum_a M_{a|x}=\1$.
On positive elements the trace coincides with the norm, so a positive map is trace-preserving if and only if it preserves the norm of positive elements.
$S_1(H)$ is the space of trace class operators whose norm is given by $\Tr|X|$ and $\ell_1^{|\mathsf{A}|}$ is a vector space $\bC^{|\mathsf{A}|}$ with the absolute sum as its norm.
Then the multi-meter $\{M_x\}_{x\in\mathsf{X}}$ is compatible if and only if there exists a single measurement map $\Pi : S_1(H) \to \ell_1$ and stochastic post-processing maps $u_x : \ell_1 \to \ell_1^{|\mathsf{A}|}$ such that
\begin{equation}
    M_x = u_x \circ \Pi, \qquad \forall x \in \mathsf{X}.
\end{equation}
The quantum side of the measurement is identical for every setting $x$, with differences arising solely from classical post-processing.
In other words, compatibility is the existence of a \emph{common linear factorization through a commutative space}.
For simplicity, the parent POVM above has countably many outcomes.
A general parent POVM has a measurable outcome space $\Omega$.
Then the sum in \eqref{eq:intro-parent-sum} becomes an integral against the POVM, and $\ell_1$ is replaced by $L_1(\Omega)$ as in \Cref{sec:compatible-inst}.
A continuous outcome space already appears in the example of \Cref{sec:motiv-closure}.

This factorization viewpoint leads to two closely related directions of study.
The first is to extend measurement compatibility to families of quantum channels, or \defn{multi-channels}, and families of instruments, or \defn{multi-instruments}, with a common input and output system.
In each case, compatibility is expressed through a common parent device followed by quantum post-processing.
We introduce a dimension parameter $d$ by allowing the common parent device to output a $d$-dimensional quantum system together with an unrestricted classical register.
For a single map, we define it \defn{$d$-factorizable} that factors through $d$-dimensional quantum system together with unrestricted classical side information, and call it \defn{$d$-compatible} for a family of maps that has a common factorization map into $d$-dimensional quantum systems.
For example, $1$-compatibility of a multi-meter corresponds to joint measurability.
\Cref{fig:dcompat-circuit} shows a $d$-compatible multi-instrument as a circuit.

The second direction is to develop a functional-analytic approach to these notions.
Spaces of quantum states and observables are Banach spaces, and operator spaces and operator systems add the matrix structure that encodes complete positivity and complete boundedness \cite{Effros2000,Pisier_2003,Paulsen_2003}.
Equipping the spaces in a factorization with these structures lets one study compatible families with the tools of Banach and operator space theory, without relying on an explicit physical realization.
Introductions to this interplay aimed at quantum information are \cite{AubrunSzarek2017,GuptaMandayamSunder2015}.
Such methods have been applied to Bell nonlocality and steering \cite{JPPVW2010prl,Yin2015,Loulidi2022} and to measurement incompatibility \cite{BluhmNechita2018,Jencova2018,Bluhm2022}, and \Cref{sec:motiv-op-spaces} surveys this line of work.

\begin{figure}[tb]
\centering
\begin{tikzpicture}[
  >=stealth,
  qwire/.style={thick},
  cwire/.style={thick, double, double distance=1.4pt},
  box/.style={draw, thick, fill=white, rounded corners=2pt, inner sep=3pt},
  lab/.style={font=\small},
]
  \draw[dashed, rounded corners=4pt] (1.2,-1.0) rectangle (5.7,1.3);
  \node[lab, anchor=north west] at (1.25,1.27) {$\Phi_x$};

  \node[box, minimum width=1.1cm, minimum height=1.5cm] (P) at (2.0,0) {$\theta$};

  \node[box, minimum width=1.1cm, minimum height=1.5cm] (N) at (4.9,0) {$\Psi_x$};

  \draw[qwire] (0.35,0) -- (P.west);
  \node[lab, left] at (0.35,0) {$\rho$};
  \node[lab, above] at (0.78,0) {$\bC^n$};

  \draw[qwire] (P.east |- 0,0.3) -- (N.west |- 0,0.3);
  \node[lab, above] at (3.45,0.3) {$\bC^d$};
  \draw[cwire] (P.east |- 0,-0.3) -- (N.west |- 0,-0.3);
  \node[lab, below] at (3.45,-0.3) {$\omega$};

  \fill (4.9,1.85) circle (2.2pt);
  \draw[cwire, ->] (4.9,1.85) -- (N.north);
  \node[lab, right] at (5.0,1.55) {$x$};

  \draw[qwire] (N.east |- 0,0.3) -- (6.6,0.3);
  \node[lab, right] at (6.6,0.3) {$\bC^m$};
  \draw[cwire] (N.east |- 0,-0.3) -- (6.6,-0.3);
  \node[lab, right] at (6.6,-0.3) {$a$};
\end{tikzpicture}
\caption[Circuit of a $d$-compatible multi-instrument]{Circuit of a $d$-compatible multi-instrument $\Phi_x = \Psi_x \circ \theta$.
Single lines carry quantum systems and double lines carry classical variables.}
\label{fig:dcompat-circuit}
\end{figure}
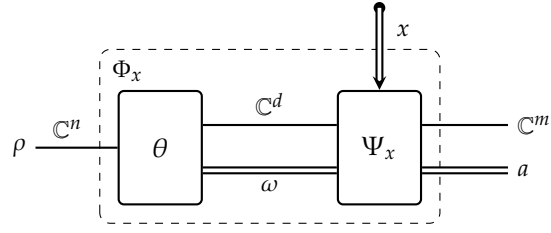

The main contributions of this paper are as follows.
\begin{itemize}
  \item We formulate $d$-compatibility for measurements, channels, and instruments as a linear factorization of Banach spaces, operator spaces, and operator systems, with a concrete example of genuine $d$-incompatible multi-channel whose element is $d$-factorizable.
  \item Using operator space-theoretic tools, we prove that the compatible multi-meters and the $d$-compatible multi-instruments $\{\Phi_x:S_1(H)\to S_1(K)\otimes^\wedge\ell_1^m\}_{x\in\mathsf X}$ form closed sets in operator norm (\Cref{thm:closurePOVM,thm:closure-instruments}). Here $H$ is separable and may be infinite-dimensional, the output $K$ and the outcome set are finite, the setting set $\mathsf X$ is arbitrary, and the classical register of the parent may be infinite.
  \item We derive $d$-compatibility criteria characterized by a compatibility functional $\gamma_{d,\cp}$ (\Cref{thm:cb-criterion,thm:cp-criterion}).
  \item We construct explicit incompatible multi-channels using the clock and shift operators, finite groups, and mutually unbiased bases.
  \item We establish a one-to-one correspondence between $d$-compatibility and a steering scenario with $d$-preparable Choi assemblage (\Cref{thm:choi-assemblage}), and present a one-way LOCC implementation of the assemblage.
  \item We identify the generalized robustness of $d$-compatibility and prove an upper bound in terms of $\gamma_{d,\cp}$ (\Cref{prop:gamma-bounds-Rg}), characterized operationally through a nontransient preparation game with quantum memory of at most dimension $d$.
\end{itemize}

\section{Motivation \& overview of the results}
\label{sec:motivation}

\subsection{Operator spaces as a tool for quantum information}
\label{sec:motiv-op-spaces}

One of the main objectives of this paper is to show how tools from the theory of operator spaces and operator systems can be used to characterize families of quantum operations and to study their properties.
The spaces of states and observables are Banach spaces, and operator spaces and operator systems add the matrix structure needed to study complete boundedness and complete positivity \cite{Effros2000,Pisier_2003,Paulsen_2003}.
These structures have already proved useful in quantum information theory.
Tensor norms of operator spaces yield unbounded violations of Bell inequalities \cite{JPPVW2010prl,JPPVW2010cmp}, and the approach extends to steering inequalities \cite{Yin2015}.
Compatibility, steering and Bell nonlocality have been described through free spectrahedra and tensor norms \cite{BluhmNechita2018,BluhmNechita2020,Bluhm2022,BluhmNechita2022cube,BluhmNechita2022tensor,Loulidi2022}, and through base-norm and order-unit spaces in general probabilistic theories \cite{Jencova2018,Jencova2022}.
The completely bounded norms of unital $k$-positive maps between operator systems are studied in \cite{AubrunDavidsonMullerHermesPaulsenRahaman2024}.
Entanglement annihilation, monogamy of entanglement, and factorization through Lorentz cones have been studied using only the cone structure, with general convex cones in place of the positive semidefinite cones \cite{AubrunMullerHermes2023,AubrunMullerHermesPlavala2025,AubrunLaPianaMullerHermes2026}.

In this paper, we view compatibility as a problem of linear factorization between Banach spaces, as in \Cref{sec:intro}.
We then equip these spaces with their operator space and operator system structures, so that results from the factorization theory of functional analysis apply.
The factorization viewpoint has a precedent in the work of Achenbach et al.\ \cite{Achenbach_2025}, who describe measurement compatibility, together with simulability, steering and Bell nonlocality, as a factorization of the multi-meter map through an intermediate system.
They formulate the factorization within general probabilistic theories and develop its theory in that setting.
We instead work with factorizations of linear maps between Banach spaces and operator spaces, for which an extensive theory is already available \cite{pisier1986factorization,JungeHab,Pisier1998}.

Our approach combines ultraproducts of operator spaces, Pisier's factorization results for $L_p$-spaces \cite[Sec.~8.b]{pisier1986factorization} and for noncommutative $L_p$-spaces \cite{Pisier1998}, and operator-space duality to study compatibility through its common factorization.
In the Schr\"odinger picture, measurements, channels and instruments are completely positive trace-preserving maps between spaces of states.
In the Heisenberg picture, their adjoints are unital completely positive maps between operator systems.
The factorization can be expressed in either picture.
In the Schr\"odinger picture, the classical register is represented by an $L_1$-space, the quantum register by $S_1^d$, and the combined register by their operator space projective tensor product $S_1^d\otimes^{\wedge}L_1(\Omega)$.
In the Heisenberg picture, the combined register is its dual, the von Neumann algebra $L_\infty(\Omega;\bM_d)$.
The choice of tensor product fixes the matrix norms used in the factorization, just as the tensor product for the joint output is part of the formulation of channel compatibility in \cite{Kuramochi2018}.
This description allows us to study compatible families through norms of linear maps and their dual characterizations, which lead to the closure theorems and the incompatibility criteria described in \Cref{sec:motiv-closure,sec:motiv-criteria}.

Operator space theory is not yet a standard tool in quantum information theory.
We have therefore tried to make the operator space approach as accessible as possible to a broad audience in the quantum information community.
\Cref{sec:prelim} is an expository section written for readers who are not familiar with Banach space and operator space theory.
It introduces Banach spaces, $C^*$-algebras, operator spaces and operator systems, and it describes measurements and instruments in this language in both the Schr\"odinger and the Heisenberg picture.
The tensor products used in the paper are collected in \Cref{app:tensor}.

\subsection{From measurements to instruments: $d$-compatible instruments}\label{sec:overview-compat}

A quantum instrument is the most general description of a quantum operation that produces both a classical outcome and a quantum output system.
Measurements are the instruments without a quantum output, and channels are the instruments without a classical outcome.
For a discrete outcome set, an instrument is a family of completely positive maps $\{\Phi_a:S_1(A)\to S_1(B)\}_a$ whose sum is trace preserving, where $S_1(A)$ is the space of trace-class operators on the input Hilbert space $A$.

The factorization formulation of measurement compatibility extends naturally to quantum channels and instruments.
We write $\Phi_x$ for the map that records both outputs of the $x$-th instrument.
A multi-instrument $\{\Phi_x\}_{x\in\mathsf X}$ is \defn{$d$-compatible} if, for every $x\in\mathsf X$,
\begin{equation}\label{eq:factorization}
  \Phi_x=\Psi_x\circ\theta,
  \qquad
  \theta:S_1(A)\longrightarrow S_1^d\otimes^{\wedge}L_1(\Omega),
\end{equation}
where $S_1^d\otimes^{\wedge}L_1(\Omega)$ describes the intermediate classical--quantum register.
The parent $\theta$ is a single instrument, and each $\Psi_x$ is a completely positive trace-preserving map from the register to the outputs of $\Phi_x$.
The precise definition is given in \Cref{sec:compatible-inst}.
Note that this definition of compatibility of multi-channel and multi-instruments differs from that of \cite{HeinosaariMiyadera2017,Haapasalo2021,Hsieh2022,Girard2021,Buscemi2023,Leppajarvi2024,Mitra2022}.
\Cref{sec:overview-prior} compares these notions in detail.

For discrete outcomes, the definition takes a more familiar form.
A multi-instrument $\{\Phi_{a|x}\}_{a,x}$ is $d$-compatible if there exist a single parent instrument $\{\theta_j:S_1(A)\to S_1^d\}_j$ and, for each setting $x$ and parent outcome $j$, an instrument $\{\Psi_{a|x,j}:S_1^d\to S_1(B)\}_a$ such that
\begin{equation}\label{eq:discrete-factorization}
  \Phi_{a|x}=\sum_j\Psi_{a|x,j}\circ\theta_j.
\end{equation}
For a single channel, the existence of such a factorization is equivalent to $d$-partially entanglement breaking \cite{Chruscinski2006}, as shown in \Cref{sec:kPEB}.

An operational motivation for $d$-compatibility comes from programmable quantum devices \cite{BuscemiChitambarZhou2020,JiChitambar2024}.
Suppose that the quantum input is received before the classical setting $x\in\mathsf X$ specifying which member $\Phi_x$ is to be implemented.
If the setting is not immediately available, the device must retain sufficient information about the quantum input until $x$ is supplied.
For incompatible measurements, this requires quantum memory.
For compatible measurements, a parent POVM can instead be performed immediately and only its classical outcome needs to be stored.

The factorization \eqref{eq:factorization} gives a quantitative version of this picture.
The parent instrument $\theta$ can be applied before $x$ is known.
Its output consists of a $d$-dimensional quantum register together with unrestricted classical side information.
Once $x$ is revealed, the corresponding map $\Psi_x$ completes the simulation.
Thus, $d$-compatibility means that the family can be programmed while retaining at most a $d$-dimensional quantum memory.
The smallest such $d$ therefore quantifies the quantum-memory cost of the family.

When $d=1$, only classical memory is available, and for measurements the definition reduces to joint measurability.
At the opposite extreme, if $d\geq\dim A$, the entire input system can be retained and the memory restriction becomes trivial.
The regime $1<d<\dim A$ therefore interpolates between classical-memory simulation and unrestricted quantum memory.

This quantum-memory cost should be distinguished from the complexity of a compatible measurement family, defined through the minimum number of outcomes of a parent POVM \cite{SkrzypczykHobanSainzLinden2020,Zhang2025costofsimulating}.
The latter quantifies the complexity of the classical simulation when no quantum memory is used.
In the present setting, the classical register is unrestricted, while $d$ quantifies how much quantum information must survive until the setting is supplied.
The quantum dimension needed to preserve the statistics of a prescribed set of measurements, with unrestricted classical side information, was studied in \cite{BluhmRauberWolf2018}, and $d$ plays the analogous role for reproducing the operations themselves, including their quantum outputs.

\subsection{Closure of $d$-compatibility}\label{sec:motiv-closure}

In a resource theory, the objects that can be obtained without the resource are called free.
In the resource theory of $d$-incompatibility, the free objects are the $d$-compatible families of quantum operations.
Convexity, closedness and compactness of the free set underlie witnesses, robustness measures and their dual characterizations in resource theories \cite{ChitambarGour2019,TR19}.
The set of $d$-compatible families is convex, because a classical random choice between implementations can be stored in the classical register.
Whether it is closed is the question answered by \Cref{thm:closurePOVM,thm:closure-instruments}, and in finite dimensions with finitely many settings and outcomes closedness makes it compact.

Both properties are used later in the paper.
In \Cref{sec:criterion}, the ultraproduct argument of the closure proof shows that the infimum defining each compatibility functional is attained.
Duality then turns this into the incompatibility criteria of \Cref{thm:cb-criterion,thm:cp-criterion}.
In \Cref{sec:robustness}, compactness and convexity make the generalized robustness vanish exactly on the $d$-compatible families and give its dual description through a preparation game (\Cref{prop:robustness-dual}).

Closedness also has a direct physical meaning.
If a family can be approximated arbitrarily well by $d$-compatible families, closedness guarantees that it can be implemented exactly with a quantum memory of dimension $d$.
The following example shows why this is not obvious.
For a qubit system, a noisy spin measurement along a direction $\hat n \in S^2$ with sharpness $r \in [0,1]$ is described by the dichotomic POVM
$\cM^r_{\hat n} = \{M^r_{a|\hat n}\}_{a=0,1}$, where
\begin{equation}\label{eq:noisy-spin-povm}
    M^r_{0|\hat n} = \tfrac{1}{2}(\1 + r\,\hat n\!\cdot\!\vec{\sigma}), \qquad
    M^r_{1|\hat n} = \tfrac{1}{2}(\1 - r\,\hat n\!\cdot\!\vec{\sigma}).
\end{equation}
Let $\cS_r = \{\cM^r_{\hat n} \}_{\hat n \in S^2}$ denote the multi-meter of all such noisy spin measurements.  
It is known that $\cS_r$ is compatible exactly when $r \le \tfrac{1}{2}$ \cite{Uola2014}.
For $r < \tfrac{1}{2}$ one can construct a \emph{finite-outcome} parent POVM achieving the simulation \cite[Thm.~2]{Zhang2025costofsimulating}.
Concretely, one may write
\begin{equation}\label{eq:S_r-fact}
    M^r_{a|\hat n} = \sum_{\omega \in \Omega} p(a|\hat n,\omega)\, F_\omega,
\end{equation}
for some POVM $\{F_\omega\}_{\omega \in \Omega}$, whose number of outcomes must grow at least on the order of $(\tfrac{1}{2}-r)^{-2/5}$ \cite{Zhang2025costofsimulating}.
Thus, as $r \to \tfrac{1}{2}$, simulating $\cS_r$ requires increasingly complex parent POVMs, and no finite-outcome parent POVM suffices at the critical value $r=\tfrac{1}{2}$.
However, compatibility does persist in the limit.  
The multi-meter $\cS_{1/2}$ can be simulated by a \emph{continuous-outcome} parent POVM, using the identity
\begin{align}\label{eq:compatibility-half}
    M^{1/2}_{a|\hat n}
    = \int_{S^2}\! d\hat\omega\;
      \Theta((-1)^a\, \hat n \!\cdot\! \hat\omega)\, \Pi_{\hat\omega},
\end{align}
where $\Theta$ denotes the Heaviside function and $\Pi_{S^2} \coloneqq \{ \Pi_{\hat\omega}\}_{\hat\omega \in S^2}$ with
\begin{equation}\label{eq:qubit-bloch-sphere-mother}
    \Pi_{\hat\omega}
        = \frac{1}{4\pi}(\1 + \hat\omega\!\cdot\!\vec{\sigma}).
\end{equation}
Moreover, the convergence $\cS_r \to \cS_{1/2}$ is uniform in the trace norm:
\begin{equation}
\norm{M^r_{a|\hat n} - M^{1/2}_{a|\hat n}}_1
    = \frac{1}{2} - r,
    \qquad \text{for all } a\in\{0,1\} \text{ and } \hat n \in S^2.
\end{equation}
All norms on the finite-dimensional space $\bM_2$ are equivalent up to constants, so the convergence is also uniform in the operator norm of the maps $M^r_{\hat n}\colon S_1(\bC^2)\to\ell_1^2$, $\rho\mapsto(\Tr[M^r_{a|\hat n}\rho])_a$, which is the norm used throughout.

This example shows that a sequence of compatible multi-meters can converge to another compatible multi-meter, even though the limiting simulation requires a continuous outcome space.
It raises a natural question.
\emph{If a sequence of compatible multi-meters converges in operator norm, is the limit again compatible?}
The same question arises for $d$-compatible multi-instruments, whose parents also carry a quantum register of dimension $d$.
In general no explicit parent is available for the limit, and the parents of the approximating families may become more and more complex.

The answer is yes for $d$-compatible multi-instruments, and hence for multi-meters and multi-channels.

\begin{theorem*}[Closure of $d$-compatibility, informal version of \Cref{thm:closurePOVM,thm:closure-instruments}]
Let $H$ and $K$ be separable Hilbert spaces with $\dim K < \infty$, let $\mathsf X$ be an arbitrary set of settings, and let $m$ be finite.
If a sequence of $d$-compatible multi-instruments $\{\Phi^{(n)}_x:S_1(H)\to S_1(K)\otimes^\wedge\ell_1^m\}_{x\in\mathsf X}$ converges in operator norm for every $x$, then the limit is $d$-compatible.
For $d=1$ and $K=\bC$, operator-norm limits of compatible multi-meters are compatible.
\end{theorem*}

We first present, as a warm-up, the compactness argument for finite-dimensional measurements with finitely many settings and outcomes (\Cref{subsec:finite-case}).
There the parents have a bounded number of outcomes and lie in a fixed compact set.
In general the parents act on different outcome spaces with no uniform bound, so a compactness argument needs a common space for them.
For measurements, post-processing equivalence classes provide such a space, which is compact in the weak topology defined by state discrimination probabilities \cite{Kuramochi2020}.
Our proofs instead use ultraproducts and operator space theory from functional analysis.
The ultraproduct combines the varying outcome spaces of the parents into one space.
The parents define a limit map on this space, and the quantum register of dimension $d$ is preserved.

\subsection{Compatibility criteria and incompatible MUB measure-and-prepare channels}\label{sec:motiv-criteria}

The factorization formulation also gives quantitative criteria for $d$-compatibility.
We state the criteria for multi-channels.
They extend to multi-instruments by regarding an instrument as a channel with block-diagonal output, which leaves the functionals and $d$-compatibility unchanged (\Cref{sec:criterion}).
For channels on finite-dimensional systems, we work with the dual maps $\Phi_x^*:\bM_m\to\bM_n$.
A factorization then reads $\Phi_x^*=\theta\circ\Psi_x$ with $\theta:L_\infty(\Omega;\bM_d)\to\bM_n$ and $\Psi_x:\bM_m\to L_\infty(\Omega;\bM_d)$.
We first consider completely bounded factorizations.
We define the c.b.\ \defn{compatibility functional}
\begin{equation}\label{eq:intro-gamma-cb}
  \gamma_{d,\cb}(\Phi^*) := \inf \norm{\theta}_{\cb} \sup_x \norm{\Psi_x}_{\cb},
\end{equation}
where the infimum is over all completely bounded factorizations $\Phi_x^* = \theta \circ \Psi_x$.
This follows the standard use of factorization norms in Banach space and operator space theory \cite{pisier1986factorization,JungeHab}\cite[Sec.~7.2]{Pisier1998}.
We similarly define the c.p.\ compatibility functional $\gamma_{d,\cp}(\Phi^*)$ through completely positive factorizations, so that $\gamma_{d,\cb}\le\gamma_{d,\cp}$.
The family $\Phi^*$ of completely bounded maps factors through $L_\infty(\Omega;\bM_d)$ by complete contractions if and only if $\gamma_{d,\cb}(\Phi^*)\le1$ (\Cref{prop:cb-unit-ball}).
A family $\Phi$ of channels is $d$-compatible precisely when $\gamma_{d,\cp}(\Phi^*)\le1$ (\Cref{prop:cp-gamma}).
Since each $\Phi_x^*$ is unital, $\gamma_{d,\cp}(\Phi^*)\ge\gamma_{d,\cb}(\Phi^*)\ge\norm{\Phi_x^*}_{\cb}=1$, so this condition is the same as $\gamma_{d,\cp}(\Phi^*)=1$.

These functionals are defined through all possible factorizations, so they are difficult to evaluate directly.
Operator-space duality expresses them as suprema over witnesses $T$ (\Cref{thm:cb-criterion,thm:cp-criterion}).
For the c.b.\ functional, a witness is a family $T=(T_x)_x$ of linear maps $T_x:S_1^m\to\bM_n$, viewed as one map $T:\ell_\infty(\mathsf X,S_1^m)\to\bM_n$.
It is normalized by its norm at matrix level $d$,
\begin{equation}
  \norm{T}_{(d)} := \norm{\id_{\bM_d}\otimes T : \bM_d(\ell_\infty(\mathsf X,S_1^m)) \to \bM_d(\bM_n)}.
\end{equation}

\begin{theorem*}[c.b.\ criterion, \Cref{thm:cb-criterion}]
Let $\Phi_x^*:\bM_m\to\bM_n$, $x\in\mathsf X$, be linear maps, with $\mathsf X$ finite.
Then
\begin{equation}
  \gamma_{d,\cb}(\Phi^*) = \sup_{\norm{T}_{(d)}\le1} \norm{\sum_x (T_x\otimes\Phi_x^*)(\chi_m)}_{\bM_n\minten\bM_n},
\end{equation}
where $\chi_m=\sum_{s,t=1}^m\ketbra{s}{t}\otimes\ketbra{s}{t}$.
In particular, a family of channels is not $d$-compatible if $\norm{\sum_x(T_x\otimes\Phi_x^*)(\chi_m)}>\norm{T}_{(d)}$ for some witness $T$.
\end{theorem*}

The value of a witness is the operator norm of a single matrix, and the dimension $d$ enters only through the normalization $\norm{T}_{(d)}$.
The c.p.\ criterion of \Cref{thm:cp-criterion} has the same form, with $*$-preserving witnesses, the operator norm $\norm{A_+}=\max\{\lambda_{\max}(A),0\}$ of the positive part in place of the operator norm, and a positive version of the normalization, characterizing $d$-compatibility exactly.
Although the c.b.\ criterion gives only a necessary condition for $d$-compatibility, since $\gamma_{d,\cb}\le\gamma_{d,\cp}$, a single witness with a violation is enough to certify incompatibility.

Deciding $d$-compatibility numerically is hard in general.
For multi-meters with $d=1$, compatibility of finitely many measurements with finitely many outcomes is decided by a single semidefinite program \cite{WolfPerezGarciaFernandez2009,CavalcantiSkrzypczyk2016,Skrzypczyk2020}.
For $d\ge2$, or for channels and instruments, no {single SDP} formulation is available.
By \Cref{thm:choi-assemblage}, the problem asks for a common source state of Schmidt number at most $d$.
Sets of states with bounded Schmidt number are not semidefinite representable in general, since already the separable states have no finite semidefinite description \cite{Fawzi2021}.
Deciding whether a single channel is entanglement breaking, which is the case $d=1$, is NP-hard \cite{Gharibian2010}.
Existing methods therefore rely on hierarchies and relaxations.
For multi-meters, $d$-compatibility is equivalent to $d$-preparability of a steering assemblage \cite{Jones2023}, for which a complete semidefinite hierarchy \cite{Gois2023} and further semidefinite relaxations \cite{Alessandro2025} are available.
For channels, see-saw methods give inner bounds and the polarization hierarchy \cite{PLAVALA202515} gives outer bounds, as in \Cref{ex:Cn}.
These methods give numerical bounds in fixed dimensions, whereas the witnesses above give analytic bounds for families of arbitrary dimension.

For multi-meters and $d=1$, the operator-space duality of \Cref{thm:cb-criterion,thm:cp-criterion} is related to the convex separation behind the incompatibility witnesses of \cite{Carmeli2019}.
For $d=1$ and dichotomic measurements, a similar functional is the compatibility norm of \cite[Sec.~3]{BluhmNechita2022tensor}, which is extended to general probabilistic theories in \cite[Sec.~9]{Bluhm2022}.
Here the intermediate space $L_\infty(\Omega;\bM_d)$ is noncommutative, so completely bounded norms replace Banach space norms.
The witnesses for channels in \cite{CarmeliJMP2019} concern compatibility through a joint channel, which is a different notion (\Cref{sec:overview-prior}).

Families of pinching channels onto mutually unbiased bases (MUBs) provide a concrete test of these criteria.
The pinching channel onto an orthonormal basis $\{\ket{e_r}\}_{r=1}^n$ of $\bC^n$ is $\Phi(\rho)=\sum_{r=1}^n\bra{e_r}\rho\ket{e_r}\,\ketbra{e_r}{e_r}$, which removes the off-diagonal entries of $\rho$ in this basis.
Two orthonormal bases $\{\ket{e_r}\}$ and $\{\ket{f_s}\}$ are mutually unbiased if $\abs{\braket{e_r}{f_s}}^2=1/n$ for all $r$ and $s$.
Pairs of mutually unbiased bases are a standard example of strongly incompatible measurements.
For qubit spin measurements subject to white noise, orthogonal directions need the most noise before they become jointly measurable \cite{Busch1986}.
For a pair of Fourier-conjugate bases in any finite dimension, the exact white-noise threshold for joint measurability is known \cite{Carmeli2012}.
Each channel can be implemented by measuring in its basis and preparing the corresponding basis state, so it is entanglement breaking and requires only classical memory when implemented separately \cite{HSR2003}.
A common implementation must retain enough information to reproduce any of these channels after the choice of basis is supplied, potentially requiring much larger quantum memory than individual channel requires.

For the pinching channels $\Phi=(\Phi_x)_{x=1}^k$ onto $k$ mutually unbiased bases of $\mathbb C^n$, we obtain
\begin{equation}\label{eq:intro-qip-mub-criterion}
  \gamma_{d,\mathrm{cb}}(\Phi^*)
  \geq \sqrt{\frac{k}{d}}\,(1-\frac1n).
\end{equation}
The family is consequently not $d$-compatible whenever $d<k(1-1/n)^2$ (\Cref{thm:dtwo}).
In particular, whenever a complete family of $n+1$ mutually unbiased bases exists, its pinching channels are not $d$-compatible for any $d\leq n-2$.
This gives an explicit lower bound on the memory required by the entire family.

For $k=2$ we take the computational and Fourier bases, which are the eigenbases of the clock and shift operators.
For this pair, an algebraic argument gives the stronger conclusion that every common implementation requires $d\geq n$ (\Cref{prop:mult-domain}).
Thus, even this pair of entanglement-breaking channels requires retaining the full input dimension, despite unrestricted classical information.
The argument extends to conditional expectations associated with an arbitrary finite group (\Cref{thm:finite-gp}).

\subsection{Choi state characterization, LOCC and steering}\label{sec:motiv-choi}

Channel-state duality represents a channel by its Choi state, the output of the channel applied to one half of a maximally entangled state \cite{Choi1975,Jamiolkowski1972}.
Questions about operations then become questions about bipartite states, to which the tools of entanglement theory apply.
Compatibility problems have such dual descriptions.
A multi-meter is incompatible exactly when it can demonstrate steering \cite{Quintino2014,Uola2014}, and every steering problem with an invertible marginal is equivalent to a joint-measurability problem \cite{Uola_2018}, with an extension to continuous-variable systems through channel--state duality \cite{Kiukas2017}.
Operationally, while $d$-compatibile resource asks for a common implementation of a family of maps, the dual question asks how much entanglement is needed to prepare the states that record the action of these maps on an entangled input.
The requirement of a common parent becomes a requirement of a common source state.
This allows us to relate a bound on quantum memory to a bound on entanglement dimension.

In finite dimensions, let $\phi_{RA}^+$ be the normalized maximally entangled state on a reference copy $R$ and the input system $A$.
The Choi state assemblage associated with a multi-instrument is given by
\begin{equation}\label{eq:intro-qip-choi}
  \sigma_{a|x}=(\id_{R}\otimes\Phi_{a|x})(\phi_{RA}^+).
\end{equation}
The assemblage is said to be \defn{$d$-preparable} if it can be obtained by applying local instruments to one part of a common bipartite source state of Schmidt number at most $d$.
Here the Schmidt number is the minimum, over pure-state decompositions of the source, of the largest Schmidt rank in the decomposition.
The source is identical for every setting $x$, with differences arising from the choice of local instrument.

We prove that a finite multi-instrument on finite-dimensional systems is $d$-compatible if and only if its Choi state assemblage is $d$-preparable (\Cref{thm:choi-assemblage}).
For a single channel, the characterization recovers the equivalence between $d$-partial entanglement breaking, Schmidt number \cite{Terhal2000} at most $d$ of the Choi state, and the existence of a Kraus decomposition whose operators all have rank at most $d$ \cite{Chruscinski2006}.
\Cref{cor:k-PEB} gives the corresponding statement for a single instrument.

The entanglement characterization is also useful when the input and output belong to different laboratories.
Suppose that Alice holds the quantum input, Bob receives the setting and must produce the output, and they share a bipartite state.
Then the family is $d$-compatible if and only if they can implement it using a shared state of Schmidt number at most $d$ and local operations with one-way classical communication from Alice to Bob, with Alice's operation independent of the setting (\Cref{sec:one-way-LOCC}).
Thus, the minimum quantum-memory dimension required for a common implementation equals the minimum Schmidt number required for this one-way LOCC realization.

The same characterization gives a way to certify entanglement dimension through steering with both classical and quantum outputs, extending the relation between measurement simulation and high-dimensional steering \cite{Designolle2021,Jones2023}.
Certifying the Schmidt number of a shared state is a basic task in the study of high-dimensional entanglement \cite{Friis2019}, and steering allows it with only one trusted party \cite{Designolle2021}.
Suppose that an untrusted party receives a setting $x$ and applies an instrument to its share of a bipartite source.
The party returns an outcome $a$ and a quantum system $B$ to a trusted verifier, who characterizes the resulting assemblage on the retained reference and the returned system.
For every fixed invertible reference marginal, we establish a one-to-one correspondence between multi-instruments and state assemblages (\Cref{prop:steering-equivalent}).
This correspondence identifies $d$-compatible multi-instruments with $d$-preparable assemblages.
For measurements, it recovers the steering--joint-measurability correspondence \cite{Uola_2018} and its extension to simulation with a bounded quantum dimension \cite{Jones2023}.
For channels and instruments, the correspondence provides a dimension parameter for steering with a quantum output.
This scenario differs from channel steering \cite{Piani2015,Zjawin2023}, in which the untrusted party returns only a classical outcome.
This parameter bounds the intermediate quantum register of a common implementation, whereas dimension-bounded steering \cite{Moroder2016} bounds the Hilbert space of the trusted party, and prepare-and-measure methods certify the dimension of a single channel \cite{Engineer2025}.

\subsection{Generalized robustness and preparation games}\label{sec:motiv-robustness}

Beyond deciding the amount of quantum memory required for compatibility, another way of quantifying incompatibility is through robustness \cite{Busch2013,Haapasalo15,Heinosaari2015}.
In the resource theory of $d$-incompatibility, whose free objects are the $d$-compatible families (\Cref{sec:motiv-closure}), a robustness measure records the least amount of noise, drawn from a chosen noise set, that must be mixed into a family before it becomes free \cite{ChitambarGour2019,Regula18}.
For measurements, the usual noise models are white noise and arbitrary noise, and mixing with jointly measurable families has also been considered \cite{DFK19}.
The resulting measures can be computed by semidefinite programming \cite{Uola2020,CavalcantiSkrzypczyk2016} or bounded by linear-programming relaxations \cite{Porto2026lp}.
Optimization-free nonlinear witnesses give another way to certify and quantify incompatibility \cite{Lee2026}.
Analogous noise thresholds have been studied for other simulation resources, for example the noise that a measurement tolerates before it becomes simulable by projective measurements \cite{Oszmaniec2017}.
In general probabilistic theories, simulation by restricted sets of observables also leads to quantifiers such as the minimal number of simulating observables \cite{Filippov2018}.
However, different noise models can order the same families differently and even single out different families as the most incompatible \cite{DFK19}.
For channels and instruments the noise must itself be a family of channels or instruments.
Several choices of such noise have been considered for channels \cite{Haapasalo15,Haapasalo2021,Hsieh2022} and for instruments \cite{MitraFarkas2023}.
Robustness-based quantifiers have been studied for channel incompatibility \cite{Minagawa2026} and, through thermalization, for instrument incompatibility \cite{Hsieh2024thermo}.

We therefore use the generalized robustness $\mathcal R_d$, which allows arbitrary noise.
It is the least $s\geq0$ such that $(\Phi+s\Xi)/(1+s)$ is $d$-compatible for some family $\Xi$ of channels, so it depends only on the set of $d$-compatible families and bounds the robustness against every specific noise model from below.
It is convex, vanishes exactly on the $d$-compatible families, and cannot increase under common channel pre-processing, setting-dependent channel post-processing, or classical randomization of the settings.
In convex resource theories, one plus the generalized robustness is moreover the largest advantage of a resource over all free objects in a discrimination task \cite{TR19,TRBLA19}.
For measurements this gives the state-discrimination characterization of incompatibility \cite{SSC19,UKSYG19}.
Its qualitative form was also obtained from incompatibility witnesses \cite{Carmeli2019}, and the analogous identity for single measurements is proved in \cite{Oszmaniec2019}.
A similar characterization holds for discrimination and exclusion games in which a state and a set of measurements are used together as resources \cite{Ducuara2025}.
Through the steering correspondence, the generalized robustness also relates to the robustness of steering \cite{PianiWatrous2015,CavalcantiSkrzypczyk2016}.

We extend this operational meaning to the quantum-memory constraint through a \emph{memory-bounded nontransient preparation game} (\Cref{subsec:robustness-duality}).
The referee sends half of a maximally entangled state to the player, announces the setting $x$ after a storage delay, and tests the returned quantum system jointly with the retained half.
Strategies that keep a $d$-dimensional quantum register and unrestricted classical information during the delay realize exactly the $d$-compatible families.
As in the nontransient guessing game of \cite{Ji2024}, the setting arrives only after a delay.
There the player returns its quantum system before the delay, so there is no quantum memory available and gives only a classical guess after it.
Here the player returns the quantum system after the setting is announced.
Writing $w(\Phi;M)$ for the winning probability in a game $M$ and $\mathcal F_d$ for the set of $d$-compatible families, we prove in \Cref{prop:robustness-dual} that
\begin{equation}\label{eq:intro-qip-game}
  1+\mathcal R_d(\Phi)
  =\sup_M
  \frac{w(\Phi;M)}{\displaystyle\sup_{\Sigma\in\mathcal F_d}w(\Sigma;M)},
\end{equation}
so one plus the generalized robustness is the largest advantage over all players whose quantum memory has dimension at most $d$.

For a family with input dimension $n$, we obtain the sharp bound $\mathcal R_d(\Phi)\leq n/d-1$ for $1\leq d\leq n$, attained by the identity channel (\Cref{cor:universal-bound}), and we compute $\mathcal R_d$ for partially depolarizing channels.
For the MUB pinching families of \Cref{sec:motiv-criteria}, the witness behind \eqref{eq:intro-qip-mub-criterion} is a game in which the referee sends a basis state and reveals its basis only after the delay.
The pinching family always wins, while every $d$-compatible strategy wins with probability at most $\sqrt{d/k}+1/n$, which gives $\mathcal R_d(\Phi)\geq(\sqrt k(1-\frac1n)-\sqrt d)/(\sqrt d+\sqrt k/n)$, complemented by the upper bound $\min\{k-1,\,n/d-1\}$ (\Cref{prop:mub-robustness}).
These bounds concern the channel family, whereas the white-noise robustness of the basis measurements is studied in \cite{DesignolleSkrzypczykFrowisBrunner2019}, and mutually unbiased measurements are not optimal for every steering task \cite{Bavaresco2017}.

Finally, the compatibility functional bounds the robustness from above, $1+\mathcal R_d(\Phi)\leq\gamma_{d,\cp}(\Phi^*)$ (\Cref{prop:gamma-bounds-Rg}), and it is monotone under the free operations of \Cref{lem:gamma-monotone}.
Its dual description is a signed version of the preparation game with positive and negative payoffs, and restricting to positive payoffs recovers $1+\mathcal R_d$ \eqref{eq:gap-positivity}, so the two quantities have related game interpretations but need not agree.

\subsection{Different notions of incompatibility in the literature}\label{sec:overview-prior}

Compatibility has been extended from measurements to channels and instruments in several inequivalent ways, which differ in what the common parent produces and in which post-processings are allowed.
\Cref{tab:notions} writes these notions, together with ours, as factorizations through a common parent, and the paragraphs below compare them with $d$-compatibility of this paper.
Comprehensive accounts of measurement incompatibility can be found in \cite{Heinosaari_2016,Guhne2023}.

\begin{table}[htbp]
\centering
\footnotesize
\renewcommand{\arraystretch}{1.15}
\begin{tabular}{|>{\raggedright\arraybackslash}p{0.30\linewidth}|p{0.65\linewidth}|}
\hline
\defn{Notion} & \defn{Factorization} \\
\hline
Compatible multi-meter \cite{BLN25}, joint measurability \cite{Heinosaari_2016,Guhne2023,Buscemi2020}
 & multi-meter $M_x:S_1(A)\to\ell_1^n$ with $M_x=u_x\circ\Pi$, common measurement $\Pi:S_1(A)\to L_1(\Omega)$, classical post-processing $u_x:L_1(\Omega)\to\ell_1^n$. \\
\hline
$n$-simulable multi-meter \cite{Ioannou2022,Jones2023}, $n$-preparable state assemblage \cite{Jones2023,Designolle2021}
 & $M_x=\Psi_x\circ\theta$, common instrument $\theta:S_1(A)\to S_1^n\otimes^\wedge\ell_1^k$ compressing the input into $n$ dimensions and recording an outcome, measurements $\Psi_x:S_1^n\otimes^\wedge\ell_1^k\to\ell_1^{|\mathsf A|}$.
 For an assemblage, this is the factorization of its steering-equivalent observables \cite{Jones2023}.
 Genuine $(n+1)$-dimensional steering \cite{Designolle2021} is the failure of $n$-preparability.\\
\hline
Programmable instrument device \cite{Ji2024}
 & multi-instrument $\Phi_x:S_1(A)\to S_1(B)\otimes^\wedge\ell_1^n$ with $\Phi_x=(\id_B\otimes S_x)\circ\theta$, common channel $\theta:S_1(A)\to S_1(B)\otimes^\wedge S_1(E)$ that releases the quantum output before the program arrives, multi-meter $S_x:S_1(E)\to\ell_1^n$ on the memory $E$ \cite[Prop.~1]{Ji2024}.
 The memory $E$ is unrestricted. \\
\hline
Classically compatible multi-instrument \cite{Heinosaari2014strong,Buscemi2023}, traditional compatibility \cite{Mitra2022}
 & multi-instrument $\Phi_x:S_1(A)\to S_1(B)\otimes^\wedge\ell_1^{n_x}$ with $\Phi_x=(\id_B\otimes u_x)\circ\theta$, common instrument $\theta:S_1(A)\to S_1(B)\otimes^\wedge\ell_1^k$, classical post-processing $u_x:\ell_1^k\to\ell_1^{n_x}$ of its outcome.
 All members share the quantum output $B$, which is not post-processed, so they all induce the same channel.\\
\hline
Simple (free) programmable instrument device \cite{Ji2024}
 & $\Phi_x=(\id_B\otimes u_x)\circ\theta$, common instrument $\theta:S_1(A)\to S_1(B)\otimes^\wedge\ell_1^k$, classical post-processing $u_x:\ell_1^k\to\ell_1^n$, equivalently a compatible memory multi-meter $S$ above \cite[Thm.~2]{Ji2024}. \\
\hline
Compatible multi-channel \cite{HeinosaariMiyadera2017,Haapasalo2021,Girard2021,Hsieh2022}, parallel compatibility of channels \cite{Buscemi2023}, broadcasting of states \cite{Barnum1996}
 & multi-channel $\Phi_x:S_1(A)\to S_1(B_x)$ with $\Phi_x=\Tr_{\hat x}\circ\theta$, common channel $\theta:S_1(A)\to S_1(\bigotimes_yB_y)$, post-processing the partial trace over the other outputs.
 Register $\bigotimes_yB_y$, all outputs produced in one run.
 Broadcasting of states is the case $\Phi_x=\id$, with the marginal conditions imposed only on the states to be broadcast. \\
\hline
Parallel compatible multi-instrument \cite{Heinosaari_2016,Mitra2022,Buscemi2023}, joint instrument \cite{Leppajarvi2024}
 & multi-instrument $\Phi_x:S_1(A)\to S_1(B_x)\otimes^\wedge\ell_1^{n_x}$ with $\Phi_x=(\Tr_{\hat x}\otimes m_x)\circ\theta$, common instrument $\theta:S_1(A)\to S_1(\bigotimes_yB_y)\otimes^\wedge\ell_1^{\prod_yn_y}$, $m_x$ the marginalization onto the $x$-th outcome.
 All outputs are produced in one run, as in broadcasting, and the members may have different quantum outputs $B_x$. \\
\hline
q-compatible multi-instrument \cite{Buscemi2023}
 & multi-instrument $\Phi_x:S_1(A)\to S_1(B_x)\otimes^\wedge\ell_1^{n_x}$ with $\Phi_x=\Psi_x\circ\theta$, common instrument $\theta:S_1(A)\to S_1^d\otimes^\wedge\ell_1^k$ for some $d$, and post-processings $\Psi_x:S_1^d\otimes^\wedge\ell_1^k\to S_1(B_x)\otimes^\wedge\ell_1^{n_x}$ of the form $\Psi_x(\rho\otimes e_\omega)=\sum_a\mu(a|\omega,x)\,\cD_{a,\omega,x}(\rho)\otimes e_a$ with channels $\cD_{a,\omega,x}:S_1^d\to S_1(B_x)$.
 Thus the outcome $a$ is a classical post-processing of $\omega$, and the register only undergoes a channel chosen by $a$, $\omega$ and $x$. \\
\hline
Projective simulation of an instrument \cite{Khandelwal2025}
 & single instrument $\Phi=\Psi\circ\theta$, parent $\theta(\rho)=\sum_{\lambda,a'}q_\lambda\,E_{a'|\lambda}\rho E_{a'|\lambda}\otimes e_{(a',\lambda)}$ a randomly chosen projective instrument on $A$, post-processing $\Psi$ an outcome-conditioned channel followed by classical relabeling. \\
\hline
$G$-joint measurability \cite{Lobo2026}, partial joint measurability \cite{Masini2024}
 & $M_x=\Psi_x\circ\theta$, common instrument $\theta:S_1(A)\to S_1(A')\otimes^\wedge\ell_1^k$ with an unrestricted quantum register $A'$, measurements $\Psi_x$ whose outcomes in a prescribed set $G_x$ are functions of the classical outcome alone.
 In partial joint measurability, each $G_x$ is either the set of all outcomes or empty. \\
\hline
$d$-factorizable channel (this paper), $d$-partially entanglement breaking channel \cite{Chruscinski2006}, channel of Schmidt number at most $d$ \cite{Huang2006}, $d$-compressible channel \cite{Sekatski2026}
 & single channel $\Phi=\Psi\circ\theta$, instrument $\theta:S_1(A)\to S_1^d\otimes^\wedge L_1(\Omega)$, channel $\Psi:S_1^d\otimes^\wedge L_1(\Omega)\to S_1(B)$, equivalently, $\Phi$ has a Kraus decomposition whose operators all have rank at most $d$, or its Choi state has Schmidt number at most $d$ (\Cref{cor:k-PEB}).
 Entanglement breaking channels \cite{HSR2003} are the case $d=1$. \\
\hline
$d$-compatible multi-channel and multi-instrument (this paper), $d$-simulable set of instruments \cite{Sekatski2026}
 & $\Phi_x=\Psi_x\circ\theta$, common instrument $\theta:S_1(A)\to S_1^d\otimes^\wedge L_1(\Omega)$, instruments $\Psi_x:S_1^d\otimes^\wedge L_1(\Omega)\to S_1(B)\otimes^\wedge L_1(\mu)$, channels for a multi-channel.
 Register $S_1^d$ with an unrestricted classical outcome and arbitrary post-processing depending on the classical label. \\
\hline
\end{tabular}
\caption{Compatibility and simulation notions from the literature, and the notions of this paper, written as factorizations through a common parent. Throughout, $A$ is the input system, $B$ the quantum output, $x$ the classical label of the member.}
\label{tab:notions}
\end{table}

\emph{Measurements.}
For measurements, $d$-compatibility with $d=1$ is joint measurability.
For general $d$ it coincides with the $n$-simulability of \cite{Ioannou2022,Jones2023} when $d=n$, where the compressing instrument serves as the parent $\theta$ and the measurements on the compressed system serve as the post-processings $\Psi_x$.
By \cite{Jones2023}, a multi-meter is $n$-simulable exactly when the assemblage it steers from the purification of a full-rank state, for example a maximally entangled state, is $n$-preparable \cite{Designolle2021}.
\Cref{prop:steering-equivalent} extends this correspondence to the $d$-preparable assemblages of \Cref{def:d-preparable}, which carry a quantum output.
Our definition thus extends this dimension parameter from measurements to channels and instruments.
Other notions measure the size of a compatible part rather than the memory of a common implementation.
The compatibility dimension of \cite{Loulidi2021} is the largest dimension of a subspace on which a multi-meter becomes compatible, while partial and $G$-joint measurability \cite{Masini2024,Lobo2026} require a classical simulation only for prescribed inputs or outcomes and leave the quantum register of the parent unrestricted.

\emph{Channels.}
For channels, compatibility is usually defined through a broadcasting parent \cite{HeinosaariMiyadera2017}.
Two channels are compatible when they are marginals of a common channel into the joint output $\bigotimes_yB_y$, so that all outputs are produced in one run.
Sufficiently noisy pairs are compatible, while the identity channel is not compatible with itself by the no-broadcasting theorem \cite{Barnum1996}.
Through the Choi isomorphism, this compatibility problem becomes a quantum marginal problem that can be treated by semidefinite programming \cite{Haapasalo2021,Hsieh2022}, joint channels can also be constructed algebraically \cite{Girard2021}, and for general outcome algebras the compatibility relation depends on the tensor product used for the joint output \cite{Kuramochi2018}.
In $d$-compatibility, by contrast, only one of the post-processings $\Psi_x$ is applied in a given run.
A broadcasting realization of $n$ channels with output space $B$ gives a $d$-compatible realization with $d=(\dim B)^n$, taking the broadcasting channel as the parent and the partial traces as post-processings.
The restriction in our definition therefore concerns the dimension of the intermediate quantum system rather than the simultaneous production of the outputs.

\emph{Instruments.}
For instruments, compatibility has to account for the post-measurement state as well as the outcome statistics, and several notions are in use.
Classical, or traditional, compatibility \cite{Heinosaari2014strong,Buscemi2023,Mitra2022} asks for a common instrument whose outcome is classically post-processed while its quantum output is passed on unchanged, so that all members induce the same channel.
This is $d$-compatibility with $d=\dim B$ in which the register is forwarded unchanged and only the classical label is post-processed, and it is also the factorization of the simple programmable instrument devices of \cite{Ji2024}.
The programmable instrument devices of \cite{Ji2024}, like the programmable measurement devices of \cite{Buscemi2020}, ask whether a family can be programmed without quantum memory, whereas $d$ quantifies how much quantum memory is needed.

Parallel compatibility \cite{Heinosaari_2016,Mitra2022,Buscemi2023} extends the broadcasting formulation to instruments through a joint instrument and the resulting incompatibility is quantified with robustness monotones in \cite{MitraFarkas2023}.
When all outputs coincide, a parallel realization is $d$-compatible with $d=(\dim B)^{|\mathsf X|}$, as in the channel case.
These notions are unified in \cite{Buscemi2023} through q-compatibility, which contains classical and parallel compatibility as special cases, makes every multi-channel compatible and reduces to joint measurability for multi-meters.
Classical compatibility, q-compatibility and the weaker notion of no-exclusivity are organized there into a strict hierarchy of resource theories of communication, and parallel compatibility is not part of this hierarchy.
In q-compatibility the post-processing may apply a channel to the register but chooses its outcome from the classical label alone, whereas $d$-compatibility fixes the dimension of the register and allows general instruments $\Psi_x$ that can also measure it.
A q-compatible realization with a $d$-dimensional register is therefore $d$-compatible, but not conversely, since incompatible measurements are $\dim A$-compatible without being q-compatible.
Finally, the simulation of a single instrument by a randomly chosen projective instrument followed by quantum post-processing is studied in \cite{Khandelwal2025}, which gives a necessary Schmidt-number condition on Choi operators that is also sufficient for qubit inputs.

During the preparation of this paper, we became aware of the related work \cite{Sekatski2026} on the bottleneck dimension of quantum operations.
In finite dimensions, its $d$-simulability of a set of instruments \cite[Def.~3]{Sekatski2026} is our $d$-compatibility, and its $d$-compressibility \cite[Def.~2]{Sekatski2026} of a single channel or instrument is our $d$-factorizability.

\subsection{Organization of the paper}\label{sec:overview-org}

The detailed development begins in \Cref{sec:prelim} with the operator space/operator system description of measurements, channels, and instruments.
\Cref{sec:compatible-inst} defines $d$-compatibility of channels and instruments, develops their formulation through a common factorization and presents examples of compatible and incompatible families.
The closure theorems are established in \Cref{sec:closure}, followed by the completely bounded and completely positive criteria and their applications to mutually unbiased bases in \Cref{sec:criterion}.
\Cref{sec:kPEB} establishes the Choi-state equivalence of the $d$-compatibility and its connection to steering and one-way LOCC characterizations, and the connection to partially entanglement breaking characterization of $d$-factorizablity.
Finally, \Cref{sec:robustness} develops generalized robustness, its preparation-game interpretation, and the resulting quantitative bounds.
The functional-analytic background and technical proofs are collected in the appendices.

\section{Notations}
\begingroup

Font conventions: \emph{regular} letters $H,K,A,B,\dots$ denote Hilbert spaces; \emph{sans-serif} $\mathsf X,\mathsf A$ denote classical spaces and index sets, $x\in\mathsf X$ a setting (the classical label of a member of a family) and $a\in\mathsf A$ an outcome, so that $\Phi_{a|x}$ is the component with outcome $a$ of the member $\Phi_x$; a sequence is indexed by a superscript, $\Phi^{(n)}_x$; \emph{calligraphic} $\cM,\cN,\cS,\cU$ denote operator algebras, operator spaces, instruments, and sets; and $\cB(\cdot)$ denotes the bounded operators.

\footnotesize
\renewcommand{\arraystretch}{1.3}
\begin{longtable}{S p{0.74\linewidth}}
\toprule
\multicolumn{1}{l}{\textbf{Symbol}} & \textbf{Meaning}\\
\midrule
\endhead

\cB(H) & bounded operators on $H$\\
V_{\sa} & self-adjoint part of a space with an involution \\
\bM_n & $n\times n$ complex matrices, $\bM_n=\cB(\bC^n)$\\
\CB(V,W) & completely bounded maps $V\to W$\\
J(\Phi) & Choi operator of a channel or instrument $\Phi$\\
\ket{\chi_A},\ \ket{\chi_n} & unnormalized maximally entangled vector $\sum_s\ket{s}\otimes\ket{s}$ on $H_A\otimes H_A$, on $\bC^n\otimes\bC^n$\\
\ket{\chi_\rho} & $(\rho^{1/2}\otimes\1)\ket{\chi_A}$, a purification of the state $\rho$ on $H_A$\\
\chi_n=\pure{\chi_n}\in\bM_{n^2} & unnormalized maximally entangled state on $\bC^n\otimes\bC^n$ and the Choi matrix of $\id_{\bM_n}$\\
\phi^{+}_n=\chi_n/n\in\bM_{n^2} & normalized maximally entangled state on $\bC^n\otimes\bC^n$\\
\kket{K},\ \bbra{K} & vectorization (double-ket) of an operator $K$ and its dual\\
S_1(H) & trace-class operators on $H$ (predual of $\cB(H)$)\\
S_p^d & Schatten-$p$ class on $\bC^d$ ($S_1^d$ trace class, $S_\infty^d$ compact/operator)\\
\ell_\infty(\mathsf A),\ \ell_1^m,\ \ell_2^{\mathsf A} & sequence spaces over a (finite) index set\\
\ell_1(\mathsf X,V),\ \ell_\infty(\mathsf X,V) & $\ell_1$- and $\ell_\infty$-direct sums of copies of an operator space $V$ indexed by $\mathsf X$\\
L_1(\Omega),\ L_\infty(\Omega) & classical function spaces on $(\Omega,\Sigma,\mu)$\\
L_\infty(\Omega;\bM_d) & $\bM_d$-valued essentially bounded functions, $L_\infty(\Omega)\otimes_{\min} \bM_d$; the classical--quantum register\\
L_1(\Omega;S_1) & $S_1$-valued integrable functions, $L_1(\Omega)\otimes^{\wedge}S_1$; predual of $L_\infty(\Omega;\cB)$\\
C(K) & continuous functions on a compact set $K$\\
\uprod{X_i} & Banach-space ultraproduct of $\{X_i\}$ along a non-principal ultrafilter $\cU$ ($\cU$-limit)\\
\norm{\cdot},\ \norm{\cdot}_\op & operator norm\\
\norm{\cdot}_2,\ \ip{\cdot}{\cdot}_{\HS} & Hilbert--Schmidt (Frobenius) norm $\norm{\alpha}_2=(\Tr\,\alpha^{*}\alpha)^{1/2}$ and inner product\\
\norm{\cdot}_{\cb} & completely bounded norm\\
\norm{\cdot}_{(k)}=\norm{\id_{\bM_k}\otimes\,\cdot\,} & $k$-th amplification norm (see \eqref{eq:d-norm}); for an operator $a$, $\norm{a}_p$ is the Schatten-$p$ norm\\
\SN,\ \SR & Schmidt number; Schmidt rank\\
\bottomrule
\end{longtable}
\renewcommand{\arraystretch}{1}
\endgroup

\section{Mathematical preliminaries}\label{sec:prelim}

In this section, we collect the basic mathematical notions used throughout the paper.
Our goal is to maintain a self-contained exposition for readers that are not familiar with Banach space and operator space theory.
We recall the necessary background on quantum measurements and quantum instruments, and how compatible measurements can be phrased as a linear factorization of Banach spaces.
Furthermore, we introduce operator system and operator space formulations of quantum instruments.
We refer interested readers to standard texts such as \cite{Effros2000}, \cite{Pisier_2003} or \cite{Paulsen_2003} for a more thorough treatment of these topics.
Readers familiar with the operator space-theoretic construction of quantum instruments may skip to \Cref{sec:compatible-inst}.

\subsection{Banach spaces, $C^*$-algebras, von Neumann algebras, operator systems, and operator spaces in quantum information theory}

We briefly review several classes of normed structures that play a central role in quantum information theory: Banach spaces, $C^*$-algebras, operator systems, and operator spaces.
Each of them arises naturally in the mathematical formulation of quantum states, measurements, channels, and instruments.

A \defn{Banach space} is a complex vector space $V$ equipped with a norm $\|\cdot\|$ such that $(V,\|\cdot\|)$ is complete.
For example, the space of trace-class operators $S_1(H)$ on a Hilbert space $H$ is a Banach space with respect to the trace norm, and it serves as the natural state space for quantum mechanics in the Schr\"odinger picture.
Similarly, classical probability distributions form Banach spaces such as $\ell_1(\mathsf A)$ or $L_1(\Omega)$.

We record the classical (commutative) sequence and function spaces used throughout, in dual pairs.
For a set $\mathsf A$, the \defn{$\ell_1$-space}
\begin{equation}
    \ell_1(\mathsf A)=\{\, x:\mathsf A\to\bC \ \mid \norm{x}_1:=\sum_{a\in\mathsf A}\abs{x(a)} < \infty \,\}
\end{equation}
of absolutely summable families has the (discrete) classical probability distributions on $\mathsf A$ as its positive, unit-norm elements.
Its dual is the \defn{$\ell_\infty$-space}
\begin{equation}
    \ell_\infty(\mathsf A)=\{\, x:\mathsf A\to\bC \ \mid \norm{x}_\infty := \sup_{a\in\mathsf A}\abs{x(a)}<\infty \,\}
\end{equation}
of bounded families, the commutative $C^*$-algebra of classical observables on $\mathsf A$.
For $\mathsf A=\{1,\dots,n\}$  we denote $\ell_1^n$ and $\ell_\infty^n$.
The continuous-outcome analogues live over a $\sigma$-finite measure space $(\Omega,\Sigma,\mu)$: the \defn{$L_1$-space}
\begin{equation}
    L_1(\Omega)=\{\, f:\Omega\to\bC \text{ measurable} \ \mid \norm{f}_1:=\int_\Omega\abs{f}\,d\mu<\infty \,\}
\end{equation}
of integrable functions (identified up to $\mu$-null sets), whose positive, unit-norm elements are the classical probability densities on $\Omega$, and its dual the \defn{$L_\infty$-space}
\begin{equation}
    L_\infty(\Omega)=\{\, f:\Omega\to\bC \text{ measurable} \ \mid \norm{f}_\infty:=\esssup_{\omega\in\Omega}\abs{f(\omega)}<\infty \,\}
\end{equation}
of essentially bounded functions, the commutative von Neumann algebra of classical observables.
The sequence spaces are the special case of counting measure on $\mathsf A$, $\ell_p(\mathsf A)=L_p(\mathsf A,\text{counting})$.
In each pair the $\ell_1/L_1$ side is a Schr\"odinger (state) space and the $\ell_\infty/L_\infty$ side its Heisenberg (observable) dual.

In infinite dimensions, passing between the Schr\"odinger and Heisenberg pictures requires a continuity condition, which we recall here \cite{Takesaki_1979,BuschQM2016}.
For a Banach space $X$, the \defn{weak$^*$ topology} on the dual $X^*$ is the weakest topology in which $f\mapsto f(x)$ is continuous for every $x\in X$.
For $\cB(H)=S_1(H)^*$ with the pairing $\langle a,\rho\rangle=\Tr(a\rho)$, a net of observables $a_\alpha$ converges weak$^*$ to $a$ exactly when $\Tr(a_\alpha\rho)\to\Tr(a\rho)$ for every density operator $\rho$, that is, when all expectation values converge.
A \defn{von Neumann algebra} is a unital $*$-subalgebra $\cM\subseteq\cB(H)$ that is weak$^*$ closed.
It is the dual of its \defn{predual} $\cM_*$, the space of weak$^*$ continuous linear functionals on $\cM$.
For example, $\cB(H)_*=S_1(H)$, $L_\infty(\Omega)_*=L_1(\Omega)$ and $L_\infty(\Omega;\bM_d)_*=L_1(\Omega;S_1^d)$.
A \defn{state} on $\cM$ is a positive linear functional $\varphi$ with $\varphi(\1)=1$; it assigns the expectation value $\varphi(a)$ to an observable $a$ and the probability $\varphi(E)$ to an effect $0\le E\le\1$.
A positive functional on a von Neumann algebra is \defn{normal} if it is weak$^*$ continuous, equivalently if it preserves suprema of bounded increasing nets of positive elements.
For a state $\varphi$ this is countable additivity of probabilities: if $E_1,E_2,\dots$ are the effects of mutually exclusive outcomes, then the event that one of them occurs has probability $\varphi(\sum_kE_k)=\sum_k\varphi(E_k)$.
Likewise, a positive map between von Neumann algebras is normal if it is weak$^*$ continuous, equivalently if composing it with any normal state gives a normal state, so that in the Schr\"odinger picture it sends physical states to physical states.
The normal states on $\cB(H)$ are exactly the functionals $a\mapsto\Tr(\rho a)$ given by density operators $\rho$, the positive trace-class operators with $\Tr\rho=1$ (density matrices when $\dim H<\infty$).
A unital completely positive map $\cB(K)\to\cB(H)$ is the adjoint of a channel $S_1(H)\to S_1(K)$ exactly when it is normal, and a POVM $L_\infty(\Omega)\to\cB(H)$ is normal exactly when it is weakly $\sigma$-additive.
In finite dimensions every state and every linear map is normal, so normality only matters for infinite-dimensional systems and continuous outcome spaces.

A \defn{$C^*$-algebra} is a Banach space $\cA$ equipped with a bilinear multiplication $\cA\times \cA \to \cA$ such that $\|ab\| \le \|a\|\|b\|$ and an involution $^*\colon \cA\to \cA$ satisfying the ``$C^*$-identity'' $\|a^* a\|=\|a\|^2$.
The algebraic structure allows one to model compositions of physical operations.
The fundamental example of a $C^*$-algebra in quantum information theory is the algebra $\cB(H)$ of bounded operators on a Hilbert space, as well as its commutative subalgebras that represent classical observables.
$C^*$-algebras provide the adequate abstract setting for quantum observables, with positive linear functionals representing states and completely positive maps representing physical operations.
Quantum channels in the Heisenberg picture correspond to unital completely positive maps between $C^*$-algebras.

The framework of Banach spaces and $C^*$-algebras alone is not sufficient for describing quantum channels and multipartite systems.
This leads to the notion of \emph{operator spaces}.

\begin{definition}[Operator space]
An \defn{(abstract) operator space} is a complex vector space $V$ equipped with a sequence of (matrix) norms $\{\|\cdot\|_{\bM_n(V)}\}_{n\ge 1}$ on the matrix spaces $\bM_n(V)$ satisfying \defn{Ruan's axioms}:
\begin{itemize}
\item [(R1)]  $\|x\oplus y\|_{\bM_{n+m}(V)} = \max\{\|x\|_{\bM_n(V)}, \|y\|_{\bM_m(V)}\},$
\item [(R2)] $\|\alpha x \beta\|_{\bM_m(V)} \le \|\alpha\|\,\|x\|_{\bM_n(V)}\,\|\beta\|$
\end{itemize}
for all $x\in \bM_n(V)$, $y\in \bM_m(V)$, and all matrices
$\alpha\in M_{m,n}(\mathbb{C})$, $\beta\in M_{n,m}(\mathbb{C})$ (\defn{scalar matrices}, i.e.\ with entries in $\mathbb{C}$).
\end{definition}

Ruan's representation theorem states that a vector space equipped with such a sequence of matrix norms is an operator space if and only if it admits a completely isometric embedding into $\cB(H)$ for some Hilbert space $H$~\cite[Theorem~2.3.5]{Effros2000}.
In particular, any closed subspace of $\cB(H)$ inherits a canonical operator space structure, which is called \defn{(concrete) operator space}.
The trace class $S_1(H)$ is normed by the trace norm rather than the operator norm of $\cB(H)$, so it does not inherit a structure in this way.
It carries the predual operator-space structure, as a subspace of the dual $\cB(H)^*$ with the matrix norms of $\bM_n(\cB(H)^*)=\CB(\cB(H),\bM_n)$, where completely bounded maps are defined below \cite{Effros2000,Pisier_2003}.
With this structure a bounded linear map $\Phi:S_1(H)\to S_1(K)$ and its adjoint $\Phi^*:\cB(K)\to\cB(H)$ have the same completely bounded norm, which is also known as the diamond norm of $\Phi$.

The appropriate morphisms between two operator spaces are the \defn{completely bounded} (CB) maps.
Given operator spaces $V$, $W$ and a linear map $T:V\to W$, its CB norm is defined by
\begin{equation}
\|T\|_{cb} := \sup_{n\ge 1} \|\id_n\otimes T : \bM_n(V)\to \bM_n(W)\|,
\end{equation}
and $T$ is called completely bounded if $\|T\|_{cb}<\infty$.
For a fixed $d \ge 1$ the $d$-th term of this supremum is the \defn{$d$-th amplification norm}
\begin{equation}\label{eq:d-norm}
\norm{T}_{(d)} := \norm{\id_d \otimes T : \bM_d(V) \to \bM_d(W)},
\end{equation}
so that $\norm{T} \le \norm{T}_{(d)} \le \norm{T}_{\cb}$ and $\norm{T}_{\cb} = \sup_d \norm{T}_{(d)}$.
It is the amplification norm of a witness, rather than its full c.b. norm, that enters the completely bounded criterion of \Cref{thm:cb-criterion}.
Note that the index of an amplification norm is written in parentheses, $\norm{T}_{(d)}$, while a bare numerical subscript on an operator, $\norm{a}_p$ for a matrix $a \in \bM_n$ or an element of $S_p(H)$, denotes the Schatten-$p$ norm $\norm{a}_p = (\Tr \abs{a}^p)^{1/p}$.
Quantum channels, viewed as linear maps between $S_1(H)$ spaces, are always completely bounded, and the CB norm is also referred to as the diamond norm in quantum information theory \cite{Kitaev1997,Watrous2018}.
Thus, operator spaces provide an adequate framework for the quantitative analysis of quantum channels.

In order to introduce a notion of positivity, one defines an \defn{abstract operator system} to be an abstract operator space $\cS$ equipped with a family of matrix-positive cones $\{\bM_n(\cS)_+\}_{n\ge 1}$ and a distinguished order unit, called \defn{Archimedean order unit} with certain axioms.
These axioms ensure a complete order isomorphism between an abstract operator system to a concrete operator system \cite{ChoiEffros1977,Paulsen_2003}, defined as follows.

\begin{definition}[Operator system]
A \defn{(concrete) operator system} is a unital self-adjoint subspace $\cS\subset \cB(H)$, i.e., a linear subspace satisfying:
\begin{enumerate}[(i)]
\item $x\in\cS \Rightarrow x^*\in\cS$,
\item $\1_{H}\in\cS$.
\end{enumerate}
\end{definition}

Unlike $C^*$-algebras, operator systems need not be closed under products, but they retain enough order structure to support a meaningful notion of positivity at all matrix levels.

In quantum information theory, operator systems naturally model spaces of observables, ranges of quantum measurements, and domains/codomains of quantum channels.
The Schr\"odinger (predual) and Heisenberg (dual) pictures are connected by Banach-space duality, but they sit in two different functional-analytic frameworks: the Schr\"odinger picture uses Banach spaces and operator spaces (e.g.\ $S_1(H)$, $\ell_1^n$, $L_1(\Omega)$), while the Heisenberg picture uses $C^*$-algebras and operator systems (e.g.\ $\cB(H)$, $\ell_\infty^n$, $L_\infty(\Omega)$).
Here $\ell_\infty^n=(\mathbb{C}^n,\norm{\cdot}_\infty)$ and $\ell_1^n=(\mathbb{C}^n,\norm{\cdot}_1)$ denote $\mathbb{C}^n$ with the supremum and the $\ell_1$ norm, respectively; they are mutually dual, $\ell_1^n=(\ell_\infty^n)^{*}$, and more generally $\ell_\infty(\mathsf A)$ (bounded functions) and $\ell_1(\mathsf A)$ (summable functions) on a set $\mathsf A$ form the analogous dual pair.
In the Heisenberg picture a quantum measurement, for example, is the unital completely positive (UCP) map
\begin{equation}
    M^* : \ell_\infty(\mathsf{A})\to \cB(H),
\end{equation}
where the commutative $C^*$-algebra $\ell_\infty(\mathsf{A})$ is seen as an operator system.
Channels and instruments are morphisms that are completely bounded in the operator-space sense (Schr\"odinger picture) and completely positive in the operator-system sense (Heisenberg picture).
\Cref{tab:schrodinger-vs-heisenberg} summarizes the parallel description of the basic objects used throughout this paper.

Throughout this paper, we use Banach space language to formulate measurement compatibility as a factorization problem of positive maps and operator system/operator space language to capture complete positivity and stability under tensoring with ancillas.

\begin{table}
\centering
\small
\renewcommand{\arraystretch}{1.4}
\begin{tabular}{|p{0.18\linewidth}|p{0.36\linewidth}|p{0.36\linewidth}|}
\hline
 & \defn{Schr\"odinger (predual) picture} & \defn{Heisenberg (dual) picture} \\
\hline
Carrier objects
 & quantum states (density operators) $\rho \in S_1(H)$, $\rho\ge0$, $\Tr\rho=1$; \newline classical distributions $p\in \ell_1^n$
 & observables / POVM elements in $\cB(H)$, with states as normal states $a\mapsto\Tr(\rho a)$; \newline classical outcomes in $\ell_\infty^n$ \\
\hline
Mathematical framework
 & Banach space / operator space \newline $(S_1(H),\ \ell_1^n,\ L_1(\Omega))$ \newline predual $\cM_*$ of a von Neumann algebra
 & $C^*$-algebra / operator system \newline $(\cB(H),\ \ell_\infty^n,\ L_\infty(\Omega))$ \newline von Neumann algebra $\cM=(\cM_*)^*$ \\
\hline
Quantum channel
 & $\Phi : S_1(H) \to S_1(K)$, \newline completely positive and trace-preserving
 & $\Phi^{*} : \cB(K) \to \cB(H)$, \newline unital completely positive \\
\hline
Measurement \mbox{(POVM)}
 & $M : S_1(H) \to \ell_1^n$, \newline positive and trace-preserving
 & $M^{*} : \ell_\infty^n \to \cB(H)$, UCP; \newline POVM elements $M_a := M^{*}(e_a)$ with $M_a\ge 0$, $\sum_a M_a = \1$ \\
\hline
Continuous-outcome POVM
 & $\Pi : S_1(H) \to L_1(\Omega)$, \newline positive and trace-preserving
 & $\Pi^{*} : L_\infty(\Omega) \to \cB(H)$, \newline weakly $\sigma$-additive (that is, normal) UCP \\
\hline
Quantum instrument
 & $\Phi : S_1(H) \to S_1(K)\otimes^{\wedge}\ell_1^m$, \newline completely positive and trace-preserving
 & $\Phi^{*} : \cB(K)\otimes_{\min}\ell_\infty^m \to \cB(H)$, \newline unital completely positive \\
\hline
Morphism property
 & completely bounded \newline (operator-space sense)
 & completely positive \newline (operator-system sense) \newline and normal (weak$^*$ continuous) \\
\hline
\end{tabular}
\caption{Parallel formulation of the basic objects of this paper in the Schr\"odinger (predual) and Heisenberg (dual) pictures.}
\label{tab:schrodinger-vs-heisenberg}
\end{table}

\subsection{Measurements in the Schr\"odinger and Heisenberg Pictures}

We recall two equivalent descriptions of quantum measurements and fix notation that will be used throughout the paper.  
Let $S_1(H)$ denote the Banach space of trace-class operators on a Hilbert space $H$, equipped with the trace norm $\|\cdot\|_1$.
In the Schr\"odinger (predual) picture, a measurement with $n$ outcomes is a linear map
\begin{equation}
    M : S_1(H) \to \ell_1^n
\end{equation}
that is positive and trace-preserving.  
A density operator $\rho \in S_1(H)$ satisfies $\rho \ge 0$ and $\|\rho\|_1=1$, while a classical probability distribution $p = (p_1,\dots,p_n)$ is a positive vector in $\ell_1^n$ with $\|p\|_1 = 1$.
Indeed, a measurement is a positive linear mapping from quantum states to classical probability distributions.

More commonly, measurements are described in the Heisenberg (dual) picture via \defn{positive operator-valued measures} (POVMs).
Given a measurement map $M:S_1(H)\to\ell_1^n$, its Banach space dual $M^*:\ell_\infty^n\to \cB(H)$ satisfies
\begin{equation}
    M_i \equiv M^*(e_i),
\end{equation}
where $\{e_i\}_{i=1}^n$ is the canonical basis of $\ell_\infty^n$.
The operators $\{M_i\}_{i=1}^n$ form a POVM, i.e. $M_i \ge 0$ for all $i$, and $\sum_i M_i = \1_{H}$.
Note that $\ell_\infty$ and $\cB(H)$ spaces are $C^*$-algebras with units $\1_{\ell_\infty} = (1,\cdots,1)$ and $\1_{\cB(H)} = I_{|H|}$.
Then, the dual map $M^*$ of a measurement is a unital completely positive map, and conversely every unital positive map defines a POVM via $M_a = M^*(e_a)$.
We will interchangeably write a POVM either as a collection of operators or as a unital completely positive map.
More generally, for a $\sigma$-finite measurable space $(\Omega,\Sigma,\mu)$, a
POVM is a weakly $\sigma$-additive unital completely positive map
\begin{equation}
    \Pi: L_\infty(\Sigma) \to \cB(H).
\end{equation}
In the discrete case, $L_\infty(\Sigma)$ is identified with
$\ell_\infty(\mathsf{A})$ and $\mu$ with the counting measure.

\subsection{Quantum instruments}\label{subsec:instruments}

A \defn{(discrete) quantum instrument} is a finite collection of completely positive trace-nonincreasing maps
\begin{equation}
    \Phi = \{ \Phi_a : S_1(A) \to S_1(B) \}_{a\in\mathsf{A}}
\end{equation}
such that their sum $\sum_{a\in\mathsf{A}} \Phi_a$ is trace-preserving; in other words, it is a quantum channel.  
Each $\Phi_a$ represents the post-measurement state update conditioned on the measurement outcome $a$.  

It is convenient to combine an instrument into a single completely positive map with classical side information.  
Let $\{\ketbra{a}{a}\}_{a\in\mathsf{A}}$ be the canonical basis of $\ell_1^{\mathsf{A}} \embeds S_1(\ell_2^{\mathsf{A}})$.
Define a CPTP map
\begin{equation}
    \Phi : S_1(A) \to S_1(B) \otimes_\pi \ell_1^{\mathsf{A}},
\end{equation}
by
\begin{equation}\label{eq:instrument_discrete}
    \Phi(\rho) = \sum_{a\in\mathsf{A}} \Phi_a(\rho) \otimes \ketbra{a}{a}.
\end{equation}
Given two Banach spaces $V$ and $W$, the Banach space projective tensor product $V\otimes_\pi W$ is the completion of the algebraic tensor product $V \otimes W$ with respect to the projective tensor norm
\begin{equation}
    \norm{u}_{\pi} := \inf \{ \sum_{i\leq n} \abs{\lambda_i}\norm{v_i}_V \norm{w_i}_W \mid u = \sum_{i\leq n} \lambda_i v_i\otimes w_i, \lambda_i \in \bC, n<\infty \}.
\end{equation}
Note that for simple tensor $v \otimes w$ we have $\norm{v \otimes w}_\pi = \norm{v}_V \cdot \norm{w}_W$.
This representation identifies discrete instruments with completely positive trace-preserving maps into $S_1(B)\otimes_\pi \ell_1^{\mathsf{A}}$, the classical--quantum Banach tensor product.

Let $(\Omega,\Sigma,\mu)$ be a $\sigma$-finite measure space.  
A \defn{continuous} quantum instrument~\cite{DaviesLewis1970,Ozawa_1984,Pellonpaa2013} is a family of completely positive maps
\begin{equation}
    \Phi = \{ \Phi_\omega : S_1(H) \to S_1(K) \}_{\omega\in\Omega}
\end{equation}
such that the integral
\begin{equation}\label{eq:instrument_integral}
    \Phi_E(\rho) := \int_E \Phi_\omega(\rho) \otimes \ketbra{\omega}{\omega} \, d\mu(\omega)
\end{equation}
is well-defined for every measurable set $E \in \Sigma$, and the induced map
\begin{equation}
    \Phi_\Omega : S_1(H) \to S_1(K) \otimes_\pi L_1(\Omega)
\end{equation}
is trace-preserving.\footnote{Note that only the induced map $\Phi_\Omega$ is required to be trace-preserving.
The maps $\Phi_\omega$ are densities with respect to $\mu$, and they need not be trace-nonincreasing individually.
For example, $\Phi_\omega(\rho)=2\omega\rho$ on $\Omega=[0,1]$ with Lebesgue measure defines an instrument.
This is the density form of the standard notion of an instrument \cite{DaviesLewis1970,Ozawa_1984}, which assigns to each $E\in\Sigma$ the completely positive map $\rho\mapsto\int_E\Phi_\omega(\rho)\,d\mu(\omega)$.
These maps are trace-nonincreasing, since for $\rho\geq0$ their traces are bounded by the trace of $\Phi_\Omega(\rho)$, which equals $\Tr\rho$.
The representation of general instruments by pointwise Kraus operators with respect to a scalar measure is given in \cite{Pellonpaa2013}, see also \cite{BuschQM2016}.}
Here $L_1(\Omega)$ is the classical $L_1$-space, and the codomain again carries the projective tensor norm.
We call this channel $\Phi_\Omega$ the channel induced by the instrument $\Phi$.
We can think of the set of quantum channels $N:S_1(H) \to S_1(K)$ as a subset of instruments with $\Omega$ being a singleton.
We define the \defn{dual instrument} to be the dual map $\Phi^* \equiv \Phi_{\Omega}^*:\cB(K) \otimes_{\min} L_\infty(\Omega) \to \cB(H)$, obtained by dualizing the induced channel $\Phi_\Omega:S_1(H)\to S_1(K)\otimes_\pi L_1(\Omega)$.
Its domain $\cB(K) \otimes_{\min} L_\infty(\Omega)$ is the observable-side (Heisenberg-picture) counterpart of the state space $S_1(K)\otimes_\pi L_1(\Omega)$: the completion of the algebraic tensor product $\cB(K) \otimes L_\infty(\Omega)$ in the minimal $C^*$-tensor norm
\begin{equation}
    \norm{u}_{\min} = \esssup_{\omega\in \Omega} \norm{u(\omega)}_{\cB(K)},
\end{equation}
realized concretely as the essentially bounded $\cB(K)$-valued functions $L_\infty(\Omega;\cB(K))$.
Here the output space $K$ is finite-dimensional.
Then the algebraic tensor product $\cB(K)\otimes L_\infty(\Omega)$ is already complete, it equals $L_\infty(\Omega;\cB(K))$, and it is the dual of $S_1(K)\otimes_\pi L_1(\Omega)$.
For infinite-dimensional $K$, the dual is instead the von Neumann tensor product $\cB(K)\bar\otimes L_\infty(\Omega)$, the weak$^*$ closure of $\cB(K)\otimes_{\min}L_\infty(\Omega)$, and the two differ when $L_\infty(\Omega)$ is also infinite-dimensional.

\subsection{Operator-space perspective on instruments}\label{subsec:os-instruments}

Quantum instruments may also be viewed within the framework of operator spaces.

Before introducing the operator-space description, we fix the roles of the tensor products in play, as several appear and it is easy to lose track of them.
They are organized along two axes: the \emph{picture}---Schr\"odinger (states, in the trace class $S_1$) versus Heisenberg (observables, in the bounded operators $\cB$)---and the \emph{category}---Banach spaces versus operator spaces, the latter additionally recording matrix norms.
\begin{itemize}
  \item $\otimes_\pi$ (Banach projective): the classical-quantum state space at the Banach level, $S_1(K)\otimes_\pi L_1(\Omega)$, in which an instrument is a trace-preserving completely positive map.
  \item $\otimes^\wedge$ (operator-space projective): its operator-space refinement $S_1(K)\otimes^\wedge L_1(\Omega)$, carrying the matrix norms.
  \item $\otimes_{\min}$ (minimal): the observable space $\cB(K)\otimes_{\min}L_\infty(\Omega)=L_\infty(\Omega;\cB(K))$ on which the dual instrument acts, dual to the state space above.
  On the $C^*$-algebras $\cB(K)$ and $L_\infty(\Omega)$ the minimal $C^*$-tensor product and the operator-space minimal tensor product (defined below) coincide, so we write $\otimes_{\min}$ for both.
\end{itemize}

The set $S_1(H)$ is a canonical operator space, and the c-q space $S_1(K)\otimes_\pi L_1(\Omega)$ embeds naturally into the operator-space projective tensor product
\begin{equation}
    S_1(K)\otimes^\wedge L_1(\Omega),
\end{equation}
where $L_1(\Omega)$ is identified with the predual $L_\infty(\Omega)_*$, a subspace of the dual $L_\infty(\Omega)^*$, and equipped with its induced operator-space structure.
The projective operator space tensor norm is defined by
\begin{equation}
    \|x\|_\wedge := \inf \{ \|\alpha\|_2 \,\|v\|_{\bM_n(V)} \,\|w\|_{\bM_n(W)} \,\|\beta\|_2 \mid x = \alpha (v\otimes w) \beta \},
\end{equation}
with the infimum over all $n$ and factorizations $x = \alpha (v\otimes w)\beta$ into scalar matrices $\alpha,\beta$, i.e. matrices with entries in $\bC$.
Here, $\norm{\cdot}_2$ is the Hilbert--Schmidt (Frobenius) norm $\norm{\alpha}_2 := (\Tr\,\alpha^{*}\alpha)^{1/2} = (\sum_{i,j}|\alpha_{ij}|^{2})^{1/2}$.
This formula gives the norm on $V\otimes W$, the first matrix level.
There $\alpha\in M_{1,n^2}(\bC)$ is a row and $\beta\in M_{n^2,1}(\bC)$ is a column, so their Hilbert--Schmidt and operator norms agree.
The matrix norm on $\bM_k(V\otimes W)$ is the same infimum over factorizations with $\alpha\in M_{k,n^2}(\bC)$ and $\beta\in M_{n^2,k}(\bC)$, with the operator norms $\norm{\alpha}\,\norm{\beta}$ in place of $\norm{\alpha}_2\,\norm{\beta}_2$, see \Cref{app:tensor}.

When the classical factor is the finite-dimensional $\ell_1^d$, the tensor product is the $\ell_1$-direct sum $\ell_1^d(S_1(K))=\bigoplus_{i=1}^{d}S_1(K)$, on which the Banach projective and operator-space projective tensor products coincide for \emph{any} $S_1(K)$, including an infinite-dimensional one.
The same holds whenever the quantum factor is finite-dimensional, where $S_1^k\otimes^\wedge L_1(\Omega)=L_1(\Omega;S_1^k)$ is a Bochner space, for any measure space $\Omega$.
For a Banach space $X$ and a measure space $(\Omega,\Sigma,\mu)$, the \defn{Bochner space} $L_1(\Omega;X)$ is the $X$-valued analogue of $L_1(\Omega)$: the ($\mu$-equivalence classes of) Bochner-integrable functions $f:\Omega\to X$, normed by $\norm{f}_{L_1(\Omega;X)}=\int_\Omega\norm{f(\omega)}_X\,d\mu(\omega)$ and reducing to $L_1(\Omega)$ when $X=\bC$.
More generally, the two projective norms coincide on $S_1(K)\otimes L_1(\Omega)$ for every Hilbert space $K$ and every $\sigma$-finite measure space $\Omega$.
The reason is that every bounded map $T:L_1(\Omega)\to\cB(K)$ is completely bounded with $\norm{T}_{\cb}=\norm{T}$.
Indeed, $T$ is the restriction to $L_1(\Omega)\subset L_\infty(\Omega)^*$ of the adjoint of the bounded map $S:S_1(K)\to L_\infty(\Omega)$ given by $S(\sigma)(f)=\Tr(T(f)\sigma)$.
A bounded map into the commutative $C^*$-algebra $L_\infty(\Omega)$ is completely bounded with the same norm, so $\norm{T}_{\cb}\le\norm{S^*}_{\cb}=\norm{S}_{\cb}=\norm{S}=\norm{T}$.
Hence $\CB(L_1(\Omega),\cB(K))=\cB(L_1(\Omega),\cB(K))$ isometrically, and the dualities of \Cref{app:tensor} give $(S_1(K)\otimes^\wedge L_1(\Omega))^*=(S_1(K)\otimes_\pi L_1(\Omega))^*$ isometrically.
Therefore the two norms agree on $S_1(K)\otimes L_1(\Omega)$, and both completions are the Bochner space $L_1(\Omega;S_1(K))$.
Thus $\otimes_\pi$ and $\otimes^\wedge$ give the same Banach space on the classical--quantum state space, and $\otimes^\wedge$ only adds the matrix norms on $\bM_n(S_1(K)\otimes^\wedge L_1(\Omega))$.
The operator-space structure is nevertheless needed.
Its matrix norms define complete boundedness of maps into and out of $S_1(K)\otimes^\wedge L_1(\Omega)$, and they enter the operator-space ultraproducts of \Cref{sec:closure}.

If $\rho = \int \rho_\omega \otimes \ketbra{\omega}{\omega}\, d\mu(\omega)$ is a c--q state, then $\rho \in S_1(K)\otimes^\wedge L_1(\Omega)$ and $\|\rho\|_\wedge \le 1$.
A quantum instrument is therefore a completely positive, trace-preserving map
\begin{equation}
    \Phi : S_1(H) \to S_1(K)\otimes^\wedge L_1(\Omega),
\end{equation}
and the dual instrument is a UCP map
\begin{equation}
    \Phi^* : \cB(K)\otimes_{\min} L_\infty(\Omega) \longrightarrow \cB(H).
\end{equation}

Here $\otimes_{\min}$ denotes the \defn{operator space minimal tensor product}.
Given operator spaces $V\subset \cB(H)$ and $W\subset \cB(K)$, their minimal tensor product $V\otimes_{\min} W$ is defined by representing the algebraic tensor product $V\otimes W$ inside $\cB(H)\otimes \cB(K)$ and completing it with respect to the norm
\begin{equation}
    \|u\|_{\min} := \sup \{ \|(\pi_V\otimes \pi_W)(u)\| \; \mid \; \pi_V:V\to \cB(H'),\ \pi_W:W\to \cB(K')\text{ are complete contractions} \}, \label{eq:min-norm}
\end{equation}
where $\pi_V$ and $\pi_W$ range over all completely contractive representations.
A linear map $\phi$ between operator spaces is \defn{completely contractive} if $\norm{\phi}_{\cb}\le1$, equivalently if $\id_n\otimes\phi$ is contractive for every $n$.
Equivalently, when $V$ and $W$ are concretely realized as operator subspaces of $\cB(H)$ and $\cB(K)$, one may identify
\begin{equation}
    V\otimes_{\min} W \subset \cB(H\otimes K)
\end{equation}
via the tensor product of operators and take the operator norm of the resulting operator on $H\otimes K$.
This yields the same norm as \eqref{eq:min-norm}.

The two occurrences of $\otimes_{\min}$ agree: for $C^*$-algebras such as $\cB(K)$ and $L_\infty(\Omega)$, the minimal $C^*$-tensor product is the operator-space minimal tensor product, so the symbol is unambiguous.
This space is also the operator-space dual of the state space when $K$ is finite-dimensional, $(S_1(K)\otimes^\wedge L_1(\Omega))^{*}=\cB(K)\otimes_{\min}L_\infty(\Omega)=L_\infty(\Omega;\cB(K))$, which is why the dual instrument is valued in the minimal tensor product.
The minimal tensor product is the canonical operator-space tensor product for the Heisenberg picture, and the dual instrument $\Phi^*$ naturally takes values in $\cB(H)$ when defined on the operator system $\cB(K)\otimes_{\min}L_\infty(\Omega)$.

\section{Measurement, channel, and instrument compatibility as a linear factorization of Banach spaces} \label{sec:compatible-inst}

This section develops a unified language for compatibility: measurement, channel, and instrument compatibility are each expressed as the factorization of a family of positive maps through a common intermediate Banach (operator) space.
We treat the three settings in turn, separating the channel case from the general instrument case.

\subsection{Measurement compatibility as a Banach-space factorization}

Let $\cM = \{M_{a|x}\}_{a\in\mathsf{A},\,x\in\mathsf{X}}$ be a family of POVMs acting on a Hilbert space $H$, and let
\begin{equation}
    M_x^* \colon \ell_\infty(\mathsf{A}) \longrightarrow \cB(H)
\end{equation}
denote the associated unital positive maps in the Heisenberg picture.  
The family $\cM$ is said to be \defn{compatible} (or \defn{jointly measurable}) if there exists a measurable space $(\Omega,\Sigma)$, a POVM  $\Pi^* : L_\infty(\Omega) \longrightarrow \cB(H)$, and stochastic maps $v_x \colon \ell_\infty(\mathsf{A}) \to L_\infty(\Omega)$ such that
\begin{equation}\label{eq:factorization_heisenberg}
    M_x^* = \Pi^* \circ v_x,
    \qquad \forall\, x \in \mathsf{X} .
\end{equation}
Thus, every measurement $M_x$ is obtained from a common parent measurement $\Pi$ by classical postprocessing.

Equivalently, in the Schr\"odinger (predual) picture, the associated measurement maps $M_x : S_1(H) \longrightarrow \ell_1(A)$ factor through a common channel $\Pi : S_1(H) \longrightarrow L_1(\Omega)$ as
\begin{equation}\label{eq:factorization_schrodinger}
    M_x = u_x \circ \Pi,
\end{equation}
where $u_x : L_1(\Omega) \to \ell_1(A)$ are stochastic postprocessing maps.
In other words, the diagram
\begin{equation}
  \begin{tikzcd}[column sep=tiny,row sep=large]
     S_1(H) \ar[rr,"M_x"] \ar[dr,swap,"\Pi"]
      & & \ell_1(A) \\
    & L_1(\Omega) \ar[ru,swap,"u_x"] &
  \end{tikzcd}
\end{equation}
commutes for every $x\in\mathsf{X}$.
Hence, measurement compatibility is precisely a \emph{factorization property of positive Banach-space maps}.

\subsection{Channel compatibility}

A quantum channel is the special case of an instrument with a trivial classical register, obtained by discarding the $L_1(\mu)$ factor.
Specializing the factorization picture of the previous subsection to this case yields the notion of channel compatibility, which we record before treating general instruments.

\begin{definition}[$d$-factorizable channel]
A quantum channel $\Phi : S_1(A) \to S_1(B)$ is \defn{$d$-factorizable} if there exist an instrument $\Psi : S_1(A) \to S_1^d \otimes^{\wedge} L_1(\Omega)$ and a completely positive, trace-preserving map $Q : S_1^d \otimes^{\wedge} L_1(\Omega) \to S_1(B)$ for some measure space $(\Omega, \Sigma)$ such that $\Phi = Q \circ \Psi$.
\end{definition}

\begin{definition}[jointly $d$-factorizable / $d$-compatible channels]
A family of quantum channels $\{\Phi^{(i)} : S_1(A) \to S_1(B)\}_{i\in \mathsf{I}}$ is \defn{jointly $d$-factorizable}, or \defn{$d$-compatible}, if there exists a single channel $\Psi : S_1(A) \to S_1^d \otimes^{\wedge} L_1(\Omega)$ such that for each $i \in \mathsf{I}$, $\Phi^{(i)} = Q^{(i)} \circ \Psi$ for some completely positive, trace-preserving map $Q^{(i)} : S_1^d\otimes^{\wedge}L_1(\Omega) \to S_1(B)$.
Equivalently, the following diagram commutes for all $i \in \mathsf{I}$.
\begin{equation}
  \begin{tikzcd}[column sep=tiny,row sep=large]
    S_1(A) \ar[rr,"\Phi^{(i)}"] \ar[dr,swap,"\Psi"]
      & & S_1(B) \\
    & S_1^d \otimes^{\wedge} L_1(\Omega) \ar[ru,swap,"Q^{(i)}"] &
  \end{tikzcd}
\end{equation}
\end{definition}

\begin{example}[$d$-factorizability of partial depolarizing channels]\label{ex:depolarizing}
Let $\dim A=\dim B=D$, and consider the depolarizing channel
\begin{equation}
    \Phi_p(\rho)\coloneqq p\rho+(1-p)\Tr(\rho)\frac{\1_D}{D},\qquad0\le p\le1.
\end{equation}
Its normalized Choi state is the isotropic state
\begin{equation}
    \rho_{\Phi_p}\coloneqq \frac{1}{D}J(\Phi_p)= p\pure{\phi_D^+}+(1-p)\frac{\1_{D^2}}{D^2},
\end{equation}
where $\ket{\phi_D^+}\coloneqq \frac{1}{\sqrt D}\sum_{i=1}^D\ket{i}\otimes\ket{i}$. 

\Cref{cor:k-PEB} shows that a channel $\Phi$ is $d$-factorizable if and only if the Schmidt number of its Choi operator is at most $d$, namely,
\begin{equation}
    \Phi\text{ is $d$-factorizable}\quad\Longleftrightarrow\quad \SN(J(\Phi))\le d.
\end{equation}
Combining with the fact in Ref.~\cite{Terhal2000} that Isotropic states satisfy
\begin{equation}
    \SN(\rho_{\Phi_p})\le d \quad\Longleftrightarrow\quad\bra{\phi_D^+}\rho_{\Phi_p}\ket{\phi_D^+}\le\frac{d}{D},
\end{equation}
we therefore obtain the sharp threshold
\begin{equation}
    \Phi_p\text{ is $d$-factorizable}
    \quad\Longleftrightarrow\quad
    p\le\frac{dD-1}{D^2-1}.
    \label{eq:partial-depol}
\end{equation}
For $d=1$, \eqref{eq:partial-depol} recovers the usual entanglement-breaking threshold.
\end{example}

It is immediate that a multi-channel is $d$-incompatible if any of the channels is not $d$-factorizable. Such incompatibility is inherited from an individual channel and is not genuinely collective. We now give a more nontrivial example in which every channel is individually $d$-factorizable, yet the family is $d$-incompatible.

\begin{example}[A genuine $d$-incompatible multi-channel with $d$-factorizable channels]\label{ex:Cn}
Let $H=\bC^n$ with basis $\{\ket{i}\}_{i=1}^n$, and let $\mathsf E_n\coloneqq \{\{i,j\}:1\le i<j\le n\}$ and $P_{\{i,j\}}\coloneqq \ketbra{i}{i}+\ketbra{j}{j}$. For $e\in\mathsf E_n$ and $0\le p\le1$, define
\begin{equation}
    \Phi_e^{(p)}(\rho)\coloneqq pP_e\rho P_e+\Tr[(\1_n-pP_e)\rho]\frac{P_e}{2},
    \label{eq:2d-perserving channel}
\end{equation}
and the family $\cC_n^{(p)}\coloneqq \{\Phi_e^{(p)}\}_{e\in\mathsf E_n}$. Since the output of $\Phi_e^{(p)}$ is supported on the two-dimensional space $P_eH$, every individual channel is $2$-factorizable.

At $p=1$, however, the exact compatibility dimension is
\begin{equation}
    \cC_n^{(1)}\text{ is $d$-compatible}
    \quad\Longleftrightarrow\quad
    d\ge n.
    \label{eq:exact-boundary}
\end{equation}
Indeed, if $d\ge n$, the multi-channel factorizes through the common $n$-dimensional Hilbert space $H$.

Conversely, suppose that $\cC_n^{(1)}$ has a common $d$-dimensional parent $\Psi$, with post-processings $Q_e$ satisfying
$\Phi_e^{(1)}=Q_e\circ\Psi$. On the subspace $P_eH$, one has
$\Phi_e^{(1)}(\rho)=\rho$, so $Q_e$ recovers $\Psi$ on every two-dimensional coordinate subspace. If $\{K_\alpha\}_\alpha$ are Kraus operators of $\Psi$, the Knill--Laflamme conditions give
\begin{equation}
    P_eK_\alpha^*K_\beta P_e = c_{\alpha\beta}^{(e)}P_e\quad\forall e\in\mathsf E_n\Longrightarrow   K_\alpha^*K_\beta=c_{\alpha\beta}\1_n\qquad\forall \alpha,\beta,
\end{equation}
since varying $e$ shows that every $K_\alpha^*K_\beta$ has vanishing off-diagonal entries and equal diagonal entries.

Thus $\Psi$ is reversible on the entire input space, making
$\id_{n}$ $d$-factorizable. By \Cref{cor:k-PEB}, this would imply
\begin{equation}
    n=\SN(J(\id_{n}))\le d.
\end{equation}
This proves \eqref{eq:exact-boundary}.

For the noisy family, define
\begin{equation}
    \cD_p^{(n)}(\rho)
    \coloneqq 
    p\rho+(1-p)\Tr(\rho)\frac{\1_n}{n},
    \qquad
    \Theta_e(\rho)
    \coloneqq 
    P_e\rho P_e+\Tr[(\1_n-P_e)\rho]\frac{P_e}{2}.
\end{equation}
A direct calculation gives $\Phi_e^{(p)}=\Theta_e\circ\cD_p^{(n)}$. Since
\begin{equation}
    \cD_p^{(n)}\text{ is $d$-factorizable}
    \quad\Longleftrightarrow\quad
    p\le q_{d,n}\coloneqq \frac{dn-1}{n^2-1},
\end{equation}
the family $\cC_n^{(p)}$ is $d$-compatible throughout this range.

Defining its compatibility threshold by
\begin{equation}
    p_{d,n}^{\star}
    \coloneqq 
    \sup\{p\in[0,1]:\cC_n^{(p)}
    \text{ is $d$-compatible}\},
\end{equation}
we therefore have $1>p_{d,n}^{\star}\ge q_{d,n}$. 

This bound can be further tightened. For $(n,d)=(3,2)$, the analytic bound is only $q_{2,3}=5/8$. A Choi-matrix see-saw SDP gives a compatible realization at $p\approx0.8854$, while a polarization SDP outer relaxation inspired by Ref.~\cite{PLAVALA202515} gives the upper bound $p\approx 0.8938$
\begin{align}
0.8854 \lesssim p_{2,3}^{\star}  \lesssim 0.8938 .
 \label{eq:numerical-interval}
\end{align}
\end{example}
The lower bound is obtained by a Choi-matrix seesaw search using a parent instrument with $24$ classical outcomes. Increasing this number yielded no improved feasible model within the tested initialization and iteration time;
The upper bound is obtained from a finite,
polarization-inspired SDP relaxation, together with PPT
constraints with respect to the relevant Choi matrix~\cite{Kim-Zhang2026}.
Neither bound is known to be tight. Thus, further optimization may improve both the lower and upper bounds.

\subsection{Instrument compatibility}

Just as measurement compatibility admits a Banach-space factorization formulation, instrument factorizability and compatibility naturally appear as factorization problems in operator spaces; we begin with the factorizability of a single instrument.

Let $(\Omega,\Sigma,\mu)$ be a measure space.  
Recall that a quantum instrument can be represented as a completely positive, trace-preserving map
\begin{equation}
    T : S_1(A) \longrightarrow S_1(B) \otimes^\wedge L_1(\mu),
\end{equation}
where $\otimes^\wedge$ denotes the operator-space projective tensor product.

\begin{definition} A quantum instrument  $T : S_1(A) \longrightarrow S_1(B) \otimes^\wedge L_1(\mu)$ is called \defn{$d$-factorizable} if there exist a quantum instrument  $N : S_1(A) \to S_1^d \otimes^\wedge L_1(\nu)$ for some measure space $(\Omega',\Sigma',\nu)$, and a completely positive, trace-preserving map $P : S_1^d \otimes^\wedge L_1(\nu) \to S_1(B) \otimes^\wedge L_1(\mu)$
such that $T = P \circ N$.
\end{definition}

Now we define compatible instruments.

\begin{definition} A family of quantum instruments $\{ T^{(i)} : S_1(A) \to S_1(B) \otimes^\wedge L_1(\mu) \}_{i\in \mathsf{I}}$ is called \defn{jointly $d$-factorizable}, or \defn{$d$-compatible}, if there exist a single quantum instrument $\Psi : S_1(A) \to S_1^d \otimes^\wedge L_1(\nu)$ such that for each $i \in \mathsf{I}$, $T^{(i)} = Q^{(i)} \circ \Psi$ for some completely positive, trace-preserving map $Q^{(i)} : S_1^d \otimes^\wedge L_1(\nu) \to S_1(B) \otimes^\wedge L_1(\mu)$.
Equivalently, the following diagram commutes for all $i \in \mathsf{I}$.
\begin{equation}
  \begin{tikzcd}[column sep=tiny,row sep=large]
    S_1(A) \ar[rr,"T^{(i)}"] \ar[dr,swap,"\Psi"]
      & & S_1(B) \otimes^\wedge L_1(\mu) \\
    & S_1^d \otimes^\wedge L_1(\nu) \ar[ru,swap,"Q^{(i)}"] &
  \end{tikzcd}
\end{equation}
\end{definition}

If $d \ge \dim A$, one may choose $\Psi$ as an isometric embedding, making the definition trivial.
Therefore, the regime of interest is when $d < \dim A$. When $\dim B = d = 1$, it recovers standard measurement compatibility, and when $d=1$, this construction reduces to the a notion of instrument compatibility.



There are two immediate ways to construct incompatible multi-instruments.

\begin{itemize}
    \item Let $\mc M=\{\{M_{a|x}\}_a\}_x$ be an incompatible multi-meter and define measure-and-prepare instruments
\begin{equation*}
    T_{a|x}(\rho) = \Tr(M_{a|x}\rho)\tau_{a|x}.
\end{equation*}
Each instrument $\{T_{a|x}\}_a$ is individually $1$-factorizable, whereas the multi-instrument is $1$-incompatible because its induced multi-meter is $\mc M$.

    \item Any family of individually $d$-factorizable but collectively $d$-incompatible channels can be regarded as a multi-instrument with a single classical outcome.
\end{itemize}

In both constructions, however, the incompatibility is already present in one of the two marginals: the classical-outcome marginal in the first construction and the quantum-output marginal in the second. We now give a multi-instrument whose $d$-incompatibility is genuinely instrument-level. Neither its induced multi-meter nor its family of outcome-forgotten channels is incompatible; the incompatibility resides entirely in the correlations between the classical outcome and the quantum output.

\begin{example}[Incompatible instruments from incompatible measurements]
Let $\mathsf A=\mathbb Z_m$, and let
$\cM=\{\{M_{a|x}\}_{a\in\mathsf A}\}_{x\in\mathsf X}$ be a family of POVMs. On $K=\bC^m$, define
\begin{equation}
    T_{a|x}(\rho) \coloneqq  \frac{1}{m}\sum_{a'\in\mathsf A}\Tr(M_{a'|x}\rho)\ketbra{a+a'}{a+a'}.
\end{equation}
Each instrument $T^{(x)}=\{T_{a|x}\}_a$ is $1$-factorizable. Moreover,
\begin{equation}
    T_{a|x}^*(\1_K)=\frac{\1_H}{m}, \qquad \sum_aT_{a|x}(\rho)=\Tr(\rho)\frac{\1_K}{m},
\end{equation}
so the measurement incompatibility is stored entirely in the correlation
between the outcome and the output state. In fact,
\begin{equation}
    \cT\coloneqq \{T^{(x)}\}_x\text{ is $1$-compatible}
    \quad\Longleftrightarrow\quad
    \cM\text{ is jointly measurable}.
\end{equation}
The forward construction follows immediately from any parent POVM of $\cM$.
Conversely, if
\begin{equation}
    T_{a|x}(\rho) =\sum_\lambda\Tr(G_\lambda\rho)\sigma_{a|x,\lambda},
\end{equation}
define $q(a'|x,\lambda)\coloneqq \sum_a\bra{a+a'}\sigma_{a|x,\lambda}\ket{a+a'}$, and $M_{a'|x}=\sum_\lambda q(a'|x,\lambda)G_\lambda$. Thus, every incompatible $\cM$ produces individually $1$-factorizable but
collectively $1$-incompatible instruments.

\end{example}

\section{Closure of compatible quantum instruments}\label{sec:closure}

A natural operational question for any joint-measurability property is whether it is preserved under approximation: \emph{if a sequence of compatible multi-meters (or $d$-compatible multi-instruments) converges in operator norm, does the limit inherit the same joint structure?}
Physically, a positive answer means that joint measurability is \emph{robust}: an experimenter who can realize parent measurements to arbitrary precision can also realize a parent measurement in the limit, with no exotic limiting behaviour.
In this section we prove that this is indeed the case, both for compatible POVMs and for $d$-compatible quantum instruments on separable Hilbert spaces with finite-dimensional output.

The main obstruction to a direct argument is the \emph{moving intermediate space}.
A compatible multi-meter $\cM^{(n)}=\{M^{(n)}_x\}_{x\in\mathsf{X}}$ factors as
\begin{equation}\label{eq:closure-section-factorization}
   M^{(n)}_x \;=\; u^{(n)}_x \circ \Pi^{(n)},
   \qquad \Pi^{(n)} : S_1(H) \to L_1(\Omega_n), \qquad
   u^{(n)}_x : L_1(\Omega_n) \to \ell_1^m,
\end{equation}
where $(\Omega_n,\Sigma_n,\mu_n)$ is the outcome space of a parent measurement that may depend on $n$ in arbitrary ways.
Even if all $M^{(n)}_x$ converge in operator norm, their factorizations \emph{a priori} live in different Banach spaces, so the parent measurements $\Pi^{(n)}$ have no canonical limit.
When the outcome spaces are uniformly finite, a routine compactness argument suffices (\Cref{lem:closurePOVM-finite}).
Without that uniformity, we replace compactness by the \emph{ultraproduct} of Banach (and operator) spaces.
The ultraproduct yields a well-defined limit Banach space even when the $L_1(\Omega_n)$ are incomparable, and a classical result identifies the limit as an $L_1$-space again, see \cite[Theorem~8.6]{pisier1986factorization} for the statement and its history.
The image of $S_1(H)$ under the limit parent is separable.
So it lies in the $L_1$-space of a countably generated $\sigma$-subalgebra, and the classification of abelian von Neumann algebras on a separable Hilbert space embeds this $L_1$-space order-isometrically into $L_1[0,1]$ (\Cref{prop:L1-maharam}).
This restores a single intermediate space through which the limit multi-meter factors.
The embedding has a positive integral-preserving left inverse, so the post-processings extend to all of $L_1[0,1]$.
The instrument theorem follows the same blueprint with the operator-space projective tensor product $S_1^d\otimes^\wedge(\cdot)$ ``pulled out'' of the ultraproduct, an operation that is valid precisely because $d$ is finite.

To keep the main flow of the argument self-contained, the ultraproduct construction, the operator-space lemmas and the Maharam-based embedding result are collected in \Cref{app:background}, and the technical proofs of the closure theorems in \Cref{app:closure}.

\subsection{Main results}\label{subsec:closure-main}

The two main statements of this section are the closure of compatible POVMs and of $d$-compatible instruments under norm convergence.
We state both together so the operator-space generalization is in plain view.

\newcommand{\maintheoremPOVM}{
    Let $H$ be a separable Hilbert space (not necessarily finite dimensional), $m$ be a finite number, and $(\cM^{(n)})_{n\in\bN}$ be a sequence of compatible multi-meters of
    measurements of the form
    \begin{equation}
        \cM^{(n)} := \{ M^{(n)}_x : S_1(H) \to \ell_1^m \}_{x\in\mathsf{X}}.
    \end{equation}
    Suppose that the sequence converges in operator norm, in the sense that each component map has a limit
    \begin{equation}
        M_x := \lim_{n\to\infty} M^{(n)}_x \qquad (x\in\mathsf{X}),
    \end{equation}
    and denote the limiting multi-meter by
    \begin{equation}
        \cM := \{ M_x : S_1(H) \to \ell_1^m \}_{x\in\mathsf{X}}.
    \end{equation}
    Then the limiting multi-meter $\cM$ is also compatible.
}
\begin{theorem}\label{thm:closurePOVM}
\maintheoremPOVM
\end{theorem}

Recall that compatibility of $\cM^{(n)}$ in \Cref{thm:closurePOVM} means each $M^{(n)}_x$ admits a factorization of the form \eqref{eq:closure-section-factorization} for a parent measurement $\Pi^{(n)}$ with classical outcome space $(\Omega_n,\Sigma_n)$ and stochastic post-processings $u^{(n)}_x$.

\begin{theorem}[Closure of $d$-compatible instruments]\label{thm:closure-instruments}
Let $H$ be a separable Hilbert space, let $K$ be a finite-dimensional Hilbert space, and let $m,d\in\mathbb{N}$.
Fix an index set $\mathsf{X}$.  For each $n\in\mathbb{N}$ and each $x\in\mathsf{X}$, let
\begin{equation}
\Phi^{(n)}_x:S_1(H)\longrightarrow S_1(K)\otimes^\wedge \ell_1^m
\end{equation}
be a completely positive, trace-preserving map.
Assume that for every $n$ the multi-instrument $\{\Phi^{(n)}_x\}_{x\in\mathsf{X}}$ is \emph{$d$-compatible}, and moreover that for each $x\in\mathsf{X}$ the sequence $(\Phi^{(n)}_x)_{n\in\mathbb{N}}$ converges in operator norm to a limit map
\begin{equation}
\Phi_x:=\lim_{n\to\infty} \Phi^{(n)}_x:S_1(H)\to S_1(K)\otimes^\wedge \ell_1^m.
\end{equation}
Then the limiting multi-instrument $\{\Phi_x\}_{x\in\mathsf{X}}$ is also $d$-compatible.
\end{theorem}

The output systems in this paper are finite-dimensional, which is why $\dim K < \infty$ is assumed.

Operationally, \Cref{thm:closurePOVM} says that the set of compatible multi-meters on a fixed input space is closed in the operator-norm topology (componentwise), and \Cref{thm:closure-instruments} says the same about the set of $d$-compatible multi-instruments.

The two theorems are stated for sequences.
The proofs in \Cref{app:closure} work verbatim for nets, with $\cU$ an ultrafilter on the directed index set that contains all tails.
The ultraproduct maps are well defined for an arbitrary index set, since every map involved is completely positive and trace preserving, hence a complete contraction, and so the family is uniformly bounded.
Hence the compatible multi-meters and the $d$-compatible multi-instruments form closed sets in the product of the operator-norm topologies, for an arbitrary set $\mathsf{X}$ of settings.
For countable $\mathsf{X}$ this product topology is metrizable, and the sequential statements already give closedness.

In particular, \emph{incompatibility} of a family is detected by a finite operator-norm gap to the compatible set, which is a useful starting point for robustness and certification statements.

\subsection{Warm-up: finite outcomes via compactness}\label{subsec:finite-case}

When the parent POVMs are uniformly bounded in their number of outcomes and the underlying Hilbert space is finite-dimensional, all factorization maps in \eqref{eq:closure-section-factorization} live in a fixed compact set, and the closure property follows by a direct compactness argument.
We record this case first because the proof is short and isolates exactly what fails in the general setting.

We begin by describing the two parameter sets that encode the factorization maps, along with the topology they carry, as their compactness is what drives the proof.

\begin{definition}[Parameter sets and their topology]\label{def:finite-param-sets}
Fix finite $d,m,k$.
\begin{itemize}
    \item The set of \defn{$k$-outcome POVMs} on $\bM_d(\bC)$ is
    \begin{equation}
        K_{\POVM} := \{ (P_i)_{i=1}^{k} \in (\bM_d(\bC)_{\mathrm{sa}})^{k} \;:\; P_i \geq 0,\ \textstyle\sum_{i=1}^{k} P_i = \1 \},
    \end{equation}
    regarded as a subset of the finite-dimensional real vector space $(\bM_d(\bC)_{\mathrm{sa}})^{k}\cong\bR^{kd^2}$.
    \item The set of \defn{post-processings} is the set of $m\times k$ column-stochastic matrices
    \begin{equation}
        K_{\stoc} := \{ u\in\bR_{\geq 0}^{m\times k} \;:\; \textstyle\sum_{a=1}^{m} u_{ab} = 1 \ \text{for all } b \},
    \end{equation}
    regarded as a subset of $\bR^{mk}$.
\end{itemize}
Each set carries the norm topology inherited from its ambient finite-dimensional space; since all norms on a finite-dimensional space are equivalent, this topology is canonical and coincides with entrywise convergence.
\end{definition}

\begin{lemma}[Compactness of the parameter sets]\label{lem:finite-param-compact}
$K_{\POVM}$ and $K_{\stoc}$ are compact, and hence so is $K_{\POVM}\times\prod_{x\in\mathsf{X}}K_{\stoc}$.
\end{lemma}
\begin{proof}
Both sets are defined inside a finite-dimensional space by finitely many non-strict (hence closed) constraints---positivity and an affine normalization---so each is closed.
The constraints also force boundedness: $0\le P_i\le\1$ gives $\norm{P_i}\le 1$, and $u_{ab}\in[0,1]$.
By Heine--Borel each set is compact, and compactness of the product follows from Tychonoff's theorem (a subnet replacing a subsequence when $\mathsf{X}$ is infinite).
\end{proof}

It is worth being explicit about why this compactness is the right tool, and why it is compatible with the operator-norm convergence that defines our problem.
The hypothesis fixes the data $M^{(n)}_x$ to converge in operator norm; what we must produce is a single \emph{parent} measurement $\Pi$ and post-processings $u_x$ realizing the limit multi-meter.
Compactness of $K_{\POVM}\times\prod_{x\in \mathsf{X}} K_{\stoc}$ supplies these witnesses by extracting a convergent subnet of the parents $(\Pi^{(n)},u^{(n)}_x)$, but a priori the topology in which they converge need not be the operator-norm topology in which the data converge.
The finite-dimensional, bounded-outcome assumptions remove this gap on two counts.

First, by \Cref{def:finite-param-sets} each factor $K_{\POVM}$ and $K_{\stoc}$ of the parameter set sits in a finite-dimensional space, so by \Cref{lem:finite-param-compact} each factor is norm-compact and the product over $x\in\mathsf X$ is compact in the product topology, also for infinite $\mathsf X$, and on a finite-dimensional space all norms are equivalent, with constants depending only on the dimension, so ``norm convergence'' of the parents is unambiguous and coincides with entrywise convergence in each factor.

Second, because the intermediate space $\ell_1^k$ has fixed finite dimension, the maps $\Pi^{(n)}$ and $u^{(n)}_x$ are uniformly bounded in operator norm, and composition $(u,\Pi)\mapsto u\circ\Pi$ is jointly continuous for the operator norm with an explicit Lipschitz estimate (the bilinear bound $\norm{u\circ\Pi-u'\circ\Pi'}\le\norm{u}\,\norm{\Pi-\Pi'}+\norm{u-u'}\,\norm{\Pi'}$ used below).
Joint continuity is what makes the norm topology \emph{valid} for this argument: the limit extracted by compactness is automatically a limit of the compositions $u^{(n)}_x\circ\Pi^{(n)}$ in operator norm, so it must agree with the independently given operator-norm limit $M_x$.
What fails in the general setting is precisely this coincidence: once the outcome number $k$ or the dimension is unbounded, the parameter set is no longer norm-compact, and a weaker (e.g.\ weak-$*$) topology must be used, in which the joint continuity of composition---and hence the validity of passing the norm limit through the factorization---is no longer automatic.

\begin{example}[Non-uniqueness of the parent POVM]\label{ex:nonunique-parent}
Generally, the parent POVM witnessing compatibility is not unique, which is exactly why the lemma below extracts a convergent subnet of the parents rather than asserting that the parents themselves converge.
Here, we demonstrate a convergent sequence of compatible multi-meters with two distinct converging subsequences of parent POVMs.
Take $d=2$ and two \emph{noisy spin measurements} along the orthogonal axes $\hat z$ and $\hat x$, with a constant $\eta\in(0,1/\sqrt2)$:
\begin{equation}
M_{a|1}=\frac{1}{2}(\1+a\,\eta\,\sigma_z),\qquad M_{b|2}=\frac{1}{2}(\1+b\,\eta\,\sigma_x),\qquad a,b\in\{+1,-1\}.
\end{equation}
These are noisy versions of $\sigma_z$ and $\sigma_x$ measurements; although $\sigma_z$ and $\sigma_x$ are maximally incompatible, the smeared pair is jointly measurable precisely when $\eta\le1/\sqrt2$, so our choice $\eta<1/\sqrt2$ sits strictly inside the compatible region.
Define the deterministic post-processings $u_1:(a,b)\mapsto a$ and $u_2:(a,b)\mapsto b$.
Then, a four-outcome POVM $(G_{ab})_{a,b\in\{+1,-1\}}$ is a parent POVM when $\sum_b G_{ab}=M_{a|1}$ and $\sum_a G_{ab}=M_{b|2}$.
Consider the following $G_{ab}$, parametrized by $t$:
\begin{equation}
G_{ab}(t)=\frac{1}{4}(\1+a\,\eta\,\sigma_z+b\,\eta\,\sigma_x+ab\,t\,\sigma_y).
\end{equation}

The marginal constraints hold for \emph{every} real $t$, while positivity of all four effects holds iff
\begin{equation}
2\eta^2+t^2\le1,\qquad\text{i.e.}\qquad t\in[-\sqrt{1-2\eta^2},\,\sqrt{1-2\eta^2}].
\end{equation}
With the post-processings held fixed, the parents fill a nondegenerate compact segment $\{\Pi(t):|t|\le\sqrt{1-2\eta^2}\}$ $\subset K_{\POVM}$, which collapses to a single point only at the threshold $\eta=1/\sqrt2$.
Now let the compatible multi-meter be constant, $\cM^{(n)}=\cM$ for all $n$, but choose the parents $\Pi^{(n)}=\Pi(t_n)$ with $t_n=\sqrt{1-2\eta^2}\,(-1)^n$.
Then $M^{(n)}_x\to M_x$ in operator norm while the parent sequence $(\Pi^{(n)})$ does not converge: its cluster points are the two extremal joint POVMs $\Pi(\pm\sqrt{1-2\eta^2})$.

What the proof uses instead is compactness of $K_{\POVM}\times\prod_x K_{\stoc}$ (\Cref{lem:finite-param-compact}): any choice of parents admits a convergent subsequence $\Pi^{(n_j)}\to\Pi(t_\infty)$ for some limit parameter $t_\infty\in[-\sqrt{1-2\eta^2},\,\sqrt{1-2\eta^2}]$, and the bilinear estimate \eqref{eq:bilinear-bound} in the proof below forces $u_x\circ\Pi(t_\infty)=M_x$, so the limit is again a parent of $\cM$---though which value $t_\infty$ takes depends on the subsequence.
\end{example}

\begin{lemma}\label{lem:closurePOVM-finite}
    Let $d,m,k$ be finite numbers, $H$ a Hilbert space of dimension $d$, $\mathsf{X}$ a set, not necessarily finite, and $(\cM^{(n)})_{n\in\bN}$ a sequence of compatible multi-meters of
    measurements of the form
    \begin{equation}
        \cM^{(n)} := \{ M^{(n)}_x : S_1(H) \to \ell_1^m \}_{x\in\mathsf{X}},
    \end{equation}
    where each parent POVM has $k$ outcomes.

    Suppose that the sequence converges in operator norm, in the sense that each component map has a limit
    \begin{equation}
        M_x := \lim_{n\to\infty} M^{(n)}_x \qquad (x\in\mathsf{X}),
    \end{equation}
    and denote the limiting multi-meter by
    \begin{equation}
        \cM := \{ M_x : S_1(H) \to \ell_1^m \}_{x\in\mathsf{X}}.
    \end{equation}
    Then the limiting multi-meter $\cM$ is also compatible via $k$-outcome POVM.
\end{lemma}

\begin{proof}
    Assume that each $M^{(n)}_x$ factors through $\ell_1^k$ via maps $u^{(n)}_x:\ell_1^k \to \ell_1^m$ and $\Pi^{(n)}:S_1(H) \to \ell_1^k$.
    Let $(P^{(n)}_i)_i$ be the dual element of $\Pi^{(n)}$ expressed as the $k$-tuple POVM, i.e.
    \begin{equation}
       \Pi^{(n)}(\rho) \mapsto (\Tr[P^{(n)}_i\rho])_i \q \sum_i^k P^{(n)}_i = \1, \; P^{(n)}_i \geq 0,
    \end{equation}
    and express each $u^{(n)}_x$ as an $m\times k$ stochastic matrix.
    Via these equivalences (\Cref{def:finite-param-sets}), the tuple $((P^{(n)}_i)_i,(u^{(n)}_x)_x)$ is an element of $K_{\POVM} \times \prod_{x\in \mathsf{X}} K_{\stoc}$.
    By \Cref{lem:finite-param-compact} this set is compact, so there exists a convergent subnet\renewcommand{\thefootnote}{\arabic{footnote})}\footnote{A \emph{net} is the generalization of a sequence in which the index set is an arbitrary directed set rather than $\bN$; a \emph{subnet} is the corresponding generalization of a subsequence. Nets are needed here only because $\mathsf{X}$ may be uncountable, so that the product $\prod_{x\in\mathsf{X}} K_{\stoc}$ is compact but need not be sequentially compact. When $\mathsf{X}$ is countable the reader may simply read ``subsequence'' throughout.}\renewcommand{\thefootnote}{\arabic{footnote}} (a subsequence when $\mathsf{X}$ is finite) $((P^{(n_j)}_i)_i, (u^{(n_j)}_x)_{x})_{j} \to ((P_i)_i, (u_x)_x)$
    in $K_{\POVM} \times \prod_{x\in\mathsf{X}} K_{\stoc}$.
    Let $\Pi$ be the predual element corresponding $(P_i)$.
    Then
    \begin{align}
        \norm{u_x \circ \Pi - u^{(n_j)}_x \circ \Pi^{(n_j)}}
        &=  \norm{u_x \circ (\Pi - \Pi^{(n_j)}) + (u_x - u^{(n_j)}_x) \circ \Pi^{(n_j)}}\\
        &\leq \norm{u_x}\cdot \norm{\Pi - \Pi^{(n_j)}} + \norm{u_x - u^{(n_j)}_x}\cdot \norm{\Pi^{(n_j)}} \to 0 \label{eq:bilinear-bound}
    \end{align}
    Therefore, $u_x \circ \Pi = M_x$ for all $x \in \mathsf{X}$.
\end{proof}

\begin{remark} \label{rem:finite-suffices}
The same conclusion holds when $\mathsf{X}$ is finite and the parent POVM is allowed infinitely many ($k=\infty$) outcomes.
Indeed, suppose $\{M_x\}_{x\in\mathsf{X}}$ is compatible through a parent POVM $\Pi$ on an outcome space $\Omega$ (so $\Pi\geq0$, $\Pi(\Omega)=\1$) and post-processing kernels $u_x(a\mid\omega)$, $\sum_a u_x(a\mid\omega)=1$, with $M_{a|x}=\int_\Omega u_x(a\mid\omega)\,d\Pi(\omega)$.
Since $\mathsf{X}$ and $m$ are finite, there are finitely many outcome assignments $g=(g_x)_{x\in\mathsf{X}}\in[m]^{\mathsf{X}}$ and
\begin{equation}
  G_g:=\int_\Omega (\prod_{x\in\mathsf{X}} u_x(g_x\mid\omega))\,d\Pi(\omega)
\end{equation}
defines a joint POVM: $G_g\geq0$ and $\sum_g G_g=\int_\Omega 1\,d\Pi=\1$.
For each setting $y\in\mathsf{X}$ and outcome $a\in[m]$ the marginal recovers $M_y$, namely $\sum_{g:\,g_y=a}G_g=M_{a|y}$, where the sum runs over the assignments $g$ with $g_y=a$.
Indeed,
\begin{align}
  \sum_{g:g_y = a} G_g &= \sum_{g:g_y = a} \int_{\Omega} (\prod_{x\in \mathsf{X}} u_x (g_x|\omega)) \; d\Pi(\omega) \\
  &= \int_{\Omega} \sum_{g:g_y = a} (\prod_{x\neq y} u_x(g_x|\omega)) \cdot u_y(a|\omega) \;d\Pi(\omega) \\
  &= \int_{\Omega} u_y(a|\omega) \; d\Pi(\omega) = M_{a|y}.
\end{align}
Thus the infinite-outcome parent may be folded into a finite one with at most $m^{|\mathsf{X}|}$ outcomes, reducing the statement to the finite-$k$ case proved above.
Moreover, the tuples $w_g=(\delta_{a,g_x}G_g)_{x,a}$ lie in a real vector space of dimension $|\mathsf X|(m-1)d^2+d^2$, so the conic Carath\'eodory theorem lets one retain at most $k_0=|\mathsf{X}|(m-1)d^2+d^2$ of the $G_g$, rescaled by nonnegative scalars, and discard the rest while preserving every $M_{a|x}$.

\end{remark}

The two ingredients used above---compactness of the set of $k$-outcome POVMs and a fixed finite dimension on the intermediate space $\ell_1^k$---are exactly what is missing when the parent measurement is allowed to take values in $L_1(\Omega_n)$ with $\Omega_n$ unbounded.
The next subsection outlines how an ultraproduct construction recovers them.

\subsection{Strategy of the general proofs}\label{subsec:proof-strategy}

The full proofs of \Cref{thm:closurePOVM,thm:closure-instruments} are given in \Cref{app:closure}.
Here we describe the four-step plan that underlies both arguments so the reader can follow the appendix with the physical picture in mind.

\paragraph{Step 1: pass each factorization to an ultraproduct.}
Fix a free ultrafilter $\cU$ on $\mathbb{N}$ and consider the Banach-space ultraproduct $\uprod{L_1(\Omega_n)}$, which is well defined even though the underlying measure spaces $\Omega_n$ vary with $n$.
The ultraproduct of the parent measurements $\Pi^{(n)}$ and of the post-processings $u^{(n)}_x$ produces a factorization of $\uprod{M^{(n)}_x}$ through $\uprod{L_1(\Omega_n)}$.
For the instrument case, the same step is carried out for $\theta^{(n)}:S_1(H)\to S_1^d\otimes^\wedge L_1(\Omega_n)$ and the channels $\Psi^{(n)}_x$; the operator-space ultraproduct of completely positive trace-preserving maps is again completely positive and trace-preserving (\Cref{lem:cp-ultra-os}).

\paragraph{Step 2: identify the ultraproduct as an $L_1$-space.}
By a classical result, $\uprod{L_1(\Omega_n)}$ is itself an abstract $L_1$-space \cite[Theorem~8.6]{pisier1986factorization} and hence order-isometric to $L_1(\Omega,\Sigma,\mu)$ for some measure space $(\Omega,\Sigma,\mu)$, in general not $\sigma$-finite.
For the instrument case, the finite-dimensional pull-out \Cref{lem:pullout-finite} extends this identification through the operator-space tensor product:
\begin{equation*}
\uprod{S_1^d\otimes^\wedge L_1(\Omega_n)} \;\cong\; S_1^d\otimes^\wedge \uprod{L_1(\Omega_n)} \;\cong\; S_1^d\otimes^\wedge L_1(\Omega)
\end{equation*}
completely isometrically.

\paragraph{Step 3: bring the diagonal copy of $S_1(H)$ back to $L_1[0,1]$.}
The diagonal embedding $\Delta_H:S_1(H)\to \uprod{S_1(H)}$ is isometric, so the image of $\widetilde \Pi:=\uprod{\Pi^{(n)}}\circ \Delta_H$ inside the abstract $L_1(\Omega)$ is separable.
By \Cref{prop:L1-maharam} the image $\widetilde\Pi(S_1(H))$ lies in $L_1(\Omega,\Sigma_X)$ for a countably generated $\sigma$-subalgebra $\Sigma_X$.
There are an order-isometric embedding $j:L_1(\Omega,\Sigma_X)\hookrightarrow L_1[0,1]$ and a positive integral-preserving map $\cE:L_1[0,1]\to L_1(\Omega,\Sigma_X)$ with $\cE\circ j=\id$.

Pushing $\widetilde\Pi$ forward through $j$ yields a candidate parent measurement $\Pi:=j\circ\widetilde\Pi:S_1(H)\to L_1[0,1]$.
For instruments, the corresponding step uses \Cref{lem:tensor-preserve-embed}, which promotes the embedding $j$ to a complete isometry $\id_{S_1^d}\otimes j$ on the operator-space projective tensor product.

\paragraph{Step 4: descend the post-processings to a fixed Banach space.}
The diagonal embedding $\Delta_m:\ell_1^m\to\uprod{\ell_1^m}$ is also isometric.
Since $\ell_1^m$ is finite-dimensional, $\Delta_m$ is onto, and the ultraproduct post-processings $\uprod{u^{(n)}_x}$ give maps $u_x:=\Delta_m^{-1}\circ\uprod{u^{(n)}_x}\circ\cE:L_1[0,1]\to\ell_1^m$, where $L_1(\Omega,\Sigma_X)$ is regarded as a subspace of $L_1(\Omega)$.
Combining this with the identity $\uprod{M^{(n)}_x}\circ \Delta_H = \Delta_m\circ M_x$ (\Cref{lem:diag-ultra-limit}, a direct consequence of operator-norm convergence) one verifies that $M_x = u_x\circ \Pi$, which is the desired factorization of the limit multi-meter through a single parent measurement.
For instruments, an analogous argument with $\uprod{S_1(K)\otimes^\wedge\ell_1^m}$ in place of $\uprod{\ell_1^m}$ recovers post-processing channels $\Psi_x:S_1^d\otimes^\wedge L_1[0,1]\to S_1(K)\otimes^\wedge \ell_1^m$ with $\Phi_x = \Psi_x\circ \theta$.

The supporting definitions and lemmas (ultrafilters, ultralimits, the operator-space ultraproduct, Pisier's factorization theorem, the operator-space tensor-product lemmas, and the Maharam-based embedding) are collected in \Cref{app:background}, and the full proofs of \Cref{thm:closurePOVM,thm:closure-instruments} are given in \Cref{app:closure}.

\section{$d$-compatibility criterion for quantum channels}\label{sec:criterion}
\subsection{Criterion via completely bounded maps}

We give a quantitative criterion for joint $d$-factorizability in terms of a completely bounded factorization functional, certified by linear \emph{witnesses} via Hahn--Banach duality.

Let $\{\Phi_x : \bM_m\to \bM_n \}_{x\in\mathsf{X}}$ be a family of c.b. maps. We say the family is \defn{jointly $d$-factorizable via c.b. maps} if there exist a $\sigma$-finite measure space $(\Omega,\Sigma,\mu)$ and contractive c.b. maps $\theta : L_\infty(\Omega;\bM_d) \to \bM_n$ and $\Psi_x : \bM_m \to L_\infty(\Omega;\bM_d)$ with
\begin{equation}
  \Phi_x = \theta \circ \Psi_x, \q \forall\, x \in \mathsf{X}.
\end{equation}

Throughout this section the set $\mathsf X$ of settings is finite, and the map $\theta$ in every factorization is normal, that is weak$^*$ continuous.

It is convenient to assemble the family into a single map $\Phi:\ell_1(\mathsf{X}, \bM_m) \to \bM_n$,
\begin{equation}
  \Phi((z_x)_x) = \sum_x \Phi_x(z_x).
\end{equation}
Define the functional $\gamma_{d,\cb}$
\begin{equation} \label{eq:gamma-cb}
  \gamma_{d,\cb}(\Phi):=\inf\{\,\norm{\theta}_{\cb}\cdot\sup_x\norm{\Psi_x}_{\cb}
   \mid \Phi_x=\theta \circ \Psi_x \; \forall x \in \mathsf{X}\},
\end{equation}
where the infimum is over all \emph{completely bounded} factorizations $\Phi_x=\theta\circ\Psi_x$ through $L_\infty(\Omega;\bM_d)$ for some $(\Omega,\Sigma,\mu)$, and $\gamma_{d,\cb}(\Phi):=+\infty$ if no such factorization exists.

The functional $\gamma_{d,\cb}$ measures the best available factorization constant.
The following proposition confirms that it detects joint $d$-factorizability exactly, the family being factorizable precisely when $\gamma_{d,\cb}\le1$.

\begin{proposition}\label{prop:cb-unit-ball}
  $\{\Phi_x\}_x$ is jointly $d$-factorizable via c.b. maps if and only if $\gamma_{d,\cb}(\Phi) \leq 1$.
\end{proposition}
\begin{proof}
  ($\Rightarrow$) A contractive factorization has $\norm{\theta}_{\cb}, \sup_x \norm{\Psi_x}_{\cb} \leq 1$, so the product in \cref{eq:gamma-cb} is $\leq 1$ and the infimum satisfies $\gamma_{d,\cb} \leq 1$.

  ($\Leftarrow$) Suppose $\gamma_{d,\cb} \leq 1$.
  If $\Phi = 0$, then $\theta = \Psi = 0$ satisfies the condition.
  Assume $\Phi \neq 0$.
  The ultraproduct argument of \Cref{thm:closure-instruments} applies with complete positivity and trace preservation replaced by complete boundedness, the completely bounded unit ball being closed under operator-space ultraproducts; one passes between the Heisenberg and Schr\"{o}dinger pictures by taking adjoints, which preserves the c.b. norm, so the factorizing maps may be taken normal.
  Therefore, the set of families that admit a contractive factorization is closed and the infimum in \cref{eq:gamma-cb} is attained.

  If $\Phi = \theta \circ \Psi$, then we can rescale $\theta' = \norm{\theta}_{\cb}^{-1}\cdot \theta$ and $\Psi_x' = \norm{\theta}_{\cb} \Psi_x$.
  Then $\norm{\theta'}_{\cb} = 1$ and $\sup_x \norm{\Psi_x'}_{\cb} = \norm{\theta}_{\cb} \sup_x \norm{\Psi_x}_{\cb} = \gamma_{d,\cb}(\Phi) \leq 1$.
\end{proof}

The main theorem is our completely bounded criterion for joint $d$-factorizability. It expresses $\gamma_{d,\cb}$ as a supremum over linear \emph{witnesses} $T$, so a family fails to be $d$-factorizable exactly when some witness $T$ drives $\norm{T\otimes\Phi}$ above $1$.

\newcommand{\CBCriterion}{
  Let $\{\Phi_x:\bM_m \to \bM_n\}_{x\in \mathsf{X}}$ be linear maps.
  Then
  \begin{equation} \label{eq:cb-criterion}
    \gamma_{d,\cb}(\Phi) = \sup \norm{T \otimes \Phi: \ell_\infty(\mathsf{X},S_1^m) \otimes_{\min} \ell_1(\mathsf{X}, \bM_m) \to \bM_n \otimes_{\min} \bM_n},
  \end{equation}
  where $T\otimes \Phi$ is defined as
  \begin{equation}\label{eq:T}
    T \otimes \Phi = \sum_{x}T_x \otimes \Phi_x
  \end{equation}
  the supremum is over all linear $T$ with
  \begin{equation}
    \norm{T:\ell_\infty(\mathsf{X}, S_1^m)\to \bM_n}_{(d)} \leq 1.
  \end{equation}
}
\begin{theorem}\label{thm:cb-criterion}
  \CBCriterion
\end{theorem}
See \Cref{app:cb-criterion} for the proof.

In \eqref{eq:T} the two factors carry the operator space structures given by the pairing~\eqref{eq:scalar-pairing}.
Here $S_1^m$ is the operator space dual of $\bM_m$ under the pairing $\ip{f}{y}=\Tr[f^Ty]=\sum_{s,t}f_{st}y_{st}$ of \eqref{eq:scalar-pairing} in \Cref{app:pairings}, and $\ell_\infty(\mathsf X,S_1^m)$ is the dual of $\ell_1(\mathsf X,\bM_m)$ in the same way.
With these structures the map norm dominates a single Choi evaluation (\Cref{lem:choi}):
\begin{equation}\label{eq:cb-criterion-choi}
  \norm{T \otimes \Phi} \mathrel{\ge} \norm{\sum_x (T_x \otimes \Phi_x)(\chi_m)}_{\bM_n \minten \bM_n}.
\end{equation}
The suprema of the two sides over $\norm{T}_{(d)}\le1$ coincide, so \Cref{thm:cb-criterion} holds with either quantity.
Here, $\chi_m = \sum_{s,t=1}^m \ketbra{s}{t}\otimes\ketbra{s}{t} \in \bM_{m^2}$ is the unnormalized maximally entangled state (or the Choi matrix of $\id_{\bM_m}$).
In particular, the Choi quantity is the ordinary operator norm of a single matrix in $\bM_n\minten\bM_n$, a concrete finite-dimensional quantity that the proof evaluates.

\subsection{Criterion via completely positive maps}
\Cref{thm:cb-criterion} characterizes joint $d$-factorizability in terms of completely bounded maps.
Quantum channels, however, factor through \emph{completely positive} maps, and completely positive factorization is genuinely more restrictive than its completely bounded counterpart.
The distinction is categorical: completely bounded maps are the morphisms of operator spaces, whereas completely positive maps are the morphisms of operator systems.
Since every completely positive map is completely bounded, a completely positive factorization is in particular a completely bounded one, so the infimum defining $\gamma_{d,\cp}$ ranges over a strictly smaller family than the one defining $\gamma_{d,\cb}$.
Consequently the completely bounded equivalence \eqref{eq:cb-criterion} survives only as a necessary condition for completely positive factorizability.
We now establish a completely positive analogue of the equivalence \eqref{eq:cb-criterion}.%

Throughout this section $\Phi=(\Phi_x)_{x\in\mathsf{X}}$ denotes a family of completely positive maps $\Phi_x:\bM_m\to\bM_n$.
As in the completely bounded case we identify it with the single map $\Phi:\ell_1(\mathsf{X},\bM_m)\to\bM_n$, $\Phi((z_x)_x)=\sum_x\Phi_x(z_x)$.

We say that the family $\Phi$ is \defn{jointly $d$-factorizable via c.p. maps} if there exist a $\sigma$-finite measure space $(\Omega,\Sigma,\mu)$ and completely positive contractions $\theta:L_\infty(\Omega;\bM_d)\to\bM_n$ and $\Psi_x:\bM_m\to L_\infty(\Omega;\bM_d)$ with $\Phi_x=\theta\circ\Psi_x$ for all $x$.
The family is \defn{$d$-compatible} if such a factorization exists with \emph{unital} maps, i.e. dual objects of quantum channels.

We define the c.p. factorization functional as
\begin{equation}\label{eq:gamma-cp}
  \gamma_{d,\cp}(\Phi):=\inf\{\,\norm{\theta}_{\cb}\cdot\sup_x\norm{\Psi_x}_{\cb}\mid \Phi_x=\theta\circ\Psi_x\ \forall x \},
\end{equation}
where the infimum is over all c.p. factorizations through some $L_\infty(\Omega; \bM_d)$.
We set $\gamma_{d,\cp}(\Phi):=+\infty$ if no such factorization exists.

The completely positive contractive factorizations of $\Phi$ through $L_\infty(\,\cdot\,;\bM_d)$ form a class closed under operator-space ultraproducts.
Consequently, the infimum in \eqref{eq:gamma-cp} is attained.
The proof of \Cref{thm:closure-instruments} applies with the normalization $\theta(\1)=\1$ replaced by the sub-unital conditions $\theta(\1)\le\1$ and $\Psi_x(\1)\le\1$.
These are non-strict order inequalities, which are preserved under ultralimits exactly as the equality is.
Complete positivity is likewise preserved by \Cref{lem:cp-ultra-os}, applied to the preduals as in that proof, so the ultraproduct of contractive completely positive factorizations is again one, and it is identified with a factorization through some $L_\infty(\Omega;\bM_d)$ with $\sigma$-finite $\Omega$ and normal $\theta$, as in the strategy of \Cref{subsec:proof-strategy}.

\newcommand{\CPGamma}{
  A family $\Phi=(\Phi_x)_{x\in\mathsf{X}}$ of contractive c.p. maps $\Phi_x:\bM_m\to\bM_n$ is $d$-factorizable via c.p. maps if and only if $\gamma_{d,\cp}(\Phi) \leq 1$.
  Moreover, a unital $\Phi$ is $d$-compatible if and only if $\gamma_{d,\cp}(\Phi) \leq 1$.
}
\begin{proposition}\label{prop:cp-gamma}
  \CPGamma
\end{proposition}
See \Cref{app:cp-gamma-proof} for the proof.

Before stating the c.p. criterion, we introduce the quantity $\norm{\cdot}_{(d)+}$.

For a linear map $u:\bM_m\to\bM_d$ let $J(u):=(\id_{\bM_m}\otimes u)(\chi_m)$ be its Choi matrix.
The positive elements of $\ell_\infty(\mathsf X,S_1^m)\otimes\bM_d$ are the tuples $(z_x)_{x\in\mathsf X}$ with $z_x\ge0$ in $S_1^m\otimes\bM_d=\bM_m\otimes\bM_d$, and by Choi's theorem \cite{Choi1975}, \cite[Thm.~3.14]{Paulsen_2003}, $u$ is completely positive if and only if $J(u)\ge0$.
Hence
\begin{equation}\label{eq:mincone}
  (\ell_\infty(\mathsf X,S_1^m)\otimes\bM_d)_+ = \{(J(u_x))_{x\in\mathsf X}\mid u_x:\bM_m\to\bM_d\ \text{completely positive}\}.
\end{equation}

For a self-adjoint matrix $A$ let $A_+:=\frac12(\abs{A}+A)$ be its positive part, so that $A=A_+-A_-$ with $A_\pm\ge0$ and $A_+A_-=0$.
Its operator norm is
\begin{equation}\label{eq:plus-part}
  \norm{A_+}=\max\{\lambda_{\max}(A),0\}=\sup\{\Tr[A\rho]\mid \rho\geq0,\ \Tr\rho\leq1\}.
\end{equation}
For a ${*}$-preserving $T=(T_x)_x$, $T_x:S_1^m\to\bM_n$, we define
\begin{equation}\label{eq:dplus}
  \norm{T}_{(d)+}:=\sup\{\norm{(\sum_x (T_x\otimes u_x)(\chi_m))_+} \mid u_x:\bM_m\to\bM_d\ \text{completely positive},\ u_x(\1)\leq\1\}.
\end{equation}
By the Choi correspondence the admissible inputs are exactly the tuples of Choi matrices of completely positive maps $u_x:\bM_m\to\bM_d$ with $\norm{u_x}_{\cb}\leq 1$, equivalently $u_x(\1)\le\1$ \cite[Prop.~3.6]{Paulsen_2003}, which is the form used in \eqref{eq:cpmatrix}.
The quantity $\norm{T}_{(d)+}$ is positively homogeneous and subadditive in $T$ but is not a norm, since it can vanish for $T\neq0$, for example for $T_x=-\id$ when $m=n$.

\newcommand{\CPCriterion}[1]{
Let $\mathsf X$ be finite, let $\Phi_x:\bM_m\to\bM_n$ be completely positive, $x\in \mathsf{X}$, and let $K\ge0$.
Then $\gamma_{d,\cp}(\Phi)\le K$ if and only if
\begin{equation}#1
  \norm{(\sum_x (T_x\otimes\Phi_x)(\chi_m))_+}\ \le\ K\,\norm{T}_{(d)+}
\end{equation}
for every ${*}$-preserving $T=(T_x)_x$, $T_x:S_1^m\to\bM_n$.
}
\begin{theorem}[c.p. criterion]\label{thm:cp-criterion}
  \CPCriterion{\label{eq:cpcriterion}}
\end{theorem}

Since $T_x$'s are $*$-preserving maps, their Choi matrices $J(T_x)$ are self-adjoint.
Then the condition \eqref{eq:cpcriterion} yields the following: for every finite tuple of self-adjoint matrices $J(T_x) \in \bM_m \otimes \bM_n$
\begin{equation}\label{eq:cpmatrix}
  \norm{(\sum_x (\Phi_x\otimes\id_{\bM_n})(J(T_x)))_+}
  \ \le\ K \sup_{(u_x)} \norm{(\sum_x (u_x\otimes\id_{\bM_n})(J(T_x)))_+},
\end{equation}
where the supremum is over all c.p. maps $u_x:\bM_m \to \bM_d$ with $\norm{u_x}_{\cb} \leq 1$.

For the proof, see \Cref{app:cp-criterion}.

\begin{remark}[Multi-instruments]\label{rem:criterion-instruments}
The results of this section extend to multi-instruments.
In the Heisenberg picture, an instrument with outcome set $\mathsf A$ and output $\bM_k$ is a unital c.p.\ map $\Phi_x:\ell_\infty^{\mathsf A}(\bM_k)\to\bM_n$.
Here $\ell_\infty^{\mathsf A}(\bM_k)\subseteq\bM_{|\mathsf A|k}$ is the subalgebra of block-diagonal matrices, as in \Cref{sec:d-PEB}.
Factorizations through $L_\infty(\Omega;\bM_d)$, $d$-compatibility and the functionals $\gamma_{d,\cb}$ and $\gamma_{d,\cp}$ are defined for such maps exactly as above.
Let $E:\bM_{|\mathsf A|k}\to\ell_\infty^{\mathsf A}(\bM_k)$ be the conditional expectation onto the diagonal blocks and let $\iota$ be the inclusion.
Both maps are unital, completely positive and completely contractive, and $E\circ\iota=\id$.
The unital c.p.\ map $\tilde\Phi_x:=\Phi_x\circ E$ on $\bM_{|\mathsf A|k}$ is the instrument regarded as a channel with block-diagonal output.
A factorization $\Phi_x=\theta\circ\Psi_x$ gives the factorization $\tilde\Phi_x=\theta\circ(\Psi_x\circ E)$.
Conversely, a factorization $\tilde\Phi_x=\theta\circ\tilde\Psi_x$ gives $\Phi_x=\tilde\Phi_x\circ\iota=\theta\circ(\tilde\Psi_x\circ\iota)$.
Composition with $E$ or $\iota$ preserves complete positivity and unitality and does not increase c.b.\ norms.
Hence $\gamma_{d,\cb}(\tilde\Phi)=\gamma_{d,\cb}(\Phi)$ and $\gamma_{d,\cp}(\tilde\Phi)=\gamma_{d,\cp}(\Phi)$, and $\Phi$ is $d$-compatible if and only if $\tilde\Phi$ is.
\Cref{prop:cb-unit-ball,prop:cp-gamma} and the criteria of this section therefore apply to the multi-instrument $\Phi$ through $\tilde\Phi$, with $m=|\mathsf A|k$.
Moreover, $(T_x\otimes\tilde\Phi_x)(\chi_m)=(T_x\circ E\otimes\tilde\Phi_x)(\chi_m)$, and replacing $T_x$ by $T_x\circ E$ does not increase the normalization of the witness.
It therefore suffices to use witnesses supported on the diagonal blocks.
\end{remark}

\subsection{Incompatible entanglement-breaking channels from MUB} \label{sec:eb-mub}
\subsubsection{The clock--shift MUB pair}\label{sec:clock-shift}
Let $\{\ket{k} \mid 0 \leq k \leq n-1 \}$ be the computational basis and $\omega = e^{2\pi i/n}$ be a root of unity.

Introduce the \emph{clock} and \emph{shift} unitaries
\begin{equation}
  Z \ket{k} = \omega^k \ket{k}, \q\q X \ket{k} = \ket{k+1},
\end{equation}
where addition of indices is always modulo $n$.
They obey the Weyl commutation relation $ZX=\omega\,XZ$.
Let us define
\begin{equation}
  A_n = \ell_\infty^n = \Span \{ Z^k \mid 0 \leq k \leq n-1 \}, \q\q
  B_n = \Span \{ X^k \mid 0 \leq k \leq n-1 \},
\end{equation}
each of which generates a maximal abelian subalgebra of $\bM_n$.
These are two noncommuting copies of $\ell_\infty^n$ whose products span the space of $(n \times n)$-matrices.

Let $\Phi_1, \Phi_2:\bM_n \to \bM_n$ be the pinching maps onto $A_n$ and $B_n$, respectively.
For mutually orthogonal projections $\{P_i\}_i \subset \bM_n$ with $\sum_i P_i = \1$, the associated \emph{pinching map} is given by
\begin{equation}
  \cP(\rho) = \sum_i P_i \rho P_i.
\end{equation}

\newcommand{\MultDomainStatement}{
  There is no joint u.c.p. factorization
  \begin{equation}
    \Phi_x = \theta \circ \Psi_x, \q\q x = 1,2
  \end{equation}
  through $L_\infty(\Omega;\bM_d)$ for any measure space $\Omega$ and $d \lneqq n$, that is, the pair $(\Phi_1,\Phi_2)$ is not $d$-compatible for $d<n$.
}
\begin{proposition}\label{prop:mult-domain}
  \MultDomainStatement
\end{proposition}
See \Cref{app:mult-domain-proof} for the proof.

The clock--shift pair is the case $G=\bZ_n$ of a construction available for every finite group, in which the second algebra is the group algebra and is no longer commutative.
It is presented in \Cref{sec:finite-gp} below, where the two channels are individually $1$- and $\floor{\sqrt N}$-factorizable while the pair requires $d\ge N=|G|$.
The analytic criterion of \Cref{thm:cb-criterion} applies to the pair as the case $k=2$ of the bound for $k$ mutually unbiased bases proved next.

\subsubsection{Analytic bounds for $k$ mutually unbiased bases}\label{sec:m-MUB}\label{sec:d-incomp-MUB}

Now suppose $\bM_n$ carries $k$ mutually unbiased bases, each of which spanning maximal abelian matrix algebras $A_n^{(1)},\dots,A_n^{(k)}$.
Further suppose their trace-zero parts
\begin{equation}
P_n^{(x)}=\{z\in A_n^{(x)}\mid \Tr(z)=0\}
\end{equation}
are \emph{mutually orthogonal} in $S_2^n$.
For the clock--shift pair these are the spaces $H_n=\Span\{Z^u\mid1\le u\le n-1\}$ and $K_n=\Span\{X^v\mid1\le v\le n-1\}$, which are orthogonal because $\ip{Z^u}{X^v}_{\HS}=\Tr[Z^{-u}X^v]=\sum_i\omega^{-ui}\bra{i}X^v\ket{i}=0$ for $u,v\neq0$.
The witness passes $S_1^n \to S_2^n$ and then projects onto the $k$ orthogonal $(n-1)$-dimensional subspaces.
Let $\Phi_x:\bM_n\to\bM_n$ be the pinching onto $A_n^{(x)}$,  let $T_x$ be the Hilbert--Schmidt orthogonal projection onto $\bar P_n^{(x)}=\{\bar z\mid z\in P_n^{(x)}\}$, and assemble the witness
\begin{equation}
T:S_1^n\to\ell_1^k(\bM_n),\qquad T(z)=(T_1(z),\dots,T_k(z)).
\end{equation}
Its amplification norms satisfy $\norm{T}_{(d)}\le\sqrt k\sqrt d$ by \Cref{prop:Td-norm} in the appendix.
At $d=1$ this is the estimate $\norm{T}\le\sqrt k$, which follows from the domination of the operator norm by the Hilbert--Schmidt norm, the Cauchy--Schwarz passage from the $\ell_1^k$-sum to the $\ell_2^k$-sum, and the mutual orthogonality of the projections $T_x$, which gives $(\sum_x\norm{T_{x}z}_2^2)^{1/2}\le\norm{z}_2\le\norm{z}_1$.
\Cref{prop:Td-norm} holds for any $k$ mutually unbiased bases, in particular for the conjugate bases $\{\bar f_r^{(x)}\}_r$ used here.
In \eqref{eq:cb-criterion} the same components enter through the map $(w_x)_x\mapsto\sum_xT_x(w_x)$ from $\ell_\infty^k(S_1^n)$ to $\bM_n$, which under the pairing $\ip{f}{y}=\Tr[f^Ty]$ is the Banach adjoint of the map of \Cref{prop:Td-norm} for the bases $\{f_r^{(x)}\}_r$.
It also has $d$-norm at most $\sqrt k\sqrt d$, because $\norm{u^*}_{(d)}=\norm{u}_{(d)}$ for every linear map $u$ by Smith's lemma applied to $\bM_d(W^*)=\CB(W,\bM_d)$, see \cite[Prop.~3.2.2]{Effros2000}.

For the lower estimate it is cleanest to realize the pairing through the left--right action.

Let $\bM_n\otimes\bM_n$ act on $S_2^n$ by
\begin{equation}
\sigma(a\otimes b)\,\xi \;=\; a\,\xi\,b^T .
\label{eq:sigma-action}
\end{equation}
Multiplicativity reads $\sigma(a\otimes b)\sigma(a'\otimes b')\xi=aa'\xi b'^Tb^T=\sigma(aa'\otimes bb')\xi$, and $\sigma(a\otimes b)^*=\sigma(a^*\otimes b^*)$ for the Hilbert--Schmidt inner product.
So $\sigma$ is a unital $^*$-homomorphism, and since $\bM_n\otimes\bM_n\iso\bM_{n^2}$ is simple, $\sigma$ is injective, hence isometric for the minimal tensor norm.
Indeed, the kernel of $\sigma$ is a two-sided ideal of a simple algebra, hence $\{0\}$ or the whole algebra, and it is not the whole algebra because $\sigma$ is unital; an injective $^*$-homomorphism between $C^*$-algebras is isometric.

The action $\xi\mapsto a\xi b$ without the transpose represents $\bM_n\otimes\bM_n^{\op}$ instead.
It computes the operator norm of the partial transpose, which is not the minimal tensor norm on $\bM_n\otimes\bM_n$.

For each $x$ fix the rank-one spectral projections $f_r^{(x)}\in A_n^{(x)}$, so that $\sum_r f_r^{(x)}=\1$, and set
\begin{equation}
z_x \;=\; \sum_r f_r^{(x)}\otimes \bar f_r^{(x)} \;\in\; \bM_n\otimes_{\min}S_1^n,\qquad \norm{z_x}\le1,
\end{equation}
where the norm bound holds because $z_x$ is the flip of the Choi matrix of the unital completely positive pinching $\Phi_x$, whose completely bounded norm is one, and the flip of the two tensor factors is a complete isometry for the minimal tensor norm.
The tuple $(z_1,\dots,z_k)$ with $\max_x\norm{z_x}\le1$ is an admissible test element for $T\otimes\Phi=\sum_xT_x\otimes\Phi_x$ in \eqref{eq:cb-criterion}.
Indeed, the canonical inclusion $\ell_\infty^k(S_1^n\otimes_{\min}\bM_n)\embeds\ell_\infty^k(S_1^n)\otimes_{\min}\ell_1^k(\bM_n)$ is contractive, and the map depends only on the diagonal components $z_{xx}=(\pi_x\otimes\pi_x)(z)$, each coordinate projection $\pi_x$ being a complete contraction, so that $\max_x\norm{z_{xx}}\le\norm{z}$.
Since the pinching $\Phi_x$ fixes $A_n^{(x)}$ and $T_x$ removes the trace on  $\bar A_n^{(x)}$, we have $T_x(\bar f_r^{(x)})=\bar f_r^{(x)}-\Tr(\bar f_r^{(x)}\tfrac{1}{\sqrt n}\1)\,\tfrac{1}{\sqrt n}\1=\bar f_r^{(x)}-\frac1n\1$, and
\begin{equation}\label{eq:proj}
y_x \;=\; (\Phi_x\otimes T_x)(z_x)
\;=\; \sum_r f_r^{(x)}\otimes(\bar f_r^{(x)}-\frac{1}{n}\1).
\end{equation}
Take the maximally mixed unit vector $\hat{\1}=\tfrac{1}{\sqrt n}\1$ in $S_2^n$.
Since $(\bar f_r^{(x)})^T=(f_r^{(x)})^*=f_r^{(x)}$, the second tensor factor of $y_x$ acts on the right by $f_r^{(x)}-\frac1n\1$.
Using $f_r^{(x)}f_r^{(x)}=f_r^{(x)}$ and $\sum_r f_r^{(x)}=\1$,
\begin{equation}
\sigma(y_x)\,\hat{\1}
=\frac{1}{\sqrt n}(\sum_r f_r^{(x)}f_r^{(x)}-\frac{1}{n}\sum_r f_r^{(x)})
=(1-\frac{1}{n})\hat{\1} .
\label{eq:mub-eig}
\end{equation}
Then for each setting $x$, $\hat{\1}$ is an eigenvector, which is the advantage of the left--right action over the vectorized picture, where the transpose bookkeeping has to be arranged basis by basis.
Under the vectorization $\bC^n\otimes\bC^n\iso S_2^n$, $\ket{s}\ket{t}\mapsto\ketbra{s}{t}$, $\hat{\1}$ gets mapped to the maximally entangled state $\ket{\phi_n^+}$, so the two pictures test the same state.
Summing the $k$ contributions and evaluating at $\hat{\1}$,
\begin{equation}
\norm{\sum_x y_x}_{\bM_n \minten \bM_n}
= \norm{\sum_x \sigma(y_x)}_{B(S_2^n)}
= \sup_{\norm{\xi}_2 \leq 1}\norm{\sum_x \sigma(y_x) \xi}_2
\geq \norm{\sigma(\sum_x y_x)\hat{\1}}_2
= k(1-\frac{1}{n}).
\label{eq:mbound}
\end{equation}

The witness $T$ and the test elements $z_x$ do not depend on $d$, and only the normalization of $T$ does.

\newcommand{\DTwoStatement}[1]{Let $\Phi = (\Phi_x:\bM_n\to \bM_n)_{x=1,\dots, k}$ be the pinching maps onto $k$ mutually unbiased bases of $\bM_n$.
Then, for every $d\ge1$,
\begin{equation}#1
\gamma_{d,\cb}(\Phi) \geq \frac{\sqrt{k}}{\sqrt{d}}(1-\frac{1}{n}) .
\end{equation}
In particular, the multi-channel is not $d$-compatible whenever $d<k(1-\frac{1}{n})^2$.
For $d=1$ this is the case when $k>(\frac{n}{n-1})^2$, and for a full family of $k=n+1$ mutually unbiased bases it is the case for every $d\le n-2$.
}
\begin{theorem}
\label{thm:dtwo}\DTwoStatement{\label{eq:gammad}}
\end{theorem}
See \Cref{app:dtwo-proof} for the proof.

For $n$ a prime power, there are $k=n+1$ mutually unbiased bases \cite{Ivanovic81,WF89}, so the full family of pinchings is not $d$-compatible for any $d\le n-2$.
For the clock--shift pair, $k=2$, the bound $\sqrt2(1-\frac{1}{n})$ exceeds $1$ only for $n\ge4$, and \Cref{prop:mult-domain} is stronger, excluding every $d<n$.

\subsubsection{Two MUB from a finite group}\label{sec:finite-gp}

The clock/shift pair is the case of a construction available for every finite group, in which the second algebra is no longer necessarily commutative.
The two structural properties we are looking for are that the two subalgebras generate the whole matrix algebra and that trace-zero parts of the two are orthogonal to each other, which is also known as commuting square structure in von Neumann algebras.
These properties allow for the proof of \Cref{prop:mult-domain} to hold for the finite group case.
Similar construction is used in \cite{GJL2018tro}.

Let $G$ be a finite group, $N = |G|$ and $\{\ket{g} \}_{g \in G}$ be the canonical ONB of $\ell_2(G)$, so $\cB(\ell_2(G)) = \bM_N$.
Define $A$ to be the multiplication algebra and $B$ the group algebra
\begin{align}
  A &:= \ell_\infty(G) = \{ m_f \mid f:G\to \bC \}, \q\q m_f \ket{g} := f(g) \ket{g}\\
  B &:= \Span \{ \lambda_g \mid g \in G \}, \q\q \lambda_h \ket{g} = \ket{hg}
\end{align}
where $\lambda$ is the left regular representation.
Note that both $A$ and $B$ are unital $*$-subalgebras of $\bM_N$.
We define $\Phi_A, \Phi_B$ to be the conditional expectation onto $A$ and $B$ respectively.
Denote $\tau = \tfrac{1}{N} \Tr$ as the normalized trace.
For a $C^*$-subalgebra $M \subset \bM_N$, the conditional expectation $E_M:\bM_N \to M$ is defined as the unique $\tau$-preserving unital completely positive map such that
\begin{equation} \label{eq:E_M}
  \tau(E_M(z)y) = \tau(zy), \q z \in \bM_N,\ y \in M.
\end{equation}
It is straightforward to show that the following maps satisfy \eqref{eq:E_M}
\begin{align}
  \Phi_A(z) &= \sum_{g\in G} \ip{g}{zg} \ketbra{g}{g}\\
  \Phi_B(z) &= \sum_{g\in G} \tau(\lambda_g^* z) \lambda_g.
\end{align}
For all $a\in A$ and $b \in B$, $\tau$ satisfy (\Cref{lem:gp-orth})
\begin{equation}
  \tau(a b) = \tau(a) \tau(b).
\end{equation}
This says precisely that $a - \tau(a)\1$ and $b - \tau(b)\1$ are Hilbert-Schmidt orthogonal, meaning that the trace-zero parts $\cric{A} = A \ominus \bC\1$ and $\cric{B} = B \ominus \bC\1$ are orthogonal in $S_2^N$.
For $G = \bZ_n$, this is precisely the dephasing $\Phi_x$'s of the clock-shift section.

Moreover, a straightforward calculation shows that $m_{\delta_g}\lambda_{gh^{-1}} = \ketbra{g}{h}$ for any $g,h \in G$, and it follows that $A$ and $B$ generate the matrix algebra $\bM_N$.

Note that $\Phi_A$ and $\Phi_B$ are $d$-factorizable for certain $d$.
Indeed, $\Phi_A$ is the pinching map onto the $\{\ket{g}\}$ basis, which can be seen as a measure and prepare channel.
This shows that $\Phi_A$ is 1-factorizable.
To show the $d$-factorizability of $\Phi_B$, first recall that regular representation contains all irreducible representations, and we can write $B$ as
\begin{equation}
  B \cong \bigoplus_\pi H_\pi \otimes H_\pi^* \cong \bigoplus_\pi \bM_{n_\pi},
\end{equation}
where $\pi$ index is over all irreducible representations $H_\pi$ with $n_\pi = \dim H_\pi$.
This leads to
\begin{equation}
  \sum_{\pi} n_\pi^2 = N,
\end{equation}
and $n_\pi \leq \sqrt{N}$ for any $\pi$.
Due to the inclusion $\bigoplus_\pi \bM_{n_\pi} \subseteq \ell_\infty(\bM_{\floor{\sqrt{N}}})$, $\Phi_B$ is $d$-factorizable for $d=\floor{\sqrt N}$, and hence for every $d\ge\floor{\sqrt N}$, once each block is padded to make the inclusion unital.

\newcommand{\FiniteGPStatement}[1]{
  Let $G$ be a finite group with $N = |G|$, and let $\Phi_A$, $\Phi_B$ be the conditional expectations above.
  Then,
  \begin{enumerate}[(1)]
    \item $\Phi_A$ is entanglement breaking, i.e. $1$-factorizable;
    \item $\Phi_B$ is $d$-factorizable for every $d\ge\floor{\sqrt N}$;
    \item #1 the pair $(\Phi_A, \Phi_B)$ admits no joint $d$-factorization via u.c.p.\ maps for any $d \leq N-1$.
  \end{enumerate}
}
\begin{theorem}\label{thm:finite-gp}
  \FiniteGPStatement{\label{it:finite-3}}
\end{theorem}
See \Cref{app:finite-gp-proof} for the proof.

What is good about this example is that the two channels are individually cheap in terms of resources.
$\Phi_A$ needs a purely classical register and $\Phi_B$ a quantum system of size $\leq \sqrt{N}$, and jointly they need $N$.
So the obstruction is quadratically larger than either constitutent.
Extending the construction to three or more subalgebras with pairwise commuting-square structure is, however, a hard problem.
Pairs of subalgebras whose trace-zero parts are orthogonal, called orthogonal in \cite{Popa1983} and complementary or quasi-orthogonal in \cite{Petz2007}, are well understood, but systems of several pairwise complementary subalgebras are not.
For maximal abelian subalgebras the question is precisely the existence of mutually unbiased bases, whose maximal number is unknown outside prime-power dimensions \cite{Durt2010,BSTW2007}, and for systems mixing commutative and non-commutative subalgebras there are outright non-existence results already in $\bM_4$ and $\bM_{n^2}$ \cite{OPS2007,Weiner2010}.
Hence, it is practically very challenging to find examples of $k$-MUB multi-channels in dimensions that are not prime powers, or of their non-commutative analogues, for $k > 2$.

\section{$d$-preparable Choi state characterization of $d$-compatibility}\label{sec:kPEB}

The main result of this section characterizes the $d$-compatibility of a multi-instrument by $d$-preparable Choi operators.
For a single channel it reduces to the theorem of \cite{Huang2006,Chruscinski2006} that a channel is $d$-partially entanglement breaking exactly when it has a Kraus decomposition by operators of rank at most $d$, together with its factorization form.
Throughout this section $\dim H_A$ and $\dim H_B$ are finite, instruments have a finite outcome set $\mathsf{A}$, and multi-instruments have a finite index set $\mathsf{X}$.

\subsection{$d$-preparable assemblage}\label{sec:d-preparable}

Recall that the \defn{Schmidt rank} $\SR(\ket{\psi})$ of a pure state $\ket{\psi}\in H_A\otimes H_B$ is the number of nonzero coefficients in its Schmidt decomposition, or equivalently the rank of either reduced state.
The \defn{Schmidt number} of a positive operator $\sigma$ on $H_A\otimes H_B$ is
\begin{equation}
  \SN(\sigma) := \min\{\, \max_k \SR(\ket{\psi_k}) \;:\; \sigma=\sum_k p_k\,\pure{\psi_k} \,\},
\end{equation}
the minimum over all decompositions of $\sigma$ into subnormalized pure states.
It does not increase under a channel applied to one side, and it equals $1$ exactly when $\sigma$ is separable.

The Choi states $J(\Phi_{a|x})$ of a multi-instrument form a family of positive operators indexed by an outcome $a$ and a setting $x$.
The following definition asks whether such a family has a common source of bounded Schmidt number.

\begin{definition}[$d$-preparable state assemblage, extending {\cite{Designolle2021,Jones2023}}]\label{def:d-preparable}
Let $H_A$ and $H_B$ be finite-dimensional Hilbert spaces and let $\mathsf A$ and $\mathsf X$ be finite sets.
A \defn{state assemblage} on $H_A\otimes H_B$ is a family $\{\sigma_{a|x}\}_{a\in\mathsf A,\,x\in\mathsf X}$ of positive operators on $H_A\otimes H_B$ such that $\sum_{a\in\mathsf A}\Tr_B\sigma_{a|x}=\rho$ for all $x\in\mathsf X$ and some positive operator $\rho$ on $H_A$ (no signalling condition).
It is \defn{$d$-preparable} if there exist a finite-dimensional Hilbert space $K$, a positive $\sigma \in S_1(H_A\otimes K)$ with $\SN(\sigma)\le d$, and instruments $\cE_x=(\cE_{a|x})_{a\in\mathsf A}:S_1(K)\to\ell_1^{\mathsf A}(S_1(H_B))$, $x\in\mathsf X$, such that
\begin{equation}\label{eq:d-preparable}
  \sigma_{a|x}=(\id\otimes\cE_{a|x})(\sigma)\q\q\forall a\in\mathsf A\q \forall x\in\mathsf X.
\end{equation}
\end{definition}

For $H_B=\bC$ the instruments $\cE_x$ are POVMs $(A_{a|x})_{a\in\mathsf A}$ on $K$, and \eqref{eq:d-preparable} reads $\sigma_{a|x}=\Tr_K[(\1\otimes A_{a|x})\sigma]$.
This is the $n$-preparability of state assemblages of \cite{Designolle2021,Jones2023} with $n=d$.
The definition extends their notion by a quantum output $H_B$ next to the classical outcome $a$.

The definition arises from the scenario of \Cref{fig:scenario}.
There are two parties, a verifier and an untrusted party.
The verifier holds a quantum system $H_A$.
The untrusted party holds a quantum system $K$.
The two systems are in a joint state $\sigma$ on $H_A\otimes K$.
The verifier trusts its own laboratory.
That is, it knows the dimension of every system it holds and it can perform any measurement on these systems.
It knows neither $K$ nor $\sigma$ nor what the untrusted party does.
In each round the verifier chooses a setting $x\in\mathsf X$ and announces it to the untrusted party.
The untrusted party applies an instrument $\cE_x$ to its system $K$.
The instrument produces a classical outcome $a\in\mathsf A$ and an output quantum system $H_B$.
The untrusted party sends the outcome $a$ and the system $H_B$ to the verifier.
At the end of the round the verifier holds the outcome $a$ and the two quantum systems $H_A$ and $H_B$.
The operator $\sigma_{a|x}=(\id\otimes\cE_{a|x})(\sigma)$ of \eqref{eq:d-preparable} describes what the verifier holds.
The verifier can determine the operators $\sigma_{a|x}$ by measuring $H_A\otimes H_B$ over many rounds with the same setting $x$, because it trusts both systems.
So the family $\{\sigma_{a|x}\}$ is known to the verifier, while $K$, $\sigma$ and the instruments $\cE_x$ are not.

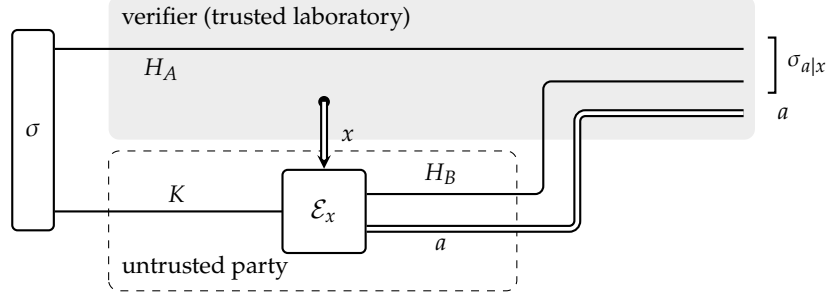
\begin{figure}[tb]
\centering
\begin{tikzpicture}[
  >=stealth,
  qwire/.style={thick},
  cwire/.style={thick, double, double distance=1.4pt},
  box/.style={draw, thick, fill=white, rounded corners=2pt, inner sep=3pt},
  lab/.style={font=\small},
]
  \fill[gray!14, rounded corners=4pt] (1.0,-0.05) rectangle (9.55,1.85);
  \node[lab, anchor=north west] at (1.05,1.83) {verifier (trusted laboratory)};
  \draw[dashed, rounded corners=4pt] (1.0,-2.05) rectangle (6.4,-0.2);
  \node[lab, anchor=south west] at (1.05,-2.03) {untrusted party};

  \node[box, minimum width=0.55cm, minimum height=2.65cm] (src) at (0,0.075) {$\sigma$};

  \draw[qwire] (src.east |- 0,1.15) -- (9.4,1.15);
  \node[lab, below] at (1.7,1.15) {$H_A$};
  \draw[qwire] (src.east |- 0,-1.0) -- (3.3,-1.0);
  \node[lab, above] at (1.9,-1.0) {$K$};

  \node[box, minimum width=1.1cm, minimum height=1.1cm] (E) at (3.85,-1.0) {$\cE_x$};

  \fill (3.85,0.45) circle (2.2pt);
  \draw[cwire, ->] (3.85,0.45) -- (E.north);
  \node[lab, right] at (3.95,0.0) {$x$};

  \draw[qwire, rounded corners=3pt] (E.east |- 0,-0.77) -- (6.75,-0.77) |- (9.4,0.72);
  \node[lab, above] at (5.4,-0.77) {$H_B$};
  \draw[cwire, rounded corners=3pt] (E.east |- 0,-1.23) -- (7.2,-1.23) |- (9.4,0.3);
  \node[lab, below] at (5.4,-1.23) {$a$};

  \draw[thick] (9.72,1.3) -- (9.82,1.3) -- (9.82,0.57) -- (9.72,0.57);
  \node[lab, right] at (9.85,0.935) {$\sigma_{a|x}$};
  \node[lab, right] at (9.72,0.3) {$a$};
\end{tikzpicture}
\caption{$d$-preparable assemblage $\{\sigma_{a|x}\}$ with $\SN(\sigma) \leq d$}
\label{fig:scenario}
\end{figure}

The non-signalling condition has the following meaning in this scenario.
At the end of a round with setting $x$, and averaged over the outcome $a$, the system $H_A$ is in the state $\sum_a\Tr_B\sigma_{a|x}$.
This state equals $\Tr_K\sigma$ for every $x$, because $\cE_x$ is trace preserving.
So it does not depend on $x$.
The physical reason is that $H_A$ never leaves the laboratory of the verifier and the untrusted party acts only on $K$.
The independence of $x$ is the non-signalling condition of the definition, with $\rho=\Tr_K\sigma$, and every family produced in this scenario satisfies it.
The output system $H_B$ may have $x$ dependence since the untrusted party chooses the instrument $\cE_x$ after learning $x$.
In a state assemblage of \cite{Designolle2021,Jones2023} there is no output system, and the non-signalling condition reads $\sum_a\sigma_{a|x}=\rho$ for every $x$.

The assemblage $\{\sigma_{a|x}\}$ is $d$-preparable exactly when it can be produced in this scenario from a shared state $\sigma$ of Schmidt number at most $d$, for some system $K$ and some instruments $\cE_x$ on $K$.
The Schmidt number refers to the cut between $H_A$ and $K$.
The system $K$ and the instruments $\cE_x$ are arbitrary, so the bound $d$ constrains the shared state only.
Suppose the verifier finds that its assemblage is not $d$-preparable.
Then no shared state of Schmidt number at most $d$ produces these states, whatever the system $K$ and the instruments on $K$ are.
The verifier has thereby certified that the state shared with the untrusted party had Schmidt number larger than $d$, trusting only its own laboratory.
For $H_B=\bC$ this is the steering scenario of \cite{Designolle2021,Jones2023}, with the verifier as the trusted party.

\subsection{Connection to quantum steering}\label{sec:steering}
The definition grades steering by the Schmidt number of the shared state, as in \cite{Designolle2021}.
For $d=1$ the shared state is separable, $\sigma=\sum_{\lambda\in\mathsf L}\rho_\lambda\otimes\tau_\lambda$ with a finite set $\mathsf L$ and positive operators $\rho_\lambda$ on $H_A$ and $\tau_\lambda$ on $K$, and \eqref{eq:d-preparable} reads
\begin{equation}\label{eq:lhs}
  \sigma_{a|x}=\sum_{\lambda\in\mathsf L}\rho_\lambda\otimes\tau_{a|x,\lambda},\q\tau_{a|x,\lambda}:=\cE_{a|x}(\tau_\lambda)\ge0,\q\sum_{a\in\mathsf A}\Tr\tau_{a|x,\lambda}=\Tr\tau_\lambda .
\end{equation}
Conversely, every family of the form \eqref{eq:lhs} with positive $\rho_\lambda$ and $\tau_{a|x,\lambda}$ and with $c_\lambda:=\sum_a\Tr\tau_{a|x,\lambda}>0$ independent of $x$ is $1$-preparable.
Take $K=\ell_2^{\mathsf L}$, $\sigma=\sum_\lambda c_\lambda\rho_\lambda\otimes\pure{\lambda}$ and $\cE_{a|x}(\eta)=\sum_\lambda\bra\lambda\eta\ket\lambda\,\tau_{a|x,\lambda}/c_\lambda$.
In \eqref{eq:lhs} the system $H_A$ of the verifier is correlated with the outcome $a$ and the returned system $H_B$ only through the classical variable $\lambda$.
For $H_B=\bC$ this is a local hidden state model, so the $1$-preparable state assemblages are the unsteerable ones \cite{Designolle2021,Jones2023}, and we use the same word for a quantum output.
An assemblage that is not $1$-preparable is steerable, and one that is not $(d-1)$-preparable shows genuine $d$-dimensional steering in the terminology of \cite{Designolle2021}.
The sets of $d$-preparable assemblages increase with $d$.
For $H_B=\bC$, \cite{Jones2023} prove that the assemblage steered from a maximally entangled state by a multi-meter is $d$-preparable exactly when the multi-meter is $d$-simulable, which is $d$-compatibility in our terms.
At $d=1$ this is the equivalence between joint measurability and unsteerability \cite{Quintino2014,Uola2014}.
\Cref{thm:choi-assemblage} below is this statement for multi-instruments, and \Cref{prop:steering-equivalent} extends it to arbitrary state assemblages.

We apply the definition to the Choi states $J(\Phi_{a|x})$ of the components of a multi-instrument.
They arise in the scenario with $K=H_A$, $\sigma=\pure{\chi_A}$ and $\cE_x=\Phi_x$.
That is, the verifier prepares the maximally entangled state $\ket{\chi_A}$, keeps one half, sends the other half to the untrusted party, and the untrusted party applies $\Phi_x$ to it.
Any source $\sigma$ of the Choi states satisfies $\Tr_K\sigma=\1_A$, since taking the partial trace over $H_B$ in \eqref{eq:d-preparable} and summing over $a$ gives $\Tr_K\sigma=\sum_{a\in\mathsf A}\Tr_BJ(\Phi_{a|x})$, and $\cE_x$ and $\Phi_x$ are trace preserving.

So $\sigma$ is itself the Choi state of a channel $S_1(H_A)\to S_1(K)$.
For a multi-meter the Choi states form the state assemblage steered from a maximally entangled state.

\newcommand{\ChoiAssemblageStatement}[1]{
Let $\mathsf A$ and $\mathsf X$ be finite sets and let $H_A$ and $H_B$ be finite-dimensional Hilbert spaces.
Let $\Phi=\{\Phi_x\}_{x\in\mathsf X}$ be a multi-instrument with members $\Phi_x:S_1(H_A)\to\ell_1^{\mathsf A}(S_1(H_B))$.
The following are equivalent.
\begin{enumerate}[(1)]
  \item $\Phi$ is $d$-compatible.
  \item $\Phi$ is $d$-compatible with a common parent into $\ell_1^m(S_1^d)$ for some finite $m$.
  \item There are finite sets $\mathsf J$ and $\mathsf R$, operators $A_j:H_A\to\bC^d$, $j\in\mathsf J$, with $\sum_jA_j^\dagger A_j=\1_A$, and operators $B_{a|x,j,r}:\bC^d\to H_B$, $a\in\mathsf A$, $r\in\mathsf R$, with $\sum_{a,r}B^\dagger_{a|x,j,r}B_{a|x,j,r}=\1_d$ for all $x\in\mathsf X$ and $j\in\mathsf J$, such that
  \begin{equation}#1
    \Phi_{a|x}(\rho)=\sum_{j\in\mathsf J}\sum_{r\in\mathsf R}B_{a|x,j,r}A_j\,\rho\,A_j^\dagger B^\dagger_{a|x,j,r}
    \q\text{for all }x\in\mathsf X,\ a\in\mathsf A\text{ and }\rho\in S_1(H_A).
  \end{equation}
  \item The assemblage $(J(\Phi_{a|x}))_{a\in\mathsf A,\,x\in\mathsf X}$ of Choi states of $\Phi$ is $d$-preparable.
\end{enumerate}
}
\begin{theorem}[Choi-state characterization of $d$-compatible multi-instrument]\label{thm:choi-assemblage}
\ChoiAssemblageStatement{\label{eq:common-kraus}}
\end{theorem}

Condition (3) is a Kraus form of $d$-compatibility.
All members share the Kraus operators $A_j$ of a common first stage into $\bC^d$, and only the second factors depend on the classical label $x$.

See \Cref{app:choi-assemblage-proof} for the proof.

\subsection{One-way LOCC implementation}\label{sec:one-way-LOCC}
The conditions of \Cref{thm:choi-assemblage} are also equivalent to the following.
There exist finite-dimensional Hilbert spaces $K_1,K_2$, a state $\tau$ on $K_1\otimes K_2$ with $\SN(\tau)\le d$, a POVM $(M_\lambda)_{\lambda\in\mathsf L}$ on $H_A\otimes K_1$ with $\mathsf L$ finite, and instruments $\cD_{x,\lambda}:S_1(K_2)\to\ell_1^{\mathsf A}(S_1(H_B))$ such that
\begin{equation}\label{eq:one-way-locc}
  \Phi_x(\rho)=\sum_{\lambda\in\mathsf L}\cD_{x,\lambda}(\Tr_{AK_1}[(M_\lambda\otimes\1_{K_2})(\rho\otimes\tau)])\q\text{for all }x\in\mathsf X\text{ and }\rho\in S_1(H_A).
\end{equation}
In words, $\Phi$ is implementable with a shared state $\tau$ of Schmidt number at most $d$ and one-way classical communication from the party holding the input $\rho$ to the party holding the label $x$.
\Cref{fig:one-way-locc} shows the implementation.

To see this, we first show that the condition (2) in \Cref{thm:choi-assemblage} gives \eqref{eq:one-way-locc} by teleporting the 
$d$-dimensional output of the parent instrument.
Assume $\Phi_x = \sum_j \Psi_{x,j}\circ \theta_j$ for all $x \in \mathsf{X}$.
Let $K_1=K_2=\bC^d$, let $\tau:=\pure{\phi^+_d}_{K_1K_2}$, and let $\ket{\beta_{st}}:=(U_{st}\otimes\1)\ket{\phi^+_d}$, $s,t=0,\dots,d-1$, be the Bell basis of $\bC^d\otimes K_1$ generated by the Weyl unitaries $U_{st}$, so that for every $\xi\in S_1^d$
\begin{equation}\label{eq:teleportation-identity}
  \Tr_{\bC^d\otimes K_1}[(\pure{\beta_{st}}_{\bC^d \otimes K_1} \otimes \1_{K_2})(\xi\otimes\tau)] = \frac{1}{d^2}\,U_{st}^\dagger\xi U_{st}.
\end{equation}
The sender applies the instrument $(\theta_j)_j$ to the input and measures the register together with $K_1$ in the Bell basis.
This is the POVM $M_{(j,s,t)}:=(\theta_j^*\otimes\id_{K_1})(\pure{\beta_{st}})$ on $H_A\otimes K_1$, where $\theta_j^*$ is the adjoint of $\theta_j$, and $\sum_{j,s,t}M_{(j,s,t)}=\sum_j\theta_j^*(\1_d)\otimes\1_{K_1}=\1_A\otimes\1_{K_1}$.
By \eqref{eq:teleportation-identity},
\begin{align}
  \frac{1}{d^2}\,U_{st}^\dagger \theta_j(\rho) U_{st} &= \Tr_{\bC^d\otimes K_1}[(\pure{\beta_{st}}_{\bC^d \otimes K_1} \otimes \1_{K_2})(\theta_j(\rho) \otimes\tau)]\\
  &= \Tr_{AK_1}[(M_{(j,s,t)}\otimes\1_{K_2})(\rho\otimes\tau)].
\end{align}
Defining the receiver's instruments $\cD_{x,(j,s,t)}(\eta):=\Psi_{x,j}(U_{st}\eta U_{st}^\dagger)$ gives
\begin{equation}
  \frac{1}{d^2}\sum_{j,s,t}\cD_{x,(j,s,t)}(U_{st}^\dagger\theta_j(\rho)U_{st})=\sum_j\Psi_{x,j}(\theta_j(\rho))=\Phi_x(\rho),
\end{equation}
which is \eqref{eq:one-way-locc}.

Conversely, we claim \eqref{eq:one-way-locc} gives condition (4).
Replacing $K_2$, $\tau$ and $\cD_{x,\lambda}$ by $K_2\otimes\bC^d$, $\tau\otimes\pure{0}$ and $\cD_{x,\lambda}\circ\Tr_{\bC^d}$, we may assume $\dim K_2\ge d$.
Let $H_{A'}$ be a copy of $H_A$ and apply \eqref{eq:one-way-locc} to the $H_A$ half of $\pure{\chi_A}$ on $H_{A'}\otimes H_A$.
With
\begin{equation}
  \sigma_\lambda:=\Tr_{AK_1}[(\1_{A'}\otimes M_\lambda\otimes\1_{K_2})(\pure{\chi_A}\otimes\tau)],
\end{equation}
a positive operator on $H_{A'}\otimes K_2$, linearity gives $J(\Phi_{a|x})=\sum_\lambda(\id\otimes\cD_{a|x,\lambda})(\sigma_\lambda)$ for the components $\cD_{a|x,\lambda}$ of $\cD_{x,\lambda}$.
Put $K:=K_2\otimes\ell_2^{\mathsf L}$, $\sigma:=\sum_\lambda\sigma_\lambda\otimes\pure{\lambda}$ and $\cE_{a|x}(\eta):=\sum_\lambda\cD_{a|x,\lambda}((\1\otimes\bra{\lambda})\eta(\1\otimes\ket{\lambda}))$, the components of an instrument $\cE_x$, so that \eqref{eq:d-preparable} holds.

It remains to see that $\SN(\sigma)\le d$.
Write $\tau=\sum_jp_j\pure{\tau_j}$ with $\SR(\tau_j)\le d$ and $\ket{\tau_j}=(\1_{K_1}\otimes V_j)\ket{\hat\tau_j}$ for isometries $V_j:\bC^d\to K_2$ and vectors $\hat\tau_j\in K_1\otimes\bC^d$.
Since $M_\lambda$ acts on $H_A\otimes K_1$ only,
\begin{equation}
  \sigma_\lambda=\sum_jp_j(\1_{A'}\otimes V_j)X_{\lambda,j}(\1_{A'}\otimes V_j)^\dagger,
  \q X_{\lambda,j}:=\Tr_{AK_1}[(\1_{A'}\otimes M_\lambda\otimes\1_{\bC^d})(\pure{\chi_A}\otimes\pure{\hat\tau_j})]\ge0 .
\end{equation}
A positive operator on $H_{A'}\otimes\bC^d$ has Schmidt number at most $d$, and the Schmidt number is preserved by isometries on the second factor and does not increase under positive combinations or under tensoring with $\pure{\lambda}$.
Hence $\SN(\sigma_\lambda)\le d$ for every $\lambda$, and $\SN(\sigma)\le d$.

\begin{figure}[tb]
\centering
\begin{tikzpicture}[
  >=stealth,
  qwire/.style={thick},
  cwire/.style={thick, double, double distance=1.4pt},
  box/.style={draw, thick, fill=white, rounded corners=2pt, inner sep=3pt},
  lab/.style={font=\small},
]
  \draw[dashed, rounded corners=4pt] (0.9,0.5) rectangle (5.0,2.4);
  \node[lab, anchor=north west] at (0.95,2.38) {sender};
  \draw[dashed, rounded corners=4pt] (0.9,-2.15) rectangle (8.7,-0.1);
  \node[lab, anchor=south west] at (0.95,-2.13) {receiver};

  \node[box, minimum width=0.55cm, minimum height=2.5cm] (tau) at (0,0.05) {$\tau$};

  \draw[qwire] (-0.3,1.9) -- (3.3,1.9);
  \node[lab, left] at (-0.3,1.9) {$\rho$};
  \node[lab, below] at (1.9,1.9) {$H_A$};

  \draw[qwire] (tau.east |- 0,1.0) -- (3.3,1.0);
  \node[lab, below] at (1.9,1.0) {$K_1$};
  \draw[qwire] (tau.east |- 0,-0.9) -- (5.25,-0.9);
  \node[lab, above] at (2.6,-0.9) {$K_2$};

  \node[box, minimum width=1.1cm, minimum height=1.5cm] (M) at (3.85,1.45) {$M$};

  \draw[cwire, ->, rounded corners=3pt] (M.south) -- (3.85,0.2) -| (5.55,-0.6);
  \node[lab, right] at (3.95,0.42) {$\lambda$};

  \draw[cwire] (-0.3,-1.5) -- (5.25,-1.5);
  \node[lab, left] at (-0.3,-1.5) {$x$};

  \node[box, minimum width=1.3cm, minimum height=1.2cm] (D) at (5.9,-1.2) {$\cD_{x,\lambda}$};

  \draw[qwire] (D.east |- 0,-0.95) -- (8.4,-0.95);
  \node[lab, right] at (8.75,-0.95) {$H_B$};
  \draw[cwire] (D.east |- 0,-1.45) -- (8.4,-1.45);
  \node[lab, right] at (8.75,-1.45) {$a$};
\end{tikzpicture}

\caption{One-way LOCC implementation \eqref{eq:one-way-locc} of a $d$-compatible multi-instrument.
The sender holds the input $\rho$ and the half $K_1$ of the shared state $\tau$ with $\SN(\tau)\le d$, measures $H_A\otimes K_1$ with the POVM $(M_\lambda)_{\lambda\in\mathsf L}$ and sends the outcome $\lambda$ to the receiver.
The receiver holds the label $x$ and the half $K_2$ and applies the instrument $\cD_{x,\lambda}$ to $K_2$.}
\label{fig:one-way-locc}
\end{figure}
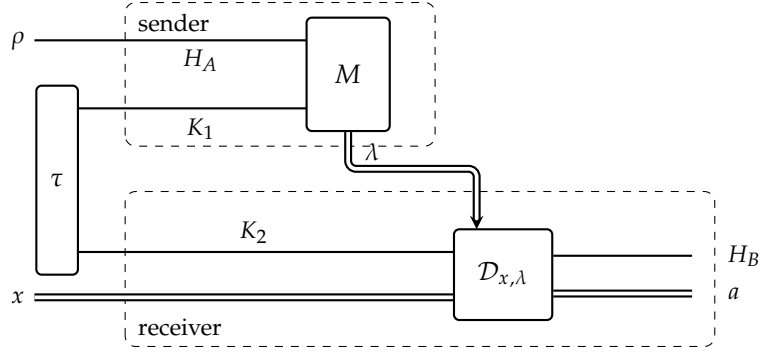

\begin{proposition}[One-to-one mapping between steering and $d$-compatibility problems]\label{prop:steering-equivalent}
A state assemblage is $d$-preparable if and only if the corresponding multi-instrument is $d$-compatible.
\end{proposition}

\begin{proof}
Let $\rho$ be a positive operator on $H_A$ with support projection $P$, let $\rho^{-1/2}$ be the pseudoinverse of $\rho^{1/2}$, so that $\rho^{1/2}\rho^{-1/2}=\rho^{-1/2}\rho^{1/2}=P$, and let $\ket{\chi_\rho}:=(\rho^{1/2}\otimes\1)\ket{\chi_A}$.
The map
\begin{equation}\label{eq:steering-equivalent}
  \Phi=\{\Phi_x\}_{x\in\mathsf X}\longmapsto\{\sigma_{a|x} := (\id\otimes\Phi_{a|x})(\pure{\chi_\rho})\}_{a\in\mathsf A,\,x\in\mathsf X}
\end{equation}
is a bijection from the multi-instruments with members $\Phi_x:S_1(\supp\rho)\to\ell_1^{\mathsf A}(S_1(H_B))$ onto the state assemblages on $H_A\otimes H_B$ with $\sum_a\Tr_B\sigma_{a|x}=\rho$.
Its inverse is
\begin{equation}\label{eq:steering-equivalent-inverse}
  J(\Phi_{a|x})=(\rho^{-1/2}\otimes\1)\,\sigma_{a|x}\,(\rho^{-1/2}\otimes\1)\q\text{for all }a\in\mathsf A\text{ and }x\in\mathsf X.
\end{equation}

Conjugation by $\rho^{\pm1/2}\otimes\1$ acts on the first factor, so it commutes with $\id\otimes\cE$ for every map $\cE$ on the second factor, and it does not increase the Schmidt number with respect to the cut at $H_A$.
In particular $(\id\otimes\Phi_{a|x})(\pure{\chi_\rho})=(\rho^{1/2}\otimes\1)J(\Phi_{a|x})(\rho^{1/2}\otimes\1)$.
These operators are positive, and their partial traces over $H_B$ sum to $\sum_{a\in\mathsf A}\rho^{1/2}\Tr_BJ(\Phi_{a|x})\rho^{1/2}=\rho$, so \eqref{eq:steering-equivalent} maps multi-instruments to state assemblages with reduced state $\rho$.
Conversely, for a state assemblage with reduced state $\rho$ the right-hand side of \eqref{eq:steering-equivalent-inverse} is positive, and its partial traces over $H_B$ sum to $\rho^{-1/2}\rho\rho^{-1/2}=P$, the identity of $\supp\rho$.
By the Choi correspondence it is the family of Choi states of the components of a unique multi-instrument $\Phi$.
The two maps are inverse to each other, because each $\sigma_{a|x}$ satisfies $\sigma_{a|x}=(P\otimes\1)\sigma_{a|x}(P\otimes\1)$: the operators $((\1-P)\otimes\1)\sigma_{a|x}((\1-P)\otimes\1)$ are positive, and their traces sum over $a$ to $\Tr[(\1-P)\rho]=0$, so \eqref{eq:steering-equivalent} is a bijection.
If $\sigma_{a|x}=(\id\otimes\cE_{a|x})(\sigma)$ with $\SN(\sigma)\le d$, then $J(\Phi_{a|x})=(\id\otimes\cE_{a|x})(\sigma')$ for $\sigma':=(\rho^{-1/2}\otimes\1)\sigma(\rho^{-1/2}\otimes\1)$ with $\SN(\sigma')\le d$, so $\Phi$ is $d$-compatible by \Cref{thm:choi-assemblage}.
If $\Phi$ is $d$-compatible, then $J(\Phi_{a|x})=(\id\otimes\cE_{a|x})(\sigma')$ with $\SN(\sigma')\le d$ by \Cref{thm:choi-assemblage}, and $\sigma_{a|x}=(\id\otimes\cE_{a|x})(\sigma)$ for $\sigma:=(\rho^{1/2}\otimes\1)\sigma'(\rho^{1/2}\otimes\1)$ with $\SN(\sigma)\le d$.
The last claim follows from \eqref{eq:steering-equivalent}, because $\pure{\chi_\rho}$ has Schmidt rank $\rank\rho\le\dim H_A$.

\end{proof}

For $H_B=\bC$ the Choi states of a multi-meter are the transposes of its effects, and \eqref{eq:steering-equivalent-inverse} gives the steering-equivalent observables of \cite{Uola_2018}.
The proposition is then the one-to-one mapping between steering and joint measurability problems of \cite{Uola_2018} for $d=1$, and the correspondence between $d$-preparability and $d$-simulability of \cite{Jones2023} for general $d$.
The proposition adds the quantum output $H_B$, and the object that corresponds to an assemblage with quantum output is a multi-instrument.

\subsection{$d$-PEB characterization of $d$-factorizability}\label{sec:d-PEB}

When the family $\Phi$ is a singleton, $d$-compatibility reduces to $d$-factorizability, which is characterized by the $d$-PEB property.
A channel $\Phi$ is \defn{$d$-partially entanglement-breaking} ($d$-PEB) if $\mathrm{SN}(J(\Phi)) \le d$ \cite{Chruscinski2006}.

To extend the $d$-PEB property to instruments $\Phi = \{\Phi_a:S_1(H_A) \to S_1(H_B)\}_{a\in \mathsf{A}}$, we define the \defn{joint Choi matrix}
\begin{equation}\label{eq:joint-choi}
    J(\Phi) = \sum_{a \in \mathsf{A}} \pure{a} \otimes J(\Phi_a).
\end{equation}

This is the Choi operator of $\Phi$ regarded as a channel $S_1(H_A)\to S_1(\ell_2^{\mathsf A}\otimes H_B)$.
Here $\ell_1^{\mathsf A}(S_1(H_B))$ is identified with the operators on $\ell_2^{\mathsf A}\otimes H_B$ that are block diagonal in the register, $(\eta_a)_a\mapsto\sum_a\pure{a}\otimes\eta_a$.
Under this identification an instrument is a channel with block diagonal range.
Its Choi operator is a positive block diagonal operator whose partial trace over the output is $\1_A$, and every positive block diagonal operator with this partial trace is the Choi operator of a unique instrument.
Schmidt ranks and Schmidt numbers of operators on $H_A\otimes\ell_2^{\mathsf A}\otimes H_B$ refer to the cut $H_A:\ell_2^{\mathsf A}\otimes H_B$, and $\Phi$ is $d$-PEB when $\SN(J(\Phi))\le d$, as for channels.

\newcommand{\KPEBStatement}{
Let $\Phi=\{\Phi_a\}_{a\in\mathsf A}$ be an instrument with components $\Phi_a:S_1(H_A)\to S_1(H_B)$.
The following are equivalent.
\begin{enumerate}[(i)]
  \item $\Phi$ is $d$-PEB, that is, $\SN(J(\Phi))\le d$.
  \item $\Phi$ admits a Kraus decomposition
  \begin{equation}
    \Phi_a(\rho)=\sum_{j=1}^mK_{a,j}\rho K_{a,j}^\dagger,\q\rank(K_{a,j})\le d,\q\sum_{a,j}K_{a,j}^\dagger K_{a,j}=\1_A,
  \end{equation}
  for some $m\in\bN$.
  \item $\Phi$ factors through $\ell_1^m(S_1^d)$ via channels for some $m\in\bN$, that is, $\Phi$ is $d$-factorizable in the sense of \Cref{sec:compatible-inst}.
\end{enumerate}
In particular, a channel is $d$-PEB if and only if it has a Kraus decomposition by operators of rank at most $d$, if and only if it factors through $\ell_1^m(S_1^d)$ via channels for some $m$.
For channels the equivalence of (i) and (ii) is due to \cite{Huang2006,Chruscinski2006}, see also \cite{Skowronek2009}.
}
\begin{corollary}[$d$-PEB characterization of $d$-factorizable instruments]\label{cor:k-PEB}
\KPEBStatement
\end{corollary}

See \Cref{app:k-PEB-proof} for the proof.

For $d=1$, \Cref{cor:k-PEB} recovers the entanglement breaking channels, which are exactly the channels that can be simulated by measurement and state preparation and form the free objects of the resource theory of quantum memories \cite{Rosset2018}.
For general $d$, condition (iii) says that $\Phi$ can be simulated with a classical memory together with a $d$-dimensional quantum memory, so the Schmidt number of the Choi state measures the quantum memory that a channel requires~\cite{NamikiTokunaga2012}.
This quantity is what \cite{Engineer2025} call the dimensionality of $\Phi$, and lower bounds on it have recently been certified in prepare-and-measure experiments.

\section{Generalized robustness of $d$-incompatibility}
\label{sec:robustness}\label{subsec:robustness}

The $d$-compatible families form the free set of a resource theory where the quantum memory is restricted to a $d$-dimensional register while the classical memory is unrestricted.
By \eqref{eq:one-way-locc} they are also the families that a sender holding the input and a receiver holding the setting can implement by one-way LOCC assisted by a shared state of Schmidt number at most $d$, so the register dimension is at the same time an entanglement dimension.
Pre- and post-processing by fixed quantum channels and classical relabeling of the family preserve the free set, because each of these operations can be absorbed into a common implementation without enlarging the quantum register.
This places the factorization framework in the resource-theoretic approach to incompatibility \cite{Ji2024,DFK19,Haapasalo15}.

For a closed convex free set the robustness measures are among the standard quantifiers of a resource theory \cite{ChitambarGour2019,Regula18}.
A robustness measure records the least amount of noise, drawn from a chosen noise set, that brings an object into the free set.
When the noise set consists of all objects of the theory the measure is called the \emph{generalized robustness}, and it then depends on nothing but the free set.
For incompatibility it was introduced in \cite{Haapasalo15}.

The role of generalized robustness in resource theory is twofold: one from its definition and one through convex duality.
By definition, it measures the least amount of arbitrary physical noise that makes an object free.
In other words, it quantifies the resource by the noise that an object tolerates.
Moreover, it is a resource monotone, i.e. it vanishes exactly on the free objects, does not increase under free operations, here the pre- and post-processing by quantum channels and the classical relabeling and randomization of the settings described above, and is convex.
So it quantifies the resource in a way that complies with the order structure induced by the free operations.
Without monotonicity, the noise tolerance would be a property of how the object happens to be presented, not of its resource, and it could not be used to compare objects.
With it, the tolerance orders the objects, and that order is respected by free conversions.
This is the reason to prefer it over noise-specific measures, whose orderings disagree across noise sets \cite{DFK19}.

Convex duality adds a second description.
One plus the generalized robustness is the largest value that the object attains on a positive witness that is at most one on every free object.
Operationally, it is the largest ratio between the success probability of the object and that of the best free object in a suitable discrimination task \cite{TR19,TRBLA19}.
For measurement incompatibility the discrimination task behind the generalized robustness is state discrimination \cite{SSC19,UKSYG19}, and for programmable instruments it is the nontransient guessing game of Ji and Chitambar \cite[Sec.~V]{Ji2024}.

Throughout this section $\mathsf X$ is finite and we work in the Heisenberg picture, as in \Cref{sec:criterion}.
A family of channels is then a family of u.c.p.\ maps $\Phi_x:\bM_m\to\bM_n$, $x\in\mathsf X$, and it is $d$-compatible when $\Phi_x=\theta\circ\Psi_x$ for u.c.p.\ maps $\Psi_x:\bM_m\to L_\infty(\Omega;\bM_d)$ and $\theta:L_\infty(\Omega;\bM_d)\to\bM_n$, which is the dual form of the definition in \Cref{sec:compatible-inst}.
We write $\mathcal F_d$ for the set of $d$-compatible families with fixed $m$ and $n$. 
We formulate the results for multi-channels.
The multi-instrument case can be treated by replacing the output by the corresponding block-diagonal classical--quantum output, as in \Cref{rem:criterion-instruments}.

\begin{definition}[Generalized $d$-compatibility robustness]
Let $\Phi=(\Phi_x)_{x\in\mathsf X}$ be a family of u.c.p. maps $\Phi_x:\bM_m\to\bM_n$.
We define its generalized $d$-compatibility robustness as
\begin{align}
\mathcal R_d(\Phi):=\inf\{s\geq0 \mid \left(\frac{\Phi_x+s\Xi_x}{1+s}\right)_{x\in\mathsf X} \text{is $d$-compatible for some u.c.p. family }(\Xi_x)_{x\in\mathsf X}\}.
\label{eq:d-compatibility-robustness}
\end{align}
\end{definition}

The measure $\mathcal R_d$ is the generalized robustness for the free set $\mathcal F_d$, so its monotone properties follow from those of $\mathcal F_d$ by the general theory \cite{Regula18,TR19}.
The set $\mathcal F_d$ is closed by \Cref{thm:closure-instruments}, and since the u.c.p.\ families form a compact set, $\mathcal F_d$ is even compact.
It is convex, because a classical random choice between free implementations can be stored in the classical register without enlarging the quantum register.
It is preserved by the free operations, as noted at the beginning of the section.
Hence $\mathcal R_d(\Phi)=0$ if and only if $\Phi$ is $d$-compatible, $\mathcal R_d$ is convex, and it does not increase under free operations.
In addition, $\mathcal R_d$ is nonincreasing in $d$, since $\mathcal F_d\subseteq\mathcal F_{d+1}$.
For $d=1$ and multi-meters, $\mathcal R_1$ is the generalized incompatibility robustness of measurements \cite{DesignolleFarkasKaniewski2019,BuscemiChitambarZhou2020}.

The rest of the section develops both roles for $\mathcal R_d$.
In \Cref{subsec:robustness-bounds} the factorization functional $\gamma_{d,\cp}$ of \Cref{sec:criterion} bounds $\mathcal R_d$ from above, so the criterion of that section also measures the failure of $d$-compatibility.
\Cref{subsec:robustness-duality} identifies the discrimination task for $d$-incompatibility, which we call the memory-bounded nontransient preparation game.
\Cref{subsec:random-noise} applies the witness bound to the pinching channels onto mutually unbiased bases and to the clock and shift pair.

\subsection{Upper bound from the factorization functional}\label{subsec:robustness-bounds}

The factorization functional $\gamma_{d,\cp}$ gives a natural way to quantify the robustness of $d$-incompatibility.
For a u.c.p.\ family $\Phi$, it satisfies $\gamma_{d,\cp}(\Phi)\geq1$, with equality exactly when $\Phi\in\mathcal F_d$ by \Cref{prop:cp-gamma}.
The inequality holds because $1=\norm{\Phi_x}_{\cb}\le\norm{\theta}_{\cb}\norm{\Psi_x}_{\cb}$ for every factorization $\Phi_x=\theta\circ\Psi_x$.
So $\gamma_{d,\cp}-1$ vanishes exactly on the free set, which makes it a criterion.
For it to be a measure of the resource it must also respect the free operations, for the reason given at the beginning of the section.
The next lemma shows that $\gamma_{d,\cp}$ cannot increase under pre- or post-processing by fixed quantum channels, also when the pre-processing is an instrument whose outcome is forwarded to the post-processing, or under classical relabeling and randomization of the settings.
This has to be checked separately from the monotonicity of $\mathcal R_d$, because the two quantities are related only by the bound \eqref{eq:robustness-gamma-bound} below and not by a functional relation.

\begin{lemma}[Monotonicity of $\gamma_{d,\cp}$]\label{lem:gamma-monotone}
Let $\Lambda:\bM_n\to\bM_{n'}$ and $\Lambda'_x:\bM_{m'}\to\bM_m$, $x\in\mathsf X$, be u.c.p.\ maps, representing quantum channels in the Heisenberg picture.
Then, for every family $\Phi=(\Phi_x:\bM_m \to \bM_n)_x$,
\begin{equation}\label{eq:gamma-monotone}
  \gamma_{d,\cp}(\Lambda\circ\Phi)\leq\gamma_{d,\cp}(\Phi),
  \qquad
  \gamma_{d,\cp}(\Phi\circ\Lambda')\leq\gamma_{d,\cp}(\Phi),
\end{equation}
where $\Lambda\circ\Phi=(\Lambda\circ\Phi_x)_x$ and $\Phi\circ\Lambda'=(\Phi_x\circ\Lambda'_{x})_x$.
More generally, let $\Lambda_z:\bM_n\to\bM_{n'}$, $z\in\mathsf Z$, be c.p.\ maps with $\sum_{z\in\mathsf Z}\Lambda_z$ unital, representing a quantum instrument with finite outcome set $\mathsf Z$ in the Heisenberg picture, and let $\Lambda'_{x,z}:\bM_{m'}\to\bM_m$, $x\in\mathsf X$, $z\in\mathsf Z$, be u.c.p.\ maps.
Then
\begin{equation}\label{eq:gamma-monotone-side}
  \gamma_{d,\cp}(\Phi')\leq\gamma_{d,\cp}(\Phi),
  \qquad
  \Phi'_x:=\sum_{z\in\mathsf Z}\Lambda_z\circ\Phi_x\circ\Lambda'_{x,z}.
\end{equation}
Classical relabeling $\Phi'_y=\Phi_{f(y)}$ by an arbitrary function $f$ and randomization $\Phi'_y=\sum_xp(x|y)\Phi_x$ of the settings also cannot increase $\gamma_{d,\cp}$.
\end{lemma}

Since the lemma is formulated in the Heisenberg picture, the algebraic and the physical order are reversed.
The composition $\Lambda\circ\Phi$ describes the channel $\Lambda_*$ acting on the input state before $\Phi_*$, and $\Phi\circ\Lambda'_x$ describes $\Lambda'_{x,*}$ acting on the output.
The post-processing may depend on the setting, because it acts after the setting has arrived, whereas the pre-processing $\Lambda$ is the same for all settings.
In \eqref{eq:gamma-monotone-side} the pre-processing is the instrument $(\Lambda_{z,*})_z$, and its outcome $z$ is sent past $\Phi$ as classical side information (\Cref{fig:side-channel}), so the post-processing $\Lambda'_{x,z,*}$ may depend on $z$ as well as on $x$.
The two inequalities in \eqref{eq:gamma-monotone} are the case $|\mathsf Z|=1$.

\begin{figure}[tb]
\centering
\begin{tikzpicture}[
  >=stealth,
  qwire/.style={thick},
  cwire/.style={thick, double, double distance=1.4pt},
  box/.style={draw, thick, fill=white, rounded corners=2pt, inner sep=3pt},
  lab/.style={font=\small},
]
  \draw[dashed, rounded corners=4pt] (0.65,-1.35) rectangle (7.0,2.2);
  \node[lab, anchor=north west] at (0.7,2.17) {$\Phi'_x$};

  \node[box, minimum width=1.0cm, minimum height=1.5cm] (L) at (1.35,0) {$\Lambda$};
  \node[box, minimum width=1.2cm, minimum height=1.5cm] (F) at (3.55,0) {$\Phi_x$};
  \node[box, minimum width=1.4cm, minimum height=2.2cm] (P) at (6.05,0) {$\Lambda'_{x,z}$};

  \draw[qwire] (-0.05,0) -- (L.west);
  \node[lab, above] at (0.3,0) {$\bC^{n'}$};

  \draw[qwire] (L.east) -- (F.west);
  \node[lab, below] at (2.37,0) {$\bC^{n}$};

  \draw[cwire, ->, rounded corners=3pt] (L.east |- 0,0.45) -- (2.3,0.45) |- (6.05,1.75) -- (P.north);
  \node[lab, above] at (4.2,1.75) {$z$};

  \draw[qwire] (F.east |- 0,0.3) -- (P.west |- 0,0.3);
  \node[lab, above] at (4.75,0.3) {$\bC^m$};
  \draw[cwire] (F.east |- 0,-0.3) -- (P.west |- 0,-0.3);
  \node[lab, below] at (4.75,-0.3) {$a$};

  \fill (3.55,-1.8) circle (2.2pt);
  \node[lab, left] at (3.45,-1.8) {$x$};
  \draw[cwire, ->] (3.55,-1.8) -- (F.south);
  \draw[cwire, ->] (3.55,-1.8) -| (P.south);

  \draw[qwire] (P.east |- 0,0.3) -- (7.9,0.3);
  \node[lab, right] at (7.9,0.3) {$\bC^{m'}$};
  \draw[cwire] (P.east |- 0,-0.3) -- (7.9,-0.3);
  \node[lab, right] at (7.9,-0.3) {$b$};
\end{tikzpicture}
\caption{The operation $\Phi\mapsto\Phi'$ of \eqref{eq:gamma-monotone-side} in the Schr\"odinger picture ($\Lambda$ refers to the predual $\Lambda_*$ and so do $\Phi_x$ and $\Lambda_{x,z}'$).
The pre-processing instrument $(\Lambda_z)_z$ acts before the setting $x$ is known, and its outcome $z$ bypasses $\Phi_x$ as classical side information, so the post-processing $\Lambda'_{x,z}$ may depend on both $x$ and $z$.
Single lines carry quantum systems and double lines carry classical variables. $a$, $b$ correspond to classical outputs when $\Lambda, \Lambda_{x,z}', \Phi_x$ are considered as instruments.}
\label{fig:side-channel}
\end{figure}
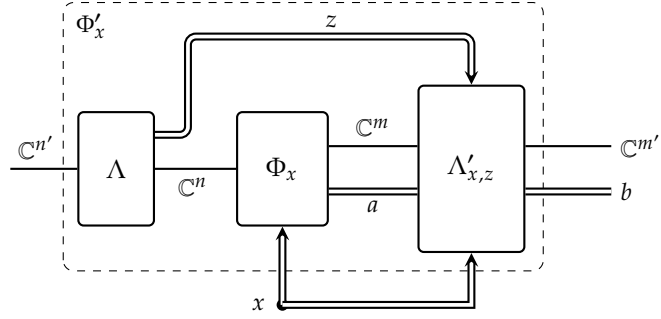

\begin{proof}
Take any c.p.\ factorization $\Phi_x=\theta\circ\Psi_x$ through $L_\infty(\Omega;\bM_d)$.
The family $\Lambda\circ\Phi$ factors through the same algebra with parent map $\Lambda\circ\theta$, and
\begin{equation}
  \norm{\Lambda\circ\theta}_{\cb}
  \leq\norm{\Lambda}_{\cb}\norm{\theta}_{\cb}
  =\norm{\theta}_{\cb},
\end{equation}
since a u.c.p.\ map is a complete contraction.
Similarly, $\Phi_x\circ\Lambda'_{x}=\theta\circ(\Psi_x\circ\Lambda'_{x})$, with $\norm{\Psi_x\circ\Lambda'_{x}}_{\cb}\leq\norm{\Psi_x}_{\cb}$.
In either case the cost $\norm{\theta}_{\cb}\sup_x\norm{\Psi_x}_{\cb}$ does not increase, and taking the infimum in \eqref{eq:gamma-cp} proves \eqref{eq:gamma-monotone}.
For a randomized setting $\Phi'_y=\sum_x p(x|y)\Phi_x$, use the same parent and the maps $\Psi'_y=\sum_x p(x|y)\Psi_x$.
These satisfy $\sup_y\norm{\Psi'_y}_{\cb}\leq\sup_x\norm{\Psi_x}_{\cb}$.
The deterministic case $p(x|y)=\delta_{x,f(y)}$ is the relabeling $\Phi'_y=\Phi_{f(y)}$ by a function $f$, which need not be bijective, so settings may be dropped or repeated, and a bijective $f$ only renames the settings.

For \eqref{eq:gamma-monotone-side}, put $\Omega':=\Omega\times\mathsf Z$, so that $L_\infty(\Omega';\bM_d)=\ell_\infty^{\mathsf Z}(L_\infty(\Omega;\bM_d))$, and define the c.p.\ maps
\begin{equation}
  \theta'((f_z)_z):=\sum_{z\in\mathsf Z}\Lambda_z(\theta(f_z)),
  \qquad
  \Psi'_x(y):=(\Psi_x(\Lambda'_{x,z}(y)))_{z\in\mathsf Z}.
\end{equation}
Then
\begin{equation}
\theta'\circ\Psi'_x=\sum_z\Lambda_z\circ\Phi_x\circ\Lambda'_{x,z}=\Phi'_x.
\end{equation}
Since the norm of a c.p. map is attained at $\1$ \cite[Prop.~3.6]{Paulsen_2003}, $\theta(\1)\le\norm{\theta}_{\cb}\1$ and $\sum_z\Lambda_z$ is unital, we get
\begin{equation}
  \norm{\theta'}_{\cb}=\norm{\sum_z\Lambda_z(\theta(\1))}\le\norm{\theta}_{\cb},
\end{equation}
and
\begin{equation}
\norm{\Psi'_x}_{\cb}=\max_z\norm{\Psi_x(\Lambda'_{x,z}(\1))}=\norm{\Psi_x}_{\cb}.
\end{equation}
So the cost does not increase.
\end{proof}

If $\theta$ and $\Psi_x$ are u.c.p., then so are $\theta'$ and $\Psi'_x$, so the operation $\Phi\mapsto\Phi'$ of \eqref{eq:gamma-monotone-side} maps $\mathcal F_d$ into $\mathcal F_d$.
It is linear and maps u.c.p.\ families to u.c.p.\ families, so it maps a feasible pair $(s,\Xi)$ in \eqref{eq:d-compatibility-robustness} for $\Phi$ to the feasible pair $(s,\Xi')$ for $\Phi'$, and $\mathcal R_d(\Phi')\le\mathcal R_d(\Phi)$.
The side information has to be classical.
If a quantum system $\bC^k$ is forwarded instead, the same construction only gives a factorization through $L_\infty(\Omega;\bM_{kd})$, and $\mathcal R_d$ can increase: routing the input around $\Phi$ through $\bC^n$ turns any free family into copies of $\id_n$, and since every member of a $d$-compatible family is $d$-factorizable, their robustness is at least $\mathcal R_d(\id_n)=n/d-1>0$ by \eqref{eq:identity-robustness} when $d<n$.

Thus $\gamma_{d,\cp}-1$ is a faithful resource monotone on u.c.p.\ families, although it may take the value $+\infty$.
It is finite exactly when every member of the family is $d$-partially entanglement breaking ($d$-PEB).
Indeed, a c.p.\ factorization $\Phi_x=\theta\circ\Psi_x$ through $L_\infty(\Omega;\bM_d)$ writes the Choi operator of $(\Phi_x)_*$ as
\begin{equation}
  J((\Phi_x)_*) = \int_\Omega(\id\otimes(\Psi_{x,\omega})_*)(\sigma_\omega)\,d\mu(\omega),
\end{equation}
where $(\sigma_\omega)_\omega$ is the Choi operator of $\theta_*$, an integrable family of positive operators on $\bC^n\otimes\bC^d$, so the Choi operator of $(\Phi_x)_*$ has Schmidt number at most $d$ and $\Phi_x$ is $d$-PEB.
Conversely, each $d$-PEB member has a u.c.p.\ factorization by \Cref{cor:k-PEB}, and finitely many such factorizations combine into a common one on the disjoint union $\Omega=\bigsqcup_x\Omega_x$ of their classical registers, with parent $\theta(f):=\sum_x\theta_x(f|_{\Omega_x})$, of c.b.\ norm $|\mathsf X|$, and with $\Psi_x$ placed in the $x$-th summand, so that $\gamma_{d,\cp}(\Phi)\le|\mathsf X|$ whenever every member is $d$-PEB.
For $d=1$ this says that $\gamma_{1,\cp}$ is finite exactly for families of entanglement-breaking channels \cite{HSR2003}.
Completing a c.p.\ factorization to a unital one supplies a noise family that makes $\Phi$ $d$-compatible, giving the following bound on $\mathcal R_d$.

\begin{theorem}\label{prop:gamma-bounds-Rg}
The c.p. factorization functional provides a direct upper bound on $\mathcal R_d(\Phi)$:
\begin{align}
\mathcal R_d(\Phi)\leq\gamma_{d,\cp}(\Phi)-1.
\label{eq:robustness-gamma-bound}
\end{align}
\end{theorem}
\begin{proof}
We may assume that $\gamma_{d,\cp}(\Phi)$ is finite, since otherwise the bound is immediate.
Indeed, fix $K>\gamma_{d,\cp}(\Phi)$.
By \cref{eq:gamma-cp}, there exists a c.p. factorization $\Phi_x=\theta\circ\Psi_x$ through $L_\infty(\Omega;\bM_d)$ which, after rescaling, satisfies $\norm{\theta}_{\cb}\leq1$ and $\sup_x\norm{\Psi_x}_{\cb}\leq K$.
Since the maps are c.p., this implies $\theta(\1)\leq\1$ and $\Psi_x(\1)\leq K\1$.
Choose normal states $\omega$ on $L_\infty(\Omega;\bM_d)$ and $\tau_x$ on $\bM_m$, and define
\begin{align}
\widetilde\theta(y)&:=\theta(y)+\omega(y)\bigl(\1-\theta(\1)\bigr),\qquad
\widetilde\Psi_x(z):=K^{-1}\Psi_x(z)+\tau_x(z)\bigl(\1-K^{-1}\Psi_x(\1)\bigr).
\label{eq:ucp-completions}
\end{align}
Both $\widetilde\theta$ and $\widetilde\Psi_x$ are u.c.p., and hence $\Lambda_x:=\widetilde\theta\circ\widetilde\Psi_x$ defines a $d$-compatible family.
Moreover, $K\Lambda_x-\Phi_x$ is completely positive and
\begin{align}
(K\Lambda_x-\Phi_x)(\1)=(K-1)\1.
\end{align}
Therefore, for $K>1$,
\begin{align}
\Xi_x:=\frac{K\Lambda_x-\Phi_x}{K-1}
\end{align}
is u.c.p. and satisfies
\begin{align}
\frac{\Phi_x+(K-1)\Xi_x}{K}=\Lambda_x.
\end{align}
It follows that $\mathcal R_d(\Phi)\leq K-1$.
\end{proof}

Note that The gap in \eqref{eq:robustness-gamma-bound} can be infinite.
For the identity channel on $\bM_n$ and $d<n$, $\gamma_{d,\cp}(\id_n)=+\infty$ because $\SN(J(\id_n))=n$, while $\mathcal R_d(\id_n)=n/d-1$ by \eqref{eq:identity-robustness} below.

The bound \eqref{eq:robustness-gamma-bound} can also be derived from the witness form of $\mathcal R_d$ in \Cref{prop:robustness-dual}, the analogue of \cite[Thm.~4]{Ji2024},
\begin{equation}
  1 + \cR_d(\Phi) = \sup \{ \ip{W}{\Phi} \mid W \geq 0, \ip{W}{\sigma} \leq 1 \text{ for all $d$-compatible $\sigma$ } \}
\end{equation}
is the factorization functional \eqref{eq:gamma-cp} except the constraint extends to all self-adjoint $W$ instead of positive $W$, and $\ip{W}{\sigma}\le1$ is imposed for all $\sigma$ in the larger set $\hat\cF_d$ of \eqref{eq:C-set}.

We will make this connection more apparent in the next section.

\subsection{Memory-bounded nontransient preparation game}\label{subsec:robustness-duality}

We first describe a game where the quantum memory of the player is bounded, and state the duality between the generalized robustness and the advantage in this game.
We then let the referee penalize the player.
This signed version of the game is dual to the factorization functional, and comparing the two dualities identifies the gap in \eqref{eq:robustness-gamma-bound}.

\begin{figure}[tb]
\centering
\begin{tikzpicture}[
  >=stealth,
  qwire/.style={thick},
  cwire/.style={thick, double, double distance=1.4pt},
  box/.style={draw, thick, fill=white, rounded corners=2pt, inner sep=3pt},
  lab/.style={font=\small},
]
  \fill[gray!14, rounded corners=4pt] (1.0,0.15) rectangle (11.9,2.8);
  \node[lab, anchor=north west] at (1.05,2.78) {referee};
  \draw[dashed, rounded corners=4pt] (1.0,-2.3) rectangle (8.55,-0.05);
  \node[lab, anchor=south west] at (1.05,-2.28) {player};

  \fill[gray!25] (5.45,-2.05) rectangle (6.25,-0.3);
  \node[lab, anchor=south] at (5.85,-2.02) {delay};

  \node[box, minimum width=0.55cm, minimum height=1.5cm] (src) at (1.75,1.25) {$\phi^+$};

  \node[box, minimum width=1.3cm, minimum height=1.2cm] (I) at (3.95,-1.05) {\small instrument};
  \node[box, minimum width=1.3cm, minimum height=1.2cm] (C) at (7.6,-1.05) {\small channel};
  \node[box, minimum width=1.1cm, minimum height=1.5cm] (M) at (9.95,1.35) {};
  \draw[thick] ([xshift=-0.32cm,yshift=-0.22cm]M.center) arc (180:0:0.32cm);
  \draw[thick,->] ([yshift=-0.22cm]M.center) -- ([xshift=0.26cm,yshift=0.3cm]M.center);
  \node[lab, above] at (M.north) {$\{M_x,\1-M_x\}$};

  \draw[qwire] (src.east |- 0,1.65) -- (M.west |- 0,1.65);
  \node[lab, above] at (3.4,1.65) {$\bC^n$};

  \draw[qwire, rounded corners=3pt] (src.east |- 0,0.85) -- (2.55,0.85) |- (I.west |- 0,-0.9);
  \node[lab, anchor=east] at (2.42,-0.5) {$\bC^n$};

  \draw[qwire] (I.east |- 0,-0.8) -- (C.west |- 0,-0.8);
  \node[lab, above] at (5.1,-0.8) {$\bC^d$};
  \draw[cwire] (I.east |- 0,-1.3) -- (C.west |- 0,-1.3);
  \node[lab, below] at (5.1,-1.3) {$j$};

  \fill (7.6,0.75) circle (2.2pt);
  \draw[cwire, ->] (7.6,0.75) -- (C.north);
  \node[lab, right] at (7.7,0.25) {$x$};

  \draw[qwire, rounded corners=3pt] (C.east |- 0,-1.05) -- (8.85,-1.05) |- (M.west |- 0,1.05);
  \node[lab, right] at (8.95,0.3) {$\bC^m$};

  \draw[cwire] (M.east |- 0,1.35) -- (11.7,1.35);
  \node[lab, above right] at (10.6,1.4) {win if $M_x$};
\end{tikzpicture}
\caption{Memory-bounded nontransient preparation game.}
\label{fig:delayed-game}
\end{figure}
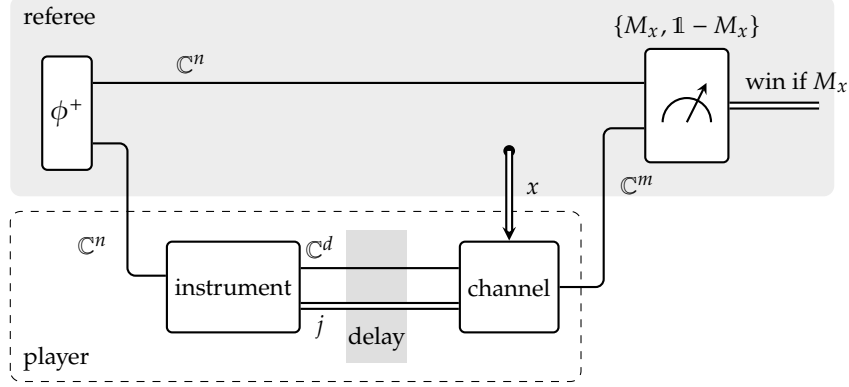

\paragraph{Memory-bounded nontransient preparation game.}
The referee prepares a maximally entangled state $\phi^+_n$ on $\bC^n\otimes\bC^n$, keeps one half, and sends the other half to the player with $d$-dimensional quantum memory.
After a delay referee announces a uniformly chosen setting $x$, and the player returns an $m$-dimensional quantum system.
The referee jointly measures the returned system and the retained half with a two-outcome measurement $\{M_x, \1 - M_x\}$.
The player wins if the outcome corresponds to $M_x$.
The player may apply any instrument to the received half before the delay, but can carry only a $d$-dimensional quantum system and a classical outcome $j$ across it.
After the delay the player applies a channel that may depend on $x$ and on $j$, and returns its output.
These strategies are exactly the $d$-compatible families of \Cref{sec:compatible-inst}.
Note that the game is nontransient in the sense of the nontransient guessing game in \cite{Ji2024}, since the setting arrives only after a delay, and the memory which survives the delay is bounded by $d$.
The two games differ in what is due after the delay.
There the player returns the quantum system before the delay and only a classical guess after it, so a player without quantum memory may still apply any instrument to the input, the identity included, and keeps its classical outcome for the guess.
Here the quantum system itself is due after the delay, so a player without quantum memory must measure before the delay and prepare afterwards.
Both games coincide on multi-meters, reducing at $d=1$ to state discrimination with post-measurement information \cite{Carmeli2018,Buscemi2020}.

Formally, a game is specified by a tuple $M=(M_x)_{x\in\mathsf X}$ of observables $0\leq M_x\leq\1$ on $\bC^n\otimes\bC^m$.
The player is described by a u.c.p.\ family $\Phi=(\Phi_x:\bM_m\to\bM_n)_{x\in\mathsf X}$, one map for each setting, whose preduals $(\Phi_x)_*:S_1^n\to S_1^m$ are the channels applied to the received half.
The winning probability of this player is
\begin{equation}\label{eq:game-score}
  w(\Phi;M):=\frac{1}{|\mathsf X|}\sum_{x\in\mathsf X}\Tr[M_x(\id\otimes(\Phi_x)_*)(\phi^+_n)]
  =\frac{1}{|\mathsf X|}\ip{M}{\Phi}.
\end{equation}
For a tuple $W=(W_x)_{x\in\mathsf X}$ of matrices $W_x\in\bM_n\otimes\bM_m$ we write
\begin{equation}\label{eq:game-pairing}
  \ip{W}{\Phi}
  :=\sum_{x\in\mathsf X}\Tr[W_x(\id\otimes(\Phi_x)_*)(\phi^+_n)]
  =\frac1n \sum_{x\in\mathsf X}\Tr[W_xJ((\Phi_x)_*)],
\end{equation}
the Choi state dual pairing.
On elementary tensors $e \otimes f$ the pairing reads
\begin{equation}\label{eq:choi-pairing-transpose}
  \frac1n \Tr[(e\otimes f)J((\Phi_x)_*)] = \frac1n \Tr[e^T \Phi_x(f)],
\end{equation}
where $e^T$ is the transpose in the basis defining $\chi_n$.

For a self-adjoint tuple $W$, define
\begin{equation}\label{eq:robustness-free-value}
  \nu_d(W):=\sup_{\Sigma\in\mathcal F_d}\ip{W}{\Sigma}.
\end{equation}
The best winning probability with a $d$-dimensional quantum memory is
\begin{equation}
w^*(M):=\sup_{\Sigma\in\mathcal F_d}w(\Sigma;M)=\frac{\nu_d(M)}{|\mathsf X|},
\end{equation}
and we call the ratio $w(\Phi;M)/w^*(M)=\ip{M}{\Phi}/\nu_d(M)$ the \defn{advantage} of a family $\Phi$ in the game $M$, the factor by which it beats every memory-bounded player.
For measurements, the advantage in state discrimination is the quantity that the generalized robustness controls, and for states it is the advantage in channel or subchannel discrimination \cite{TR19,UKSYG19}.
The advantage $\ip{M}{\Phi}/\nu_d(M)$ is unchanged when the observables are rescaled, $M_x\mapsto cM_x$ with $c>0$, and conversely every tuple $W$ of positive matrices becomes a game after division by $c\geq\max_x\norm{W_x}$.

Hence, it suffices to optimize the advantage over games and over positive tuples $W$.

Let $\cC_{\cp}$ denote the cone of families $\Delta=(\Delta_x)_{x\in\mathsf X}$ of completely positive maps $\Delta_x:\bM_m\to\bM_n$.

\begin{proposition}[Robustness duality; cf.\ {\cite[Thm.~4]{Ji2024}}]\label{prop:robustness-dual}
For every u.c.p.\ family $\Phi$,
\begin{equation}\label{eq:robustness-dual}
  1+\mathcal R_d(\Phi)
  =\sup_{M}\frac{w(\Phi;M)}{w^*(M)}
  =\sup\{\ip{W}{\Phi}\mid W_x\geq0\ \forall x,\ \nu_d(W)\leq1\},
\end{equation}
the first supremum over all games.
In words, $1+\mathcal R_d(\Phi)$ is the largest advantage of $\Phi$ over every player with a $d$-dimensional nontransient quantum memory.
\end{proposition}

\begin{proof}
Let 
\begin{equation}
D:=\mathcal F_d-\cC_{\cp}=\{\Sigma-\Delta\mid\Sigma\in\mathcal F_d,\ \Delta\in\cC_{\cp}\}.
\end{equation}
We first show that
\begin{equation}\label{eq:one-plus-Rd}
1+\mathcal R_d(\Phi)
= \inf \{1+s \mid s\geq 0,\ \frac{\Phi+s\Xi}{1+s} \in \cF_d \; \Xi \text{ u.c.p} \}
=\inf\{t>0\mid\Phi\in tD\},
\end{equation}
i.e. $1+\cR_d$ is the Minkowski functional of $D$.
Let $S$ and $T$ be the feasible sets on the left and right of \eqref{eq:one-plus-Rd}.
If $(\Phi+s\Xi)/(1+s)=\Sigma\in\mathcal F_d$ with $\Xi$ u.c.p., then $(1+s)\Sigma-\Phi=s\Xi\in\cC_{\cp}$, so $\Phi\in(1+s)D$ and $S\subseteq T$.
Conversely, let $\Phi=t\Sigma-\Delta$ with $\Sigma\in\mathcal F_d$ and $\Delta\in\cC_{\cp}$.
Evaluating at $\1$ gives $\Delta_x(\1)=(t-1)\1$, so $t\geq1$.
If $t>1$, then $\Xi:=\Delta/(t-1)$ is u.c.p.\ and $(\Phi+(t-1)\Xi)/t=\Sigma$, so $s=t-1$ is feasible in $S$.
If $t=1$, then $\Delta_x(\1)=0$ forces $\Delta_x=0$, because $0\leq\Delta_x(z)\leq\norm{z}\Delta_x(\1)$ for $z\geq0$, and hence $\Phi=\Sigma\in\mathcal F_d$ and $\mathcal R_d(\Phi)=0$, showing $T \subseteq S$, which proves the claim.

The set $D$ is convex, it contains $0$, and it is closed, being the sum of the compact set $\mathcal F_d$ and the closed cone $-\cC_{\cp}$.
By the bipolar theorem, in the real vector space of ${*}$-preserving families with the pairing \eqref{eq:game-pairing}, the Minkowski functional of $D$ is the support function of its polar,
\begin{equation}
  \inf\{t>0\mid\Phi\in tD\}=\sup\{\ip{W}{\Phi}\mid W\in D^\circ\},
  \qquad D^\circ=\{W\mid\ip{W}{X}\leq1\ \text{for all}\ X\in D\}.
\end{equation}
It remains to identify $D^\circ$.
If $W\in D^\circ$, then $\ip{W}{\Sigma}-r\ip{W}{\Delta}\leq1$ for all $\Sigma\in\mathcal F_d$, $\Delta\in\cC_{\cp}$ and $r>0$.
Letting $r\to\infty$ gives $\ip{W}{\Delta}\geq0$ for all $\Delta\in\cC_{\cp}$, that is $W_x\geq0$ for all $x$, and $\Delta=0$ gives $\nu_d(W)\leq1$.
Conversely, $W_x \geq 0$ for any $x$ and $\nu_d(W) \leq 1$ imply $\ip{W}{\Sigma-\Delta}\leq\nu_d(W)\leq1$.
This proves
\begin{equation}
1+\mathcal R_d(\Phi) = \sup \{ \ip{W}{\Phi} \mid W_x \geq 0\; \forall x,\nu_d(W) \leq 1\}.
\end{equation}
For the equality with the supremum over games, note that $\nu_d(W)>0$ whenever $W\geq0$ is nonzero, because the completely depolarizing family $N_0$, $N_{0,x}(z)=\Tr(z)\1/m$, is $1$-compatible, hence lies in $\mathcal F_d$, and $\ip{W}{N_0}=\frac{1}{nm}\sum_x\Tr W_x>0$.
Since the constraint $\nu_d(W)\leq1$ and the objective scale linearly in $W$, the constrained supremum equals the supremum of the ratios $\ip{W}{\Phi}/\nu_d(W)$ over nonzero positive tuples, which by definition is the supremum of the advantage over all games.
\end{proof}

By \eqref{eq:robustness-dual}, the upper bound $1 + \cR_d(\Phi) \leq \gamma_{d,\cp}(\Phi)$ also bounds the largest advantage.
Conversely, every single game gives a lower bound on $\mathcal R_d$ through its advantage.
\Cref{prop:robustness-dual} can be compared to an earlier result in \cite{sudarsanan2024higher}, where the authors relate the robustness of incompatibility to the task of distinguishing higher-order processes.

For a family with a single member, $d$-compatibility is $d$-factorizability of that channel, and $\mathcal R_d$ is its robustness of $d$-factorizability.
We compute it for the partial depolarizing channel $\Phi_p$ of \Cref{ex:depolarizing} on $\bM_n$, which is the identity channel at $p=1$.
The computation rests on the threshold \eqref{eq:partial-depol}, and its purpose is the universal bound of \Cref{cor:universal-bound} below.
Consider the game with the single observable $M=\phi^+_n$, whose winning probability is the entanglement fidelity,
\begin{equation}
w(\Phi;\phi^+_n)=\ip{\phi^+_n}{\Phi}=\Tr[\phi^+_n(\id\otimes\Phi_*)(\phi^+_n)].
\end{equation}
For $\Phi_p$ it equals $p+(1-p)/n^2$, and \eqref{eq:partial-depol} says that $\Phi_p$ is $d$-factorizable exactly when this fidelity is at most $d/n$, which equals the free value of the game $\nu_d(\phi^+_n)$.
To see this, note that the Choi state of a $d$-factorizable channel has Schmidt number at most $d$ by \Cref{cor:k-PEB}.
A state $\rho$ of Schmidt number at most $d$ has singlet fraction $\Tr[\phi^+_n\rho]\leq d/n$ \cite{Terhal2000}, and $\Phi_p$ at the threshold attains this value.
Hence \Cref{prop:robustness-dual} gives
\begin{equation}
1+\mathcal R_d(\Phi_p)\geq\frac nd(p+\frac{1-p}{n^2}).
\end{equation}
Conversely, let $\Xi$ be the channel with normalized Choi state $(\1-\phi^+_n)/(n^2-1)$, which is a channel because this state has reference marginal $\1/n$, and whose entanglement fidelity is zero.
The mixture $(\Phi_p+s\Xi)/(1+s)$ has an isotropic Choi state, so it is a partial depolarizing channel $\Phi_{p'}$, with entanglement fidelity $(p+(1-p)/n^2)/(1+s)$.
For
\begin{equation}
  s=\frac nd (p+\frac{1-p}{n^2})-1\geq0
\end{equation}
this fidelity is $d/n$, that is $p'=(dn-1)/(n^2-1)$, and the mixture is $d$-factorizable by \eqref{eq:partial-depol}.
When this $s$ is negative, $\Phi_p$ is itself $d$-factorizable by \eqref{eq:partial-depol}.
The two bounds agree, so
\begin{equation}\label{eq:depolarizing-robustness}
  \mathcal R_d(\Phi_p)
  =\max \{ 0,\ \frac nd (p+\frac{1-p}{n^2})-1 \}
  =\max \{0,\ \frac{p(n^2-1)-(dn-1)}{dn}\},
\end{equation}
whose zero set is the threshold \eqref{eq:partial-depol}.
The noise that makes $\Phi_p$ $d$-factorizable is exactly the noise that lowers its entanglement fidelity to $d/n$.
At $p=1$,
\begin{equation}\label{eq:identity-robustness}
  \mathcal R_d(\id_n)=\frac nd-1 ,
\end{equation}
and for $d=1$ this is the generalized robustness $n-1$ of the maximally entangled state.

\begin{corollary}[Universal bound]\label{cor:universal-bound}
Every u.c.p.\ family $\Phi=(\Phi_x:\bM_m\to\bM_n)_{x\in\mathsf X}$ satisfies
\begin{equation}\label{eq:universal-bound}
  \mathcal R_d(\Phi)\le\frac nd-1 ,
\end{equation}
with equality for the family $(\id_n)_{x\in\mathsf X}$ of copies of the identity.
\end{corollary}
\begin{proof}
The computation of \eqref{eq:identity-robustness} applies to the family of copies of the identity, with the game $M_x=\phi^+_n$ for every $x$ and the noise family of copies of $\Xi$, and gives the robustness $n/d-1$.
Every family with input dimension $n$ is a setting-dependent post-processing of it, $\Phi_x=\id_n\circ\Phi_x$, and $\mathcal R_d$ is nonincreasing under such post-processing, since composing a $d$-compatible family $(\id_n+s\Xi_x)/(1+s)$ with $\Phi_x$ gives the $d$-compatible family $(\Phi_x+s\,\Xi_x\circ\Phi_x)/(1+s)$.

\end{proof}

The family of copies of the identity channel $\id_n$ is thus the most $d$-incompatible resource with input dimension $n$, and $\mathcal R_d$ measures the distance from it.

\paragraph{The signed preparation game.}
We now let the referee penalize the player.
A signed game is specified by a tuple $S=(S_x)_{x\in\mathsf X}$ of self-adjoint matrices on $\bC^n\otimes\bC^m$ with $\norm{S_x}\leq1$.
Write $S_x=S_x^+-S_x^-$ for the Jordan decomposition, with $S_x^\pm\geq0$ and $S_x^+S_x^-=0$, so that $S_x^++S_x^-=|S_x|\leq\1$.
The referee jointly measures the returned system and the retained half with the three-outcome measurement $\{S_x^+,S_x^-,\1-|S_x|\}$.
The player wins on the first outcome with payoff $+1$, loses on the second with payoff $-1$, and neither wins nor loses on the third with $0$ payoff.
The rest of the game is unchanged.
The expected payoff of a player $\Phi$ is
\begin{equation}\label{eq:signed-game-score}
  w(\Phi;S) := \frac{1}{|\mathsf X|}\ip{S}{\Phi}=w_+(\Phi)-w_-(\Phi),
  \qquad w_\pm(\Phi) := \frac{1}{|\mathsf X|}\ip{S^\pm}{\Phi},
\end{equation}
the winning probability minus the losing probability.
A game of the previous paragraphs is the case $S_x\geq0$, in which the player never loses and the expected payoff is the winning probability, and so are the discrimination tasks of \mbox{\cite{Ji2024,TR19,TRBLA19}}.
Every self-adjoint tuple becomes a signed game after division by $c\geq\max_x\norm{S_x}$.
For a signed game $\nu_d(S)$ can be zero or negative, so the advantage is not defined as a ratio, and the duality takes the constrained form on the right of \eqref{eq:robustness-dual}.
The factorization functional admits a dual description of exactly this shape, and comparing the two dual programs identifies the gap in \eqref{eq:robustness-gamma-bound}.
Let
\begin{equation}\label{eq:C-set}
  \hat\cF_d:=\{\sigma\in\cC_{\cp}\mid\gamma_{d,\cp}(\sigma)\leq1\},
\end{equation}
extending $\cF_d$ from unital to sub-unital families.
By \Cref{prop:cp-gamma} this is the set of c.p.\ families admitting a factorization through some $L_\infty(\Omega;\bM_d)$ with contractive c.p.\ maps, and a unital family lies in $\hat\cF_d$ if and only if it is $d$-compatible.

The set $\hat\cF_d$ is convex and closed, and by \Cref{lem:cp-polar-form} and \Cref{lem:gauge-polar} for every c.p.\ family $\Phi$
\begin{equation}\label{eq:gamma-polar}
  \gamma_{d,\cp}(\Phi)=\sup\{\ip{S}{\Phi}\mid S=(S_x)_x\ \text{self-adjoint},\ \ip{S}{\sigma}\leq1\ \text{for all}\ \sigma\in \hat\cF_d\},
\end{equation}
the polar form of \Cref{thm:cp-criterion}.

In the language of the game, with $w(\Phi;S)=\ip{S}{\Phi}/|\mathsf X|$ as in \eqref{eq:signed-game-score} for every self-adjoint tuple $S$, \eqref{eq:gamma-polar} reads
\begin{equation}\label{eq:gamma-game}
  \gamma_{d,\cp}(\Phi)=|\mathsf X|\sup\{w(\Phi;S)\mid S\ \text{self-adjoint},\ \ip{S}{\sigma}\leq1\ \text{for all}\ \sigma\in\hat\cF_d\},
\end{equation}
the analogue of $1+\mathcal R_d(\Phi)=|\mathsf X|\sup\{w(\Phi;W)\mid W\geq0,\ \nu_d(W)\leq1\}$ in \eqref{eq:robustness-dual}.
The proof of \eqref{eq:gamma-polar} is the same as that of \eqref{eq:robustness-dual}, the bipolar theorem applied to $\hat\cF_d$ in place of $\mathcal F_d-\cC_{\cp}$.
The supremum in \eqref{eq:gamma-game} runs over self-adjoint tuples of every norm.
This is needed, because signed games with $\norm{S_x}\leq1$ have $w(\Phi;S)\leq1$, while $\gamma_{d,\cp}(\id_n)=+\infty$ for $d<n$.

The two dual programs differ in the positivity of the witness and in the normalization set, $\hat\cF_d$ in place of $\mathcal F_d$, and the second difference is immaterial for positive witnesses.
Indeed, each $\sigma\in\hat\cF_d$ factors through contractive c.p.\ maps, and the completions \eqref{eq:ucp-completions} with $K=1$ turn this factorization into a $d$-compatible family $\Sigma$ with $\Sigma-\sigma$ completely positive.
So $\ip{W}{\sigma}\leq\ip{W}{\Sigma}\leq\nu_d(W)$ for every $W\geq0$.

Hence
\begin{equation}\label{eq:gap-positivity}
  1+\mathcal R_d(\Phi)=\sup\{\ip{W}{\Phi}\mid W\geq0,\ \ip{W}{\sigma}\leq1\ \text{for all}\ \sigma\in \hat\cF_d\}\leq\gamma_{d,\cp}(\Phi),
\end{equation}
which is a second proof of \eqref{eq:robustness-gamma-bound} and identifies the gap as the positivity constraint on the witness.

\subsection{Mutually unbiased bases and the clock--shift example}\label{subsec:random-noise}

Since the generalized robustness optimizes over the noise, a lower bound on $\mathcal R_d$ is a lower bound on the robustness against every noise set, white noise included, and the converse fails.
For MUB measurements, the white-noise robustness has been bounded, and computed exactly in some cases, in Ref.~\cite{DesignolleSkrzypczykFrowisBrunner2019}.
Here we bound the generalized robustness of the pinching channels onto mutually unbiased bases of \Cref{sec:m-MUB} from both sides.

Let $\Phi=(\Phi_x)_{x=1}^k$ be the pinching channels onto $k$ mutually unbiased bases of $\bM_n$ as in \Cref{sec:m-MUB},
\begin{equation}
  \Phi_x(z) = \sum_r f_r^{(x)}zf_r^{(x)},
\end{equation}
where $f_r^{(x)}$ are the minimal projections of the maximal abelian algebra $A_n^{(x)}$, whose trace-zero parts $P_n^{(x)}$ are mutually orthogonal in $S_2^n$.

For a family $\Lambda = (\Lambda_x:\bM_n \to \bM_n)_x$ of linear maps, define
\begin{equation}
  Y_{\Lambda} := \sum_x(\Lambda_x\otimes T_x)(z_x),
\end{equation}
where $T_x$ is the Hilbert--Schmidt orthogonal projection onto $\bar P_n^{(x)}$ and $z_x=\sum_rf_r^{(x)}\otimes\overline{f_r^{(x)}}$, as in \Cref{sec:m-MUB}.
Explicitly, $Y_\Lambda=\sum_x\sum_r\Lambda_x(f_r^{(x)})\otimes(\overline{f_r^{(x)}}-\frac1n\1)$.
Define
\begin{equation}\label{eq:mub-witness}
  W_x:=\frac{1}{\sqrt k}(T_x\otimes\id)(J(\Phi_x))
  =\frac{1}{\sqrt k}\sum_{r=1}^n(\overline{f_r^{(x)}}-\frac{\1}{n})\otimes f_r^{(x)},
  \qquad x=1,\dots,k,
\end{equation}
where $J(\Phi_x)=\sum_r\overline{f_r^{(x)}}\otimes f_r^{(x)}$ is the Choi matrix of $\Phi_x$, which is the flip of $z_x$, and the bar is entrywise complex conjugation in the basis defining $\chi_n$.
By \eqref{eq:choi-pairing-transpose}, for every family $\Lambda$,
\begin{equation}\label{eq:witness-value}
  \ip{W}{\Lambda}
  =\frac{1}{n\sqrt k}\sum_{x=1}^k\sum_{r=1}^n\Tr [(f_r^{(x)}-\frac1n\1)\,\Lambda_x(f_r^{(x)})]
  =\frac{1}{\sqrt k}\ip{\hat{\1}}{\sigma(\sum_x(\Lambda_x\otimes T_x)(z_x))\hat{\1}},
\end{equation}
using $\ip{\hat{\1}}{\sigma(a\otimes b)\hat{\1}}=\Tr(ab^T)/n$ and $(\overline{f_r^{(x)}})^T=f_r^{(x)}$ with $\sigma$ being the left--right action \eqref{eq:sigma-action}.

We have the chain
\begin{equation}\label{eq:chain}
  \norm{T \otimes \Lambda} \geq \norm{Y_\Lambda} = \norm{\sigma(Y_\Lambda)}_{B(S_2^n)} \geq \abs{\ip{\hat{\1}}{\sigma(Y_\Lambda)\hat{\1}}} = \sqrt{k}\, \abs{\ip{W}{\Lambda}},
\end{equation}
where the first inequality is the admissibility of the tuple $(z_x)_x$ in \eqref{eq:cb-criterion}, shown in \Cref{sec:m-MUB}, and the second tests the operator norm on the unit vector $\hat{\1}$.
For the pinching family the last inequality is an equality.
Indeed,
\begin{equation}
\sigma(y_x(\Phi))\xi=\Phi_x(\xi)-\frac1n\xi, 
\end{equation}
so on $S_2^n$,
\begin{equation}
\sigma(Y_{\Phi})=\sum_x\Phi_x-\frac kn\,\id,
\end{equation}
a sum of $k$ orthogonal projections with pairwise intersection $\bC\1$ and mutually orthogonal trace-zero parts, whose eigenvalues are $k(1-\frac1n)$ on $\hat{\1}$, $1-\frac kn$ on $\bigoplus_xP_n^{(x)}$, and $-\frac kn$ on the orthogonal complement.
Hence
\begin{equation}
\norm{Y_{\Phi}}=k(1-\frac1n)=\sqrt k\,\ip{W}{\Phi}.
\end{equation}

Using \eqref{eq:gamma-ratio} and $\norm{T}_{(d)} \leq \sqrt{k}\sqrt{d}$ from \Cref{prop:Td-norm}, we recover \eqref{eq:gammad},
\begin{equation}
  \gamma_{d,\cb}(\Phi)
  \geq \frac{\norm{T\otimes \Phi}}{\norm{T}_{(d)}}
  \geq \frac{\ip{W}{\Phi}}{\sqrt{d}}
  =\sqrt{\frac kd}\Bigl(1-\frac1n\Bigr).
\end{equation}

Adding a constant to each component makes the witness positive,
\begin{equation}\label{eq:mub-witness-positive}
  W_x^+:=W_x+\frac{1}{n\sqrt k}\,\1\otimes\1=\frac{1}{\sqrt k}\sum_{r=1}^n\overline{f_r^{(x)}}\otimes f_r^{(x)}\ \geq 0,
\end{equation}
using $\sum_rf_r^{(x)}=\1$, and $\ip{W^+}{\Xi}=\ip{W}{\Xi}+\sqrt k/n$ for every unital family $\Xi$.
The positive completion turns the witness into a preparation game, making the connection with \eqref{eq:robustness-dual} explicit.
The observables $M_x:=\sqrt k\,W_x^+=\sum_r\overline{f_r^{(x)}}\otimes f_r^{(x)}$ are projections, so $M=(M_x)_x$ is a game, and it is the following one.
The referee chooses $x$ and $r$ uniformly, sends the state $f_r^{(x)}$, and reveals $x$ after the storage delay.
The player returns a state that the referee tests with $f_r^{(x)}$.
The winning probability of a u.c.p.\ family $\Lambda$ is
\begin{equation}\label{eq:mub-game-score}
  w_{k-\rm MUB}(\Lambda) := w(\Lambda;M)
  =\frac{1}{kn}\sum_{x,r}\Tr[\Lambda_x(f_r^{(x)})f_r^{(x)}]
  =\frac{\ip{W}{\Lambda}}{\sqrt k}+\frac1n,
\end{equation}
by \eqref{eq:game-score} and \eqref{eq:witness-value}.
For the pinching family, $\Phi_x(f_r^{(x)})=f_r^{(x)}$ for all $x$ and $r$, so $w_{k-\rm MUB}(\Phi)=1$.
For every $d$-compatible family $\Sigma$,
\begin{equation}\label{eq:w-mub-bound}
w_{k-\rm MUB}(\Sigma)\leq \sqrt{\frac{d}{k}}+\frac1n
\end{equation}
because $\sqrt k\,\ip{W}{\Sigma}\leq\norm{T\otimes\Sigma}\leq\norm{T}_{(d)}\,\gamma_{d,\cb}(\Sigma)\leq\sqrt k\sqrt d$, by \eqref{eq:chain}, \eqref{eq:gamma-ratio}, and $\gamma_{d,\cb}(\Sigma)\leq1$ from \Cref{prop:cb-unit-ball}.

\begin{proposition}[Robustness of MUB pinching channels]\label{prop:mub-robustness}
For the pinching family $\Phi$ onto $k$ mutually unbiased bases of $\bM_n$,
\begin{equation}\label{eq:Rd-mub}
  \max\left\lbrace0,
    \frac{\sqrt k(1-1/n)-\sqrt d}{\sqrt d+\sqrt k/n}\right\rbrace
  \leq\mathcal R_d(\Phi)
  \leq\min\left\lbrace k-1,\ \frac nd-1\right\rbrace.
\end{equation}
\end{proposition}

\begin{proof}
By \Cref{prop:robustness-dual} and \eqref{eq:w-mub-bound}
\begin{equation}
  1 + \cR_d(\Phi) = \sup_{M'}\frac{w(\Phi;M')}{w^*(M')} \geq \frac{w_{k-\mathrm{MUB}}(\Phi)}{\sup_{\Sigma \in \cF_d} w_{k-\mathrm{MUB}}(\Sigma)} \geq \frac{1}{\sqrt{d/k} + 1/n},
\end{equation}
and rearranging gives the lower bound.

On the upper bound, the one of the min values $n/d-1$ is from \Cref{cor:universal-bound}.
To show the bound $k-1$, consider the c.p.\ factorization of $\Phi$ through $\ell_\infty^{kn}=L_\infty(\Omega;\bM_1)$ with $\Omega=\{1,\dots,k\}\times\{1,\dots,n\}$, given by
\begin{equation}
  \Psi_x(z)(x',r):=\delta_{xx'}\Tr(f_r^{(x)}z),
  \qquad
  \theta(g):=\sum_{x',r}g(x',r)\,f_r^{(x')}.
\end{equation}
Indeed
\begin{equation}
  \theta \circ \Psi_x (z) = \sum_{x',r} \delta_{x,x'} \Tr[f_r^{(x)}z]f_r^{(x')} = \sum_r \Tr[f_r^{(x)}z]f_r^{(x)} = \Phi_x(z).
\end{equation}
Furthermore, $\norm{\Psi_x}_{\cb}=\norm{\Psi_x(\1)}_\infty=1$ and $\norm{\theta}_{\cb}=\norm{\theta(\1)}=k$ since a c.p.\ map attains its c.b.\ norm at the unit \cite[Prop.~3.6]{Paulsen_2003}.
Hence $\gamma_{1,\cp}(\Phi)\leq k$, and $\gamma_{d,\cp}\leq\gamma_{1,\cp}$ because a factorization through $\bM_1$ is one through a corner of $\bM_d$.
\end{proof}

The bound $k-1$ is obtained with white noise.
The mixture $\frac1k\Phi+\frac{k-1}{k}N_0$, where $N_{0,x}(z)=\Tr(z)\1/n$ is the completely depolarizing family, is $1$-compatible, being the strategy that measures in a uniformly random basis $x'$, prepares the outcome if $x'$ equals the announced setting $x$, and prepares $\1/n$ otherwise, so $k-1$ also bounds the white-noise robustness.

The white-noise robustness computed in \cite{DesignolleSkrzypczykFrowisBrunner2019} is that of the basis measurements rather than of the measure-and-prepare channels $\Phi_x$, and it corresponds to the case $d=1$ here, since the family $\eta\Phi+(1-\eta)N_0$ is $1$-compatible exactly when the depolarized measurements $\{\eta f_r^{(x)}+(1-\eta)\1/n\}_r$ are jointly measurable.

The clock--shift pair of \Cref{sec:clock-shift} gives a concrete instance of these bounds.
The two bases are the computational basis and its Fourier transform.
Their pinching channels are
\begin{equation}\label{eq:clockshift-pinchings}
  \Phi_Z(z)=\frac1n\sum_{a=0}^{n-1}Z^azZ^{-a},
  \qquad
  \Phi_X(z)=\frac1n\sum_{a=0}^{n-1}X^azX^{-a}.
\end{equation}
Taking $k=2$ and $d=1$ gives
\begin{equation}\label{eq:Rd-clockshift}
  \mathcal R_1(\Phi_Z,\Phi_X)
  \geq\max\left\lbrace0,
    \frac{\sqrt2(1-1/n)-1}{1+\sqrt2/n}\right\rbrace,
\end{equation}
together with the upper bound $\mathcal R_1(\Phi_Z,\Phi_X)\leq\sqrt n/(\sqrt n+2)$, the white-noise robustness of a pair of mutually unbiased bases \cite{DesignolleSkrzypczykFrowisBrunner2019}, which is also the white-noise robustness of $(\Phi_Z,\Phi_X)$ at $d=1$.
For example, at $n=4$ the lower bound is $(8\sqrt2-11)/7\approx0.0448$.
It is positive for $n\geq4$ and tends to $\sqrt2-1$ as $n\to\infty$.
For $n=2,3$ it gives only the trivial bound, although the algebraic argument in \Cref{prop:mult-domain} still proves incompatibility.
More generally, the lower bound in \eqref{eq:Rd-mub} is positive whenever $d<k(1-1/n)^2$, so it quantifies the same range of memory dimensions excluded by \Cref{thm:dtwo}.

\section{Conclusions and future directions}
We formulated the compatibility of quantum measurements, channels and instruments as a common factorization through $S_1^d\otimes^\wedge L_1(\Omega)$, a $d$-dimensional quantum register together with an unrestricted classical register.
The least such $d$ is the quantum memory that a common implementation has to keep until the setting $x$ is supplied.
The operator-space formulation gave four groups of results.
First, the $d$-compatible multi-instruments with a separable input space and finite-dimensional outputs form a closed set for an arbitrary set of settings (\Cref{thm:closurePOVM,thm:closure-instruments}).
This holds even when the limit has no parent with finitely many outcomes.
Second, operator-space duality expresses the c.b.\ and c.p.\ compatibility functionals as suprema over witnesses (\Cref{thm:cb-criterion,thm:cp-criterion}).
For the pinching channels onto $k$ mutually unbiased bases of $\bC^n$ these witnesses exclude every $d<k(1-1/n)^2$ (\Cref{thm:dtwo}).
For the clock--shift pair and its finite-group analogue an algebraic argument excludes every $d$ below the full dimension (\Cref{prop:mult-domain,thm:finite-gp}), although each channel alone needs a much smaller register.
Third, in finite dimensions and for finitely many settings and outcomes, $d$-compatibility is equivalent to the $d$-preparability of the Choi state assemblage (\Cref{thm:choi-assemblage}) and to an implementation by one-way LOCC with a shared state of Schmidt number at most $d$ (\Cref{sec:one-way-LOCC}).
The bijection of \Cref{prop:steering-equivalent} identifies the $d$-compatible multi-instruments with the $d$-preparable state assemblages of any fixed invertible reduced state.
Fourth, one plus the generalized robustness $\mathcal R_d$ is the largest advantage over all players with a $d$-dimensional quantum memory in the memory-bounded nontransient preparation game (\Cref{prop:robustness-dual}).
The robustness $\mathcal R_d$ is at most $\gamma_{d,\cp}-1$ (\Cref{prop:gamma-bounds-Rg}), and it is at most $n/d-1$ for every family with input dimension $n$ (\Cref{cor:universal-bound}).

\paragraph{Future directions.}
Several questions remain open.
The bijection of \Cref{prop:steering-equivalent} suggests a resource theory of steering with quantum outputs, whose free objects include the $d$-preparable assemblages with a fixed reduced state.
We do not develop this theory in full here, and there is room to develop it along the lines of the resource theory of steering \cite{Gallego2015}.
For multi-instruments, \Cref{fig:side-channel} shows an example of an operation that should be free.
Another direction is to generalize the factorization framework to higher-order quantum processes such as quantum combs \cite{Chiribella_2008,Chiribella_2009}.
The closure theorems allow a separable input space but assume a finite-dimensional output.
We do not prove the case of an infinite-dimensional output here.
Likewise, the Choi characterization, the criteria and the robustness duality are proved only in finite dimensions.
Their extension to continuous-variable systems is open.

\vspace{3em}

\paragraph*{Acknowledgments.}
FL was supported by National Science Foundation Grant No.~2426103.
MJ is partially supported by NSF DMS 2247114 and NSF DMS 2556195.
YZ acknowledges the support of the Institute for Quantum Computing, the University of Waterloo, and the Perimeter Institute. Research at the Institute for Quantum Computing is supported by Innovation, Science and Economic Development Canada.

\printbibliography

@article{Heisenberg_1927,
	author = {Heisenberg, W. },
	da = {1927/03/01},
	doi = {10.1007/BF01397280},
	id = {Heisenberg1927},
	isbn = {0044-3328},
	journal = {Zeitschrift f{\"u}r Physik},
	number = {3},
	pages = {172--198},
	title = {{\"U}ber den anschaulichen Inhalt der quantentheoretischen Kinematik und Mechanik},
	ty = {JOUR},
	url = {https://doi.org/10.1007/BF01397280},
	volume = {43},
	year = {1927},
}

@article{sudarsanan2024higher,
doi = {10.1088/1367-2630/ad93f6},
url = {https://doi.org/10.1088/1367-2630/ad93f6},
year = {2024},
month = {dec},
publisher = {IOP Publishing},
volume = {26},
number = {12},
pages = {123003},
author = {Sudarsanan Ragini, Nidhin and Sazim, Sk},
title = {Higher-order incompatibility improves distinguishability of causal quantum networks},
journal = {New Journal of Physics},
eprint = {2405.20080},
primaryclass = {quant-ph},
archiveprefix = {arXiv}
}

@article{Bohr_1928,
	author = {Bohr, N. },
	da = {1928/04/01},
	doi = {10.1038/121580a0},
	id = {BOHR1928},
	isbn = {1476-4687},
	journal = {Nature},
	number = {3050},
	pages = {580--590},
	title = {The Quantum Postulate and the Recent Development of Atomic Theory1},
	ty = {JOUR},
	url = {https://doi.org/10.1038/121580a0},
	volume = {121},
	year = {1928}
}

@article{Heinosaari_2016,
   title={An invitation to quantum incompatibility},
   volume={49},
   ISSN={1751-8121},
   url={http://dx.doi.org/10.1088/1751-8113/49/12/123001},
   DOI={10.1088/1751-8113/49/12/123001},
   number={12},
   journal={Journal of Physics A: Mathematical and Theoretical},
   publisher={IOP Publishing},
   author={Heinosaari, Teiko and Miyadera, Takayuki and Ziman, Mário},
   year={2016},
   month=feb, pages={123001}
}

@article{Uola_2018,
    author = {Uola, Roope and Budroni, Costantino and G\"uhne, Otfried and Pellonp\"a\"a, Juha-Pekka},
    title = {{One-to-one mapping between steering and joint measurability problems}},
    journal = {Physical Review Letters},
    volume = {115},
    number = {23},
    pages = {230402},
    year = {2015},
    doi = {10.1103/PhysRevLett.115.230402},
    eprint = {1507.08633},
    archivePrefix = {arXiv}
}

@article{Zhang2025costofsimulating,
  doi = {10.22331/q-2025-10-31-1902},
  url = {https://doi.org/10.22331/q-2025-10-31-1902},
  title = {Cost of {S}imulating {E}ntanglement in {S}teering {S}cenarios},
  author = {Zhang, Yujie and Zhang, Jiaxuan and Chitambar, Eric},
  journal = {{Quantum}},
  issn = {2521-327X},
  publisher = {{Verein zur F{\"{o}}rderung des Open Access Publizierens in den Quantenwissenschaften}},
  volume = {9},
  pages = {1902},
  month = oct,
  year = {2025}
}

@article{Skrzypczyk2020,
  ids     = {SkrzypczykHobanSainzLinden2020},  % duplicate keys used in the .tex files
  doi = {10.1103/physrevresearch.2.023292},
  url = {https://doi.org/10.1103%2Fphysrevresearch.2.023292},
  year = 2020,
  month = {jun},
  publisher = {American Physical Society ({APS})},
  volume = {2},
  number = {2},
  pages = {023292},
  author = {Paul Skrzypczyk and Matty J. Hoban and Ana Bel{\'{e}}n Sainz and Noah Linden},
  title = {Complexity of compatible measurements},
  journal = {Physical Review Research}
}

@book{pisier1986factorization,
  title={Factorization of Linear Operators and Geometry of Banach Spaces},
  author={Pisier, G. and Conference Board of the Mathematical Sciences},
  isbn={9780821807101},
  lccn={lc85018605},
  series={Regional conference series in mathematics},
  url={https://books.google.com/books?id=5SbaBwAAQBAJ},
  year={1986},
  publisher={Conference Board of the Mathematical Sciences}
}

@article{Pisier1998,
  author = {Pisier, Gilles},
  journal = {Ast{\'{e}}risque},
  publisher = {Soci{\'{e}}t{\'{e}} Math{\'{e}}matique de France},
  title = {{Non-commutative vector valued $L_p$-spaces and completely $p$-summing maps}},
  volume = {247},
  year = {1998}
}

@book{Pisier_2003,
    place={Cambridge},
    series={London Mathematical Society Lecture Note Series},
    title={Introduction to Operator Space Theory},
    publisher={Cambridge University Press},
    author={Pisier, Gilles},
    year={2003},
    collection={London Mathematical Society Lecture Note Series}
}

@article{Chruscinski2006,
archivePrefix = {arXiv},
arxivId = {quant-ph/0511244},
author = {Chru{\'{s}}ci{\'{n}}ski, Dariusz and Kossakowski, Andrzej},
doi = {10.1007/s11080-006-7264-7},
eprint = {quant-ph/0511244},
issn = {12301612},
journal = {Open Systems \& Information Dynamics},
number = {1},
pages = {17--26},
primaryClass = {quant-ph},
title = {{On partially entanglement breaking channels}},
volume = {13},
year = {2006}
}

@book{Effros2000,
  title={Operator Spaces},
  author={Effros, E.G. and Ruan, Z.J.},
  isbn={9780198534822},
  lccn={lc00056685},
  series={London Mathematical Society monographs},
  url={https://books.google.com/books?id=v7mj8Dy84k8C},
  year={2000},
  publisher={Clarendon Press}
}

@article{Heinrich_1980,
author = {Heinrich, Stefan},
journal = {Journal f\"{u}r die reine und angewandte Mathematik},
pages = {72-104},
title = {Ultraproducts in Banach space theory.},
url = {http://eudml.org/doc/152195},
volume = {313},
year = {1980},
}

@book{Diestel1977,
  title={Vector Measures},
  author={Diestel, J. and Uhl, J.J.},
  isbn={9780821815151},
  lccn={lc77009625},
  series={Mathematical surveys and monographs},
  url={https://books.google.com/books?id=fEKZAwAAQBAJ},
  year={1977},
  publisher={American Mathematical Society}
}

@book{Cohn1994,
  title={Measure Theory},
  author={Cohn, D.L.},
  isbn={9780817630034},
  lccn={80014768},
  url={https://books.google.com/books?id=vRxV2FwJvoAC},
  year={1994},
  publisher={Birkh{\"a}user Boston}
}

@article{Ozawa_1984,
    author = {Ozawa, Masanao},
    title = {Quantum measuring processes of continuous observables},
    journal = {Journal of Mathematical Physics},
    volume = {25},
    number = {1},
    pages = {79-87},
    year = {1984},
    month = {01},
    issn = {0022-2488},
    doi = {10.1063/1.526000},
    url = {https://doi.org/10.1063/1.526000},
    eprint = {https://pubs.aip.org/aip/jmp/article-pdf/25/1/79/19265867/79_1_online.pdf},
}

@article{Chiribella_2008,
  title={Quantum circuits architecture},
  author={Chiribella, Giulio and D'Ariano, Giacomo Mauro and Perinotti, Paolo},
  journal={Physical Review Letters},
  volume={101},
  number={6},
  pages={060401},
  year={2008}
}

@misc{Kim-Zhang2026,
  author       = {Yujie Zhang},
  title        = {Numerical Bounds on Quantum-Channel $d$-Compatibility},
  year         = {2026},
  howpublished = {Python source code},
  url          = {https://github.com/yujie4phy/corner-channel-bounds}
}

@article{Chiribella_2009,
  title={Theoretical framework for quantum networks},
  author={Chiribella, Giulio and D'Ariano, Giacomo Mauro and Perinotti, Paolo},
  journal={Physical Review A},
  volume={80},
  number={2},
  pages={022339},
  year={2009}
}

@Article{Achenbach_2025,
author={Achenbach, Tim
and Bluhm, Andreas
and Lepp{\"a}j{\"a}rvi, Leevi
and Nechita, Ion
and Pl{\'a}vala, Martin},
title={Factorization of multimeters: a unified view on nonclassical quantum phenomena},
journal={Letters in Mathematical Physics},
year={2026},
month={May},
day={21},
volume={116},
number={3},
pages={56},
issn={1573-0530},
doi={10.1007/s11005-026-02088-2},
url={https://doi.org/10.1007/s11005-026-02088-2}
}

@book{Paulsen_2003,
    place={Cambridge},
    series={Cambridge Studies in Advanced Mathematics},
    title={Completely Bounded Maps and Operator Algebras},
    publisher={Cambridge University Press},
    author={Paulsen, Vern},
    year={2003},
    collection={Cambridge Studies in Advanced Mathematics}
}

@article{Skowronek2009,
  title={Cones of positive maps and their duality relations},
  author={Lukasz Skowronek and Erling St{\o}rmer and Karol Życzkowski},
  journal={Journal of Mathematical Physics},
  year={2009},
  volume={50},
  pages={062106},
  url={https://api.semanticscholar.org/CorpusID:9085540}
}

@phdthesis{JungeHab,
  author = {Junge, Marius},
  title  = {Factorization Theory for Spaces of Operators},
  type   = {Habilitationsschrift},
  school = {Universit\"at Kiel},
  year   = {1996},
}

@book{rudin1987real,
  title={Real and Complex Analysis},
  author={Rudin, Walter},
  isbn={978-0070542340},
  edition={3rd},
  year={1987},
  publisher={McGraw-Hill},
  address={New York}
}

@article{Ji2024,
  ids     = {JiChitambar2024},  % duplicate keys used in the .tex files
  title   = {Incompatibility as a resource for programmable quantum instruments},
  author  = {Ji, Kaiyuan and Chitambar, Eric},
  journal = {PRX Quantum},
  volume  = {5},
  number  = {1},
  pages   = {010340},
  year    = {2024},
  publisher = {American Physical Society},
  doi     = {10.1103/PRXQuantum.5.010340},
  eprint  = {2112.03717},
  archivePrefix = {arXiv},
  primaryClass  = {quant-ph}
}

@article{Regula18,
  title   = {Convex geometry of quantum resource quantification},
  author  = {Regula, Bartosz},
  journal = {Journal of Physics A: Mathematical and Theoretical},
  volume  = {51},
  number  = {4},
  pages   = {045303},
  year    = {2018},
  publisher = {IOP Publishing},
  doi     = {10.1088/1751-8121/aa9100},
  eprint  = {1707.06298},
  archivePrefix = {arXiv},
  primaryClass  = {quant-ph}
}

@article{TR19,
  title   = {General Resource Theories in Quantum Mechanics and Beyond: Operational Characterization via Discrimination Tasks},
  author  = {Takagi, Ryuji and Regula, Bartosz},
  journal = {Physical Review X},
  volume  = {9},
  number  = {3},
  pages   = {031053},
  year    = {2019},
  publisher = {American Physical Society},
  doi     = {10.1103/PhysRevX.9.031053},
  eprint  = {1901.08127},
  archivePrefix = {arXiv},
  primaryClass  = {quant-ph}
}

@article{TRBLA19,
  title   = {Operational Advantage of Quantum Resources in Subchannel Discrimination},
  author  = {Takagi, Ryuji and Regula, Bartosz and Bu, Kaifeng and Liu, Zi-Wen and Adesso, Gerardo},
  journal = {Physical Review Letters},
  volume  = {122},
  number  = {14},
  pages   = {140402},
  year    = {2019},
  publisher = {American Physical Society},
  doi     = {10.1103/PhysRevLett.122.140402},
  eprint  = {1809.01672},
  archivePrefix = {arXiv},
  primaryClass  = {quant-ph}
}

@article{SSC19,
  title   = {All Sets of Incompatible Measurements give an Advantage in Quantum State Discrimination},
  author  = {Skrzypczyk, Paul and \v{S}upi\'c, Ivan and Cavalcanti, Daniel},
  journal = {Physical Review Letters},
  volume  = {122},
  number  = {13},
  pages   = {130403},
  year    = {2019},
  publisher = {American Physical Society},
  doi     = {10.1103/PhysRevLett.122.130403},
  eprint  = {1901.00816},
  archivePrefix = {arXiv},
  primaryClass  = {quant-ph}
}

@article{UKSYG19,
  title   = {Quantifying Quantum Resources with Conic Programming},
  author  = {Uola, Roope and Kraft, Tristan and Shang, Jiangwei and Yu, Xiao-Dong and G\"uhne, Otfried},
  journal = {Physical Review Letters},
  volume  = {122},
  number  = {13},
  pages   = {130404},
  year    = {2019},
  publisher = {American Physical Society},
  doi     = {10.1103/PhysRevLett.122.130404},
  eprint  = {1812.09216},
  archivePrefix = {arXiv},
  primaryClass  = {quant-ph}
}

@article{DFK19,
  ids     = {DesignolleFarkasKaniewski2019},  % duplicate keys used in the .tex files
  title   = {Incompatibility robustness of quantum measurements: a unified framework},
  author  = {Designolle, S\'ebastien and Farkas, M\'at\'e and Kaniewski, J\k{e}drzej},
  journal = {New Journal of Physics},
  volume  = {21},
  number  = {11},
  pages   = {113053},
  year    = {2019},
  publisher = {IOP Publishing},
  doi     = {10.1088/1367-2630/ab5020},
  eprint  = {1906.00448},
  archivePrefix = {arXiv},
  primaryClass  = {quant-ph}
}

@article{Haapasalo15,
  title   = {Robustness of incompatibility for quantum devices},
  author  = {Haapasalo, Erkka},
  journal = {Journal of Physics A: Mathematical and Theoretical},
  volume  = {48},
  number  = {25},
  pages   = {255303},
  year    = {2015},
  publisher = {IOP Publishing},
  doi     = {10.1088/1751-8113/48/25/255303},
  eprint  = {1502.04881},
  archivePrefix = {arXiv},
  primaryClass  = {quant-ph}
}

@article{Choi74,
  title   = {A Schwarz inequality for positive linear maps on $C^*$-algebras},
  author  = {Choi, Man-Duen},
  journal = {Illinois Journal of Mathematics},
  volume  = {18},
  number  = {4},
  pages   = {565--574},
  year    = {1974},
  publisher = {Duke University Press},
  doi     = {10.1215/ijm/1256051007}
}

@book{Pisier96,
  author    = {Pisier, Gilles},
  title     = {The Operator {H}ilbert Space {OH}, Complex Interpolation and Tensor Norms},
  series    = {Memoirs of the American Mathematical Society},
  number    = {585},
  volume    = {122},
  publisher = {American Mathematical Society},
  address   = {Providence, RI},
  year      = {1996},
  issn      = {0065-9266}
}

@article{Ioannou2022,
  ids     = {IoannouEtAl2022},  % duplicate keys used in the .tex files
  author  = {Ioannou, Marie and Sekatski, Pavel and Designolle, S\'ebastien and
             Jones, Benjamin D. M. and Uola, Roope and Brunner, Nicolas},
  title   = {Simulability of High-Dimensional Quantum Measurements},
  journal = {Physical Review Letters},
  volume  = {129},
  number  = {19},
  pages   = {190401},
  year    = {2022},
  doi     = {10.1103/PhysRevLett.129.190401},
  eprint  = {2202.12980},
  archivePrefix = {arXiv},
  primaryClass  = {quant-ph}
}

@article{Jones2023,
  ids     = {JonesEtAl2023},  % duplicate keys used in the .tex files
  author  = {Jones, Benjamin D. M. and Uola, Roope and Cope, Thomas and Ioannou, Marie and
             Designolle, S\'ebastien and Sekatski, Pavel and Brunner, Nicolas},
  title   = {Equivalence between simulability of high-dimensional measurements and
             high-dimensional steering},
  journal = {Physical Review A},
  volume  = {107},
  number  = {5},
  pages   = {052425},
  year    = {2023},
  doi     = {10.1103/PhysRevA.107.052425},
  eprint  = {2207.04080},
  archivePrefix = {arXiv},
  primaryClass  = {quant-ph}
}

@article{Buscemi2020,
  ids     = {BuscemiChitambarZhou2020},  % duplicate keys used in the .tex files
  author  = {Buscemi, Francesco and Chitambar, Eric and Zhou, Wenbin},
  title   = {Complete Resource Theory of Quantum Incompatibility as Quantum Programmability},
  journal = {Physical Review Letters},
  volume  = {124}, number = {12}, pages = {120401}, year = {2020},
  doi     = {10.1103/PhysRevLett.124.120401},
  eprint  = {1908.11274}, archivePrefix = {arXiv}
}

@article{Buscemi2023,
  author  = {Buscemi, Francesco and Kobayashi, Kodai and Minagawa, Shintaro and
             Perinotti, Paolo and Tosini, Alessandro},
  title   = {Unifying different notions of quantum incompatibility into a strict hierarchy
             of resource theories of communication},
  journal = {Quantum},
  volume  = {7},
  pages   = {1035},
  year    = {2023},
  doi     = {10.22331/q-2023-06-07-1035},
  eprint  = {2211.09226},
  archivePrefix = {arXiv},
  primaryClass  = {quant-ph}
}

@article{Piani2015,
  author  = {Piani, Marco},
  title   = {Channel steering},
  journal = {Journal of the Optical Society of America B},
  volume  = {32}, number = {4}, pages = {A1--A7}, year = {2015},
  doi     = {10.1364/JOSAB.32.0000A1},
  eprint  = {1411.0397}, archivePrefix = {arXiv}
}

@article{Zjawin2023,
  author  = {Zjawin, Beata and Schmid, David and Hoban, Matty J. and Sainz, Ana Bel\'en},
  title   = {The resource theory of nonclassicality of channel assemblages},
  journal = {Quantum},
  volume  = {7}, pages = {1134}, year = {2023},
  doi     = {10.22331/q-2023-10-10-1134}
}

@article{HeinosaariMiyadera2017,
  author  = {Heinosaari, Teiko and Miyadera, Takayuki},
  title   = {Incompatibility of quantum channels},
  journal = {Journal of Physics A: Mathematical and Theoretical},
  volume  = {50}, number = {13}, pages = {135302}, year = {2017},
  doi     = {10.1088/1751-8121/aa5f6b},
  eprint  = {1608.01794}, archivePrefix = {arXiv}
}

@article{Haapasalo2021,
  author  = {Haapasalo, Erkka and Kraft, Tristan and Miklin, Nikolai and Uola, Roope},
  title   = {Quantum marginal problem and incompatibility},
  journal = {Quantum},
  volume  = {5}, pages = {476}, year = {2021},
  doi     = {10.22331/q-2021-06-15-476},
  eprint  = {1909.02941}, archivePrefix = {arXiv}
}

@article{Girard2021,
  author  = {Girard, Mark and Pl\'avala, Martin and Sikora, Jamie},
  title   = {Jordan products of quantum channels and their compatibility},
  journal = {Nature Communications},
  volume  = {12}, pages = {2129}, year = {2021},
  doi     = {10.1038/s41467-021-22275-0},
  eprint  = {2009.03279}, archivePrefix = {arXiv}
}

@article{Hsieh2022,
  author  = {Hsieh, Chung-Yun and Lostaglio, Matteo and Ac\'in, Antonio},
  title   = {Quantum channel marginal problem},
  journal = {Physical Review Research},
  volume  = {4}, number = {1}, pages = {013249}, year = {2022},
  doi     = {10.1103/PhysRevResearch.4.013249},
  eprint  = {2102.10926}, archivePrefix = {arXiv}
}

@article{Barnum1996,
  author  = {Barnum, Howard and Caves, Carlton M. and Fuchs, Christopher A. and
             Jozsa, Richard and Schumacher, Benjamin},
  title   = {Noncommuting Mixed States Cannot Be Broadcast},
  journal = {Physical Review Letters},
  volume  = {76}, number = {15}, pages = {2818--2821}, year = {1996},
  doi     = {10.1103/PhysRevLett.76.2818},
  eprint  = {quant-ph/9511010}, archivePrefix = {arXiv}
}

@article{Busch1986,
  author  = {Busch, Paul},
  title   = {Unsharp reality and joint measurements for spin observables},
  journal = {Physical Review D},
  volume  = {33}, number = {8}, pages = {2253--2261}, year = {1986},
  doi     = {10.1103/PhysRevD.33.2253}
}

@article{Terhal2000,
  ids     = {TerhalHorodecki1999},  % duplicate keys used in the .tex files
  author  = {Terhal, Barbara M. and Horodecki, Pawe{\l}},
  title   = {Schmidt number for density matrices},
  journal = {Physical Review A},
  volume  = {61},
  number  = {4},
  pages   = {040301(R)},
  year    = {2000},
  doi     = {10.1103/PhysRevA.61.040301},
  eprint  = {quant-ph/9911117},
  archivePrefix = {arXiv}
}

@article{Khandelwal2025,
  author  = {Khandelwal, Shishir and Tavakoli, Armin},
  title   = {Simulating Quantum Instruments with Projective Measurements and
             Quantum Postprocessing},
  journal = {Physical Review Letters},
  volume  = {135},
  number  = {4},
  pages   = {040202},
  year    = {2025},
  doi     = {10.1103/bhr5-g71p},
  eprint  = {2503.00956},
  archivePrefix = {arXiv},
  primaryClass  = {quant-ph},
  note    = {APS assigned the non-numeric DOI 10.1103/bhr5-g71p}
}

@article{Carmeli2019,
  author  = {Carmeli, Claudio and Heinosaari, Teiko and Toigo, Alessandro},
  title   = {Quantum Incompatibility Witnesses},
  journal = {Physical Review Letters},
  volume  = {122},
  number  = {13},
  pages   = {130402},
  year    = {2019},
  doi     = {10.1103/PhysRevLett.122.130402},
  eprint  = {1812.02985},
  archivePrefix = {arXiv},
  primaryClass  = {quant-ph}
}

@article{CarmeliJMP2019,
  author  = {Carmeli, Claudio and Heinosaari, Teiko and Miyadera, Takayuki and Toigo, Alessandro},
  title   = {Witnessing incompatibility of quantum channels},
  journal = {Journal of Mathematical Physics},
  volume  = {60},
  number  = {12},
  pages   = {122202},
  year    = {2019},
  doi     = {10.1063/1.5126496},
  eprint  = {1906.10904},
  archivePrefix = {arXiv},
  primaryClass  = {quant-ph}
}

@article{Guhne2023,
  author  = {G\"uhne, Otfried and Haapasalo, Erkka and Kraft, Tristan and
             Pellonp\"a\"a, Juha-Pekka and Uola, Roope},
  title   = {Colloquium: Incompatible measurements in quantum information science},
  journal = {Reviews of Modern Physics},
  volume  = {95}, number = {1}, pages = {011003}, year = {2023},
  doi     = {10.1103/RevModPhys.95.011003},
  eprint  = {2112.06784}, archivePrefix = {arXiv}
}

@article{Leppajarvi2024,
  ids     = {LeppajarviSedlak2022,LeppajarviSedlak2024},  % duplicate keys used in the .tex files
  author  = {Lepp\"aj\"arvi, Leevi and Sedl\'ak, Michal},
  title   = {Incompatibility of quantum instruments},
  journal = {Quantum},
  volume  = {8},
  pages   = {1246},
  year    = {2024},
  doi     = {10.22331/q-2024-02-12-1246},
  eprint  = {2212.11225},
  archivePrefix = {arXiv},
  primaryClass  = {quant-ph}
}

@article{Mitra2022,
  ids     = {MitraFarkas2021,MitraFarkas2022},  % duplicate keys used in the .tex files
  author  = {Mitra, Arindam and Farkas, M\'at\'e},
  title   = {Compatibility of quantum instruments},
  journal = {Physical Review A},
  volume  = {105},
  number  = {5},
  pages   = {052202},
  year    = {2022},
  doi     = {10.1103/PhysRevA.105.052202},
  eprint  = {2110.00932},
  archivePrefix = {arXiv},
  primaryClass  = {quant-ph}
}

@article{HSR2003,
  author  = {Horodecki, Micha{\l} and Shor, Peter W. and Ruskai, Mary Beth},
  title   = {Entanglement breaking channels},
  journal = {Reviews in Mathematical Physics},
  volume  = {15},
  number  = {6},
  pages   = {629--641},
  year    = {2003},
  doi     = {10.1142/S0129055X03001709},
  eprint  = {quant-ph/0302031},
  archivePrefix = {arXiv}
}

@article{GJL2018tro,
  author  = {Gao, Li and Junge, Marius and LaRacuente, Nicholas},
  title   = {Capacity Estimates via Comparison with {TRO} Channels},
  journal = {Communications in Mathematical Physics},
  volume  = {364}, pages = {83--121}, year = {2018},
  doi     = {10.1007/s00220-018-3249-y},
  eprint  = {1609.08594}, archivePrefix = {arXiv},
  note    = {Page range read off secondary listings; confirm against the Springer page}
}

@article{Rosset2018,
  author  = {Rosset, Denis and Buscemi, Francesco and Liang, Yeong-Cherng},
  title   = {Resource Theory of Quantum Memories and Their Faithful Verification with Minimal Assumptions},
  journal = {Physical Review X},
  volume  = {8},
  pages   = {021033},
  year    = {2018},
  doi     = {10.1103/PhysRevX.8.021033},
  eprint  = {1710.04710},
  archivePrefix = {arXiv},
  primaryClass  = {quant-ph}
}

@article{Ivanovic81,
  author  = {Ivanovi{\'c}, I. D.},
  title   = {Geometrical description of quantal state determination},
  journal = {Journal of Physics A: Mathematical and General},
  volume  = {14},
  number  = {12},
  pages   = {3241--3245},
  year    = {1981},
  doi     = {10.1088/0305-4470/14/12/019}
}

@article{WF89,
  author  = {Wootters, William K. and Fields, Brian D.},
  title   = {Optimal state-determination by mutually unbiased measurements},
  journal = {Annals of Physics},
  volume  = {191},
  number  = {2},
  pages   = {363--381},
  year    = {1989},
  doi     = {10.1016/0003-4916(89)90322-9}
}

@article{BLN25,
  author  = {Bluhm, Andreas and Lepp{\"a}j{\"a}rvi, Leevi and Nechita, Ion},
  title   = {On the simulation of quantum multimeters},
  journal = {Quantum},
  volume  = {9},
  pages   = {1608},
  year    = {2025},
  doi     = {10.22331/q-2025-01-27-1608},
  eprint  = {2402.18333},
  archivePrefix = {arXiv},
  primaryClass  = {quant-ph}
}

@article{Popa1983,
  author  = {Popa, Sorin},
  title   = {Orthogonal pairs of $*$-subalgebras in finite von {N}eumann algebras},
  journal = {Journal of Operator Theory},
  volume  = {9},
  number  = {2},
  pages   = {253--268},
  year    = {1983}
}

@article{Petz2007,
  author  = {Petz, D{\'e}nes},
  title   = {Complementarity in quantum systems},
  journal = {Reports on Mathematical Physics},
  volume  = {59},
  number  = {2},
  pages   = {209--224},
  year    = {2007},
  doi     = {10.1016/S0034-4877(07)00010-9},
  eprint  = {quant-ph/0610189},
  archivePrefix = {arXiv}
}

@article{OPS2007,
  author  = {Ohno, Hiromichi and Petz, D{\'e}nes and Sz{\'a}nt{\'o}, Andr{\'a}s},
  title   = {Quasi-orthogonal subalgebras of $4\times4$ matrices},
  journal = {Linear Algebra and its Applications},
  volume  = {425},
  pages   = {109--118},
  year    = {2007},
  doi     = {10.1016/j.laa.2007.03.020}
}

@article{Weiner2010,
  author  = {Weiner, Mih{\'a}ly},
  title   = {On orthogonal systems of matrix algebras},
  journal = {Linear Algebra and its Applications},
  volume  = {433},
  number  = {3},
  pages   = {520--533},
  year    = {2010},
  doi     = {10.1016/j.laa.2010.03.017},
  eprint  = {1002.0017},
  archivePrefix = {arXiv},
  note    = {arXiv title: On quasi-orthogonal systems of matrix algebras}
}

@article{BSTW2007,
  author  = {Boykin, P. Oscar and Sitharam, Meera and Tiep, Pham Huu and Wocjan, Pawel},
  title   = {Mutually unbiased bases and orthogonal decompositions of {L}ie algebras},
  journal = {Quantum Information and Computation},
  volume  = {7},
  number  = {4},
  pages   = {371--382},
  year    = {2007},
  doi     = {10.26421/QIC7.4-6},
  eprint  = {quant-ph/0506089},
  archivePrefix = {arXiv}
}

@article{Durt2010,
  author  = {Durt, Thomas and Englert, Berthold-Georg and Bengtsson, Ingemar and {\.Z}yczkowski, Karol},
  title   = {On mutually unbiased bases},
  journal = {International Journal of Quantum Information},
  volume  = {8},
  number  = {4},
  pages   = {535--640},
  year    = {2010},
  doi     = {10.1142/S0219749910006502}
}

@article{DesignolleSkrzypczykFrowisBrunner2019,
  author  = {Designolle, S{\'e}bastien and Skrzypczyk, Paul and Fr{\"o}wis, Florian and Brunner, Nicolas},
  title   = {Quantifying Measurement Incompatibility of Mutually Unbiased Bases},
  journal = {Physical Review Letters},
  volume  = {122},
  pages   = {050402},
  year    = {2019},
  doi     = {10.1103/PhysRevLett.122.050402}
}

@article{PLAVALA202515,
  author  = {Pl{\'a}vala, Martin and Ligthart, Laurens T. and Gross, David},
  title   = {The polarization hierarchy for polynomial optimization over convex bodies, with applications to nonnegative matrix rank},
  journal = {Linear Algebra and its Applications},
  volume  = {723},
  pages   = {15--32},
  year    = {2025},
  doi     = {10.1016/j.laa.2025.05.019}
}

@article{ChoiEffros1977,
  author  = {Choi, Man-Duen and Effros, Edward G.},
  title   = {Injectivity and operator spaces},
  journal = {Journal of Functional Analysis},
  volume  = {24},
  number  = {2},
  pages   = {156--209},
  year    = {1977},
  doi     = {10.1016/0022-1236(77)90052-0}
}

@article{Kitaev1997,
  author  = {Kitaev, A. Yu.},
  title   = {Quantum computations: algorithms and error correction},
  journal = {Russian Mathematical Surveys},
  volume  = {52},
  number  = {6},
  pages   = {1191--1249},
  year    = {1997},
  doi     = {10.1070/RM1997v052n06ABEH002155}
}

@book{Watrous2018,
  author    = {Watrous, John},
  title     = {The Theory of Quantum Information},
  publisher = {Cambridge University Press},
  address   = {Cambridge},
  year      = {2018},
  doi       = {10.1017/9781316848142}
}

@article{Alessandro2025,
  author  = {D'Alessandro, Nicola and Roch i Carceller, Carles and Tavakoli, Armin},
  title   = {Semidefinite Relaxations for High-Dimensional Entanglement in the Steering Scenario},
  journal = {Physical Review Letters},
  volume  = {134},
  number  = {9},
  pages   = {090802},
  year    = {2025},
  doi     = {10.1103/PhysRevLett.134.090802}
}

@article{Gois2023,
  ids     = {deGoisEtAl2023},
  author  = {de Gois, Carlos and Pl{\'a}vala, Martin and Schwonnek, Ren{\'e} and G{\"u}hne, Otfried},
  title   = {Complete Hierarchy for High-Dimensional Steering Certification},
  journal = {Physical Review Letters},
  volume  = {131},
  number  = {1},
  pages   = {010201},
  year    = {2023},
  doi     = {10.1103/PhysRevLett.131.010201}
}

@article{MitraFarkas2023,
  author  = {Mitra, Arindam and Farkas, M{\'a}t{\'e}},
  title   = {Characterizing and Quantifying the Incompatibility of Quantum Instruments},
  journal = {Physical Review A},
  volume  = {107},
  pages   = {032217},
  year    = {2023},
  doi     = {10.1103/PhysRevA.107.032217}
}

@article{Designolle2021,
  author  = {Designolle, S{\'e}bastien and Srivastav, Vatshal and Uola, Roope and Herrera Valencia, Natalia and McCutcheon, Will and Malik, Mehul and Brunner, Nicolas},
  title   = {Genuine High-Dimensional Quantum Steering},
  journal = {Physical Review Letters},
  volume  = {126},
  number  = {20},
  pages   = {200404},
  year    = {2021},
  doi     = {10.1103/PhysRevLett.126.200404},
  eprint  = {2007.02718},
  archivePrefix = {arXiv},
  primaryClass  = {quant-ph}
}

@article{Bavaresco2017,
  author  = {Bavaresco, Jessica and Quintino, Marco T\'ulio and Guerini, Leonardo and
             Maciel, Thiago O. and Cavalcanti, Daniel and Terra Cunha, Marcelo},
  title   = {Most incompatible measurements for robust steering tests},
  journal = {Physical Review A},
  volume  = {96}, number = {2}, pages = {022110}, year = {2017},
  doi     = {10.1103/PhysRevA.96.022110},
  eprint  = {1704.02994}, archivePrefix = {arXiv}
}

@article{Bluhm2022,
  author  = {Bluhm, Andreas and Jen\v{c}ov\'a, Anna and Nechita, Ion},
  title   = {Incompatibility in General Probabilistic Theories, Generalized Spectrahedra,
             and Tensor Norms},
  journal = {Communications in Mathematical Physics},
  volume  = {393},
  number  = {3},
  pages   = {1125--1198},
  year    = {2022},
  doi     = {10.1007/s00220-022-04379-w},
  eprint  = {2011.06497},
  archivePrefix = {arXiv},
  primaryClass  = {quant-ph}
}

@article{BluhmNechita2018,
  author  = {Bluhm, Andreas and Nechita, Ion},
  title   = {Joint measurability of quantum effects and the matrix diamond},
  journal = {Journal of Mathematical Physics},
  volume  = {59}, number = {11}, pages = {112202}, year = {2018},
  eprint  = {1807.01508}, archivePrefix = {arXiv}
}

@article{BluhmNechita2020,
  author  = {Bluhm, Andreas and Nechita, Ion},
  title   = {Compatibility of Quantum Measurements and Inclusion Constants for the Matrix Jewel},
  journal = {SIAM Journal on Applied Algebra and Geometry},
  volume  = {4}, number = {2}, pages = {255--296}, year = {2020},
  doi     = {10.1137/19M123837X},
  eprint  = {1809.04514}, archivePrefix = {arXiv}
}

@article{BluhmNechita2022cube,
  author  = {Bluhm, Andreas and Nechita, Ion},
  title   = {Maximal violation of steering inequalities and the matrix cube},
  journal = {Quantum},
  volume  = {6}, pages = {656}, year = {2022},
  doi     = {10.22331/q-2022-02-21-656},
  eprint  = {2105.11302}, archivePrefix = {arXiv}
}

@article{BluhmNechita2022tensor,
  author  = {Bluhm, Andreas and Nechita, Ion},
  title   = {A tensor norm approach to quantum compatibility},
  journal = {Journal of Mathematical Physics},
  volume  = {63}, number = {6}, pages = {062201}, year = {2022},
  eprint  = {2202.13993}, archivePrefix = {arXiv}
}

@article{BluhmRauberWolf2018,
  author  = {Bluhm, Andreas and Rauber, Lukas and Wolf, Michael M.},
  title   = {Quantum Compression Relative to a Set of Measurements},
  journal = {Annales Henri Poincar\'e},
  volume  = {19}, number = {6}, pages = {1891--1937}, year = {2018},
  doi     = {10.1007/s00023-018-0660-z}
}

@article{Busch2013,
  author  = {Busch, Paul and Heinosaari, Teiko and Schultz, Jussi and Stevens, Neil},
  title   = {Comparing the degrees of incompatibility inherent in probabilistic physical theories},
  journal = {Europhysics Letters},
  volume  = {103}, number = {1}, pages = {10002}, year = {2013},
  doi     = {10.1209/0295-5075/103/10002},
  eprint  = {1210.4142}, archivePrefix = {arXiv}
}

@article{Carmeli2012,
  author  = {Carmeli, Claudio and Heinosaari, Teiko and Toigo, Alessandro},
  title   = {Informationally complete joint measurements on finite quantum systems},
  journal = {Physical Review A},
  volume  = {85}, number = {1}, pages = {012109}, year = {2012},
  doi     = {10.1103/PhysRevA.85.012109},
  eprint  = {1111.3509}, archivePrefix = {arXiv}
}

@article{Carmeli2018,
  author  = {Carmeli, Claudio and Heinosaari, Teiko and Toigo, Alessandro},
  title   = {State discrimination with postmeasurement information and incompatibility of quantum measurements},
  journal = {Physical Review A},
  volume  = {98}, number = {1}, pages = {012126}, year = {2018},
  doi     = {10.1103/PhysRevA.98.012126},
  eprint  = {1804.09693}, archivePrefix = {arXiv}
}

@article{CavalcantiSkrzypczyk2016,
  author  = {Cavalcanti, Daniel and Skrzypczyk, Paul},
  title   = {Quantitative relations between measurement incompatibility, quantum steering,
             and nonlocality},
  journal = {Physical Review A},
  volume  = {93}, number = {5}, pages = {052112}, year = {2016},
  doi     = {10.1103/PhysRevA.93.052112}
}

@article{ChitambarGour2019,
  author  = {Chitambar, Eric and Gour, Gilad},
  title   = {Quantum resource theories},
  journal = {Reviews of Modern Physics},
  volume  = {91}, number = {2}, pages = {025001}, year = {2019},
  doi     = {10.1103/RevModPhys.91.025001},
  eprint  = {1806.06107}, archivePrefix = {arXiv}
}

@article{Ducuara2025,
  author  = {Ducuara, Andr\'es F. and Takakura, Ryo and Hernandez, Fernando J. and Susa, Cristian E.},
  title   = {Multiobject operational tasks for measurement incompatibility},
  journal = {Physical Review Research},
  volume  = {7}, number = {3}, pages = {033050}, year = {2025},
  doi     = {10.1103/m7ln-tb1s},
  eprint  = {2412.15615}, archivePrefix = {arXiv}
}

@article{Engineer2025,
  author  = {Engineer, Sophie and Goel, Suraj and Egelhaaf, Sophie and McCutcheon, Will and
             Srivastav, Vatshal and Leedumrongwatthanakun, Saroch and Wollmann, Sabine and
             Jones, Benjamin D. M. and Cope, Thomas and Brunner, Nicolas and Uola, Roope and
             Malik, Mehul},
  title   = {Certifying high-dimensional quantum channels},
  journal = {Physical Review Research},
  volume  = {7}, number = {3}, pages = {033233}, year = {2025},
  doi     = {10.1103/bwd5-wx7j},
  eprint  = {2408.15880}, archivePrefix = {arXiv}
}

@article{Filippov2018,
  author  = {Filippov, Sergey N. and Heinosaari, Teiko and Lepp\"aj\"arvi, Leevi},
  title   = {Simulability of observables in general probabilistic theories},
  journal = {Physical Review A},
  volume  = {97}, number = {6}, pages = {062102}, year = {2018},
  doi     = {10.1103/PhysRevA.97.062102},
  eprint  = {1803.11006}, archivePrefix = {arXiv}
}

@article{Guerini2017,
  author  = {Guerini, Leonardo and Bavaresco, Jessica and Terra Cunha, Marcelo and Ac\'in, Antonio},
  title   = {Operational framework for quantum measurement simulability},
  journal = {Journal of Mathematical Physics},
  volume  = {58}, number = {9}, pages = {092102}, year = {2017},
  doi     = {10.1063/1.4994303},
  eprint  = {1705.06343}, archivePrefix = {arXiv}
}

@article{Heinosaari2015,
  author  = {Heinosaari, Teiko and Kiukas, Jukka and Reitzner, Daniel},
  title   = {Noise robustness of the incompatibility of quantum measurements},
  journal = {Physical Review A},
  volume  = {92}, number = {2}, pages = {022115}, year = {2015},
  doi     = {10.1103/PhysRevA.92.022115},
  eprint  = {1501.04554}, archivePrefix = {arXiv}
}

@article{Hsieh2024thermo,
  author  = {Hsieh, Chung-Yun and Chen, Shin-Liang},
  title   = {Thermodynamic Approach to Quantifying Incompatible Instruments},
  journal = {Physical Review Letters},
  volume  = {133}, number = {17}, pages = {170401}, year = {2024},
  doi     = {10.1103/PhysRevLett.133.170401},
  eprint  = {2402.13080}, archivePrefix = {arXiv}
}

@article{JPPVW2010cmp,
  author  = {Junge, Marius and Palazuelos, Carlos and P\'erez-Garc\'ia, David and
             Villanueva, Ignacio and Wolf, Michael M.},
  title   = {Unbounded Violations of Bipartite {Bell} Inequalities via Operator Space Theory},
  journal = {Communications in Mathematical Physics},
  volume  = {300}, number = {3}, pages = {715--739}, year = {2010},
  doi     = {10.1007/s00220-010-1125-5}
}

@article{JPPVW2010prl,
  author  = {Junge, Marius and Palazuelos, Carlos and P\'erez-Garc\'ia, David and
             Villanueva, Ignacio and Wolf, Michael M.},
  title   = {Operator Space Theory: A Natural Framework for {Bell} Inequalities},
  journal = {Physical Review Letters},
  volume  = {104}, number = {17}, pages = {170405}, year = {2010},
  doi     = {10.1103/PhysRevLett.104.170405},
  eprint  = {0912.1941}, archivePrefix = {arXiv}
}

@article{Jencova2018,
  author  = {Jen\v{c}ov\'a, Anna},
  title   = {Incompatible measurements in a class of general probabilistic theories},
  journal = {Physical Review A},
  volume  = {98}, number = {1}, pages = {012133}, year = {2018},
  doi     = {10.1103/PhysRevA.98.012133},
  eprint  = {1705.08008}, archivePrefix = {arXiv}
}

@article{Jencova2022,
  author  = {Jen\v{c}ov\'a, Anna},
  title   = {Assemblages and steering in general probabilistic theories},
  journal = {Journal of Physics A: Mathematical and Theoretical},
  volume  = {55}, number = {43}, pages = {434001}, year = {2022},
  doi     = {10.1088/1751-8121/ac97ce},
  eprint  = {2202.09109}, archivePrefix = {arXiv}
}

@article{Kiukas2017,
  author  = {Kiukas, Jukka and Budroni, Costantino and Uola, Roope and Pellonp\"a\"a, Juha-Pekka},
  title   = {Continuous-variable steering and incompatibility via state-channel duality},
  journal = {Physical Review A},
  volume  = {96}, number = {4}, pages = {042331}, year = {2017},
  doi     = {10.1103/PhysRevA.96.042331},
  eprint  = {1704.05734}, archivePrefix = {arXiv}
}

@article{Kuramochi2018,
  author  = {Kuramochi, Yui},
  title   = {Quantum incompatibility of channels with general outcome operator algebras},
  journal = {Journal of Mathematical Physics},
  volume  = {59}, number = {4}, pages = {042203}, year = {2018},
  doi     = {10.1063/1.5008300},
  eprint  = {1708.00150}, archivePrefix = {arXiv}
}

@misc{Kuramochi2020,
  author  = {Kuramochi, Yui},
  title   = {Compact convex structure of measurements and its applications to simulability,
             incompatibility, and convex resource theory of continuous-outcome measurements},
  year    = {2020},
  eprint  = {2002.03504},
  archivePrefix = {arXiv},
  primaryClass  = {quant-ph},
  doi     = {10.48550/arXiv.2002.03504}
}

@article{Lee2026,
  author  = {Lee, Kuan-Yi and Lin, Jhen-Dong and Miranowicz, Adam and Chen, Yueh-Nan},
  title   = {General Class of Functionals for Certifying Quantum Incompatibility},
  journal = {Physical Review Letters},
  volume  = {136}, pages = {120203}, year = {2026},
  doi     = {10.1103/rsx1-zvbp},
  eprint  = {2601.02239}, archivePrefix = {arXiv}
}

@article{Lobo2026,
  author  = {Lobo, Edwin Peter and Balanz\'o-Juand\'o, Maria and Pironio, Stefano},
  title   = {Generalized measurement incompatibility},
  journal = {arXiv preprint},
  year    = {2026},
  eprint  = {2605.16151}, archivePrefix = {arXiv}
}

@article{Loulidi2021,
  author  = {Loulidi, Faedi and Nechita, Ion},
  title   = {The compatibility dimension of quantum measurements},
  journal = {Journal of Mathematical Physics},
  volume  = {62},
  pages   = {042205},
  year    = {2021},
  doi     = {10.1063/5.0028658},
  eprint  = {2008.10317},
  archivePrefix = {arXiv},
  primaryClass  = {quant-ph}
}

@article{Loulidi2022,
  author  = {Loulidi, Faedi and Nechita, Ion},
  title   = {Measurement Incompatibility versus {Bell} Nonlocality: An Approach via Tensor Norms},
  journal = {PRX Quantum},
  volume  = {3}, number = {4}, pages = {040325}, year = {2022},
  doi     = {10.1103/PRXQuantum.3.040325},
  eprint  = {2205.12668}, archivePrefix = {arXiv}
}

@article{Masini2024,
  author  = {Masini, Michele and Ioannou, Marie and Brunner, Nicolas and
             Pironio, Stefano and Sekatski, Pavel},
  title   = {Joint-measurability and quantum communication with untrusted devices},
  journal = {Quantum},
  volume  = {8}, pages = {1574}, year = {2024},
  doi     = {10.22331/q-2024-12-23-1574},
  eprint  = {2403.14785}, archivePrefix = {arXiv}
}

@misc{Minagawa2026,
  author        = {Minagawa, Shintaro and Takakura, Ryo and Torii, Kensei},
  title         = {Joint Realizability Tradeoffs Bounded by Quantum Channel Incompatibility},
  year          = {2026},
  eprint        = {2605.11924}, archivePrefix = {arXiv}, primaryClass = {quant-ph}
}

@article{Moroder2016,
  author  = {Moroder, Tobias and Gittsovich, Oleg and Huber, Marcus and Uola, Roope and
             G\"uhne, Otfried},
  title   = {Steering Maps and Their Application to Dimension-Bounded Steering},
  journal = {Physical Review Letters},
  volume  = {116}, number = {9}, pages = {090403}, year = {2016},
  doi     = {10.1103/PhysRevLett.116.090403},
  eprint  = {1412.2623}, archivePrefix = {arXiv}
}

@article{Oszmaniec2017,
  author  = {Oszmaniec, Micha\l{} and Guerini, Leonardo and Wittek, Peter and Ac\'in, Antonio},
  title   = {Simulating Positive-Operator-Valued Measures with Projective Measurements},
  journal = {Physical Review Letters},
  volume  = {119}, number = {19}, pages = {190501}, year = {2017},
  doi     = {10.1103/PhysRevLett.119.190501},
  eprint  = {1609.06139}, archivePrefix = {arXiv}
}

@article{Oszmaniec2019,
  author  = {Oszmaniec, Micha\l{} and Biswas, Tanmoy},
  title   = {Operational relevance of resource theories of quantum measurements},
  journal = {Quantum},
  volume  = {3}, pages = {133}, year = {2019},
  doi     = {10.22331/q-2019-04-26-133},
  eprint  = {1901.08566}, archivePrefix = {arXiv}
}

@article{PianiWatrous2015,
  author  = {Piani, Marco and Watrous, John},
  title   = {Necessary and Sufficient Quantum Information Characterization of
             {Einstein--Podolsky--Rosen} Steering},
  journal = {Physical Review Letters},
  volume  = {114}, number = {6}, pages = {060404}, year = {2015},
  doi     = {10.1103/PhysRevLett.114.060404}
}

@article{Porto2026lp,
  author  = {Porto, Lucas E. A. and Designolle, S\'ebastien and Pokutta, Sebastian and
             Quintino, Marco T\'ulio},
  title   = {Measurement incompatibility and quantum steering via linear programming},
  journal = {Quantum},
  volume  = {10}, pages = {2141}, year = {2026},
  doi     = {10.22331/q-2026-06-19-2141},
  eprint  = {2506.03045}, archivePrefix = {arXiv}
}

@article{Quintino2014,
  author  = {Quintino, Marco T\'ulio and V\'ertesi, Tam\'as and Brunner, Nicolas},
  title   = {Joint Measurability, {Einstein--Podolsky--Rosen} Steering, and {Bell} Nonlocality},
  journal = {Physical Review Letters},
  volume  = {113}, number = {16}, pages = {160402}, year = {2014},
  doi     = {10.1103/PhysRevLett.113.160402},
  eprint  = {1406.6976}, archivePrefix = {arXiv}
}

@article{Uola2014,
  author  = {Uola, Roope and Moroder, Tobias and G\"uhne, Otfried},
  title   = {Joint Measurability of Generalized Measurements Implies Classicality},
  journal = {Physical Review Letters},
  volume  = {113}, number = {16}, pages = {160403}, year = {2014},
  doi     = {10.1103/PhysRevLett.113.160403},
  eprint  = {1407.2224}, archivePrefix = {arXiv}
}

@article{Uola2020,
  author  = {Uola, Roope and Costa, Ana C. S. and Nguyen, H. Chau and G\"uhne, Otfried},
  title   = {Quantum steering},
  journal = {Reviews of Modern Physics},
  volume  = {92}, number = {1}, pages = {015001}, year = {2020},
  doi     = {10.1103/RevModPhys.92.015001},
  eprint  = {1903.06663}, archivePrefix = {arXiv}
}

@article{Yin2015,
  author  = {Yin, Zhi and Marciniak, Marcin and Horodecki, Micha{\l}},
  title   = {Operator space approach to steering inequality},
  journal = {Journal of Physics A: Mathematical and Theoretical},
  volume  = {48}, number = {13}, pages = {135303}, year = {2015},
  eprint  = {1405.1945}, archivePrefix = {arXiv}
}

@article{AubrunMullerHermes2023,
  author  = {Aubrun, Guillaume and M{\"u}ller-Hermes, Alexander},
  title   = {Annihilating Entanglement Between Cones},
  journal = {Communications in Mathematical Physics},
  volume  = {400},
  pages   = {931--976},
  year    = {2023},
  doi     = {10.1007/s00220-022-04621-5},
  eprint  = {2110.11825},
  archivePrefix = {arXiv}
}

@article{AubrunMullerHermesPlavala2025,
  author  = {Aubrun, Guillaume and M{\"u}ller-Hermes, Alexander and Pl{\'a}vala, Martin},
  title   = {Monogamy of entanglement between cones},
  journal = {Mathematische Annalen},
  volume  = {391},
  year    = {2025},
  doi     = {10.1007/s00208-024-02935-4},
  eprint  = {2206.11805},
  archivePrefix = {arXiv}
}

@misc{AubrunLaPianaMullerHermes2026,
  author  = {Aubrun, Guillaume and La Piana, Francesca and M{\"u}ller-Hermes, Alexander},
  title   = {Factorization through {Lorentz} cones},
  year    = {2026},
  eprint  = {2606.27825},
  archivePrefix = {arXiv},
  primaryClass  = {math.FA}
}

@article{AubrunDavidsonMullerHermesPaulsenRahaman2024,
  author  = {Aubrun, Guillaume and Davidson, Kenneth R. and M{\"u}ller-Hermes, Alexander and Paulsen, Vern I. and Rahaman, Mizanur},
  title   = {Completely bounded norms of $k$-positive maps},
  journal = {Journal of the London Mathematical Society},
  volume  = {109},
  number  = {6},
  pages   = {e12936},
  year    = {2024},
  doi     = {10.1112/jlms.12936},
  eprint  = {2401.12352},
  archivePrefix = {arXiv}
}

@article{Choi1975,
  author  = {Choi, Man-Duen},
  title   = {Completely positive linear maps on complex matrices},
  journal = {Linear Algebra and its Applications},
  volume  = {10},
  number  = {3},
  pages   = {285--290},
  year    = {1975},
  doi     = {10.1016/0024-3795(75)90075-0}
}

@article{Jamiolkowski1972,
  author  = {Jamio{\l}kowski, Andrzej},
  title   = {Linear transformations which preserve trace and positive semidefiniteness of operators},
  journal = {Reports on Mathematical Physics},
  volume  = {3},
  number  = {4},
  pages   = {275--278},
  year    = {1972},
  doi     = {10.1016/0034-4877(72)90011-0}
}

@article{Friis2019,
  author  = {Friis, Nicolai and Vitagliano, Giuseppe and Malik, Mehul and Huber, Marcus},
  title   = {Entanglement certification from theory to experiment},
  journal = {Nature Reviews Physics},
  volume  = {1},
  pages   = {72--87},
  year    = {2019},
  doi     = {10.1038/s42254-018-0003-5}
}

@article{Gharibian2010,
  author  = {Gharibian, Sevag},
  title   = {Strong {NP}-hardness of the quantum separability problem},
  journal = {Quantum Information and Computation},
  volume  = {10},
  number  = {3--4},
  pages   = {343--360},
  year    = {2010},
  doi     = {10.26421/QIC10.3-4-11},
  eprint  = {0810.4507},
  archivePrefix = {arXiv}
}

@article{Fawzi2021,
  author  = {Fawzi, Hamza},
  title   = {The Set of Separable States has no Finite Semidefinite Representation Except in Dimension $3\times2$},
  journal = {Communications in Mathematical Physics},
  volume  = {386},
  pages   = {1319--1335},
  year    = {2021},
  doi     = {10.1007/s00220-021-04163-2},
  eprint  = {1905.02575},
  archivePrefix = {arXiv}
}

@article{Heinosaari2014strong,
  author  = {Heinosaari, Teiko and Miyadera, Takayuki and Reitzner, Daniel},
  title   = {Strongly Incompatible Quantum Devices},
  journal = {Foundations of Physics},
  volume  = {44},
  pages   = {34--57},
  year    = {2014},
  doi     = {10.1007/s10701-013-9761-1},
  eprint  = {1209.1382},
  archivePrefix = {arXiv},
  primaryClass  = {quant-ph}
}

@article{Quintino2015,
  title = {Inequivalence of entanglement, steering, and Bell nonlocality for general measurements},
  author = {Quintino, Marco T\'ulio and V\'ertesi, Tam\'as and Cavalcanti, Daniel and Augusiak, Remigiusz and Demianowicz, Maciej and Ac\'{\i}n, Antonio and Brunner, Nicolas},
  journal = {Phys. Rev. A},
  volume = {92},
  issue = {3},
  pages = {032107},
  numpages = {6},
  year = {2015},
  month = {Sep},
  publisher = {American Physical Society},
  doi = {10.1103/PhysRevA.92.032107},
  url = {https://link.aps.org/doi/10.1103/PhysRevA.92.032107}
}

@article{Hirsch2018,
  title = {Quantum measurement incompatibility does not imply Bell nonlocality},
  author = {Hirsch, Flavien and Quintino, Marco T{\'u}lio and Brunner, Nicolas},
  journal = {Phys. Rev. A},
  volume = {97},
  issue = {1},
  pages = {012129},
  year = {2018},
  doi = {10.1103/PhysRevA.97.012129}
}

@article{WolfPerezGarciaFernandez2009,
  author  = {Wolf, Michael M. and P\'erez-Garc\'ia, David and Fernandez, Carlos},
  title   = {Measurements Incompatible in Quantum Theory Cannot Be Measured Jointly in Any Other No-Signaling Theory},
  journal = {Physical Review Letters},
  volume  = {103},
  pages   = {230402},
  year    = {2009},
  doi     = {10.1103/PhysRevLett.103.230402},
  eprint  = {0905.2998},
  archivePrefix = {arXiv},
  primaryClass  = {quant-ph}
}

@article{Plavala2025,
  title = {All Incompatible Measurements on Qubits Lead to Multiparticle Bell Nonlocality},
  author = {Pl{\'a}vala, Martin and G{\"u}hne, Otfried and Quintino, Marco T{\'u}lio},
  journal = {Phys. Rev. Lett.},
  volume = {134},
  issue = {20},
  pages = {200201},
  year = {2025},
  doi = {10.1103/PhysRevLett.134.200201}
}

@article{TavakoliUola2020,
  title = {Measurement incompatibility and steering are necessary and sufficient for operational contextuality},
  author = {Tavakoli, Armin and Uola, Roope},
  journal = {Phys. Rev. Research},
  volume = {2},
  issue = {1},
  pages = {013011},
  year = {2020},
  doi = {10.1103/PhysRevResearch.2.013011}
}

@article{Selby2023,
  title = {Contextuality without Incompatibility},
  author = {Selby, John H. and Schmid, David and Wolfe, Elie and Sainz, Ana Bel{\'e}n and Kunjwal, Ravi and Spekkens, Robert W.},
  journal = {Phys. Rev. Lett.},
  volume = {130},
  issue = {23},
  pages = {230201},
  year = {2023},
  doi = {10.1103/PhysRevLett.130.230201}
}

@article{ZhangSchmidYingSpekkens2026,
  title = {Reassessing the Boundary between Classical and Nonclassical for Individual Quantum Processes},
  author = {Zhang, Yujie and Schmid, David and Y{\`i}l{\`e} Y{\=i}ng and Spekkens, Robert W.},
  journal = {Phys. Rev. X},
  volume = {16},
  issue = {2},
  pages = {021050},
  year = {2026},
  doi = {10.1103/vqfz-wzjg}
}

@article{Zhang2026quantifiers,
  doi = {10.22331/q-2026-07-30-2180},
  url = {https://doi.org/10.22331/q-2026-07-30-2180},
  title = {Quantifiers and witnesses for the nonclassicality of measurements and of states},
  author = {Zhang, Yujie and Yīng, Y{\`{i}}l{\`{e}} and Schmid, David},
  journal = {{Quantum}},
  issn = {2521-327X},
  publisher = {{Verein zur F{\"{o}}rderung des Open Access Publizierens in den Quantenwissenschaften}},
  volume = {10},
  pages = {2180},
  month = jul,
  year = {2026}
}

@misc{GourHeinosaariSpekkens2018,
  author = {Gour, Gilad and Heinosaari, Teiko and Spekkens, Robert W.},
  year   = {2018},
  note   = {Introduced in APS March Meeting 2018,
            Session S26: Quantum Resource Theories I,
            Los Angeles, California},
  url    = {https://meetings-archive.aps.org/mar/2018/s26/4/}
}

@book{rudin1991functional,
  title     = {Functional Analysis},
  author    = {Rudin, Walter},
  edition   = {2nd},
  series    = {International Series in Pure and Applied Mathematics},
  year      = {1991},
  publisher = {McGraw-Hill},
  address   = {New York}
}

@article{DaviesLewis1970,
  author  = {Davies, E. B. and Lewis, J. T.},
  title   = {An operational approach to quantum probability},
  journal = {Communications in Mathematical Physics},
  volume  = {17},
  number  = {3},
  pages   = {239--260},
  year    = {1970},
  doi     = {10.1007/BF01647093}
}

@article{Pellonpaa2013,
  author  = {Pellonp\"a\"a, Juha-Pekka},
  title   = {Quantum instruments: {I}. {E}xtreme instruments},
  journal = {Journal of Physics A: Mathematical and Theoretical},
  volume  = {46},
  number  = {2},
  pages   = {025302},
  year    = {2013},
  eprint  = {1202.5905},
  archivePrefix = {arXiv}
}

@book{BuschQM2016,
  author    = {Busch, Paul and Lahti, Pekka and Pellonp\"a\"a, Juha-Pekka and Ylinen, Kari},
  title     = {Quantum Measurement},
  series    = {Theoretical and Mathematical Physics},
  publisher = {Springer},
  address   = {Cham},
  year      = {2016},
  doi       = {10.1007/978-3-319-43389-9}
}

@book{AubrunSzarek2017,
  author    = {Aubrun, Guillaume and Szarek, Stanis{\l}aw J.},
  title     = {Alice and {Bob} Meet {Banach}: The Interface of Asymptotic Geometric Analysis and Quantum Information Theory},
  series    = {Mathematical Surveys and Monographs},
  volume    = {223},
  publisher = {American Mathematical Society},
  address   = {Providence, RI},
  year      = {2017},
  doi       = {10.1090/surv/223}
}

@book{GuptaMandayamSunder2015,
  author    = {Gupta, Ved Prakash and Mandayam, Prabha and Sunder, V. S.},
  title     = {The Functional Analysis of Quantum Information Theory: A Collection of Notes Based on Lectures by {Gilles Pisier}, {K. R. Parthasarathy}, {Vern Paulsen} and {Andreas Winter}},
  series    = {Lecture Notes in Physics},
  volume    = {902},
  publisher = {Springer},
  address   = {Cham},
  year      = {2015},
  doi       = {10.1007/978-3-319-16718-3},
  eprint    = {1410.7188},
  archivePrefix = {arXiv}
}

@article{Huang2006,
  author  = {Huang, Siendong},
  title   = {Schmidt number for quantum operations},
  journal = {Physical Review A},
  volume  = {73}, number = {5}, pages = {052318}, year = {2006},
  doi     = {10.1103/PhysRevA.73.052318}
}

@misc{Sekatski2026,
  author        = {Sekatski, Pavel},
  title         = {The bottleneck dimension of quantum operations},
  year          = {2026},
  eprint        = {2608.25010},
  archivePrefix = {arXiv},
  note          = {arXiv preprint, August 2026}
}

@article{NamikiTokunaga2012,
  author  = {Namiki, Ryo and Tokunaga, Yuuki},
  title   = {Schmidt-number benchmark for genuine quantum memories and gates},
  journal = {Physical Review A},
  volume  = {85}, number = {1}, pages = {010305}, year = {2012},
  doi     = {10.1103/PhysRevA.85.010305},
  eprint  = {1202.0346}, archivePrefix = {arXiv}
}

@article{Gallego2015,
  author  = {Gallego, Rodrigo and Aolita, Leandro},
  title   = {Resource Theory of Steering},
  journal = {Physical Review X},
  volume  = {5}, number = {4}, pages = {041008}, year = {2015},
  doi     = {10.1103/PhysRevX.5.041008}
}

@book{Takesaki_1979,
  title     = {Theory of Operator Algebras I},
  author    = {Takesaki, Masamichi},
  publisher = {Springer-Verlag},
  address   = {New York},
  year      = {1979},
  doi       = {10.1007/978-1-4612-6188-9}
}

\appendix

\section{Background from operator space theory and functional analysis}\label{app:background}

This appendix collects the definitions and standard results from the literature that are used in the proofs of \Cref{app:closure,app:criterion,app:mub}.
\Cref{app:tensor} fixes the tensor products, \Cref{app:cb-facts} records two facts about completely bounded maps, \Cref{app:noncommLp} recalls noncommutative $L_p$-spaces, \Cref{app:OH} recalls Pisier's operator Hilbert space, \Cref{app:ultraprod,app:pisier,app:maharam} contain the ultraproduct constructions, Pisier's factorization results and the consequence of Maharam's classification used for the closure theorems, and \Cref{app:convex} collects the results from convex analysis behind the dual characterizations.

\subsection{Tensor products}\label{app:tensor}

Throughout, $V\otimes W$ denotes the \emph{algebraic} tensor product; the symbols below specify the norm (and, for operator spaces, the matrix norms) under which it is completed.

\paragraph{Hilbert space, $\otimes$.}
For Hilbert spaces $H,K$, the space $H\otimes K$ is the completion of $H\otimes K$ in the inner product
\begin{equation}
  \langle h\otimes k,\ h'\otimes k'\rangle \;=\; \langle h,h'\rangle_H\,\langle k,k'\rangle_K .
\end{equation}

\paragraph{Banach projective, $\otimes_{\pi}$.}
For Banach spaces $V,W$, the completion of $V\otimes W$ in
\begin{equation}
  \norm{u}_{\pi} \;:=\; \inf \{ \textstyle\sum_{i=1}^{n}\norm{v_i}_V\,\norm{w_i}_W \ \mid u=\sum_{i=1}^{n} v_i\otimes w_i \}.
\end{equation}
It linearizes bounded bilinear maps, and $(V\otimes_{\pi}W)^{*}=\cB(V,W^{*})$ (bounded linear maps $V\to W^{*}$).

\paragraph{Operator-space projective, $\otimes^{\wedge}$.}
For operator spaces $V,W$, the matrix norm on $\bM_n(V\otimes W)$ is
\begin{equation}
  \norm{u}_{\wedge}\;:=\;\inf \{\, \norm{\alpha}\,\norm{v}_{M_p(V)}\,\norm{w}_{M_q(W)}\,\norm{\beta}\,\},
  \qquad u=\alpha\,(v\otimes w)\,\beta,
\end{equation}
the infimum over $p,q\in\bN$, $v\in M_p(V)$, $w\in M_q(W)$, and scalar matrices $\alpha\in M_{n,pq}$, $\beta\in M_{pq,n}$, where $(v\otimes w)_{(i,k),(j,l)}=v_{ij}\otimes w_{kl}\in M_{pq}(V\otimes W)$.
It linearizes \emph{jointly completely bounded} bilinear maps, with $(V\otimes^{\wedge}W)^{*}=\CB(V,W^{*})$, and satisfies $S_1(H)\otimes^{\wedge}S_1(K)=S_1(H\otimes K)$.

\paragraph{Operator-space minimal (spatial), $\otimes_{\min}$.}
For operator spaces $V\subseteq\cB(H)$ and $W\subseteq\cB(K)$, the algebraic tensor product $V\otimes W$ sits inside $\cB(H\otimes K)$ via $(v\otimes w)(h\otimes k)=vh\otimes wk$, and $V\minten W$ is its closure there, with the matrix norms inherited at every level:
\begin{equation}
  \norm{u}_{\bM_n(V\minten W)}\;:=\;\norm{u}_{\bM_n(\cB(H\otimes K))},
  \qquad n\in\bN,\quad u\in \bM_n(V\otimes W),
\end{equation}
where $\bM_n(\cB(H\otimes K))=\cB(\ell_2^n\otimes H\otimes K)$.
These norms do not depend on the choice of completely isometric embeddings $V\subseteq\cB(H)$, $W\subseteq\cB(K)$.

\paragraph{Minimal $C^{*}$-norm, $\otimes_{\min}$.}
For $C^{*}$-algebras $\cA,\cB$, the minimal (spatial) $C^{*}$-norm on $\cA\otimes\cB$ is
\begin{equation}
  \norm{x}_{\min}\;:=\;\sup_{\pi,\sigma}\,\norm{(\pi\otimes\sigma)(x)},\qquad x\in \cA\otimes\cB,
\end{equation}
the supremum over all representations $\pi:\cA\to\cB(H_\pi)$, $\sigma:\cB\to\cB(H_\sigma)$; equivalently the norm of $x$ under any faithful spatial representation $\cA\otimes\cB\subseteq\cB(H\otimes K)$.
It is the smallest $C^{*}$-norm on $\cA\otimes\cB$, and $\cA\minten\cB$ is its completion (the spatial $C^{*}$-tensor product).

\subsection{Completely bounded maps}\label{app:cb-facts}

\begin{lemma}\label{lem:cb-duality}
For any $\sigma$-finite $(\Omega,\Sigma,\mu)$ and any $n\in\bN$,
\begin{equation}\label{eq:normal-cb-dual}
  \CB^\sigma(L_\infty(\Omega;\bM_d),\,\bM_n)=\bM_n(L_1(\Omega;S_1^d))
\end{equation}
completely isometrically, where $\CB^\sigma$ denotes the \emph{normal} (weak-$*$ continuous) completely bounded maps; moreover $\norm{\theta}_{\cb}=\norm{\theta}_{(n)}$ for every such $\theta$.
\end{lemma}
The identity \eqref{eq:normal-cb-dual} is the normal (weak-$*$ continuous) form of the operator-space mapping duality $\CB(V,\bM_n)=\bM_n(V^*)$, taken with the predual $L_1(\Omega;S_1^d)$ in place of the dual; see \cite[Prop.~4.2.5 and Thm.~7.2.4]{Effros2000}. The norm identity $\norm{\theta}_{\cb}=\norm{\theta}_{(n)}$ is Smith's lemma \cite[Proposition~8.11]{Paulsen_2003}.

\begin{lemma}[Tomiyama; see {\cite[Exercise~3.10(ii)]{Paulsen_2003}}]
\label{lem:dbound}
Let $X,Y$ be operator spaces.
For any bounded $T:X\to Y$ one has $\norm{T}_{(d)}\le d\,\norm{T}$.
\end{lemma}

\subsection{Noncommutative \texorpdfstring{$L_p$}{Lp}-spaces}\label{app:noncommLp}

The Schatten classes $S_p(H)$ admit a representation as noncommutative $L_p$-spaces associated with the von Neumann algebra $\cB(H)$ endowed with its canonical trace, i.e. $S_p \equiv L_p(\cB(H), \Tr)$.
More generally, for a von Neumann algebra $(M,\tau)$ with a faithful normal semifinite trace, the noncommutative $L_p$-space $L_p(M,\tau)$ is defined as the completion of $M$ under the norm
\begin{equation}
  \|x\|_{L_p(M,\tau)} := \tau(|x|^p)^{1/p}.
\end{equation}
When $M = \cB(\ell_2^d)$ and $\tau$ is the usual trace, one obtains the Schatten class $S_p^d$.
In particular, $S_1^d$ is the trace-class and $S_\infty^d$ coincides with the compact operators.

Moreover, $L_1$ and $S_1$ have Banach lattice structure, which behaves well under ultraproducts.
A \defn{Banach lattice} is a Banach space equipped with a partial order and lattice operations $(x\vee y,\, x\wedge y)$ satisfying the compatibility condition
\begin{equation}
    |x|\le |y| \implies \|x\|\le \|y\|.
\end{equation}
The lattice is an \defn{abstract $L_p$-space} if whenever $x,y\ge 0$ and $x\wedge y=0$, one has
\begin{equation}
    \|x+y\|^p = \|x\|^p + \|y\|^p.
\end{equation}

For a Banach space $X$, one may consider Bochner-valued $L_1$-spaces $L_1(\mu;X)$, which serve as classical analogues of the noncommutative spaces.
In the present work, we require only the basic case $S_1(X) := S_1 \otimes_\wedge X$.
More generally, Pisier~\cite{Pisier1998} introduced noncommutative Bochner spaces $L_1(\mu;X)$ and Schatten-valued spaces $S_1(X)$ using complex interpolation, producing operator-space structures that behave well under complete boundedness and tensor norms.
The operator-space structure on $L_1(\mu;X)$ is canonically identified with the operator-space projective tensor product $L_1(\mu)\otimes^\wedge X$; see~\cite[(7.1.6)]{Effros2000} for the Banach-space level and \cite[Chapter~2]{Pisier1998}.

Moreover, if $X$ is an operator space, then
  \begin{equation} \label{eq:L1-S1}
    L_1(\mu; S_1(X)) \;\cong\; S_1(L_1(\mu; X))
  \end{equation}
completely isometrically~\cite[Proposition~2.1(ii)]{Pisier1998}. 
In particular,  $L_1(\mu; S_1^d) \iso S_1^d(L_1(\mu))$ completely isometrically.
More generally, $L_p(\mu; S_p^d) \iso S_p^d(L_p(\mu))$ completely isometrically \cite[Proposition~2.1(ii)]{Pisier1998}.

\subsection{Operator Hilbert spaces}\label{app:OH}

We recall the definition \cite[\S1]{Pisier96}.
For an operator space $E$, the \defn{conjugate} $\overline E$ is the set $\{\overline x\mid x\in E\}$ with $\overline x+\overline y=\overline{x+y}$, $\lambda\overline x=\overline{\bar\lambda x}$ and matrix norms $\norm{(\overline{x_{ij}})}_{\bM_r(\overline E)}:=\norm{(x_{ij})}_{\bM_r(E)}$.
Let $K$ be a finite-dimensional Hilbert space with orthonormal basis $(e_i)_i$.
The \defn{operator Hilbert space} $OH(K)$ is $K$ with the matrix norms
\begin{equation}\label{eq:OH-norm}
  \norm{\sum_ia_i\otimes e_i}_{\bM_r(OH(K))}:=\norm{\sum_ia_i\otimes\overline{a_i}}_{\bM_r\otimes\bM_r}^{1/2},\qquad a_i\in\bM_r,\ r\in\bN,
\end{equation}
where $\overline{a_i}$ is the entrywise complex conjugate of $a_i$ and the right-hand side is the operator norm on $\bC^r\otimes\bC^r$.
These norms do not depend on the choice of the orthonormal basis, and for $r=1$ they give the Hilbert-space norm of $K$.
By \cite[Thm.~1.1]{Pisier96}, $OH(K)$ is the unique operator space isometric to $K$ for which the canonical identification of $OH(K)^*$ with $\overline{OH(K)}$ is a complete isometry.
An \defn{$OH$ space} is an operator space completely isometric to some $OH(K)$.
The proof of \Cref{prop:Td-norm} uses the following properties of $OH$ spaces.
\begin{enumerate}[(i)]
\item\label{it:OH-dual} $OH(K)^*=\overline{OH(K)}$ completely isometrically \cite[Thm.~1.1]{Pisier96}, and by the definition of the conjugate the matrix norms of $\overline{OH(K)}$ are given by \eqref{eq:OH-norm} with $\overline{e_i}$ in place of $e_i$.
\item\label{it:OH-sub} Every subspace of an $OH$ space, with the induced matrix norms, is an $OH$ space \cite[Prop.~1.5(i)]{Pisier96}.
\item\label{it:OH-hom} Every linear map $v$ between $OH$ spaces satisfies $\norm{v}_{\cb}=\norm{v}$ \cite[Prop.~7.2(iii)]{Pisier_2003}.
\item\label{it:OH-S2} The operator-space structure on $S_2^n$ obtained by complex interpolation between $\bM_n$ and $S_1^n$ is $OH$ for the Hilbert--Schmidt inner product \cite[Remark~1.11]{Pisier1998}, and $S_2^d(S_2^n)=S_2^{dn}$ \cite[Theorem~1.9]{Pisier1998}.
\end{enumerate}

\subsection{Ultraproducts of Banach spaces, operator spaces and noncommutative \texorpdfstring{$L_p$}{Lp}-spaces}\label{app:ultraprod}

As described in \Cref{sec:closure}, we must consider the limits of factorizations
\begin{equation}
  M^{(n)}_x = u^{(n)}_x \circ \Pi^{(n)},
\end{equation}
where the intermediate spaces $L_1(\Omega_n)$ vary with $n$.
Ordinary limit constructions are inadequate because the maps $u_n$ take values in different Banach spaces, which may not embed canonically into a common ambient space.
Ultraproducts provide a powerful functional-analytic tool to overcome this difficulty.
They allow one to construct a ``limit Banach space" from a sequence $\{X_n\}$ of Banach spaces even when the $X_n$ are not naturally comparable.
The key idea is to replace pointwise convergence by an \emph{ultralimit}, defined using a \emph{free ultrafilter}.
A free ultrafilter acts as a finitely additive probability measure that identifies a ``dominant" subsequence of indices; regarding this generalized limit notion, bounded sequences always have a well-defined ultralimit, even without a classical limit.

This subsection reviews the basic definitions of ultrafilters, ultralimits, and ultraproducts for Banach spaces, and then extends these constructions to the operator-space setting pertinent to quantum channels and instruments, including the spaces $L_1(\Omega_n;S_1^d)$.
We present the standard definitions and structural properties without proof;  classical references include \cite{Heinrich_1980,Diestel1977} for Banach-space ultraproducts and \cite[Section~10.3]{Effros2000} and \cite[Section~2.8]{Pisier_2003} for operator-space ultraproducts.

\paragraph{Ultrafilters and ultralimits.}

An \defn{ultrafilter} $\cU$ on a set $X$ is a family of subsets of $X$ satisfying:
  \begin{enumerate}[(i)]
    \item (Non-degeneracy) $\varnothing \not\in \cU$;
    \item (Monotonicity) If $A\in\cU$ and $A\subset B$, then $B\in\cU$;
    \item (Closure under intersection) If $A,B\in\cU$, then $A\cap B\in\cU$;
    \item[(iv)] (Dichotomy) For every $A\subset X$, exactly one of $A$ or $A^c$ lies in $\cU$.
  \end{enumerate}
An ultrafilter is \defn{non-principal} if it contains no finite sets.
Non-principal ultrafilters on $\mathbb{N}$ exist by the axiom of choice and may be identified with the points of the Stone--Čech compactification $\beta\mathbb{N}$.
The ultrafilters discussed in this paper are non-principal only.

  Let $(X,d)$ be a metric space and $\cU$ an ultrafilter on $\mathbb{N}$.
  A point $x\in \mathsf{X}$ is the \defn{ultralimit} (or $\cU$-limit) of a sequence $(x_n)$, denoted $\lim_{\cU} x_n = x$, if for every $\varepsilon>0$,
  \begin{equation}
    \{n\in\mathbb{N} \mid d(x_n,x) < \varepsilon\} \in \cU.
  \end{equation}
Ultralimits possess useful properties: they are unique, they agree with ordinary limits when they exist, and apply to all bounded sequences in compact spaces.

\paragraph{Ultraproducts of Banach spaces.}

  Let $\{X_i\}_{i\in\mathbb{N}}$ be Banach spaces and define
  \begin{equation}
    \ell_\infty(X_i) := \{ (x_i) \in \prod_i X_i \mid \sup_i \|x_i\| < \infty\}.
  \end{equation}
  Fix a non-principal ultrafilter $\cU$ on $\mathbb{N}$ and let
  \begin{equation}
    N := \{ (x_i) \in \ell_\infty(X_i) \mid \lim_{\cU} \|x_i\| = 0\}.
  \end{equation}
  The \defn{Banach space ultraproduct} $\uprod{X_i}$ is defined as the quotient
  \begin{equation}
    \uprod{X_i} := \ell_\infty(X_i) / N.
  \end{equation}

The class of $(x_i)$ in the ultraproduct is denoted $\uprod{x_i}$, and its norm is given by
\begin{equation}
\norm{\uprod{x_i}} := \lim_{\cU} \norm{x_i}.
\end{equation}

Ultraproducts interact well with linear maps: if $u_i : X_i \to Y_i$ are bounded linear maps, their \defn{ultraproduct} $\uprod{u_i} : \uprod{X_i}\to \uprod{Y_i}$ is defined by
\begin{equation}
\uprod{u_i}( \uprod{x_i} ) := \uprod{u_i(x_i)}.
\end{equation}
If $\sup_i\norm{u_i}<\infty$, then $\uprod{u_i}$ is bounded with norm~\cite[(2.8.3)]{Pisier_2003}
\begin{equation}
\norm{\uprod{u_i}} := \lim_{\cU} \norm{u_i}.
\end{equation}
When $X_i$'s are identical, $\uprod{X_i} = \uprod{X}$ is called \defn{ultrapower} and denoted $X^{\mathbb{N}}$.
If $X$ is finite dimensional, the constant sequence map $X \rightarrow{} \uprod{X}$ is an isometric isomorphism of Banach spaces~\cite[Lemma~10.3.1]{Effros2000}.

If each $X_i$ is a Banach lattice, then $\uprod{X_i}$ admits a natural Banach lattice structure, and if each $X_i$ is an abstract $L_p$-space, then $\uprod{X_i}$ is also an abstract $L_p$-space~\cite[\S8.b and Theorem~8.6]{pisier1986factorization}.
This compatibility of lattice structure with ultraproducts is useful when working with $L_1$ and $S_1$, both of which are Banach lattices.

\paragraph{Ultraproducts of $C^*$-algebras and von Neumann algebras.}

Ultraproducts can also be formed for $C^*$-algebras and von Neumann algebras.
Since positivity and complete positivity are stable under ultraproduct constructions, ultraproducts of completely positive maps naturally arise in the study of limits of quantum channels and measurements.
The following lemma is essential for constructing limiting factorizations of compatible POVMs and instruments.
\begin{lemma}\label{lem:ucp-ultraproduct}
  Let $u_i : A_i \to B_i$ be completely positive maps between $C^*$-algebras with $\sup_i\norm{u_i}<\infty$.
  Then their ultraproduct $\uprod{u_i}$ is completely positive.
\end{lemma}
This rests on the identification $\bM_k(\uprod{A_i})=\uprod{\bM_k(A_i)}$ of \cite[(2.8.1)]{Pisier_2003} and \cite[\S10.3]{Effros2000}, under which a positive class has a positive representative.
Indeed, a positive element of the $C^*$-algebra $\bM_k(\uprod{A_i})$ is of the form $y^*y$, and if $(y_i)_i$ represents $y$, then $(y_i^*y_i)_i$ is a positive representative.

\paragraph{Operator-space lemmas used for the instrument case.}

The following four operator-space results are used in the proof of \Cref{thm:closure-instruments}.
\Cref{lem:cp-ultra-os} says that ultraproducts of completely bounded maps are completely bounded with c.b.\ norm at most $\lim_\cU\norm{u_n}_{\cb}$.
It also shows that complete positivity passes to ultraproducts, for operator systems and for the spaces $L_1(\Omega_n;S_1^d)$.
\Cref{lem:pullout-finite} pulls a finite-dimensional operator space out of an ultraproduct through the operator-space projective tensor product.
\Cref{lem:diag-ultra-limit} is the compatibility between operator-norm convergence and ultraproducts.
And \Cref{lem:tensor-preserve-embed} promotes an order-isometric embedding of $L_1$-spaces to a complete isometry after tensoring with a finite-dimensional operator space.

\begin{lemma}\label{lem:pullout-finite}
Let $d\in\bN$ and let $(X_n)$ be a family of operator spaces.
Then there is a canonical complete isometry
\begin{equation}
\uprod{S_1^d\otimes^\wedge X_n} \;\cong\; S_1^d\otimes^\wedge \uprod{X_n}.
\end{equation}
In particular, if $\dim H'=d$ then $S_1(H')\cong S_1^d$ as operator spaces and
\begin{equation}
\uprod{S_1(H')\otimes^\wedge X_n} \;\cong\; S_1(H')\otimes^\wedge \uprod{X_n}
\qquad\text{completely isometrically.}
\end{equation}
\end{lemma}
This is \cite[Lemma~10.3.8]{Effros2000}, where $T_d$ denotes the trace class $S_1^d$.
It also follows from \cite[Theorems~1.1 and~1.5]{Pisier1998}, which describe $S_1^d\otimes^\wedge X$ through factorizations over $\bM_d(X)$ with factors of fixed size $d$, together with $\bM_k(\uprod{X_n})=\uprod{\bM_k(X_n)}$ \cite[(2.8.1)]{Pisier_2003}.

\begin{lemma} \label{L:up-L1-S1}
  Let $d\in\bN$ and let $(\Omega_n,\Sigma_n,\mu_n)$ be measure spaces.
  Then $\uprod{L_1(\Omega_n; S_1^d)} \iso S_1^d(\uprod{L_1(\Omega_n)})$ completely isometrically.
\end{lemma}
\begin{proof}
  By \eqref{eq:L1-S1},
  \begin{equation}
  L_1(\Omega_n;S_1^d)\iso S_1^d(L_1(\Omega_n))=S_1^d\otimes^\wedge L_1(\Omega_n)
  \end{equation}
  completely isometrically for every $n$, and the ultraproduct of these complete isometries is a complete isometry \cite[Proposition~10.3.2]{Effros2000}.
  Hence
  \begin{equation}
  \uprod{L_1(\Omega_n;S_1^d)}\iso\uprod{S_1^d\otimes^\wedge L_1(\Omega_n)},
  \end{equation}
  and \Cref{lem:pullout-finite} with $X_n=L_1(\Omega_n)$ identifies the latter completely isometrically with
  \begin{equation}
  S_1^d\otimes^\wedge\uprod{L_1(\Omega_n)}=S_1^d(\uprod{L_1(\Omega_n)}).
  \end{equation}
\end{proof}

\begin{lemma}\label{L:up-L1-S1-order}
  Let $d\in\bN$, let $(\Omega_n,\Sigma_n,\mu_n)$ be measure spaces, and let $f\in L_1(\Omega;S_1^d)$ be the image of a class $\uprod{z_n}\in\uprod{L_1(\Omega_n;S_1^d)}$ under the identification of \Cref{L:up-L1-S1} composed with an order isometry $\uprod{L_1(\Omega_n)}\iso L_1(\Omega,\Sigma,\mu)$ as in \cite[Theorem~8.6]{pisier1986factorization}.
  Then the class has a representative with $z_n\ge0$ for all $n$ if and only if $f \geq 0$ almost everywhere, and $\lim_\cU\Tr\int_{\Omega_n}z_n\,d\mu_n=\Tr\int_\Omega f\,d\mu$.
\end{lemma}
\begin{proof}
  If every $z_n$ is positive, then so is $\langle\xi,z_n\xi\rangle\in L_1(\Omega_n)$ for every $\xi\in\bC^d$.
  The order isometry sends the class of these functions to the positive function $\langle\xi,f\xi\rangle$, and letting $\xi$ run through a countable dense subset of $\bC^d$ shows that $f$ has positive semidefinite values almost everywhere.
  Conversely, such an $f$ is a norm limit of sums $\sum_lP_l\otimes g_l$ with $P_l\ge0$ and $g_l\ge0$, and the class of each such sum has a positive representative.
  The classes with a positive representative form a closed set.
  Indeed, if $\uprod{z_n}$ is a limit of such classes, then $\lim_\cU\operatorname{dist}(z_n,P_n)=0$ for the closed cones $P_n$ of positive elements, and $p_n\in P_n$ with $\norm{z_n-p_n}\le2\operatorname{dist}(z_n,P_n)$ form a bounded positive representative.
  For a positive class both sides of the trace identity equal its norm, and every class is a combination of four positive classes, which gives the trace identity.
\end{proof}

\begin{lemma}\label{lem:cp-ultra-os}
Let $(E_n)$ and $(F_n)$ be operator spaces and let $u_n:E_n\to F_n$ be completely bounded maps with $\sup_n\|u_n\|_{cb}<\infty$.
Then the ultraproduct map
\begin{equation}
\uprod{u_n}: \uprod{E_n}\to \uprod{F_n},\qquad \uprod{x_n}\mapsto \uprod{u_n(x_n)},
\end{equation}
is completely bounded and satisfies
\begin{equation}
\norm{\uprod{u_n}}_{\cb}\le\lim_\cU \norm{u_n}_{\cb}.
\end{equation}

If, in addition, each $E_n$ and $F_n$ is an operator system (or $C^*$-algebra) and each $u_n$ is completely positive (respectively, unital completely positive), then $\uprod{u_n}$ is completely positive (respectively, unital completely positive).

The same holds for completely positive $u_n:L_1(\Omega_n;S_1^d)\to L_1(\Omega_n';S_1^{d'})$ with $d,d'\in\bN$ fixed, where a class is positive if it has a positive representative, and $\uprod{u_n}$ is trace preserving if every $u_n$ is.
For a finite set $\Omega'$, the diagonal embedding of $L_1(\Omega';S_1^{d'})$ into its ultrapower is onto, and it and its inverse are completely positive and trace preserving.
\end{lemma}
\begin{proof}
The inequality is \cite[(2.8.4)]{Pisier_2003}, see also \cite[Proposition~10.3.2]{Effros2000}.
Let $E_n=L_1(\Omega_n;S_1^d)$ and $F_n=L_1(\Omega_n';S_1^{d'})$; complete positivity of $u_n$ means that $\id_{S_1^k}\otimes u_n$ is positive on $S_1^k\otimes^\wedge E_n=L_1(\Omega_n;S_1^{kd})$ for every $k$.
By \Cref{lem:pullout-finite} and \eqref{eq:L1-S1}, $S_1^k\otimes^\wedge\uprod{E_n}=\uprod{L_1(\Omega_n;S_1^{kd})}$, and this identification sends $e\otimes\uprod{x_n}$ to $\uprod{e\otimes x_n}$, so it carries $\id_{S_1^k}\otimes\uprod{u_n}$ to $\uprod{\id_{S_1^k}\otimes u_n}$.
By \Cref{L:up-L1-S1-order} with $kd$ in place of $d$, the classes with a positive representative are those whose images in $L_1(\Omega;S_1^{kd})$ are positive semidefinite almost everywhere.
The map $\uprod{\id_{S_1^k}\otimes u_n}$ sends a class with positive representative $(z_n)$ to the class of the positive elements $(\id_{S_1^k}\otimes u_n)(z_n)$, and by the trace identity of \Cref{L:up-L1-S1-order} it preserves $\Tr\otimes\int$ when every $u_n$ does.
For finite $\Omega'$, the diagonal embedding is onto because $L_1(\Omega';S_1^{d'})$ is finite-dimensional \cite[Lemma~10.3.1]{Effros2000}.
Constant representatives show that it is completely positive and trace preserving, and a class with positive representative $(p_n)$ is the diagonal image of the positive element $\lim_\cU p_n$.
\end{proof}

\begin{lemma}\label{lem:diag-ultra-limit}
Let $X$ and $Y$ be Banach spaces, and suppose $T_n:X\to Y$ converges in operator norm to $T$.
Let $\Delta_X:X\to\uprod{X}$ and $\Delta_Y:Y\to \uprod{Y}$ be the diagonal embeddings
$\Delta_X(x)=\uprod{x}$ and $\Delta_Y(y)=\uprod{y}$.  Then
\begin{equation}
\uprod{T_n}\circ \Delta_X = \Delta_Y\circ T.
\end{equation}
\end{lemma}
This follows directly from the definition of the ultraproduct norm \cite[\S2.8]{Pisier_2003}, and the diagonal embeddings are isometric \cite[Lemma~10.3.1]{Effros2000}.

\begin{lemma}\label{lem:tensor-preserve-embed}
Let $E$ be a finite-dimensional operator space and let $\jmath:L_1(\nu)\hookrightarrow L_1(\Omega)$
be an order-isometric embedding with a positive contractive left inverse $P$.  Equip $L_1(\cdot)$ with its canonical dual operator-space
structure coming from $L_\infty(\cdot)^*$.
Then the induced map
\begin{equation}
\id_E\otimes \jmath:\;E\otimes^\wedge L_1(\nu)\longrightarrow E\otimes^\wedge L_1(\Omega)
\end{equation}
is a complete isometry.
\end{lemma}
\Cref{lem:tensor-preserve-embed} follows by applying \cite[Theorem~3.9]{Paulsen_2003} to the adjoints of $\jmath$ and $P$, which map into commutative $C^*$-algebras, then \cite[Proposition~3.2.2]{Effros2000} to see that $\jmath$ and $P$ are complete contractions, and \cite[Corollary~7.1.3]{Effros2000} to $\id_E\otimes\jmath$ and $\id_E\otimes P$, since $(\id_E\otimes P)\circ(\id_E\otimes\jmath)=\id$.

\subsection{Factorization lemmas}\label{app:pisier}

\begin{lemma}(\cite[Corollary~8.7(i)]{pisier1986factorization})\label{lem:Lp-factorization}
  Consider a family of linear maps of Banach spaces $\{u_i: X_i\to Y_i\}_{i \in \mathbb{N}}$.
  For a fixed $p \in [1,\infty)$, if each $u_i$ factors through $L_p$ as $u_i=q_i \circ p_i$ with $\sup_i\norm{p_i}<\infty$ and $\sup_i\norm{q_i}<\infty$, then the ultraproduct  $\uprod{u_i}$ factors through a $L_p(M)$ space over some measure space $M$.
\end{lemma}

In what follows, $S_p(E)$ denotes Pisier's vector-valued Schatten class over an operator space $E$ (with finite level $S_p^n(E)$, so that $S_\infty^n(E)=\bM_n(E)$), and a linear map $u:E\to F$ is \defn{completely $p$-summing} if $\id_{S_p}\otimes u:S_p\otimes_{\min}E\to S_p(F)$ is bounded, with $\pi_p^{\circ}(u)$ denoting its norm.

\begin{equation}
  \pi_p^\circ(u) := \norm{\id_{S_p} \otimes u}_{S_p \otimes_{\min} E \to S_p(F)}
\end{equation}

\begin{lemma}[{\cite[Theorem~1.5]{Pisier1998}}]\label{thm:pisier-1.5}
Let $1\le p<\infty$ and let $u\in S_p^n(E)$ with associated matrix $(u_{ij})$, $u_{ij}\in E$.
Then
\begin{equation}
\norm{u}_{S_p^n(E)}
=\inf \{\,\norm{a}_{S_{2p}^n}\,\norm{v}_{\bM_n(E)}\,\norm{b}_{S_{2p}^n}\ |\ (u_{ij})=a\,v\,b,\ \ a,b\in S_{2p}^n,\ v\in\bM_n(E)\,\}.
\end{equation}
\end{lemma}

\begin{lemma}[{\cite[Theorem~5.1]{Pisier1998}}] \label{thm:pisier-5.1}
Assume $E\subseteq\cB(H)$, let $u:E\to F$ be completely $p$-summing ($1\le p<\infty$), and let $C=\pi_p^{\circ}(u)$.
Then there are an ultrafilter $\cU$ over an index set $\mathsf{J}$ and families $(a_\alpha)_{\alpha\in\mathsf{J}}$, $(b_\alpha)_{\alpha\in\mathsf{J}}$ in the unit ball of $S_{2p}(H)$ such that for all $n$ and all $(z_{ij})\in\bM_n(E)$,
\begin{equation}
\norm{(u(z_{ij}))}_{S_p^n(F)}\le C\,\lim_{\cU}\norm{(a_\alpha\, z_{ij}\, b_\alpha)}_{S_p(\ell_2^n\otimes H)}
\quad\text{and}\quad
\norm{(u(z_{ij}))}_{\bM_n(F)}\le C\,\lim_{\cU}\norm{(a_\alpha\, z_{ij}\, b_\alpha)}_{\bM_n(S_p(H))}.
\end{equation}
Conversely, if $u$ satisfies the second inequality, then $u$ is completely $p$-summing with $\pi_p^{\circ}(u)\le C$.
\end{lemma}

\begin{lemma}[{\cite[Proposition~5.6]{Pisier1998}}]\label{prop:pisier-5.6}
Let $K$ be a Hilbert space, let $a,b\in S_{2p}(K)$, and let $M_{a,b}:\cB(K)\to S_p(K)$ be defined by $M_{a,b}(z)=azb$.
Then
\begin{equation}
\norm{M_{a,b}}_{\cb}\ \le\ \pi_p^{\circ}(M_{a,b})\ \le\ \norm{a}_{S_{2p}(K)}\,\norm{b}_{S_{2p}(K)}.
\end{equation}
\end{lemma}

\subsection{Separable subspaces of $L_1$ via Maharam's classification}\label{app:maharam}

Maharam's classification of abelian von Neumann algebras provides the geometric ingredient that, combined with Pisier's factorization, lets us pass from an abstract ultraproduct $L_1(\Omega)$ to the concrete probability space $L_1[0,1]$.
The reason is that the image of $S_1(H)$ under the ultraproduct of parent measurements is separable, and a separable subspace of any $L_1$ lies in an $L_1$-space that embeds order-isometrically into $L_1[0,1]$.
This is the content of \Cref{prop:L1-maharam} below.

We first present the key lemma used for the proof in \Cref{prop:L1-maharam}.
\begin{lemma}\label{lem:sig-subalg}
    Let $X \subset L_1(\Omega,\Sigma,\mu)$ be a separable subspace. Then,
    there exists a countably generated
    $\sigma$-subalgebra $\Sigma_X \subset \Sigma$ such that
    $X\subset L_1(\Omega,\Sigma_X, \mu) \subset L_1(\Omega, \Sigma,\mu)$.
\end{lemma}
\begin{proof}
Let $S = \{x_n\}_n\subset X$ be the countable dense set, and write $x_n = (a_n-b_n) + i(c_n-d_n)$
with non-negative real functions $\{a_n,b_n,c_n,d_n\}_{n\in \bN}$.
Define the set $A_{n,p_1,p_2,r} \subset \Sigma$
\begin{equation}
    A_{n,p_1,p_2,r} := \{ a_n^{-1}(\cB(p_1+ip_2, r)) \mid p_1,p_2, r\in \bQ\},
\end{equation}
where $\cB(p_1+ip_2,r)$ is the open ball with center $p_1+ip_2$ and radius $r$ in $\bC$,
and define $B_{n,p_1,p_2,r}$, $C_{n,p_1,p_2,r}$ and $D_{n, p_1,p_2,r}$ similarly with
$b_n, c_n$ and $d_n$ respectively. Let $\Sigma_X$ be the $\sigma$-algebra generated by
\begin{equation}
E = \{A_{n,p_1,p_2,r}, B_{n,p_1,p_2,r}, C_{n,p_1,p_2,r}, D_{n,p_1,p_2,r}\mid n\in \bN, p_1,p_2,r\in \bQ\}.
\end{equation}
For any $y\in X$, there exists $(y_n)_n \subset S$ be a sequence that converges
to $y$, and write $y_n = (a_n-b_n)+i(c_n-d_n)$.
By construction the functions $a_n,b_n,c_n,d_n$ are $\Sigma_X$-measurable, hence so is $y_n$,
and therefore $y_n \in L_1(\Omega, \Sigma_X, \mu)$, which shows $X \subset L_1(\Omega, \Sigma_X, \mu)$.

Here $L_1(\Omega,\Sigma_X,\mu)$ is closed in $L_1(\Omega,\Sigma,\mu)$.
Indeed, the inclusion is isometric, since both norms are given by the same integral.
Moreover $L_1(\Omega,\Sigma_X,\mu)$ is complete by the Riesz--Fischer theorem \cite[Theorem~3.11]{rudin1987real}, which holds for every measure.
A complete subspace of a metric space is closed, so the limit $y$ of the $y_n$ lies in $L_1(\Omega,\Sigma_X,\mu)$.

Since $\Sigma_X$ is the $\sigma$-subalgebra of $\Sigma$, we have the inclusion $L_1(\Omega, \Sigma_X, \mu) \subset L_1(\Omega, \Sigma, \mu)$.
\end{proof}

The variant of Maharam's classification theorem on Abelian von Neumann algebra \cite{Takesaki_1979} states that if $H$ is a separable Hilbert space and $\cM \subset \cB(H)$ is an abelian von Neumann algebra, then there exists a $*$-isomorphism $\varphi:\cM\to \cN$ preserving weak operator topology where $\cN$ is one of the following:
\begin{enumerate}
\item $\ell^\infty(\{1,\cdots, n\}), n\geq 1$
\item $\ell^\infty(\bN)$
\item $L^\infty[0,1]$
\item $L^\infty[0,1] \oplus \ell^\infty(\{1,\cdots, n\}), n\geq 1$
\item $L^\infty[0,1] \oplus \ell^\infty(\bN)$.
\end{enumerate}
Based on this theorem, we have the following key results that we use for the proof of \Cref{thm:closurePOVM}.

\begin{proposition}\label{prop:L1-maharam}
If $X \subset L_1(\Omega, \Sigma, \mu)$ is a separable subspace, then there exist a countably generated $\sigma$-subalgebra $\Sigma_X\subset\Sigma$ and an order-isometric embedding $j:L_1(\Omega,\Sigma_X,\mu)\embeds L_1[0,1]$ such that
\begin{equation}
    X \subset L_1(\Omega,\Sigma_X,\mu) \overset{j}{\embeds} L_1[0,1].
\end{equation}
Moreover, there is a positive map $\cE:L_1[0,1]\to L_1(\Omega,\Sigma_X,\mu)$ that preserves integrals and satisfies $\cE\circ j=\id$.
\end{proposition}

\begin{proof} Let $L_1(\Omega, \Sigma_X,\mu)$ be the induced subspace from
\Cref{lem:sig-subalg} such that
\begin{equation}
    X \subset L_1(\Omega, \Sigma_X, \mu) \subset L_1(\Omega, \Sigma, \mu).
\end{equation}
Note that if $(\Omega, \Sigma', \mu)$ is a measure space with $\mu$ $\sigma$-finite and $\Sigma'$ is countably generated, then $L_p(\Omega, \Sigma', \mu)$ is separable for all $p \in [1, \infty)$ \cite[Proposition 3.4.5]{Cohn1994}.
Since $\Sigma_X$ is countably generated, $L_2(\Omega, \Sigma_X, \mu)$ is separable.

$L_\infty(\Omega, \Sigma_X) \subset \cB(L_2(\Omega, \Sigma_X))$ is an abelian von Neumann algebra and by the classification theorem, we have an order-isometric isomorphism $L_\infty(\Omega, \Sigma_X) \iso \ell_\infty(\bN, \mu_d) \oplus L_\infty([0,1], \mu_c)$, where $\mu_d$ corresponds to the atomic measure and $\mu_c$ the non-atomic measure.
Because von Neumann algebras are uniquely determined by their preduals, we have an order-isometric isomorphism
\begin{equation}
    L_1(\Omega, \Sigma_X) \iso \ell_1(\bN) \oplus L_1[0,1],
\end{equation}
which naturally embeds into $L_1[0,1] \oplus L_1[0,1] \embeds L_1[0,2] \iso L_1[0,1]$.
\end{proof}

\subsection{Convex analysis}\label{app:convex}

\begin{lemma}[{Hahn-Banach separation for open sets \cite[Theorem~3.4(a)]{rudin1991functional}}]\label{thm:HB-open}
  Let $X$ be a topological vector space, $A \subseteq X$ a non-empty convex, open set, and $B\subseteq X$ a non-empty convex set such that $A \cap B = \emptyset$.
  Then there exists a continuous linear functional $\psi \neq 0$ and a scalar $\alpha \in \bR$ such that
  \begin{equation}
    \psi(a) \leq \alpha \leq \psi(b)
  \end{equation}
  for any $a \in A$ and $b \in B$.
\end{lemma}

\begin{lemma}[{Riesz-Markov-Kakutani representation theorem \cite[Theorem~2.14]{rudin1987real}}] \label{thm:RMK}
  Let $X$ be a locally compact Hausdorff space and $\psi$ be a positive linear functional on $C_c(X,\bR)$, the real continuous functions on $X$ with compact support.
  Then there exists a unique positive Borel measure $\mu$ on $X$ such that
  \begin{equation}
    \psi(f) = \int_X f(z)\;d\mu(z)
  \end{equation}
  for any $f\in C_c(X,\bR)$, such that for some $\Sigma$ containing the Borel $\sigma$-algebra
  \begin{enumerate}[(1)]
    \item $(X,\Sigma,\mu)$ is a complete measure space;
    \item $\mu(K) < \infty$ for any compact $K \subseteq X$;
    \item (outer regularity) $\mu(E) = \inf \{\mu(U) \mid E \subset U, \text{ $U$ open }\}$ for any Borel set $E \in \Sigma$;
    \item (inner regularity) $\mu(E) = \sup \{\mu(K) \mid K \subset E, \text{ $K$ compact } \}$ for any $E \in \Sigma$ that is open or Borel set with $\mu(E) < \infty$.
  \end{enumerate}
\end{lemma}

The following form of the bipolar theorem is the step behind \Cref{lem:cp-polar-form} below and behind the proof of \Cref{prop:robustness-dual}.

\begin{lemma}[Minkowski functional and polar]\label{lem:gauge-polar}
Let $V$ be a finite-dimensional real vector space in duality with $V'$ through a nondegenerate bilinear pairing $\ip{\cdot}{\cdot}$, and let $K\subset V$ be closed, convex, and contain $0$.
Then for every $v\in V$
\begin{equation}\label{eq:gauge-polar}
  \inf\{t>0\mid v\in tK\}=\sup\{\ip{y}{v}\mid y\in K^\circ\},
  \qquad K^\circ:=\{y\in V'\mid\ip{y}{z}\leq1\ \text{for all}\ z\in K\},
\end{equation}
both sides being $+\infty$ when $v\notin tK$ for every $t>0$.
\end{lemma}

\begin{proof}
By the bipolar theorem $K^{\circ\circ}=K$, since $K$ is closed, convex and contains $0$.
Hence for $t>0$, $v\in tK$ if and only if $v/t\in K^{\circ\circ}$, that is, if and only if $\ip{y}{v}\leq t$ for all $y\in K^\circ$.
So the left-hand side of \eqref{eq:gauge-polar} is the infimum of the $t>0$ that dominate the right-hand side, which is the right-hand side itself because $0\in K^\circ$ makes it nonnegative.
\end{proof}

\section{Proofs of the closure theorems}\label{app:closure}

This appendix proves \Cref{thm:closurePOVM,thm:closure-instruments}, using the ultraproducts of \Cref{app:ultraprod}, Pisier's factorization of \Cref{app:pisier} and the Maharam-based embedding of \Cref{app:maharam}.

\subsection{Proof of \texorpdfstring{\Cref{thm:closurePOVM}}{Theorem~\ref{thm:closurePOVM}}}

The main idea is to frame the sequence of POVMs as the ultraproduct of POVMs and to use the fact that, thanks to the separability of the input Hilbert space $H$, the ultraproduct of $L_1$ spaces can be replaced by the concrete probability space $L_1[0,1]$.
One key ingredient is Pisier's factorization of ultraproducts of $L_1$ spaces (\Cref{lem:Lp-factorization}); the other is the order-isometric embedding into $L_1[0,1]$ of the $L_1$-space generated by a separable subspace (\Cref{prop:L1-maharam}).

\begin{theorem*}[restated]
\maintheoremPOVM
\end{theorem*}

\begin{proof}
For each $n\in\mathbb{N}$, the multi-meter $\cM^{(n)}$ is compatible, so there exist a measure space $(\Omega_n,\Sigma_n,\mu_n)$, a parent measurement $\Pi^{(n)} : S_1(H) \to L_1(\Omega_n)$, and stochastic postprocessing maps $u^{(n)}_x : L_1(\Omega_n) \to \ell_1^m \qquad (x\in\mathsf{X})$ such that
\begin{equation}
  M^{(n)}_x = u^{(n)}_x \circ \Pi^{(n)} \qquad (x\in\mathsf{X}).
\end{equation}
and $\|\Pi^{(n)}\| = \|u^{(n)}_x\| = 1$ for all $n$ and $x$.

Fix a non-principal ultrafilter $\cU$ on $\mathbb{N}$.
Consider the Banach-space ultraproduct
\begin{equation}
  X := \uprod{L_1(\Omega_n)}.
\end{equation}
By \mbox{\Cref{lem:Lp-factorization}}, $X$ is an abstract $L_1$-space, and by Kakutani's representation theorem there is an order isometry of $X$ onto $L_1(\Omega,\Sigma,\mu)$ for some measure space $(\Omega,\Sigma,\mu)$, see \cite[Theorem~8.6]{pisier1986factorization}.

We fix such an order-isometric identification
\begin{equation}
  \uprod{L_1(\Omega_n)} \cong L_1(\Omega).
\end{equation}

Next, form the ultraproduct maps
\begin{align}
  &\uprod{\Pi^{(n)}} : \uprod{S_1(H)} \to \uprod{L_1(\Omega_n)} \cong L_1(\Omega)\\
  &\uprod{u^{(n)}_x} : \uprod{L_1(\Omega_n)} \cong L_1(\Omega) \to \uprod{\ell_1^m}\\
  &\uprod{M^{(n)}_x} : \uprod{S_1(H)} \to \uprod{\ell_1^m}
\end{align}
By construction,
\begin{equation}
  \uprod{M^{(n)}_x} = \uprod{u^{(n)}_x} \circ \uprod{\Pi^{(n)}} \qquad (x\in\mathsf{X}).
\end{equation}

Let $\Delta_{H} : S_1(H) \to \uprod{S_1(H)}$ and $\Delta_m : \ell_1^m \to \uprod{\ell_1^m}$ denote the diagonal isometric embeddings
\begin{equation}
  \Delta_{H}(\rho) := \uprod{\rho}, \qq \Delta_m(v) := \uprod{v},
\end{equation}

Because $M^{(n)}_x \to M_x$ in operator norm, for each fixed
$\rho\in S_1(H)$ the sequence $M^{(n)}_x(\rho)$ converges in $\ell_1^m$
to $M_x(\rho)$.
Hence the ultralimit coincides with the usual limit, and therefore
\begin{equation}
  \uprod{M^{(n)}_x(\rho)} = \uprod{M_x(\rho)}
\end{equation}
in $\uprod{\ell_1^m}$.
Equivalently,
\begin{equation}\label{eq:ultra-ident}
  \uprod{M^{(n)}_x} \circ \Delta_{H} = \Delta_m \circ M_x \qquad (x\in\mathsf{X}).
\end{equation}

Now consider the separable subspace
\begin{equation}
  Y := \overline{ \uprod{\Pi^{(n)}} ( \Delta_{H}(S_1(H)) ) } \subset L_1(\Omega).
\end{equation}
Since $S_1(H)$ is separable and $\uprod{\Pi^{(n)}}$ is bounded, $Y$ is separable.
By \Cref{prop:L1-maharam} there are a countably generated $\sigma$-subalgebra $\Sigma_X$ with $Y\subset L_1(\Omega,\Sigma_X,\mu)$, an order-isometric embedding $j:L_1(\Omega,\Sigma_X,\mu)\hookrightarrow L_1[0,1]$, and a positive integral-preserving map $\cE:L_1[0,1]\to L_1(\Omega,\Sigma_X,\mu)$ with $\cE\circ j=\id$.
We regard $L_1(\Omega,\Sigma_X,\mu)$ as a subspace of $L_1(\Omega)$.

Define a map $\Pi : S_1(H) \to L_1[0,1]$ by
\begin{equation}
  \Pi := j \circ \uprod{\Pi^{(n)}} \circ \Delta_{H}.
\end{equation}
Since $\Delta_H$ and $j$ are isometric embeddings, $j$ preserves integrals of positive functions, and $\uprod{\Pi^{(n)}}$ is the ultraproduct of positive contractions, $\Pi$ is a positive, trace-preserving map.
Thus $\Pi$ is a valid parent measurement in the Schr\"odinger picture.

For each $x\in\mathsf{X}$, define $u_x : L_1[0,1] \to \ell_1^m$ by 
\begin{equation}
  u_x := \Delta_m^{-1}\circ\uprod{u^{(n)}_x}\circ\cE.
\end{equation}
This is well defined because $\ell_1^m$ is finite-dimensional, so $\Delta_m$ is onto $\uprod{\ell_1^m}$.
Each $u_x$ is positive, because $\cE$ is positive and a positive class in $L_1(\Omega)$ has positive representatives.
It preserves total mass, because $\cE$ preserves integrals and each $u^{(n)}_x$ is stochastic.
So each $u_x$ is a stochastic postprocessing map.

We now show that $M_x = u_x\circ \Pi$.
Using \eqref{eq:ultra-ident} and the definitions above, we compute:
\begin{equation}
  \Delta_m \circ M_x
  = \uprod{M^{(n)}_x} \circ \Delta_H
  = \uprod{u^{(n)}_x} \circ \uprod{\Pi^{(n)}} \circ \Delta_H
  = \uprod{u^{(n)}_x} \circ \cE\circ j\circ\uprod{\Pi^{(n)}} \circ \Delta_H
  = \Delta_m \circ u_x \circ \Pi.
\end{equation}
Since $\Delta_m$ is an isometric embedding, it is injective, and we conclude
\begin{equation}
  M_x = u_x \circ \Pi \qquad (x\in\mathsf{X}).
\end{equation}

Thus the limiting multi-meter
\begin{equation}
  \cM = \{M_x : S_1(H)\to\ell_1^m\}_{x\in\mathsf{X}}
\end{equation}
admits a common parent measurement $\Pi$ and stochastic postprocessings $u_x$, and is therefore compatible.
\end{proof}

\subsection{Proof of \Cref{thm:closure-instruments}}

We now extend the closure result for compatible POVMs to the case of quantum instruments.
Throughout this proof, $H$ denotes a separable Hilbert space and $K$ a finite-dimensional Hilbert space, and we view instruments in the Schr\"{o}dinger picture as completely positive, trace-preserving maps
\begin{equation}
    T: S_1(H) \to S_1(K) \otimes^\wedge L_1(\Omega),
\end{equation}
where $\otimes^\wedge$ denotes the operator-space projective tensor product, as in the preliminaries.

As in the compatible measurement case, the challenge in taking limits is that the intermediate Banach/operator spaces $L_1(\Omega_n)$ may depend on $n$ and have no canonical embedding into a common space.

\begin{proof}
Let $\cU$ be a non-principal ultrafilter on $\bN$.
For each $x \in \mathsf{X}$, assume $\Phi^{(n)}_x = \Psi^{(n)}_x \circ \theta^{(n)}$ with an instrument $\theta^{(n)}: S_1(H) \to S_1^d \otimes^{\wedge} L_1(\Omega_n)$ and a quantum channel $\Psi^{(n)}_x:S_1^d \otimes^{\wedge} L_1(\Omega_n) \to S_1(K) \otimes^\wedge \ell_1^m$.
Consider the ultraproduct maps
\begin{align}
&\uprod{\Phi^{(n)}_x}:\ \uprod{S_1(H)} \to \uprod{S_1(K)\otimes^\wedge \ell_1^m}\\
&\uprod{\theta^{(n)}}:\ \uprod{S_1(H)}\to \uprod{S_1^d\otimes^\wedge L_1(\Omega_n)}\\
&\uprod{\Psi^{(n)}_x}:\ \uprod{S_1^d\otimes^\wedge L_1(\Omega_n)} \to \uprod{S_1(K)\otimes^\wedge \ell_1^m}.
\end{align}

By \Cref{lem:cp-ultra-os}, each $\uprod{\Psi^{(n)}_x}$ is completely positive and trace-preserving, and we have
\begin{equation}\label{eq:ultra-factorization}
\uprod{\Phi^{(n)}_x} \;=\; \uprod{\Psi^{(n)}_x}\circ \uprod{\theta^{(n)}}
\qquad\text{for all } x\in\mathsf{X}.
\end{equation}
Here $S_1^d\otimes^\wedge L_1(\Omega_n)=L_1(\Omega_n;S_1^d)$ and $S_1(K)\otimes^\wedge\ell_1^m=L_1(\{1,\dots,m\};S_1(K))$ by \eqref{eq:L1-S1}, so \Cref{lem:cp-ultra-os}, whose $L_1$ part rests on \Cref{L:up-L1-S1,L:up-L1-S1-order}, applies with $d'=\dim K$.
By the last part of \Cref{lem:cp-ultra-os} the inverse $\Delta_F^{-1}$ of the diagonal embedding of the finite-dimensional $F$ defined below is completely positive and trace preserving.

Since $d$ is finite, the ultraproduct stability of the operator space projective tensor product, and applying \Cref{lem:pullout-finite}, yields a complete isometry
\begin{equation}\label{eq:pullout}
\uprod{S_1^d\otimes^\wedge L_1(\Omega_n)} \;\cong\; \uprod{S_1^d} \otimes^\wedge \uprod{L_1(\Omega_n)} \;\cong\; S_1^d\otimes^\wedge \uprod{L_1(\Omega_n)}.
\end{equation}
The ultraproduct of $L_1$ spaces is an $L_1$ space and there exists a measure space $(\Omega,\Sigma,\mu)$ and an order isometry \cite[Theorem 8.6]{pisier1986factorization}
\begin{equation}\label{eq:L1-ultra-represent}
\uprod{L_1(\Omega_n)} \;\cong\; L_1(\Omega,\Sigma,\mu).
\end{equation}

Transporting the operator-space structure along this identification, we may regard the common intermediate operator space as
\begin{equation}
E \;:=\; S_1^d\otimes^\wedge L_1(\Omega,\Sigma,\mu),
\end{equation}
and rewrite $\uprod{\theta^{(n)}}$ as a map $\uprod{\theta^{(n)}}:\ \uprod{S_1(H)}\to E$.
Likewise, each $\uprod{\Psi^{(n)}_x}$ becomes a completely positive trace-preserving map $\uprod{\Psi^{(n)}_x}:\ E\to \uprod{S_1(K)\otimes^\wedge \ell_1^m}$.

Let $\Delta_{H}:S_1(H)\to \uprod{S_1(H)}$ denote the diagonal embedding and define
\begin{equation}
\widetilde{\theta} := \uprod{\theta^{(n)}}\circ \Delta_{H}:\ S_1(H)\to E.
\end{equation}
Thus $\widetilde\theta(\rho)$ is the class of $(\theta^{(n)}(\rho))_n$.
For positive $\rho\in S_1^k\otimes^\wedge S_1(H)$ the class $(\id_{S_1^k}\otimes\widetilde\theta)(\rho)$ has the positive representative $((\id_{S_1^k}\otimes\theta^{(n)})(\rho))_n$, so $\widetilde\theta$ is completely positive by \Cref{L:up-L1-S1-order} with $kd$ in place of $d$, and it is trace preserving by the trace identity there.
Since $S_1(H)$ is separable and $\widetilde{\theta}$ is bounded linear, its range $\widetilde{\theta}(S_1(H))$ is separable and its closure
\begin{equation}
Y := \overline{\widetilde{\theta}(S_1(H))}\subseteq E
\end{equation}
is a separable subspace of $E$.

Choose a basis $\{e_\alpha\}_{\alpha=1}^{d^2}$ of $S_1^d$ and let $\{\phi_\alpha\}$ be the dual basis in $\cB(\bC^d)$.
For each $y\in Y$, the coefficients
\begin{equation}
y_\alpha := (\phi_\alpha\otimes \id_{L_1(\Omega)})(y)\in L_1(\Omega)
\qquad(\alpha=1,\dots,d^2)
\end{equation}
belong to a separable subspace of $L_1(\Omega)$ as $y$ ranges over $Y$.
Let $Z\subseteq L_1(\Omega)$ be the closed separable subspace generated by all such coefficients $\{y_\alpha:\ y\in Y,\ \alpha\le d^2\}$.
Then
\begin{equation}\label{eq:Y-in-tensor}
Y \subseteq S_1^d\otimes^\wedge Z \subseteq S_1^d\otimes^\wedge L_1(\Omega).
\end{equation}

Apply \Cref{prop:L1-maharam} to the separable subspace $Z\subseteq L_1(\Omega)$.
We obtain a countably generated $\sigma$-subalgebra $\Sigma_Z$ with $Z\subseteq L_1(\Omega,\Sigma_Z,\mu)$, an order-isometric embedding $j:L_1(\Omega,\Sigma_Z,\mu)\hookrightarrow L_1[0,1]$, and a positive integral-preserving map $\cE:L_1[0,1]\to L_1(\Omega,\Sigma_Z,\mu)$ with $\cE\circ j=\id$.
Since $\cE$ is a positive contractive left inverse of $j$, \Cref{lem:tensor-preserve-embed} shows that
\begin{equation}
\id_{S_1^d}\otimes j:\ S_1^d\otimes^\wedge L_1(\Omega,\Sigma_Z,\mu)\to S_1^d\otimes^\wedge L_1[0,1]
\end{equation}
is a complete isometry.

Since $Z\subseteq L_1(\Omega,\Sigma_Z,\mu)$, the inclusion \eqref{eq:Y-in-tensor} implies that $Y\subseteq S_1^d\otimes^\wedge L_1(\Omega,\Sigma_Z,\mu)$.
Define a new map
\begin{equation}
\theta \;:=\; (\id\otimes j)\circ \widetilde{\theta}:\ S_1(H)\to S_1^d\otimes^\wedge L_1[0,1].
\end{equation}
Then $\theta$ is completely positive and trace-preserving, because $\widetilde{\theta}$ is c.p. and trace-preserving and $\id_{S_1^d}\otimes j$ is completely positive and trace preserving, since $j$ is positive and preserves integrals.

The rest of the proof follows in analogous way to the compatible measurement case.
Namely, with $F:=S_1(K)\otimes^\wedge\ell_1^m$, which is finite-dimensional, the post-processings are $\Psi_x:=\Delta_F^{-1}\circ\uprod{\Psi^{(n)}_x}\circ(\id_{S_1^d}\otimes\cE)$, and $\Psi_x\circ\theta=\Delta_F^{-1}\circ\uprod{\Psi^{(n)}_x}\circ\widetilde\theta=\Delta_F^{-1}\circ\uprod{\Phi^{(n)}_x}\circ\Delta_H=\Phi_x$.

\end{proof}

\section{Proofs for \texorpdfstring{\Cref{sec:criterion}}{the criterion section}}\label{app:criterion}

\subsection{Dual Pairings}\label{app:pairings}

The proofs in this appendix use dual pairings that we collect here.

\paragraph{Trace pairing: $\cB(H)=S_1(H)^*$ and $L_\infty(\Omega;\bM_d)=L_1(\Omega;S_1^d)^*$ via Banach-space duality.}

The trace pairing
\begin{equation}\label{eq:trace-pairing}
  \ip{a}{\rho}=\Tr(a\rho),\qquad a\in\cB(H),\ \rho\in S_1(H),
\end{equation}
identifies $\cB(H)$ with the dual of $S_1(H)$.
Combined with integration, $\ip{g}{f}=\int_\Omega\Tr(g(\omega)f(\omega))\,d\mu(\omega)$ identifies $L_\infty(\Omega;\bM_d)$ with the dual of $L_1(\Omega;S_1^d)$ for $\sigma$-finite $\mu$.
The predual $T_*$ of a map $T$ between algebras of bounded operators is defined by $\Tr[T_*(y)\,e]=\Tr[y\,T(e)]$.
We use the trace pairing only as a Banach-space duality, for these norms and for predual maps, and not to define matrix norms.

\paragraph{Scalar pairing: $\bM_m=(S_1^m)^*$ via operator-space duality.}

The criterion of \Cref{thm:cb-criterion} is stated with the scalar pairing \cite[(1.1.26)]{Effros2000}
\begin{equation}\label{eq:scalar-pairing}
  \ip{f}{y}=\Tr[f^Ty]=\sum_{s,t}f_{st}y_{st},\qquad f\in S_1^m,\ y\in\bM_m.
\end{equation}
It is the trace pairing composed with the transpose in the basis $(\ket s)_s$ that defines $\ket{\chi_m}=\sum_s\ket s\otimes\ket s$ and the Choi matrix $J(\varphi)=\sum_{s,t}E_{st}\otimes\varphi(E_{st})$ of a map $\varphi:\bM_m\to\bM_n$.
Under \eqref{eq:scalar-pairing} the map $\varphi$ corresponds to $J(\varphi)\in S_1^m\otimes\bM_n$.
So $\id_{\bM_m}$ corresponds to $\chi_m$, and $\varphi$ is completely positive exactly when $J(\varphi)\ge0$.
In vector form,
\begin{equation}
\bra{\chi_m}(X\otimes Y)\ket{\chi_m}=\Tr[X^TY].
\end{equation}
A Banach dual pairing between $V^*$ and $V$ defines the operator-space structure of $V^*$ through $\bM_n(V^*)=\CB(V,\bM_n)$ \cite[(3.2.2)]{Effros2000}.
The trace and scalar pairings give $S_1^m$ the same norm but different matrix norms, because the transpose is an isometry of $\bM_m$ but not a complete isometry.
For example, $\chi_m\in\bM_m(S_1^m)$ corresponds to $\id_{\bM_m}$ under \eqref{eq:scalar-pairing}, of c.b.\ norm one, and to the transpose under \eqref{eq:trace-pairing}, of c.b.\ norm $m$.
Only under \eqref{eq:scalar-pairing} does the Choi matrix of a map represent the map, and $S_1^{nd}=S_1^d(S_1^n)$ holds for the structure it gives.
So in \Cref{thm:cb-criterion}, its proof, \Cref{lem:choi} and \Cref{lem:theta-norm}, the spaces $\ell_\infty(\mathsf X,S_1^m)$, $S_1^n$ and $L_1(\Omega;S_1^d)$ carry the structure given by \eqref{eq:scalar-pairing}, where sums over $x\in\mathsf X$ identify $\ell_\infty(\mathsf X,S_1^m)$ with the dual of $\ell_1(\mathsf X,\bM_m)$.

\paragraph{Pairing of maps: $\CB(\bM_m,\bM_n)=(\bM_m\otimes^\wedge S_1^n)^*$ via operator-space duality.}

Since $\bM_n$ is the dual of $S_1^n$ under \eqref{eq:trace-pairing}, \cite[Corollary~7.1.5]{Effros2000} identifies $\CB(\bM_m,\bM_n)$ with the operator-space dual of the operator-space projective tensor product $\bM_m\otimes^\wedge S_1^n$ through
\begin{equation}\label{eq:cb-dual-pairing}
  \ip{T}{e\otimes f}=\Tr[f\,T(e)],\qquad T\in\CB(\bM_m,\bM_n),\ e\in\bM_m,\ f\in S_1^n.
\end{equation}
As these spaces are finite-dimensional, $\bM_m\otimes^\wedge S_1^n$ is also the dual of $\CB(\bM_m,\bM_n)$.
We write its elements as linear maps $R:\bM_n\to\bM_m$, where $\sum_kz_k\otimes\rho_k$ corresponds to $R(y)=\sum_k\Tr[\rho_ky]\,z_k$.
Then \eqref{eq:cb-dual-pairing} becomes
\begin{equation}
  \sum_k\Tr[\rho_kT(z_k)]=\Tr[R\circ T],
\end{equation}
the trace of the linear map $R\circ T$ on $\bM_m$.
Here the trace of a linear map $N$ on $\bM_k$ is
\begin{equation}
  \Tr[N]=\sum_{s,t}[N(\ketbra st)]_{st}.
\end{equation}
This is the pairing for a single pair of matrix algebras.

\paragraph{Pairing of families: $\CB(\ell_1(\mathsf X,\bM_m),\bM_n)$ via operator-space duality.}
The space $\CB(\ell_1(\mathsf X,\bM_m),\bM_n)$ is a different space from $\CB(\bM_m,\bM_n)$: its elements are families $\Phi=(\Phi_x)_{x\in\mathsf X}$ with $\Phi_x=\Phi(\,\cdot\,\delta_x):\bM_m\to\bM_n$, bounded jointly in the c.b.\ norm of $\ell_1(\mathsf X,\bM_m)$ and not one at a time.
Its elements are paired with families $R=(R_x)_x$ with $R_x:\bM_n\to\bM_m$, which are the maps $R:\bM_n\to\ell_1(\mathsf X,\bM_m)$ of the polar set in the proof of \Cref{thm:cb-criterion}.
Summing the pairing \eqref{eq:cb-dual-pairing} over $x$ gives the pairing of $\CB(\ell_1(\mathsf X,\bM_m),\bM_n)$ with these families,
\begin{equation}\label{eq:cp-pairing}
  \sum_x\Tr[R_x\circ\Phi_x]=\sum_x\sum_{p,q}\Tr[E_{qp}R_x(\Phi_x(E_{pq}))].
\end{equation}
The vector states in the proofs produce instead its transposed version
\begin{equation}\label{eq:dual-pairing}
  \sum_x\Tr[\vartheta\circ R_x\circ\vartheta\circ\Phi_x],
\end{equation}
where $\vartheta$ is the transpose.
Indeed, let $T,\Phi:\bM_m\to\bM_n$ be linear and let $a,b\in\bM_n$.
Then
\begin{equation}
  (b^*\otimes\1)(T\otimes\Phi)(\chi_m)(a\otimes\1)=\sum_{s,t}b^*T(E_{st})a\otimes\Phi(E_{st}),
\end{equation}
and the vector form of the scalar pairing gives
\begin{equation}\label{eq:vector-pairing}
\begin{aligned}
  \bra{\chi_n}(b^*\otimes\1)(T\otimes\Phi)(\chi_m)(a\otimes\1)\ket{\chi_n}
  &=\sum_{s,t}\Tr[(b^*T(E_{st})a)^T\,\Phi(E_{st})]\\
  &=\sum_{s,t}\Tr[\Phi(E_{st})^T\,b^*T(E_{st})a]\\
  &=\sum_{s,t}\Tr[E_{st}\,T_*(a\,\Phi(E_{st})^T\,b^*)]\\
  &=\Tr[\vartheta\circ R^{a,b}\circ\vartheta\circ\Phi],
\end{aligned}
\end{equation}
where $R^{a,b}(y):=T_*(a\,y\,b^*)$.
The second equality is $\Tr[X^TY]=\Tr[Y^TX]$.
The third is cyclicity of the trace and the definition of the predual $T_*$.
The last is the definition of $\Tr[N]$, since
\begin{equation}
[(R^{a,b}(\Phi(E_{st})^T))^T]_{st}=\Tr[E_{st}\,R^{a,b}(\Phi(E_{st})^T)].
\end{equation}
With $\ket\xi=(a\otimes\1)\ket{\chi_n}$ and $\ket\eta=(b\otimes\1)\ket{\chi_n}$ the left-hand side is $\bra\eta(T\otimes\Phi)(\chi_m)\ket\xi$, and every vector in $\bC^n\otimes\bC^n$ has this form with $\norm{\xi}=\norm{a}_2$.
Each proof says which of the two pairings it writes as $\ip{R}{\Phi}$.
By cyclicity of the trace, \eqref{eq:dual-pairing} at $\Phi$ is \eqref{eq:cp-pairing} at $(\vartheta\circ\Phi_x\circ\vartheta)_x$.
In the c.b.\ case the two pairings give the same polar of the set of factorizable families.
Indeed, conjugating $\theta$ and $\Psi_x$ by the pointwise transpose preserves c.b.\ norms, so this set is invariant under $\Phi_x\mapsto\vartheta\circ\Phi_x\circ\vartheta$.
On ${*}$-preserving families \eqref{eq:dual-pairing} is real.

\subsection{Proof of \Cref{thm:cb-criterion}}
\label{app:cb-criterion}

If $\Phi = 0$, the statement is trivial, so assume $\Phi \neq 0$.
Since both sides scale linearly in $\Phi$, we can assume $\norm{\Phi}_{\cb} = 1$ without loss.

Let $\cF \subset \CB(\ell_1(\mathsf{X}, \bM_m), \bM_n)$ be the set
\begin{equation}
  \cF := \{\Phi = \theta \circ \Psi \mid \norm{\theta}_{\cb}, \norm{\Psi}_{\cb} \leq 1 \},
\end{equation}
where $\Psi:\ell_1(\mathsf{X}, \bM_m) \to L_\infty(\Omega;\bM_d)$ and a normal map $\theta: L_\infty(\Omega;\bM_d) \to \bM_n$ for some $\sigma$-finite $(\Omega,\Sigma,\mu)$.

We first show that
\begin{equation}\label{eq:polar}
  \gamma := \gamma_{d, \cb}(\Phi) = \sup_{R \in \cF^\circ} \Re \ip{R}{\Phi},
\end{equation}
where $\cF^\circ$ is the polar of $\cF$
\begin{equation}
  \cF^\circ = \{ R: \bM_n \to \ell_1(\mathsf{X}, \bM_m) \mid \Re \ip{R}{X} \leq 1 \; \forall X \in \cF \}.
\end{equation}
and $\ip{\cdot}{\cdot}$ is the dual pairing \eqref{eq:cp-pairing} of the operator space $\CB(\ell_1(\mathsf{X}, \bM_m), \bM_n)$.

We then show that
\begin{equation}\label{eq:R-T-equiv}
  \sup_{R\in \cF^\circ} \Re \ip{R}{\Phi} = \sup_{\norm{T}_{(d)} \leq 1} \norm{T \otimes \Phi},
\end{equation}
which together with \eqref{eq:polar} proves the theorem.

The inequality $\gamma \geq \sup_R$ is trivial.
Let $t > \gamma$.
Then $\Phi/t \in \cF$, and by the definition of $\cF^\circ$, for every $R \in \cF^\circ$, $\Re \ip{R}{\Phi/t} \leq 1$, i.e.
\begin{equation}
  \sup_{R \in \cF^\circ} \Re \ip{R}{\Phi} \leq t.
\end{equation}
It holds for arbitrary $t$, so taking $t \to \gamma$ we have $\sup_{R \in \cF^\circ} \Re\ip{R}{\Phi} \leq \gamma$.

For the reverse inequality $\gamma \leq \sup_R$ we invoke the Hahn--Banach theorem, which completes \eqref{eq:polar}.
Since $\cF$ is balanced ($\lambda \Phi = (\lambda \theta) \circ \Psi$ and $\norm{\lambda \theta}_{\cb} \leq 1$ for any $\abs{\lambda} \leq 1$) and convex (\Cref{lem:F-convex}), $\Phi$ lies in the closure of $\cF$ if and only if $\Re \ip{R}{\Phi} \leq 1$ for all $R \in \cF^\circ$.

Equivalently, $R \in \cF^\circ$ means that
\begin{equation} \label{eq:HB-witness}
  \Re\ip{R}{\theta \Psi} \leq \norm{\theta}_{\cb} \, \norm{\Psi}_{\cb}
\end{equation}
holds for \emph{all} $\theta,\Psi$; the normalized form (with $\norm{\theta}_\cb,\norm{\Psi}_\cb\le1$ and bound $1$) and this homogeneous form are equivalent by rescaling.
Specifically,
\begin{equation}
  \sup_{R \in \cF^\circ} \Re\ip{R}{\Phi} \leq \inf_{\Phi = \theta\Psi} \norm{\theta}_{\cb}\norm{\Psi}_{\cb} = \gamma,
\end{equation}
where the infimum is over all completely bounded factorizations $\Phi = \theta \Psi$ via $L_\infty(\Omega;\bM_d)$ for some $\Omega$.
For the reverse inequality let 
\begin{equation}
  s:=\sup_{R\in\cF^\circ}\Re\ip{R}{\Phi}
\end{equation}
and $t>s$.
Then $\Re\ip{R}{\Phi/t}\le1$ for all $R\in\cF^\circ$, so $\Phi/t$ lies in $\cF^{\circ\circ}$, which is the closure of $\cF$ by the bipolar theorem in the finite-dimensional space $\CB(\ell_1(\mathsf X,\bM_m),\bM_n)$.
This closure lies in $\{\gamma_{d,\cb}\le1\}$, because $\gamma_{d,\cb}$ is the Minkowski functional of $\cF$, a seminorm that is finite and hence continuous.
It is finite because expanding each $\Phi_x$ in a basis factors it through some $\ell_\infty^k(\bM_d)$.

Finally, we show \eqref{eq:R-T-equiv}.

We can express $\Phi = (\Phi_x)_{x\in \mathsf{X}}$ and $R = (R_x)_{x \in \mathsf{X}}$, where $\Phi_x \in \CB(\bM_m, \bM_n)$ and $R_x \in \CB(\bM_m, \bM_n)^*$.
Here $R_x:\bM_n\to\bM_m$ is identified with an element of $\CB(\bM_m,\bM_n)^*$ as in \Cref{app:pairings}.
In the rest of the proof $\ip{R}{\Phi}$ is the transposed pairing \eqref{eq:dual-pairing}, which gives the same polar $\cF^\circ$ and the same \eqref{eq:polar} by \Cref{app:pairings}.

We work with the Choi quantity, because \Cref{lem:choi} gives only an inequality for an individual $T$.
\begin{align}
  \
  \norm{\sum_x(T_x\otimes \Phi_x)(\chi_m)}
  = \sup_{\norm{\xi}\leq 1} \norm{\sum_x(T_x\otimes \Phi_x)(\chi_m)\xi}
  = \sup_{\norm{\xi},\norm{\eta} \leq 1} \abs{\sum_x \bra{\eta}(T_x\otimes \Phi_x)(\chi_m)\ket{\xi}}
\end{align}
where the last equation is due to Cauchy-Schwarz.

For $a,b\in\bM_n$ with $\ket{\xi} = (a\otimes \1)\ket{\chi_n}$ and $\ket{\eta} = (b \otimes \1)\ket{\chi_n}$, \eqref{eq:vector-pairing} applied to $T_x$ and $\Phi_x$ gives
\begin{equation}
  \bra{\eta}(T_x\otimes \Phi_x)(\chi_m)\ket{\xi}=\Tr[\vartheta\circ R_x^{a,b}\circ\vartheta\circ\Phi_x]=\ip{R_x^{a,b}}{\Phi_x},
\end{equation}
where $R_x^{a,b}(y)=S_x(a\,y\,b^*)$ with $S_x:=(T_x)_*$, so that $(S_x)^*=T_x$.
This concludes
\begin{equation}\label{eq:norm-as-pairing}
  \ \norm{\sum_x(T_x\otimes \Phi_x)(\chi_m)} = \sup_{\norm{a}_2, \norm{b}_2\leq 1} \sum_x \Re \ip{R_x^{a,b}}{\Phi_x}.
\end{equation}

By \eqref{eq:HB-witness}, membership $R \in \cF^\circ$ is equivalent to the condition that $\Re\sum_x \Tr[R_x \theta \Psi_x] \leq \norm{\theta}_{\cb} \sup_x \norm{\Psi_x}_{\cb}$ for all $\theta$ and $(\Psi_x)$.
We transform this polar condition, through a reversible chain, into the boundedness $\norm{T}_{(d)} \leq 1$ of a witness $T$; combined with \eqref{eq:norm-as-pairing} this identifies $\cF^\circ$ with $\{\norm{T}_{(d)} \leq 1\}$ and yields \eqref{eq:R-T-equiv}.

\textbf{Step 1 -- Rewrite $\norm{\theta}_{\cb}$.}

Fix a normal contraction $\theta:L_\infty(\Omega;\bM_d)\to\bM_n$, given by a measurable family $(\theta_\omega)_{\omega\in\Omega}$ of maps $\theta_\omega:\bM_d\to\bM_n$ through $\theta(f)=\int_\Omega\theta_\omega(f(\omega))\,d\mu(\omega)$, and write $z_\omega\in\bM_n\otimes\bM_d$ for the Choi matrix of $\theta_\omega$, so that $z=(z_\omega)_\omega$ is an element of $\bM_n(L_1(\Omega;S_1^d))$.

Operator-space duality $\CB^\sigma(L_\infty(\Omega;\bM_d), \bM_n) = \bM_n(L_1(\Omega;S_1^d))$ for normal maps (\Cref{lem:cb-duality}) together with Smith's lemma expresses $\norm{\theta}_{\cb}$ as (\Cref{lem:theta-norm}):
\begin{equation}\label{eq:theta-sandwich}
  \norm{\theta}_{\cb}
  = \sup_{\norm{a}_2, \norm{b}_2 \leq 1} \int_\Omega \norm{(a \otimes \1) z_\omega (b \otimes \1)}_{S_1^{nd}} \, d\mu(\omega),
\end{equation}
where $(z_\omega)_\omega$ are the Choi matrices.

\textbf{Step 2 -- Reducing the polar condition.}

Fix $\theta$. By \eqref{eq:HB-witness} and homogeneity in $\Psi$, the supremum of $\Re\sum_x \Tr[R_x\theta\Psi_x]$ over $\Psi$ in the unit ball is at most $\norm{\theta}_{\cb}$.
Since $\sup_x\norm{\Psi_x}_\cb = \sup_x\esssup_\omega\norm{\Psi_{x,\omega}}_\cb$ and the maps $\Psi_{x,\omega}$ vary independently, this supremum splits:
\begin{align}
  &\sup_{\sup_x \norm{\Psi_x}_{\cb} \leq 1} \Re\sum_x \Tr [R_x \theta \Psi_x] \nonumber\\
  &\q= \sup_{\esssup_\omega\norm{\Psi_{x,\omega}}_{\cb} \leq 1} \Re\sum_x \int_\Omega \Tr [R_x \theta_\omega \Psi_{x,\omega}]\, d\mu(\omega) \nonumber\\
  &\q= \sum_x \int_\Omega \norm{(R_x\otimes\id_{\bM_d})(z_\omega)}_{S_1^d(\bM_m)}\, d\mu(\omega).
\end{align}

The last equality is the trace duality between the completely bounded norm on $\CB(\bM_m,\bM_d)$ and the  norm of Pisier's vector-valued trace class $S_1^d(\bM_m)$ of the Choi matrix \cite[(7.1.16) and (7.1.20)]{Effros2000} and \cite[Thm.~4.1 and p.~142]{Pisier_2003}.
It differs from the trace norm on $\bM_{md}$, since $\1_d\otimes\1_m$ has norm $d$ in $S_1^d(\bM_m)$ and trace norm $md$.
The suprema of the real parts equal those of the absolute values, because the unit balls are invariant under multiplication by phases.
For each $\omega$ the supremum of $\Re\Tr[R_x\theta_\omega\Psi_{x,\omega}]$ over $\norm{\Psi_{x,\omega}}_{\cb}\le1$ equals $\norm{J(R_x\theta_\omega)}_{S_1^d(\bM_m)}$, where $J(R_x\theta_\omega)=(R_x\otimes\id_{\bM_d})(z_\omega)$, integrated over $\omega$ by the $L_\infty(\Omega;\bM_d)$--$L_1(\Omega;S_1^d)$ duality.
The same equality follows without splitting the supremum pointwise in $\omega$ from $(L_1(\Omega;S_1^d)\otimes^\wedge\bM_m)^*=\CB(\bM_m,L_\infty(\Omega;\bM_d))$ \cite[(7.1.15) and Thm.~7.2.4]{Effros2000}.

Substituting this into \eqref{eq:HB-witness} and using \Cref{lem:theta-norm} for $\norm{\theta}_{\cb}$, the polar bound becomes 
\begin{equation}\label{ineq:psi-cb}
  \sum_x \int_\Omega \norm{(R_x \otimes \id_{\bM_d})(z_\omega)}_{\,S_1^d(\bM_m)}\, d\mu(\omega)
  \leq \sup_{\norm{a}_2, \norm{b}_2 \leq 1} \int_\Omega \norm{(a\otimes \1)z_\omega(b\otimes \1)}_{S_1^{nd}}\, d\mu(\omega)
\end{equation}
for the Choi matrices $z=(z_\omega)_\omega$ of $\theta$.
Every step is reversible, so \eqref{eq:HB-witness} $\Leftrightarrow$ \eqref{ineq:psi-cb}.

\textbf{Step 3 -- Pisier's non-commutative Pietsch factorization.}

Inequality \eqref{ineq:psi-cb} is a $1$-summing bound at the fixed level $d$. We recall Pisier's completely $p$-summing maps and adapt his factorization theorem to this level.
The factorization we use is Pisier's operator-space version of Pietsch's factorization theorem \cite[Theorem~5.1]{Pisier1998}, recalled in \Cref{thm:pisier-5.1}.

\begin{definition}
Let $E, F$ be operator spaces, in particular $E \subseteq \cB(H)$.
Then $u:E \to F$ is \emph{completely $p$-summing ($1\leq p < \infty$)} if $\id_{S_p} \otimes u: S_p \minten E \to S_p(F)$ is bounded.
Its $p$-summing norm is defined as
\begin{equation}
  \pi_p^\circ(u) := \norm{\id_{S_p} \otimes u}_{S_p\minten E \to S_p(F)}.
\end{equation}
\end{definition}

Define $u_0:\bM_n \to \ell_1(\mathsf{X}; \bM_m)$, $u:\bM_n \otimes \bM_d \to \ell_1(\mathsf{X}; \bM_m)\otimes\bM_d$ as
\begin{equation}
  u_0(y) = (R_x(y))_x, \q\q u = u_0 \otimes \id_{\bM_d}.
\end{equation}
and write $\norm{u(z)}:=\sum_x\norm{(R_x\otimes\id_{\bM_d})(z)}_{S_1^d(\bM_m)}$, the left-hand side of \eqref{ineq:psi-cb} for a single $z$.

For $p=1$ \cite[Theorem 5.3]{Pisier1998} expresses the norm $\norm{\cdot}_{S_1\minten E}$ as the following:
for $E\subseteq\cB(H)$ and $z\in S_1\minten E$,
\begin{equation}\label{eq:sandwich-p1}
  \norm{z}_{S_1\minten E} = \sup\{\, \norm{a z b}_{S_1(\ell_2\otimes H)} \mid a,b\ge0,\ \norm{a}_{S_2},\norm{b}_{S_2}\le1 \,\}.
\end{equation}
Restricted to $z\in\bM_n\otimes\bM_d=S_1^d\otimes_{\min}\bM_n$, this is the right-hand side of \eqref{ineq:psi-cb} for a single $z$.
Since \eqref{ineq:psi-cb} concerns only the level $d$, it does not make $u_0$ completely $1$-summing, but it gives the hypothesis of the proof of \cite[Theorem~5.1]{Pisier1998} at this level.
Indeed, taking $\Omega$ finite with counting measure, \eqref{ineq:psi-cb} states that
\begin{equation}
  \sum_k\norm{u(z_k)}\le\sup_{a,b}\sum_k\norm{(a\otimes\1)z_k(b\otimes\1)}_{S_1^{nd}}
\end{equation}
for every finite family $(z_k)_k$ in $\bM_n\otimes\bM_d$, where the supremum runs over $a,b\ge0$ in the unit ball of $S_2^n$, since $\norm{(a\otimes\1)z(b\otimes\1)}_{S_1^{nd}}=\norm{(\abs{a}\otimes\1)z(\abs{b^*}\otimes\1)}_{S_1^{nd}}$.
The proof of \cite[Theorem~5.1]{Pisier1998} starts from this inequality for finite families, applies \cite[Lemma~5.2]{Pisier1998} to the convex cone of the functions $(a,b)\mapsto\sum_k\norm{(a\otimes\1)z_k(b\otimes\1)}_{S_1^{nd}}-\sum_k\norm{u(z_k)}$, and replaces each finitely supported probability measure on the pairs $(a,b)$ by a single pair using \cite[Lemma~1.14]{Pisier1998}.
It therefore applies verbatim at the level $d$.

Hence there exists an ultrafilter $\cU$ over an index set $I$ and families $(a_\alpha)_{\alpha}, (b_\alpha)_{\alpha}$ in the unit ball of $S_2^n$ so that for all $z \in \bM_n \otimes \bM_d$
\begin{equation}\label{ineq:p-summing}
  \norm{u(z)} \leq \lim_{\cU} \norm{(a_\alpha \otimes \1_{\bM_d})z(b_\alpha \otimes \1_{\bM_d})}_{S_1^{nd}}.
\end{equation}

Since the unit ball in $S_2^n$ is compact, the ultralimits $a := \lim_\cU a_\alpha, b := \lim_\cU b_\alpha$ exist inside the unit ball, and by continuity, we can move the ultralimit inside and obtain \eqref{ineq:p-summing} with the single pair $(a,b)$ in place of the ultralimit.

With a single pair $(a,b)$ in hand, it remains to absorb it into $R$ and recognize the adjoint of the resulting map as an admissible witness.
We may assume that $a$ and $b$ are invertible.
By polar decomposition we may take $a,b\ge0$, since replacing $(a,b)$ by $(|a|,|b^*|)$ changes neither the $S_2$-norms nor the trace norm in \eqref{ineq:p-summing}.
If $a\ge0$ is not invertible, we use the approximate identity: for $\varepsilon\in(0,1)$, $a_\varepsilon:=(1-\varepsilon)a+\varepsilon\1_n/\sqrt n$ is invertible with $\norm{a_\varepsilon}_2\le1$, and $a=a_\varepsilon c_\varepsilon=c_\varepsilon a_\varepsilon$ with $c_\varepsilon:=a\,a_\varepsilon^{-1}$ and $\norm{c_\varepsilon}\le(1-\varepsilon)^{-1}$.
Hence \eqref{ineq:p-summing} holds for $(a_\varepsilon,b_\varepsilon)$ up to the factor $(1-\varepsilon)^{-2}$, which carries through the remaining bounds and disappears in the suprema of \eqref{eq:R-T-equiv} as $\varepsilon\to0$.
For invertible $a,b$ we put $S(w):=u_0(a^{-1}wb^{-1})$, so that $R(y)=S(ayb)$, and \eqref{ineq:p-summing} gives $\norm{(S\otimes\id_{\bM_d})(w)}\le\norm{w}_{S_1^{nd}}$ for all $w\in\bM_n\otimes\bM_d$.

\textbf{Step 4 -- Duality flip.}
For any linear $S : V \to W$ one has $\norm{S}_{(d)} = \norm{\id \otimes S^{*} : S_1^d(W^*) \to S_1^d(V^*)}$.
Applied with $V = S_1^n$ and $W = \ell_1(\mathsf{X}; \bM_m)$, the adjoint $T := S^{*} : \ell_\infty(\mathsf{X}; S_1^m) \to \bM_n$ satisfies $\norm{T}_{(d)} = \norm{S}_{(d)} \leq 1$ and is the witness of \eqref{eq:T}, with $R = R^{a,b}$ for this $T$ and pair $a, b$.
Applied to $T$, whose adjoint is $S$, the identity above and $S_1^{nd}=S_1^d(S_1^n)$ turn the bound on $S\otimes\id_{\bM_d}$ from Step~3 into $\norm{T}_{(d)}\le1$.
The equality $\norm{T}_{(d)}=\norm{S}_{(d)}$ uses $\norm{u^*}_{(d)}=\norm{u}_{(d)}$ for linear maps $u:V\to W$ between operator spaces, which follows from Smith's lemma applied to $\bM_d(W^*)=\CB(W,\bM_d)$, see \cite[Prop.~3.2.2]{Effros2000}.

In this step all $S_1$-spaces carry the structure of \eqref{eq:scalar-pairing}, for which $S_1^{nd}=S_1^d(S_1^n)$, and the first identity holds for the adjoint $S^{\ddagger}$ with respect to \eqref{eq:scalar-pairing}.
The trace-pairing adjoint is $S^{*}=\vartheta\circ S^{\ddagger}\circ\vartheta$, and $\id_{\bM_d}\otimes(\vartheta\circ N\circ\vartheta)$ is $\id_{\bM_d}\otimes N$ conjugated by full transposes, which are isometric at every matrix level, so $\norm{S^{*}}_{(d)}=\norm{S^{\ddagger}}_{(d)}$.

Read backwards, Steps~2--4 show that $R^{a,b}\in\cF^\circ$ whenever $\norm{T}_{(d)}\le1$ and $\norm{a}_2,\norm{b}_2\le1$. Read forwards, they show that every $R\in\cF^\circ$ has this form up to the factor $(1-\varepsilon)^{-2}$ of the approximate identity, which disappears in the suprema.
Taking the supremum of $\Re \ip{R}{\Phi}$ over both descriptions and using \eqref{eq:norm-as-pairing} and the equality of suprema in \Cref{lem:choi},
\begin{equation}
  \sup_{R \in \cF^\circ} \Re \ip{R}{\Phi}
  = \sup_{\norm{T}_{(d)} \leq 1} \; \sup_{\norm{a}_2, \norm{b}_2 \leq 1} \Re \ip{R^{a,b}}{\Phi}
  = \sup_{\norm{T}_{(d)} \leq 1} \norm{T \otimes \Phi},
\end{equation}
which is \eqref{eq:R-T-equiv}.
Together with \eqref{eq:polar} this yields $\gamma_{d, \cb}(\Phi) = \sup_{\norm{T}_{(d)} \leq 1} \norm{T \otimes \Phi}$, completing the proof.


\subsection{Lemmas for the proof of \Cref{thm:cb-criterion}}

\begin{lemma}\label{lem:F-convex}
Let $\cF \subset \CB(\ell_1(\mathsf{X}, \bM_m), \bM_n)$ be the set
\begin{equation}
  \cF := \{\Phi = \theta \circ \Psi \mid \norm{\theta}_{\cb}, \norm{\Psi}_{\cb} \leq 1 \},
\end{equation}
where $\Psi:\ell_1(\mathsf{X}, \bM_m) \to L_\infty(\Omega;\bM_d)$ and $\theta: L_\infty(\Omega;\bM_d) \to \bM_n$ is normal, for some $\sigma$-finite $(\Omega,\Sigma,\mu)$.

$\cF$ is convex.
\end{lemma}
\begin{proof}
  Let $\Phi^{(1)}, \Phi^{(2)} \in \cF$ with $\Phi_x^{(1)} = \theta^{(1)} \circ \Psi_x^{(1)}$ and $\Phi_x^{(2)} = \theta^{(2)} \circ \Psi_x^{(2)}$, where the factorizations are through $\ell_\infty^{k_1}(\bM_d)$ and $\ell_\infty^{k_2}(\bM_d)$, respectively.
  Fix $\lambda \in [0,1]$.
  Define linear maps $\Psi_x:\bM_m \to \ell_\infty^{k_1+k_2}(\bM_d)$ and $\theta:\ell_\infty^{k_1+k_2}(\bM_d) \to \bM_n$ as
  \begin{equation}
    \Psi_x(z) = \Psi_x^{(1)}(z) \oplus \Psi_x^{(2)}(z), \q\q
    \theta(y_1 \oplus y_2) = \lambda \theta^{(1)}(y_1) + (1-\lambda) \theta^{(2)}(y_2).
  \end{equation}
  For general registers $L_\infty(\Omega_1;\bM_d)$ and $L_\infty(\Omega_2;\bM_d)$, read $\ell_\infty^{k_1+k_2}(\bM_d)$ below as $L_\infty(\Omega_1\sqcup\Omega_2;\bM_d)=L_\infty(\Omega_1;\bM_d)\oplus_\infty L_\infty(\Omega_2;\bM_d)$.
  The disjoint union is again $\sigma$-finite, $\theta$ below is normal when $\theta^{(1)}$ and $\theta^{(2)}$ are.
  Then,
  \begin{equation}
  \lambda\Phi^{(1)} + (1-\lambda)\Phi^{(2)} = \lambda \theta^{(1)} \circ \Psi_x^{(1)} + (1-\lambda) \theta^{(2)} \circ \Psi_x^{(2)} = \theta \circ \Psi_x.
  \end{equation}
  Since
  \begin{equation}
  \norm{\Psi_x}_{\cb} = \norm{\Psi_x^{(1)} \oplus \Psi_x^{(2)}}_{\cb} \leq \max \{\norm{\Psi_x^{(1)}}_{\cb}, \norm{\Psi_x^{(2)}}_{\cb} \} \leq 1,
  \end{equation}
  and
  \begin{align}
    \norm{\theta}_{\cb} &= \sup_j \norm{\id_j \otimes \theta}\\
    &= \sup_j \sup_{\norm{Y_1 \oplus Y_2}_{\bM_j(\ell_\infty^{k_1+k_2}(\bM_d))}\leq 1} \norm{\lambda(\id_j \otimes \theta^{(1)})(Y_1) + (1-\lambda)(\id_j \otimes \theta^{(2)})(Y_2)}_{\bM_j(\bM_n)}\\
    &\leq \sup_j \sup_{\norm{Y_1\oplus Y_2}\leq 1} \lambda \norm{(\id_j \otimes \theta^{(1)})(Y_1)}_{\bM_j(\bM_n)} + (1-\lambda) \norm{(\id_j \otimes \theta^{(2)})(Y_2)}_{\bM_j(\bM_n)}\\
    &\leq \lambda \norm{\theta^{(1)}}_{\cb} + (1-\lambda) \norm{\theta^{(2)}}_{\cb}\\
    &\leq 1,
  \end{align}
  we have shown that $\lambda \Phi^{(1)} + (1-\lambda) \Phi^{(2)} \in \cF$.

\end{proof}

\begin{lemma}\label{lem:choi}
  Let $T:\ell_\infty(\mathsf{X}, S_1^m) \to \bM_n$ and $\Phi:\ell_1(\mathsf{X}, \bM_m) \to \bM_n$ be linear maps, with the operator space structures above, where $S_1^m$ is paired with $\bM_m$ by $\ip{f}{y}=\Tr[f^Ty]$.
  Then
    \begin{equation}\label{eq:choi-lower}
    \norm{\sum_x (T_x \otimes \Phi_x)(\chi_m)}_{\bM_n \minten \bM_n}\le\norm{T \otimes \Phi}
  \end{equation}
  and
  \begin{equation}\label{eq:choi-sup}
    \sup_{\norm{T}_{(d)}\le1}\norm{T \otimes \Phi}=\sup_{\norm{T}_{(d)}\le1}\norm{\sum_x (T_x \otimes \Phi_x)(\chi_m)}_{\bM_n \minten \bM_n},
  \end{equation}
  where $\chi_m = \sum_{s,t=1}^m \ketbra{s}{t}\otimes\ketbra{s}{t} \in \bM_m \otimes \bM_m$ is the Choi matrix of $\id_{\bM_m}$.
\end{lemma}
\begin{proof}
  Set $Y=\ell_1(\mathsf{X},\bM_m)$, so that $Y^{*}=\ell_\infty(\mathsf{X},S_1^m)$ and the two tensor factors are in  the pairing $\ip{f}{y}=\Tr[f^Ty]$ of the statement.
  The algebraic tensor $Y^{*}\otimes Y$ contains the canonical maximally entangled element
  \begin{equation}
    \omega=\sum_x e_x\otimes e_x\otimes\chi_m = J(\id_Y),
  \end{equation}
  the Choi matrix of $\id_Y$.
  In the minimal tensor norm $\norm{\omega}=\norm{\id_Y}_{\cb}=1$, and $\omega$ is a norming element of the unit ball: by Choi--Jamio\l kowski \cite[Prop.~8.1.2]{Effros2000} every $z$ with $\norm{z}_{\min}\le1$ is $z=(\id\otimes\Lambda)(\omega)$ for a complete contraction $\Lambda:Y\to Y$.
  Because $T\otimes\Phi$ is a product map, evaluation at $\omega$ contracts the $x$-diagonal against $\chi_m$,
  \begin{equation}\label{eq:choi-omega}
    (T\otimes\Phi)(\omega)=\sum_x (T_x\otimes\Phi_x)(\chi_m)\in\bM_n\minten\bM_n,
  \end{equation}
  so the lower bound $\norm{T\otimes\Phi}\ge\norm{(T\otimes\Phi)(\omega)}$ is immediate from $\norm{\omega}=1$.
  For the  equality of the suprema, take $z$ with $\norm{z}_{\min}\le1$ and write $z=(\id\otimes\Lambda)(\omega)$ as above; then $(T\otimes\Phi)(z)$ is $(T\otimes\Phi)(\omega)$ with $\Phi$ post-composed by the complete contraction $\Lambda$.
  
  Since $(\id\otimes\Lambda)(\omega)=(\Lambda^*\otimes\id)(\omega)$ for the adjoint $\Lambda^*:Y^*\to Y^*$, also $(T\otimes\Phi)(z)=((T\circ\Lambda^*)\otimes\Phi)(\omega)$.
  Here
  \begin{equation}
  \norm{T\circ\Lambda^*}_{(d)}\le\norm{T}_{(d)}\norm{\Lambda^*}_{\cb}\le\norm{T}_{(d)},
  \end{equation}
  because $\norm{\Lambda^*}_{\cb}=\norm{\Lambda}_{\cb}\le1$.
  Now let $\norm{T}_{(d)}\le1$ and $\norm{z}_{\min}\le1$.
  By \eqref{eq:choi-omega} applied to $T':=T\circ\Lambda^*$, which satisfies $\norm{T'}_{(d)}\le1$,
  \begin{equation}
    \norm{(T\otimes\Phi)(z)}=\norm{\sum_x(T'_x\otimes\Phi_x)(\chi_m)}\le\sup_{\norm{T'}_{(d)}\le1}\norm{\sum_x(T'_x\otimes\Phi_x)(\chi_m)}.
  \end{equation}
  Taking the supremum over $z$ gives $\norm{T\otimes\Phi}\le$ the right-hand side of \eqref{eq:choi-sup}, and then the supremum over $T$ gives ``$\le$'' in \eqref{eq:choi-sup}.
  The reverse inequality ``$\ge$'' follows by taking the supremum over $T$ in \eqref{eq:choi-lower}.

\end{proof}

\begin{lemma}\label{lem:theta-norm}
  For a normal contraction $\theta:L_\infty(\Omega;\bM_d)\to\bM_n$ with Choi matrices $z=(z_\omega)_\omega\in\bM_n(L_1(\Omega;S_1^d))$,
  \begin{equation}
    \norm{\theta}_{\cb} = \norm{z}_{\bM_n(L_1(\Omega;S_1^d))} = \sup_{\norm{a}_2,\norm{b}_2 \leq 1, a,b \in \bM_n} \int_\Omega \norm{(a\otimes \1) z_\omega (b\otimes 1)}_{S_1^{nd}}\,d\mu(\omega)
  \end{equation}
\end{lemma}
\begin{proof}
  For any operator space $E$, $\CB(E, \bM_n) = \bM_n(E^*)$.
  By the operator space duality \Cref{lem:cb-duality}, the normal map $\theta$ is expressed as $\theta = [\theta_{pq}]_{p,q = 1}^n$ with predual elements $\theta_{pq} = (\theta_{pq,\omega})_{\omega\in\Omega} \in L_1(\Omega;S_1^d)$, defined by
  \begin{equation}
    \Tr [\theta_{pq,\omega}^{T}\,y] = \bra{p}\theta_\omega(y)\ket{q}, \q y \in \bM_d\ (\mu\text{-a.e. }\omega).
  \end{equation}
  Here $L_1(\Omega;S_1^d)$ carries the structure given by the pairing~\eqref{eq:scalar-pairing}, which is the pairing in this definition.
  Then,
  \begin{equation}
    \norm{z}_{\bM_n(L_1(\Omega;S_1^d))}
    = \norm{[\theta_{pq}]_{p,q=1}^n}_{\bM_n(L_1(\Omega;S_1^d))}
    = \norm{\theta}_{\cb}.
  \end{equation}
  By the Smith's lemma, $\norm{\theta}_{\cb} = \norm{\theta}_{(n)}$.
  \begin{equation}
    \norm{\theta}_{(n)}
    = \sup_{\esssup_\omega \norm{Y_\omega} \leq 1} \norm{\int_\Omega (\id_n \otimes \theta_\omega)(Y_\omega)\, d\mu(\omega)}_{\bM_{n^2}}, \q Y=(Y_\omega)_\omega \in L_\infty(\Omega;\bM_{nd}).
  \end{equation}
  For any unit vectors $\ket{\xi}, \ket{\eta} \in \bC^n \otimes \bC^n$, there exist $a,b\in\bM_n$ with $\norm{a}_2 = \norm{\ket{\xi}}$, $\norm{b}_2 = \norm{\ket{\eta}}$ satisfying
  \begin{equation}
    \ket{\xi} = (a \otimes \1_n)\ket{\phi^+}, \q \ket{\eta} = (b \otimes \1_n)\ket{\phi^+},
  \end{equation}
  where $\ket{\phi^+}$ is the unnormalized maximally entangled state.
  Unwinding the Choi--Jamio\l kowski identification of $\theta_\omega$ with its Choi matrix $z_\omega\in\bM_n\otimes\bM_d$ and using the ricochet property gives
  \begin{equation}
    \bra{\xi} \int_\Omega (\id_n \otimes \theta_\omega)(Y_\omega)\, d\mu \ket{\eta}
    = \int_\Omega \Tr[Y_\omega^T(\bar a \otimes \1_d)z_\omega(b^T \otimes \1_d)]\, d\mu(\omega).
  \end{equation}
  The maps $Y\mapsto Y^T$, $a\mapsto\bar a$ and $b\mapsto b^T$ preserve the respective unit balls.
  Since the family $(Y_\omega)_\omega$ ranges over the unit ball of $L_\infty(\Omega;\bM_{nd})$, by the duality between $L_\infty(\Omega;\bM_{nd})$ and $L_1(\Omega;S_1^{nd})$,
  \begin{equation}
    \sup_{\esssup_\omega \norm{Y_\omega}\leq 1} \norm{\int_\Omega (\id_n \otimes \theta_\omega)(Y_\omega)\, d\mu}_{\bM_{n^2}}
    = \sup_{\norm{a}_2, \norm{b}_2\leq 1}
    \int_\Omega \; \norm{(a\otimes \1_d)z_\omega(b \otimes \1_d)}_{S_1^{nd}}\,d\mu(\omega).
  \end{equation}
\end{proof}

\subsection{Proof of \Cref{prop:cp-gamma}}
\label{app:cp-gamma-proof}
\begin{proposition*}[restated]
  \CPGamma
\end{proposition*}
\begin{proof}
We only show the implication $(\Leftarrow)$ since the converse is trivial.
If $\Phi=0$, take $\theta=0$ and $\Psi_x=0$, so assume $\Phi\neq0$, which forces $\theta\neq0$ in every factorization.

Suppose $\gamma_{d,\cp}(\Phi)\le1$.
Since the infimum is attained, there is a completely positive factorization $\Phi_x=\theta\circ\Psi_x$ through some $L_\infty(\Omega;\bM_d)$ with $\norm{\theta}_{\cb}\cdot\sup_x\norm{\Psi_x}_{\cb}=\gamma_{d,\cp}(\Phi)\le1$.
Rescaling $\theta\mapsto\norm{\theta}_{\cb}^{-1}\theta$ and $\Psi_x\mapsto\norm{\theta}_{\cb}\Psi_x$ multiplies by a positive scalar, which preserves complete positivity and leaves each $\Phi_x=\theta\circ\Psi_x$ intact, so we may assume $\norm{\theta}_{\cb}=1$ and $\sup_x\norm{\Psi_x}_{\cb}\le1$.
Using $\norm{u}_{\cb} = \norm{u(\1)}$ for c.p. maps (\cite[Prop.~3.6]{Paulsen_2003}), $\theta(\1)\le\1$ and $\Psi_x(\1)\le\1$ (a positive element of norm $\leq 1$ in a $C^*$-algebra is $\leq \1$, and conversely).

When each $\Phi_x$ is unital,
\begin{equation}
\1=\Phi_x(\1)=\theta(\Psi_x(\1))\le\theta(\1)\le\1,
\end{equation}
which forces $\theta$ to be u.c.p.

Define $\Delta_x:=\1-\Psi_x(\1)\ge0$.
Then $\theta(\Delta_x)=0$.
Fix any state (unital positive linear functional) $\rho$ on $\bM_m$ and set $\Psi_x'(z):=\Psi_x(z)+\rho(z)\cdot\Delta_x$.
Then $\Psi_x'$ is completely positive as a sum of the completely positive maps $\Psi_x$ and $z\mapsto \rho(z)\Delta_x$, it is unital since $\Psi_x'(\1)=\Psi_x(\1)+\Delta_x=\1$.
Then,
\begin{equation}
\theta\circ\Psi_x' = \theta\circ\Psi_x + \rho(\cdot)\,\theta(\Delta_x) = \theta\circ\Psi_x = \Phi_x.
\end{equation}
and $\Phi_x=\theta\circ\Psi_x'$ is a u.c.p.\ factorization through $L_\infty(\Omega;\bM_d)$, meaning $\Phi$ is $d$-compatible.
\end{proof}

\subsection{Proof of \Cref{thm:cp-criterion}}
\label{app:cp-criterion}

\begin{theorem*}[restated]
  \CPCriterion{}
\end{theorem*}

The proof follows the proof of \Cref{thm:cb-criterion} in \Cref{app:cb-criterion}, with the modifications below.
We work in the real vector space of $*$-preserving maps (Step~1).
The polar condition becomes an inequality for positive Choi matrices with the quantity \eqref{eq:S1plus} (Step~2).
The separation argument gives a single positive $\rho$ (Step~3).
Rescaling by $a=\rho^{1/2}$ gives a witness $T$ with $\norm{T}_{(d)+}\le1$ (Step~4).
The vector states of the c.b.\ case then identify the polar with the criterion (Step~5).
Unlike in the c.b.\ case, operator norms are replaced by norms of positive parts, because $\hat\cF_d$ is not balanced.

Let $\hat\cF_d$ be the set of contractive c.p. maps that are factorizable via c.p. contractions $\Psi:\ell_1(\mathsf{X},\bM_m) \to L_\infty(\Omega;\bM_d)$ and $\theta:L_\infty(\Omega;\bM_d)\to \bM_n$ for some $\sigma$-finite $(\Omega,\Sigma,\mu)$, the set of \eqref{eq:C-set}.

\textbf{Step 1 -- $*$-preserving maps and the polar.}

Let us denote $\CB(\bM_m, \bM_n)_{\sa}$ the locally convex real vector space of $*$-preserving maps, which is spanned by completely positive maps.
We write $\ip{R}{\Phi}$ for the pairing \eqref{eq:dual-pairing} of \Cref{app:pairings}, which is real on ${*}$-preserving families, and all polars below are taken in this pairing.

The polar of $\hat\cF_d$ in this pairing is
\begin{equation}\label{eq:cp-polar}
  \hat\cF_d^\circ = \{ R = (R_x:\bM_n \to \bM_m)_x \text{ $*$-preserving } \mid \Re \sum_x \ip{R_x}{\theta \circ \Psi_x} \leq \norm{\theta}_{\cb}\sup_x \norm{\Psi_x}_{\cb} \},
\end{equation}
where the inner condition is quantified over completely positive $\theta$ and tuples $(\Psi_x)$ of completely positive maps.

\begin{lemma}[Polar form of $\gamma_{d,\cp}$]\label{lem:cp-polar-form}
The set $\hat\cF_d$ is convex and closed, and for every completely positive family $\Phi$
\begin{equation}\label{eq:cp-polar-form}
  \gamma_{d,\cp}(\Phi)=\sup_{R\in\hat\cF_d^\circ}\Re\ip{R}{\Phi}.
\end{equation}
\end{lemma}

\begin{proof}
$\hat\cF_d$ is a convex set contained in this ambient space.
Indeed, the block construction of \Cref{lem:F-convex}, $\theta(y_1 \oplus y_2) = t\theta^{(1)}(y_1) + (1-t)\theta^{(2)}(y_2)$ and $\Psi_x = \Psi_x^{(1)} \oplus \Psi_x^{(2)}$, preserves complete positivity, since it only adds and rescales positively.
It is closed by the ultraproduct argument following \eqref{eq:gamma-cp}, and it contains $0$.

However, $\hat\cF_d$ is not balanced: $-\Phi\not\in\hat\cF_d$ for $0\neq\Phi\in\hat\cF_d$, since $-\Phi$ is not completely positive.
This does not obstruct the bipolar theorem, which needs only that $\hat\cF_d$ be convex, weak-$*$ closed and contain $0$, and therefore still gives $\hat\cF_d^{\circ\circ}=\hat\cF_d$ exactly.
For a completely positive family $\Phi$ and $t>0$, $\Phi\in t\hat\cF_d$ if and only if $\gamma_{d,\cp}(\Phi)\le t$, by homogeneity and \Cref{prop:cp-gamma}, so $\gamma_{d,\cp}$ is the Minkowski functional of $\hat\cF_d$ on the completely positive cone.
As in \eqref{eq:polar}, $\hat\cF_d^{\circ\circ}=\hat\cF_d$ therefore gives \eqref{eq:cp-polar-form}.
\end{proof}

What the missing balance changes is the shape of the polar condition: it is one-sided, $\Re\ip{R}{\Phi}\leq 1$ rather than $\abs{\ip{R}{\Phi}}\leq 1$, and constrains $R$ only along the completely positive cone.
This difference is what calls for the definition of $\norm{\cdot}_{(d)+}$ and what distinguishes the proof from the c.b. case.
By \Cref{lem:cp-polar-form} it remains to show
\begin{equation}
\sup_{R \in \hat{\cF}^\circ_d} \Re \ip{R}{\Phi} = \sup\{\norm{(\sum_x(T_x\otimes\Phi_x)(\chi_m))_+}\mid T\ {*}\text{-preserving},\ \norm{T}_{(d)+}\le1\}.
\end{equation}

Now let $R\in\hat\cF_d^\circ$.
Steps~2--4 produce $a\ge0$ with $\norm{a}_2\le1$ and a ${*}$-preserving $T$ with $\norm{T}_{(d)+}\le(1-\varepsilon)^{-1}$, for any $\varepsilon\in(0,1)$, such that $R_x(y)=(T_x)_*(a\,y\,a)$ for all $x$, and Step~5 concludes.

\textbf{Step 2 -- The polar condition for positive Choi matrices.}

We use \eqref{eq:cp-polar} only for $\theta$ on $\ell_\infty^N(\bM_d)$ with $N$ finite, so normality is not needed.
Such a $\theta$ is $\theta((f_j)_j)=\sum_j\theta_j(f_j)$ with c.p.\ maps $\theta_j:\bM_d\to\bM_n$ and we set $z_j:=\sum_{k,l}\theta_j(E_{lk})^T\otimes E_{kl}\in\bM_n\otimes\bM_d$, the Choi matrix of $\vartheta\circ\theta_j\circ\vartheta$ with its tensor factors exchanged.
Then $z_j\ge0$, $\Tr_{\bM_d}z_j=\theta_j(\1)^T$ and $\norm{\theta}_{\cb}=\norm{\theta(\1)}=\norm{\sum_j\Tr_{\bM_d}z_j}$, and as $\theta$ ranges over the c.p.\ maps on $\ell_\infty^N(\bM_d)$, the family $(z_j)_j$ ranges over all families of $N$ positive matrices.
For $v\in(\bM_m\otimes\bM_d)_{\sa}$ we set
\begin{equation}\label{eq:S1plus}
  \norm{v}_{S_1^d(\bM_m)_+}:=\sup\{\Tr[v\,J(u)]\mid u:\bM_m\to\bM_d\ \text{completely positive},\ u(\1)\le\1\},
\end{equation}
and we write $R_x(z):=(R_x\otimes\id_{\bM_d})(z)$.
This quantity is positively homogeneous and subadditive, but it is not a norm.
Since $\Tr[\vartheta\circ R_x\circ\vartheta\circ\theta_j\circ\Psi]=\Tr[R_x(z_j)\,J(\Psi)]$ for every $\Psi:\bM_m\to\bM_d$, the supremum over the components $\Psi_{x,j}$ of $\Psi_x$ with $\Psi_{x,j}(\1)\le\1$ turns \eqref{eq:cp-polar} into the necessary condition
\begin{equation}\label{eq:cp-polar-explicit}
  \sum_{x,j} \norm{R_x(z_j)}_{S_1^d(\bM_m)_+} \leq \sup_{\norm{a}_2\leq 1} \sum_j \norm{(a^* \otimes \1)z_j(a \otimes \1)}_{S_1^{nd}}
  = \sup_{\rho\geq 0, \Tr\rho \leq 1} \sum_j \Tr[z_j (\rho \otimes \1_{\bM_d})]
\end{equation}
for all finite families $(z_j)$ of positive Choi matrices.

\textbf{Step 3 -- Separation.}

Using Hahn-Banach theorem and Riesz-Markov-Kakutani representation theorem, there exists $\rho$ such that
\begin{equation} \label{eq:cp-rho}
  \sum_x \norm{R_x(z)}_{S_1^d(\bM_m)_+} \leq \Tr[z(\rho \otimes \1_{\bM_d})] \qquad (z\geq0),
\end{equation}
which is proved in \Cref{lem:HB-RMK}.

\textbf{Step 4 -- Rescaling and the witness.}

Analogous to the c.b. case, there is no loss in assuming $\rho$ to be invertible.

Put $a := \rho^{1/2} \geq 0$, so that $\norm{a}_2^2 = \Tr \rho \leq 1$, which yields
\begin{equation} \label{eq:cp-sandwich}
\sum_x \norm{R_x(z)}_{S_1^d(\bM_m)_+} \leq \Tr[z(\rho\otimes\1)]=\norm{(a\otimes \1)z(a \otimes \1)}_{S_1^{nd}} \q (z\ge0).
\end{equation}
Define $S_x$ as
\begin{equation} \label{eq:cp-S}
  S_x(y) := R_x(a^{-1}ya^{-1}).
\end{equation}
Because $a = a^*$, the conjugation $y \mapsto a^{-1}ya^{-1}$ is both $*$-preserving and positivity preserving, each $S_x$ is again $*$-preserving.
Applying \eqref{eq:cp-rho} to the positive element $(a^{-1}\otimes\1)z(a^{-1}\otimes\1)$ turns it into
\begin{equation} \label{eq:cp-Sbound}
  \sum_x \norm{S_x(z)}_{S_1^d(\bM_m)_+} \leq \norm{z}_{S_1^{nd}} \qquad (z \geq 0).
\end{equation}

Let $T = S^*$, that is, $(T_x)_*=S_x$, and let $w = (w_x)_x$ be given by $w_x=J(u_x)$ with $u_x:\bM_m\to\bM_d$ completely positive and $u_x(\1)\le\1$, so that $(\id_{\bM_d}\otimes T)(w)=\sum_x(T_x\otimes u_x)(\chi_m)$.
Because $T$ is $*$-preserving and $w \geq 0$, $(\id_{\bM_d}\otimes T)(w)$ is self-adjoint and
\begin{equation} \label{eq:sa-norm-positive}
  \norm{((\id_{\bM_d} \otimes T)(w))_+} = \sup\{\, \Tr[(\id_{\bM_d}\otimes T)(w)z] \;\mid\; z \geq 0, \norm{z}_1 \leq 1 \,\}.
\end{equation}
Note that $\Tr[(\id_{\bM_d}\otimes T)(w)\,z] = \sum_x$ $\Tr[w_x\,S_x(z)]$, and taking the supremum first over admissible $w$, whose coordinates $w_x$ are independent, and then over positive $z$ in the unit ball, we obtain by \eqref{eq:S1plus} and \eqref{eq:cp-Sbound}
\begin{equation} \label{eq:cp-Tdplus}
  \norm{T}_{(d)+} = \sup \{ \sum_x \norm{S_x(z)}_{S_1^d(\bM_m)_+} \mid z \geq 0, \norm{z}_{S_1^{nd}} \leq 1 \} \leq 1.
\end{equation}

\textbf{Step 5 -- Conclusion.}

Since the theorem only concerns $T$ and $\Phi$ that are $*$-preserving, $\sum_x(T_x \otimes \Phi_x)(\chi_m)$ is self-adjoint.
The supremum in \eqref{eq:plus-part} is attained at an extreme point $\rho=\ketbra{\xi}{\xi}$ with $\norm{\xi}\le1$, so
\begin{equation}
  \norm{(\sum_x(T_x\otimes \Phi_x)(\chi_m))_+}
  = \sup_{\norm{\xi} \leq 1} \sum_x \bra{\xi}(T_x\otimes \Phi_x)(\chi_m)\ket{\xi}.
\end{equation}

Consequently, for $a\in\bM_n$ with $\norm{a}_2 \leq 1$, $\ket{\xi} = (a \otimes \1)\ket{\chi_n}$ and $R_x^a(y) = (T_x)_*(a\,y\,a^*)$, where $(T_x)_*:\bM_n\to\bM_m$ is defined by $\Tr[(T_x)_*(y)\,e]=\Tr[y\,T_x(e)]$,
\begin{equation}
  \bra{\xi}(T_x \otimes \Phi_x)(\chi_m) \ket{\xi} = \Tr[\vartheta\circ R_x^a\circ\vartheta\circ\Phi_x],
\end{equation}
which is the $x$-th term of \eqref{eq:dual-pairing},
and
\begin{equation}\label{eq:cp-Tphi-pairing}
  \norm{(\sum_x(T_x\otimes\Phi_x)(\chi_m))_+} = \sup_{\norm{a}_2 \leq 1} \ip{R^a}{\Phi}
\end{equation}
for all $*$-preserving $T$ and $\Phi$.
The first identity is \eqref{eq:vector-pairing} with $b=a$.
Every $\xi\in\bC^n\otimes\bC^n$ is of the form $(a\otimes\1)\ket{\chi_n}$ with $\norm{\xi}=\norm{a}_2$.

In Step~4, $R_x(y)=S_x(a\,y\,a)=(T_x)_*(a\,y\,a^*)$ since $a=a^*$, so $R=R^a$ with $\norm{a}_2\le1$.
Now suppose that \eqref{eq:cpcriterion} holds for every ${*}$-preserving $T$.
By \eqref{eq:cp-Tphi-pairing}, \eqref{eq:cpcriterion} and \eqref{eq:cp-Tdplus}, $\ip{R}{\Phi}=\ip{R^a}{\Phi}\le\norm{(\sum_x(T_x\otimes\Phi_x)(\chi_m))_+}\le K\norm{T}_{(d)+}\le K$.
Taking the supremum over $R\in\hat\cF_d^\circ$, \Cref{lem:cp-polar-form} gives $\gamma_{d,\cp}(\Phi)\le K$.

For necessity, suppose that $\gamma_{d,\cp}(\Phi)\le K$.
The infimum in \eqref{eq:gamma-cp} is attained, so after rescaling there is a c.p.\ factorization $\Phi_x=\theta\circ\Psi_x$ through some $L_\infty(\Omega;\bM_d)$ with $\norm{\theta}_{\cb}\le K$ and $\Psi_x(\1)\le\1$.
For almost every $\omega$ the maps $y\mapsto\Psi_x(y)(\omega)$ are completely positive with $\Psi_x(\1)(\omega)\le\1$.
So \eqref{eq:dplus} gives $\sum_x(T_x\otimes\Psi_x)(\chi_m)\le\norm{T}_{(d)+}\1$ in $\bM_n\otimes L_\infty(\Omega;\bM_d)$ for every ${*}$-preserving $T$.
Applying the positive map $\id_{\bM_n}\otimes\theta$ gives
\begin{equation}
  \sum_x(T_x\otimes\Phi_x)(\chi_m)\le\norm{T}_{(d)+}\,\1_n\otimes\theta(\1)\le K\norm{T}_{(d)+}\,\1,
\end{equation}
which is \eqref{eq:cpcriterion}.
This completes the proof of \Cref{thm:cp-criterion}.

\subsection{Lemma for the proof of \Cref{thm:cp-criterion}}

\begin{lemma} \label{lem:HB-RMK}
  Let $R_x:\bM_n\to\bM_m$ be $*$-preserving.
  Then the following are equivalent.
  \begin{enumerate}[(i)]
    \item For every finite family $(z_j)$ of positive elements of $\bM_n \otimes \bM_d$,
        \begin{equation}
      \sum_{x,j} \norm{R_x(z_j)}_{S_1^d(\bM_m)_+} \leq \sup_{\rho \geq 0,\, \Tr\rho \leq 1} \sum_j \Tr[z_j(\rho \otimes \1_{\bM_d})] .
    \end{equation}
    \item There exists a single $\rho \geq 0$ with $\Tr \rho \leq 1$ such that, for every positive $z \in \bM_n \otimes \bM_d$,
        \begin{equation}
      \sum_{x} \norm{R_x(z)}_{S_1^d(\bM_m)_+} \leq \Tr[z(\rho \otimes \1_{\bM_d})] .
    \end{equation}
  \end{enumerate}
\end{lemma}
\begin{proof}
  $(ii) \Rightarrow (i)$ is trivial and we only show $(i) \Rightarrow (ii)$.
  First, we define $\varphi$ as
    \begin{equation}
    \varphi(z) := \sum_x \norm{R_x(z)}_{S_1^d(\bM_m)_+}
  \end{equation}
  and $f_z$ for $z \geq 0$ as
  \begin{equation}
    f_z(\rho) := \Tr[z(\rho \otimes \1)] - \varphi(z).
  \end{equation}
  Let $K = \{ \rho \in \bM_n \mid \rho \geq 0,\ \Tr \rho \leq 1 \}$, a convex and compact subset of $\bM_n$.
  The condition $(i)$ then becomes
  \begin{equation}\label{eq:HB-hyp}
    \sup_{\rho \in K} \sum_j f_{z_j}(\rho) \geq 0 \q \forall z_j \geq 0,
  \end{equation}
  and we want to show that there exists $\rho \in K$ such that $f_z(\rho) \geq 0$ for any positive $z \in \bM_n \otimes \bM_d$.

  Define a convex subset $S$ and a convex, open subset $U$ of $C(K)$, the real-valued continuous functions on $K$:
  \begin{equation}
    S = \conv \{f_z \in C(K) \mid z \geq 0 \}, \q\q U = \{g \in C(K) \mid \sup_\rho g(\rho) < 0 \}
  \end{equation}
  Here $U$ is open because $K$ is compact, so $\sup_\rho g(\rho)<0$ means that $g$ is uniformly negative, and $U\cap S=\varnothing$ is precisely \eqref{eq:HB-hyp}.
For this we use that $\varphi(tz)=t\varphi(z)$ for $t\ge0$ by \eqref{eq:S1plus}, so a convex combination $\sum_jt_jf_{z_j}$ equals $\sum_jf_{t_jz_j}$.
  Note also that $S$ is convex by construction.
  By \Cref{thm:HB-open}~\cite[Theorem~3.4(a)]{rudin1991functional}, there exists a continuous linear functional $\Lambda \neq 0$ and $\alpha \in \bR$ such that
  \begin{equation}
    \Lambda(u) \leq \alpha \leq \Lambda(s)
  \end{equation}
  for any $u \in U$ and $s \in S$.

  We claim that $\Lambda$ and $\alpha$ are both positive.
  To show $\Lambda$ is positive, we show that for any $u$ that is non-positive pointwise we have $\Lambda(u) \leq 0$.
  Indeed, if $u \leq 0$ pointwise, $-1 + \lambda u \in U$ for any $\lambda \geq 0$.
  Then, $\Lambda(-1) + \lambda \Lambda(u) \leq \alpha$ for any $\lambda \geq 0$, and hence $\Lambda(u) \leq 0$.
  Writing $1$ for the constant function on $K$, positivity and $\Lambda \neq 0$ give $\Lambda(1) > 0$, since $-\norm{g}\,1 \leq g \leq \norm{g}\,1$ forces $\Lambda=0$ otherwise.
  Moreover $-\epsilon\,1 \in U$ for any $\epsilon > 0$, so $-\epsilon \Lambda(1) \leq \alpha$ and hence $\alpha \geq 0$.
  This shows $\Lambda(f_z) \geq \alpha \geq 0$ for any $f_z \in S$, and rescaling by the positive constant $\Lambda(1)^{-1}$ we may assume $\Lambda(1) = 1$.

  By \Cref{thm:RMK}~\cite[Theorem~2.14]{rudin1987real}, applied to the positive functional $\Lambda$ on $C(K,\bR)$ with $\Lambda(1)=1$, there exists a unique probability measure $\mu$ on $K$ such that
  \begin{equation}
    \Lambda(f_z) = \int_K f_z(\sigma) \; d\mu(\sigma)
  \end{equation}
for any $z \geq 0$.
Set $\rho := \int_K \sigma \;d\mu(\sigma) \in K$.
Then $\bra{\xi}\rho\ket{\xi} = \int_K \bra{\xi}\sigma\ket{\xi}\,d\mu(\sigma) \geq 0$ and $\Tr\rho = \int_K \Tr\sigma \, d\mu(\sigma) \leq 1$.
Hence, $f_z(\rho) = \Lambda(f_z) \geq 0$ for every $z \geq 0$, which is $(ii)$.
\end{proof}

\section{Proofs for \texorpdfstring{\Cref{sec:eb-mub}}{the MUB section}}\label{app:mub}

The estimates below use the vector-valued Schatten classes and Pisier's results recalled in \Cref{app:pisier}, and Tomiyama's bound of \Cref{app:cb-facts}.

\begin{lemma}
\label{lem:fourtwo}
If $\norm{a}_4 \leq 1$ then $\norm{a^{*}a}_2 = \norm{aa^*}_2 \leq 1$.
\end{lemma}

\begin{proof}
$\norm{a^{*}a}_2^{2} = \Tr[(a^{*}a)^{2}] = \Tr[(aa^*)^{2}] = \Tr[|a|^{4}] = \norm{a}_4^{4} \leq 1$.
\end{proof}

\begin{lemma}
\label{lem:vstarv}
Let $X$ be a Banach space, $H$ a Hilbert space, and $v:X\to H$ linear and bounded, with adjoint $v^{*}:H\to \bar X^{*}$.
Then
\begin{equation}
\norm{v^{*}v} \;=\; \norm{v}^{2} .
\end{equation}
\end{lemma}
\begin{proof}
Unwinding the definition of the adjoint,
\begin{equation}
  \norm{v^*v} = \sup_{\norm{z},\norm{y} \leq 1} \abs{\ip{v^*vz}{y}} = \sup_{\norm{z},\norm{y} \leq 1} \abs{\ip{vz}{vy}} \geq \sup_{\norm{z}\leq 1} \abs{\ip{vz}{vz}} = \norm{v}^2.
\end{equation}
By Cauchy-Schwarz,
\begin{equation}
  \sup_{\norm{z},\norm{y} \leq 1} \ip{vz}{vy} 
  \leq \sup_{\norm{z},\norm{y} \leq 1} \ip{vz}{vz}^{1/2} \ip{vy}{vy}^{1/2}
  = (\sup_{\norm{z}\leq 1} \norm{vz})^2 = \norm{v}^2,
\end{equation}
giving $\norm{v^*v} = \norm{v}^2$.
\end{proof}

\begin{lemma}
Let $E$ be a operator space with its predual operator space $F$, i.e. $F^* = E$.
Then for every $u \in \bM_d(E)$ and every $1\leq p < \infty$,
\begin{equation}
  \norm{u}_{\bM_d(E)} = \sup \{ \norm{(a\otimes \1)u(b \otimes \1)}_{S_p^d(E)} \mid \norm{a}_{S_{2p}^d}, \norm{b}_{S_{2p}^d} \leq 1 \}.
\end{equation}
\end{lemma}
This is \cite[Lemma~1.7]{Pisier1998}, which holds for every operator space $E$. We include the short proof for the case used here.
\begin{proof}
``$\geq$'' inequality is a direct consequence of \Cref{thm:pisier-1.5}~\cite[Theorem~1.5]{Pisier1998}:
\begin{equation}
  \norm{(a\otimes \1)u(b \otimes \1)}_{S_p^d(E)} \leq \norm{a}_{2p} \norm{u}_{\bM_d(E)} \norm{b}_{2p} \leq \norm{u}_{\bM_d(E)}
\end{equation}
for any $a, b$ such that $\norm{a}_{2p}, \norm{b}_{2p} \leq 1$.

Consider $p = 1$.
Note that $(S_1^d(F))^* = \bM_d(E)$ completely isometrically, meaning that
\begin{equation}
  \norm{u}_{\bM_d(E)} = \sup_{\norm{w}_{S_1^d(F)}\leq 1} \abs{\ip{w}{u}}.
\end{equation}
In other words, for any $\epsilon > 0$, there exists $w$ in the unit ball of $S_1^d(F)$ satisfying
\begin{equation}
  \norm{u}_{\bM_d(E)} - \epsilon \leq \abs{\ip{w}{u}}.
\end{equation}
By \Cref{thm:pisier-1.5}~\cite[Theorem~1.5]{Pisier1998}, there exists a factorization $w = (a \otimes \1) z(b\otimes \1)$ with $\norm{a}_{2}, \norm{b}_{2} \leq 1$ and $\norm{z}_{\bM_d(F)} \leq 1 + \epsilon$, since $\norm{w}_{S_1^d(F)}\le1$ and the norm in \cite[Theorem~1.5]{Pisier1998} is an infimum over such factorizations.
Since $\ip{u}{(a\otimes \1)z(b\otimes \1)} = \ip{(a^T \otimes \1)u(b^T \otimes \1)}{z}$, we get the chain of inequality
\begin{equation}
  \norm{u}_{\bM_d(E)} - \epsilon \leq \ip{a^Tu b^T}{z} \leq \norm{a^T u b^T }_{S_1^d(E)}\norm{z}_{\bM_d(F)} \leq \norm{a^Tub^T}_{S_1^d(E)}(1+\epsilon).
\end{equation}
We get the ``$\leq$'' inequality by letting $\epsilon \to 0$.

For general $p$, we use polar decomposition with H\"{o}lder inequality.
Given $\norm{a}_2 \leq 1$ with polar decomposition $a = c|a|$, split
\begin{equation}
  a = a_1a_2, \q a_1 = c|a|^{1/p'}, \q a_2 = |a|^{1/p},
\end{equation}
where $p$ and $p'$ are H\"{o}lder conjugate, and likewise $b = b_2b_1$.

Let $w = a_2ub_2$ and $\epsilon > 0$ be arbitrary.
By \Cref{thm:pisier-1.5}~\cite[Theorem~1.5]{Pisier1998} there exist $\alpha,\beta \in S_{2p}^d$, and $v \in \bM_d(E)$ with $w = \alpha v \beta$ such that
\begin{equation}
  \norm{\alpha}_{2p}\norm{v}_{\bM_d(E)}\norm{\beta} \leq (1+\epsilon) \norm{w}_{S_p^d(E)}.
\end{equation}
Then applying \Cref{thm:pisier-1.5}~\cite[Theorem~1.5]{Pisier1998} again on $(a_1\alpha)v(\beta b_1)$ and by H\"{o}lder inequality $\norm{zy}_{2} \leq \norm{z}_{2p'}\norm{y}_{2p}$, we have
\begin{align}
  \norm{a_1wb_1}_{S_1^d(E)}
  & \leq \norm{a_1\alpha}_2 \norm{v}_{\bM_d(E)} \norm{\beta b_1}_2\\
  & \leq \norm{a_1}_{2p'}\norm{\alpha}_{2p} \norm{v}_{\bM_d(E)} \norm{\beta}_{2p} \norm{b_1}_{2p'}\\
  & \leq (1+\epsilon) \norm{a_1}_{2p'} \norm{a_2 ub_2}_{S_p^d(E)} \norm{b_1}_{2p'}.
\end{align}
Taking $\epsilon \to 0$ yields
\begin{equation}
  \norm{aub}_{S_1^d(E)} \leq \norm{a_1}_{2p'} \norm{a_2ub_2}_{S_p^d(E)} \norm{b_1}_{2p'}.
\end{equation}

Functional calculus gives
\begin{equation}
  \norm{|z|^s}_q = (\Tr |z|^{sq})^{1/q} = \norm{z}_{sq}^s,
\end{equation}
and direct calculations show that $\norm{a_1}_{2p'}, \norm{b_1}_{2p'} \leq 1$ and $\norm{a_2}_{2p}, \norm{b_2}_{2p} \leq 1$.

Combining everything gives
\begin{align}
  \norm{aub}_{S_1^d(E)}
  \leq \norm{a_1}_{2p'} \norm{a_2ub_2}_{S_p^d(E)}\norm{b_1}_{2p'}
  \leq \norm{a_2ub_2}_{S_p^d(E)}
  \leq \sup_{\norm{\alpha}_{2p}, \norm{\beta}_{2p} \leq 1} \norm{\alpha u\beta}_{S_p^d(E)}.
\end{align}
Taking sup over $a,b$, we get
\begin{equation}
  \norm{u}_{\bM_d(E)} = \sup_{\norm{a}_2,\norm{b}_2
  \leq 1} \norm{aub}_{S_1^d(E)}
  \leq \sup_{\norm{a}_{2p},\norm{b}_{2p} \leq 1} \norm{aub}_{S_p^d(E)}
\end{equation}
\end{proof}
 
\begin{corollary}\label{cor:Dp}
Let $X$ be an operator space, let $Y$ be an operator space admitting an operator-space predual $E$, i.e.\ $Y=E^{*}$ completely isometrically, and let $1\le p<\infty$ with conjugate exponent $p'$.
For any bounded $T:X\to Y$,
\begin{equation}
\norm{T}_{(d)} \;=\; \sup_{\norm{a}_{2p},\,\norm{b}_{2p}\le1}\norm{ M_{a,b}\otimes T:\bM_d(X)\to S_p^d(Y)} .
\label{eq:Dp}
\end{equation}

\end{corollary}

The two cases used in this paper are:
\begin{enumerate}[(i)]
\item $p=2$, for a map $T:X \to S_2$: \label{it:D1}
\begin{equation}
\norm{T}_{(d)} \;=\; \sup_{\norm{a}_4,\norm{b}_4\le1}\norm{ M_{a,b}\otimes T:\bM_d(X)\to S_2^d(S_2)} .
\label{eq:D1}
\end{equation}
\item $p=1$, for a map $u:X\to \cB(H)$: \label{it:D2}
\begin{equation}
\norm{u}_{(d)} \;=\; \sup_{\norm{a}_2,\norm{b}_2\le1}\norm{ M_{a,b}\otimes u:\bM_d(X)\to S_1^d(\cB(H))} .
\label{eq:D2}
\end{equation}
\end{enumerate}

\begin{proposition}
\label{prop:sqrtd}
$\norm{\id:S_1^n\to S_2^n}_{(d)}\le\sqrt d$.
\end{proposition}

\begin{proof}
Fix $a, b$ such that $\norm{a}_4, \norm{b}_4 \leq 1$ and define
\begin{equation}
  v = M_{a,b} \otimes \id_{S_1^n \to S_2^n}: \bM_d(S_1^n) \to S_2^d(S_2^n) =: H.
\end{equation}
Its adjoint
\begin{equation}
v^*=M_{a^*,b^*} \otimes \id_{S_2^n \to \bM_n}:H^* \to \bM_d(S_1^n)^*
\end{equation}
is given by the cyclicity of trace:
\begin{equation}
  \ip{v^*(w)}{z} = \ip{w}{v(z)}_{\HS}
  = \Tr[w^*(a \otimes \1)z(b\otimes \1)]
  = \Tr[(b\otimes \1)w^*(a\otimes \1)z]
  = \ip{(a^* \otimes \1)w(b^* \otimes\1)}{z}.
\end{equation}
Here we are implicitly identifying $H \iso H^*$ by Riesz representation.
The composition $v^*v$ is then
\begin{equation}
  v^*v = M_{a^*a,bb^*} \otimes \id_{S_1^n \to \bM_n}.
\end{equation}
Here $v^*v$ takes values in $\bM_d(S_1^n)^*$, and $\bM_d(S_1^n)^*=S_1^d(\bM_n)$ isometrically.
This follows from $S_1^d(\bM_n)^*=\bM_d(S_1^n)$ \cite[Thm.~4.1 and p.~142]{Pisier_2003}, since the spaces are finite-dimensional.
The pairing above is sesquilinear.
Passing to the bilinear duality conjugates entrywise, which is isometric, and depending on the pairing it may replace the identity map below by the transpose $\vartheta:S_1^n\to\bM_n$.
Since $\norm{\vartheta:S_1^n\to\bM_n}=1$ as well, the estimates below are unchanged.
At the outer level the duality is the parallel pairing $\sum_{i,j}\ip{w_{ij}}{z_{ij}}$, which is the one dual to the factorization norm of \cite[Theorem~1.5]{Pisier1998}, so no transpose occurs there. A transpose at the outer level would not be isometric.
By \Cref{lem:fourtwo}, $\norm{a^*a}_2, \norm{bb^*}_2 \leq 1$ and
\begin{equation}
  \norm{v^*v} = \norm{M_{a^*a,bb^*} \otimes \id_{S_1^n \to \bM_n}} \leq \sup_{\norm{a}_2, \norm{b}_2\leq 1} \norm{M_{a,b} \otimes \id_{S_1^n \to \bM_n}} = \norm{\id_{S_1^n \to \bM_n}}_{(d)},
\end{equation}

and by \Cref{lem:dbound}
\begin{equation}
  \norm{v^*v} \leq \norm{\id:{S_1^n \to \bM_n}}_{(d)} \leq d \norm{\id:{S_1^n \to \bM_n}} = d.
\end{equation}
Since $\norm{v^*v} = \norm{v}^2$ by \Cref{lem:vstarv}, we have
\begin{equation}
  \norm{v} = \norm{M_{a,b} \otimes \id} \leq \sqrt{d}
\end{equation}
for any $a,b$ with $\norm{a}_4, \norm{b}_4 \leq 1$, and by \eqref{eq:D1} taking the supremum over all such $a,b$ gives
\begin{equation}
  \norm{\id:{S_1^n \to S_2^n}}_{(d)} \leq \sqrt{d}.
\end{equation}
\end{proof}

We first convert the estimate into a bound on the $d$-th amplification norm of the witness.

Recall the witness $T :S_1^n \to \ell_1^k(\bM_n)$ is a tuple of Hilbert-Schmidt orthogonal projection $T_x$ onto the trace-zero subspace $P_n^{(x)}$
\begin{equation}
  T(z) = (T_1(z), \cdots, T_k(z)).
\end{equation}

Throughout, $S_2^n$ and its subspaces carry the self-dual operator-space structure of Pisier's operator Hilbert space $OH$ \cite[Thm.~1.1]{Pisier96}, recalled in \Cref{app:OH}.
In the proof, the $OH$ spaces are $S_2^n$ by \ref{it:OH-S2}, which is also the structure used in \Cref{prop:sqrtd}, the subspaces $P_n^{(x)}$ by \ref{it:OH-sub}, and the Hilbert direct sum $H:=\ell_2^k(P_n^{(x)})=\bigoplus_xP_n^{(x)}$.
We give $H$ the $OH$ structure of its Hilbert-space norm; since the $P_n^{(x)}$ are mutually orthogonal, $H$ is then completely isometric to the subspace $\bigoplus_xP_n^{(x)}$ of $S_2^n$.
The spaces $S_1^n$, $\bM_n$, $\ell_1^k(P_n^{(x)})$ and $\ell_1^k(\bM_n)$ are not $OH$.

\begin{proposition}\label{prop:Td-norm}
$\norm{T}_{(d)}\le\sqrt k\,\sqrt d$.
\end{proposition}

\begin{proof}
Factor $T$ as the following
\begin{equation}
T:\; S_1^n \xrightarrow{\ \iota_1\ } S_2^n \xrightarrow{\ T'\ } H:=\ell_2^k(P_n^{(x)}) \xrightarrow{\ \iota_2\ } \ell_1^k(P_n^{(x)}) \xrightarrow{\ \jmath = (\jmath_x) \ } \ell_1^k(\bM_n),
\end{equation}
where $\iota_1,\iota_2$ is the formal identity, $T':S_2^n \to H:=\bigoplus_{x=1}^k P_n^{(x)}$ (a Hilbert-Schmidt direct sum) is $T'\xi = (T_1\xi, \dots, T_k\xi)$, and $\jmath_x:P_n^{(x)} \to \bM_n$ is the diagonal embedding on each component.

Then
\begin{equation}
  \norm{T}_{(d)} = \norm{\id_{\bM_d} \otimes T}
  \leq \norm{\jmath}_{\cb}\norm{\iota_2}_{\cb} \norm{T'}_{\cb} \norm{\iota_1}_{(d)} .
\end{equation}

By \Cref{prop:sqrtd}, $\norm{\iota_1:S_1^n \to S_2^n}_{(d)}\le\sqrt d$.

$\underline{\norm{T':S_2^n \to \ell_2^k(P_n^{(x)})}_{\cb} \leq 1}$:
The subspaces $P_n^{(x)}$ are mutually orthogonal in $S_2^n$, so $\sum_x\norm{T_x\xi}_2^2\le\norm{\xi}_2^2$ and $T'$ is a Hilbert-space contraction.
Since every bounded map between $OH$ spaces satisfies $\norm{v}_{\cb}=\norm{v}$ \cite[Prop.~7.2(iii)]{Pisier_2003}, and subspaces of $OH$ are again $OH$ \cite[Prop.~1.5(i)]{Pisier96}, we have $\norm{T'}_{\cb}=\norm{T'} \le 1$.

$\underline{\norm{\jmath = (\jmath_x:P_n^{(x)} \to \bM_n)}_{\cb} \leq 1}$:
Each $\jmath_x$ has range in the commutative $C^*$-algebra $A_n^{(x)} \cong \ell_\infty^n$, which is a minimal operator space.
Then $\norm{\jmath_x}_{\cb} = \norm{\jmath_x}$ \cite[Prop.~1.10(ii)]{Pisier_2003}, and because the operator norm is dominated by the Hilbert-Schmidt norm on diagonal matrices, we get $\norm{\jmath_x} \leq 1$.  
The statement follows since $\ell_1$-sum of complete contractions is a complete contraction \cite[\S 2.6]{Pisier_2003}.

$\underline{\norm{\iota_2:\ell_2^k(P_n^{(x)}) \to \ell_1^k(P_n^{(x)})}_{\cb} \leq \sqrt{k}}$: 
We bound the adjoint $\iota_2^*$ at every matrix level and then use $\norm{\iota_2}_{\cb}=\norm{\iota_2^*}_{\cb}$ \cite[Prop.~3.2.2]{Effros2000}.
Since $\ell_1^k(P_n^{(x)})^*=\ell_\infty^k((P_n^{(x)})^*)$ \cite[\S 2.6]{Pisier_2003} and $(P_n^{(x)})^*=\overline{P_n^{(x)}}$, $H^*=\overline H$ by \ref{it:OH-dual} and \ref{it:OH-sub} of \Cref{app:OH}, the adjoint $\iota_2^*:\ell_\infty^k(\overline{P_n^{(x)}})\to\overline{H}$ is the formal identity $(w_x)_x\mapsto\sum_xw_x$.
Fix orthonormal bases $(e_{x,i})_i$ of the $P_n^{(x)}$; together they form an orthonormal basis of $H$.
Let $r\in\bN$ and $w=(w_x)_x\in\bM_r(\ell_\infty^k(\overline{P_n^{(x)}}))$ with $w_x=\sum_iw_{x,i}\otimes\overline{e_{x,i}}$ and $w_{x,i}\in\bM_r$.
Since $\bM_r(\ell_\infty^k(E))=\ell_\infty^k(\bM_r(E))$ \cite[\S 2.6]{Pisier_2003}, $\norm{w}=\max_x\norm{w_x}_{\bM_r(\overline{P_n^{(x)}})}$.
By \ref{it:OH-dual} of \Cref{app:OH}, \eqref{eq:OH-norm} with $\overline{e_{x,i}}$ in place of $e_i$ computes the norms in $\bM_r(\overline{P_n^{(x)}})$ and in $\bM_r(\overline H)$, so
\begin{equation}
  \norm{(\id_{\bM_r}\otimes\iota_2^*)(w)}^2=\norm{\sum_x\sum_iw_{x,i}\otimes\overline{w_{x,i}}}\le\sum_x\norm{\sum_iw_{x,i}\otimes\overline{w_{x,i}}}=\sum_x\norm{w_x}^2\le k\norm{w}^2.
\end{equation}
Hence $\norm{\id_{\bM_r}\otimes\iota_2^*}\le\sqrt k$ for every $r$, that is, $\norm{\iota_2^*}_{\cb}\le\sqrt k$, and therefore $\norm{\iota_2}_{\cb}=\norm{\iota_2^*}_{\cb}\le\sqrt k$.

\end{proof}

\subsection{Proof of \texorpdfstring{\Cref{prop:mult-domain}}{Proposition~\ref{prop:mult-domain}}}\label{app:mult-domain-proof}

\begin{proposition*}[restated]
  \MultDomainStatement
\end{proposition*}
\begin{proof}
Suppose there exists such a factorization $\Phi_x = \theta \circ \Psi_x$ through $\cM$ for $d \lneqq n$.
Given a completely positive map $\theta:\cM \to \bM_n$, its \emph{multiplicative domain} is
\begin{equation}
  m_\theta := \{ z \in \cM \mid \theta(z)\theta(w) = \theta(zw) \text{ and } \theta(w)\theta(z) = \theta(wz) \; \forall w \in \cM \}.
\end{equation}
By Choi's theorem \cite{Choi74}, \cite[Thm.~3.18]{Paulsen_2003} for u.c.p. maps such as $\theta$,
\begin{equation}
  m_\theta = \{ z \in \cM \mid \theta(z)^*\theta(z) = \theta(z^*z) \text{ and } \theta(z)\theta(z)^* = \theta(zz^*) \}.
\end{equation}
Then, $m_\theta$ is a $C^*$-subalgebra of $\cM$ and $\theta|_{m_\theta}$ is a $*$-homomorphism.

Since $\Phi_1$ is a pinching map, it acts as an identity on $A_n$.
Let $z \in A_n$, and  $y=\Psi_1(z)$.
Then by Kadison-Schwarz inequality for unital 2-positive map \cite[Prop.~3.3]{Paulsen_2003}, and in particular for u.c.p. map 
\begin{equation}
  \theta(y)^* \theta(y) \leq \theta(y^* y) = \theta(\Psi_1(z)^*\Psi_1(z))\leq \theta(\Psi_1(z^* z)) = z^* z = \theta(y)^* \theta(y),
\end{equation}
showing that $y \in m_\theta$, after applying the same chain also to $z^*\in A_n$, for which $\Psi_1(z^*)=y^*$.
The same holds for any $y'$ in the range $\Psi_2(B_n)$, i.e. $\Psi_2(B_n) \subseteq m_{\theta}$.
Let $G := \theta(m_\theta)$.
Since $\theta|_{m_\theta}$ is a $*$-homomorphism, $G$ is a $*$-subalgebra of $\bM_n$ .

Since $\theta \Psi_1 = \Phi_1$ is identity on $A_n$, we have $A_n = \theta(\Psi_1(A_n)) \subseteq G$ and likewise $B_n \subseteq G$.
Because $m_\theta$ is an algebra containing $\Psi_1(A_n)$ and $\Psi_2(B_n)$ and $\theta|_{m_\theta}$ is multiplicative, $G$ contains all products of elements of $A_n$ and $B_n$, which leads to $G = \bM_n$.

Since $L_\infty(\Omega)$ is a commutative unital $C^*$-algebra, $\cM=\bM_d(L_\infty(\Omega))\cong C(K;\bM_d)$ for a compact Hausdorff space $K$ by the Gelfand representation \cite[p.~13]{Pisier_2003}, so $\cM$ is subhomogeneous of degree $d$, that is, its irreducible representations have dimension at most $d$.
Since $\theta|_{m_\theta}$ is a $*$-homomorphism of $m_\theta$ onto $G=\bM_n$, it induces a $*$-isomorphism $\bM_n\cong m_\theta/I$ for the closed ideal $I=\ker(\theta|_{m_\theta})$ \cite[Cor.~I.8.2]{Takesaki_1979}, where $m_\theta$ is a $C^*$-subalgebra of $\cM\cong C(K;\bM_d)$.
For a $C^*$-algebra $\cN$ let $\id_{\cN}\otimes\vartheta_k$ be the map $[b_{ij}]\mapsto[b_{ji}]$ on $\bM_k(\cN)$, where $\vartheta_k$ is the transpose on $\bM_k$.
For $\cN=C(K;\bM_d)$ this map is the transpose of each entry $b_{ij}\in C(K;\bM_d)$ followed by the pointwise transpose of $\bM_k(\bM_d)=\bM_{kd}$, which is isometric, so its norm is at most $d$ \cite[Exercise~3.8]{Paulsen_2003}.
The bound passes to the $C^*$-subalgebra $m_\theta$ by restriction, then to $m_\theta/I$ because $\bM_k(m_\theta/I)=\bM_k(m_\theta)/\bM_k(I)$ \cite[Cor.~I.8.2]{Takesaki_1979} and the map preserves $\bM_k(I)$, and finally to $\bM_n$ because $*$-isomorphisms of $C^*$-algebras are completely isometric \cite[p.~92]{Paulsen_2003}.
Hence
\begin{equation}
  \sup_k \norm{\id_{\bM_n}\otimes\vartheta_k} \leq d.
\end{equation}
But by the same factorization with $\bM_n$ in place of $\bM_d$, this supremum is the c.b.\ norm of the transpose on $\bM_n$, which is $n$ \cite[Exercise~3.8]{Paulsen_2003}, a contradiction.

\end{proof}

\subsection{Proof of \texorpdfstring{\Cref{thm:dtwo}}{Theorem~\ref{thm:dtwo}}}\label{app:dtwo-proof}

\begin{theorem*}[restated]
  \DTwoStatement{}
\end{theorem*}
\begin{proof}
Since the right-hand side of \eqref{eq:cb-criterion} is a supremum over $\norm{T}_{(d)}\le1$ and both sides scale linearly in $T$, applying \Cref{thm:cb-criterion} to $T/\norm{T}_{(d)}$ gives, for every nonzero witness,
\begin{equation}\label{eq:gamma-ratio}
\gamma_{d,\cb}(\Phi)\;\ge\;\frac{\norm{T\otimes\Phi}}{\norm{T}_{(d)}} .
\end{equation}
The numerator does not depend on $d$.
The tuple $(z_1,\dots,z_k)$ is an admissible test element for $T\otimes\Phi=\sum_x T_x\otimes\Phi_x$, and \eqref{eq:mbound} yields
\begin{equation}
\norm{T\otimes\Phi}\;\ge\;\norm{\sum_x(\Phi_x\otimes T_x)(z_x)}\;\ge\;k(1-\frac{1}{n}) .
\end{equation}
By \Cref{prop:Td-norm}, $\norm{T}_{(d)}\le\sqrt k\,\sqrt d$.
Substituting both estimates into \eqref{eq:gamma-ratio} gives \eqref{eq:gammad}.
If $d<k(1-\frac{1}{n})^2$, then \eqref{eq:gammad} gives $\gamma_{d,\cb}(\Phi)>1$, and \Cref{prop:cb-unit-ball} excludes $d$-compatibility.
At $d=1$ the condition reads $\sqrt k(1-\frac{1}{n})>1$, that is $k>(\frac{n}{n-1})^2$.
For $k=n+1$,
\begin{equation}
(n+1)(1-\frac{1}{n})^2=\frac{(n+1)(n-1)^2}{n^2}=n-1-\frac{1}{n}+\frac{1}{n^2}\;>\;n-2,
\end{equation}
so every integer $d\le n-2$ satisfies the hypothesis.
\end{proof}

\subsection{Proof of \texorpdfstring{\Cref{thm:finite-gp}}{Theorem~\ref{thm:finite-gp}}}\label{app:finite-gp-proof}

\begin{theorem*}[restated]
  \FiniteGPStatement{}
\end{theorem*}
\begin{proof}
Items (1) and (2) are shown before the theorem in \Cref{sec:finite-gp}, so it remains to show \cref{it:finite-3}.
Suppose $\Phi_A = \theta \Psi_A$ and $\Phi_B = \theta \Psi_B$ with u.c.p. maps through $\cM=L_\infty(\Omega;\bM_d)$.
The multiplicative domain argument in \Cref{prop:mult-domain} applies almost verbatim since the conditional expectation acts as an identity on the subalgebras.
Following the same argument, we can show that  $\Psi_A(A),\Psi_B(B)\subseteq m_\theta\subseteq\cM$, and $m_\theta$ is a $C^*$-subalgebra of $\cM$.
Moreover, $\theta|_{m_\theta}$ is a $*$-homomorphism and $A,B \subset \theta(m_\theta)$, meaning $\bM_N \subset \theta(m_\theta)$.
The end of the proof of \Cref{prop:mult-domain}, with $\bM_N=\theta(m_\theta)$ in place of $\bM_n$, then gives $d\geq N$.
\end{proof}

\subsection{Lemma for \texorpdfstring{\Cref{sec:finite-gp}}{the finite-group section}}\label{app:gp-orth}

We use the notation of \Cref{sec:finite-gp}.

\begin{lemma}\label{lem:gp-orth}
For all $f: G \to \bC$ and $b \in B$, $\tau$ satisfy
\begin{equation}
  \tau(m_f b) = \tau(m_f) \tau(b).
\end{equation}
\end{lemma}
\begin{proof}
  By linearity it suffices to take $b = \lambda_g$.
  The diagonal entries of $m_f\lambda_g$ are
  \begin{equation}
    \bra{h}m_f\lambda_g\ket{h} = \bra{h}m_f\ket{gh} = \braket{h}{gh}\cdot f(gh) = f(h) \delta_{gh,h},
  \end{equation}
  where $\delta$ is the indicator function.
  Since $gh = h$ for all $h$ if and only if $g = e$,
  \begin{equation}
    \Tr[m_f \lambda_g] = \delta_{g,e}\sum_h f(h) = \delta_{g,e}\Tr[m_f],
  \end{equation}
  and $\Tr[\lambda_g] = N \delta_{g,e}$.
  Combining altogether we get
  \begin{equation}
    \tau[m_f \lambda_g] = \delta_{g,e} \tau(m_f) = \tau(m_f) \tau(\lambda_g).
  \end{equation}
\end{proof}

\section{Proofs for \texorpdfstring{\Cref{sec:kPEB}}{the Choi-state section}}\label{app:kPEB}

\subsection{Proof of \texorpdfstring{\Cref{thm:choi-assemblage}}{Theorem~\ref{thm:choi-assemblage}}}\label{app:choi-assemblage-proof}

\begin{theorem*}[restated]
\ChoiAssemblageStatement{}
\end{theorem*}

We prove the equivalence of the conditions (1)--(4) of \Cref{thm:choi-assemblage}.

\begin{proof}
$(1)\Leftrightarrow(2)$\q
$(1)\Rightarrow(2)$ follows from Carath\'eodory's theorem, as in \Cref{rem:finite-suffices}.
Each $\Phi_x$ is an integral of compositions $\Psi_{x,\omega}\circ\theta_\omega$ of completely positive maps $\theta_\omega:S_1(H_A)\to S_1^d$ with instruments $\Psi_{x,\omega}:S_1^d\to\ell_1^{\mathsf A}(S_1(H_B))$.
After normalizing $\Tr\theta_\omega(\1_A)=1$, the tuples $(\Psi_{x,\omega}\circ\theta_\omega)_x$ lie in a compact subset of a real vector space of dimension $N=|\mathsf X||\mathsf A|(\dim H_A)^2(\dim H_B)^2$, so $(\Phi_x)_x$ is a nonnegative combination of at most $N+1$ of them, which is (2).
Conversely, $(2)\Rightarrow(1)$ holds because $\ell_1^m(S_1^d)=S_1^d\otimes^\wedge L_1(\Omega)$ for $\Omega=\{1,\dots,m\}$ with counting measure. 

$(2)\Rightarrow(4)$\q
Condition (2) gives $\Phi_x=\sum_{j=1}^m\Psi_{x,j}\circ\theta_j$ with an instrument $(\theta_j)_{j=1}^m:S_1(H_A)\to\ell_1^m(S_1^d)$ and instruments $\Psi_{x,j}:S_1^d\to\ell_1^{\mathsf A}(S_1(H_B))$.
Put $K = \bC^d \otimes \bC^m$, and define $\sigma \in S_1(H_A) \otimes S_1(K)$ and $\cE_x:S_1(K) \to \ell_1^{\mathsf{A}}(S_1(H_B))$ by
\begin{align}
  \sigma := \sum_{j=1}^m J(\theta_j)\otimes\pure{j}, \q\q
  \cE_{a|x}(\rho) := \sum_{j=1}^m \Psi_{a|x,j}((\1\otimes\bra{j})\rho(\1\otimes\ket{j})).
\end{align}

Each $J(\theta_j)\otimes\pure{j}$ is a positive operator on $H_A\otimes(\bC^d\otimes\bC\ket{j})$ and therefore a sum of pure states of Schmidt rank at most $d$ with respect to $(H_A: K)$, so $\SN(\sigma)\le d$.
The maps $\cE_{a|x}$ are the components of an instrument $\cE_x$, and $(\id\otimes\cE_{a|x})(\sigma)=\sum_j(\id\otimes \Psi_{a|x,j})(J(\theta_j))=J(\Phi_{a|x})$.

$(4)\Rightarrow(3)$\q
Replacing $K$, $\sigma$ and $\cE_{a|x}$ by $K\otimes\bC^d$, $\sigma\otimes\pure{0}$ and $\cE_{a|x}\circ\Tr_{\bC^d}$, we may assume $\dim K\ge d$.
Write $\sigma=\sum_{j=1}^r\pure{\psi_j}$ with $\SR(\psi_j)\le d$, and choose isometries $V_j:\bC^d\to K$ and vectors $\hat\psi_j\in H_A\otimes\bC^d$ with $\ket{\psi_j}=(\1\otimes V_j)\ket{\hat\psi_j}$.
Set $\sigma_j:=\pure{\hat\psi_j}$ and $\Psi_{a|x,j}(\eta):=\cE_{a|x}(V_j\eta V_j^\dagger)$, the components of an instrument $\Psi_{x,j}:S_1^d\to\ell_1^{\mathsf A}(S_1(H_B))$.
Taking the partial trace over $H_B$ in \eqref{eq:d-preparable} with $\sigma_{a|x}=J(\Phi_{a|x})$ and summing over $a$ gives $\Tr_K\sigma=\sum_{a\in\mathsf A}\Tr_BJ(\Phi_{a|x})=\1_A$, because $\cE_x$ and $\Phi_x$ are trace preserving.
Hence $\sum_j\Tr_{\bC^d}\sigma_j=\Tr_K\sigma=\1_A$, so the $\sigma_j$ are the Choi operators of the components $\theta_j$ of an instrument $S_1(H_A)\to\ell_1^r(S_1^d)$.
We have
\begin{equation}
  J(\Phi_{a|x})=\sum_j(\id\otimes\cE_{a|x})((\1\otimes V_j)\sigma_j(\1\otimes V_j)^\dagger)=\sum_j(\id\otimes \Psi_{a|x,j})(\sigma_j)=J(\sum_j\Psi_{a|x,j}\circ\theta_j),
\end{equation}
so $\Phi_{a|x}=\sum_j\Psi_{a|x,j}\circ\theta_j$ for all $a$ and $x$, that is, $\Phi_x=\sum_j\Psi_{x,j}\circ\theta_j$.
Since $\sigma_j$ has rank one, $\theta_j(\rho)=A_j\rho A_j^\dagger$ for a single operator $A_j:H_A\to\bC^d$, and $\sum_jA_j^\dagger A_j=\1_A$ because $\sum_j\theta_j$ is trace preserving.
Choosing Kraus operators $B_{a|x,j,r}$, $r\in\mathsf R$, of the components $\Psi_{a|x,j}$ of the instruments $\Psi_{x,j}$, which satisfy $\sum_{a,r}B^\dagger_{a|x,j,r}B_{a|x,j,r}=\1_d$ because $\Psi_{x,j}$ is trace preserving, we obtain \eqref{eq:common-kraus}.

$(3)\Rightarrow(2)$\q
Since $\sum_jA_j^\dagger A_j=\1_A$, the maps $\theta_j(\rho):=A_j\rho A_j^\dagger$ form an instrument $S_1(H_A)\to\ell_1^{\mathsf J}(S_1^d)$, and since $\sum_{a,r}B^\dagger_{a|x,j,r}B_{a|x,j,r}=\1_d$, the maps $\Psi_{a|x,j}(\eta):=\sum_rB_{a|x,j,r}\eta B^\dagger_{a|x,j,r}$ form instruments $\Psi_{x,j}:S_1^d\to\ell_1^{\mathsf A}(S_1(H_B))$.
Now \eqref{eq:common-kraus} reads $\Phi_x=\sum_j\Psi_{x,j}\circ\theta_j$.
\end{proof}

\subsection{Proof of \texorpdfstring{\Cref{cor:k-PEB}}{Corollary~\ref{cor:k-PEB}}}\label{app:k-PEB-proof}

\begin{corollary*}[restated]
\KPEBStatement
\end{corollary*}
We prove the equivalence of the conditions (i)--(iii) of \Cref{cor:k-PEB}.

\begin{proof}
Apply \Cref{thm:choi-assemblage} to the singleton family $\{\Phi\}$, for which (iii) is condition (2).
If (i) holds, then condition (4) holds with $K:=\ell_2^{\mathsf A}\otimes H_B$, $\sigma:=J(\Phi)$ and the instrument $\cE$ that reads the register, $\cE_a(\eta):=(\bra{a}\otimes\1_B)\eta(\ket{a}\otimes\1_B)$, since $(\id\otimes\cE_a)(J(\Phi))$ is the diagonal block $J(\Phi_a)$ of the block diagonal operator $J(\Phi)$.
Hence (i) implies (iii).
If (iii) holds, then \eqref{eq:common-kraus} holds, and its operators $K_{a,(j,r)}:=B_{a|x,j,r}A_j$ factor through $\bC^d$ and have rank at most $d$, which is (ii).
Finally, (ii) implies (i), because $J(\Phi)=\sum_{a,j}\pure{a}\otimes\kket{K_{a,j}}\bbra{K_{a,j}}$ with $\kket{K}:=\sum_l\ket{l}\otimes K\ket{l}=(\1\otimes K)\ket{\chi_A}$, and $\ket{a}\otimes\kket{K}$ has Schmidt rank $\rank(K)$ with respect to $H_A:\ell_2^{\mathsf A}\otimes H_B$.
\end{proof}

\end{document}